\documentclass[11pt,reqno]{amsart}
\usepackage{hyperref}
\usepackage{amsmath}
\usepackage{amsthm}
\usepackage{amsfonts}
\usepackage{amssymb}

\usepackage{latexsym,amsmath}
\usepackage{geometry}
\usepackage{fixmath}

\usepackage[dvips]{graphicx}
\usepackage{graphicx}
\usepackage{fullpage}
\usepackage{appendix}

\usepackage{color} 
\usepackage{times}

\usepackage{leftidx}
\usepackage{scrextend}
\usepackage{float}
\usepackage{graphicx}
\usepackage[symbol]{footmisc}

\usepackage[shortlabels]{enumitem}

\usepackage{cleveref}

\usepackage{upgreek}
\usepackage{bm}

\usepackage{hyperref}
\usepackage{amsmath}
\usepackage{amsthm}
\usepackage{amsfonts}
\usepackage{amssymb}

\usepackage{dsfont}
\usepackage{latexsym}
\usepackage{geometry}
\usepackage{fixmath}
\usepackage{hyperref}
\usepackage{thmtools}
\usepackage{thm-restate}

\usepackage{cleveref}

\hypersetup{
    colorlinks=true,
    linkcolor=blue,
    filecolor=magenta,      
    urlcolor=cyan,
    }

\usepackage{upgreek}
\usepackage{bm}

\usepackage{algorithm}
\usepackage[noend]{algpseudocode}

\crefname{lem}{Lemma}{Lemmas}
\crefname{section}{Section}{Sections}
\crefname{lemma}{Lemma}{Lemmas}
\crefname{thm}{Theorem}{Theorems}
\crefname{corollary}{Corollary}{Corollaries}
\crefname{theorem}{Theorem}{Theorems}
\crefname{defn}{Definition}{Definitions}
\crefname{definition}{Definition}{Definitions}
\crefname{fact}{Fact}{Facts}
\crefname{figure}{Fig.}{Figures}
\crefname{clm}{Claim}{Claims}
\crefname{claim}{Claim}{Claims}
\crefname{prop}{Proposition}{Propositions}
\crefname{proposition}{Proposition}{Propositions}
\crefname{algocf}{Algorithm}{Algorithms}

\newtheorem{theorem}{Theorem}[section]
\newtheorem{lemma}[theorem]{Lemma}

\newtheorem{proposition}[theorem]{Proposition}
\newtheorem{claim}[theorem]{Claim}

\newtheorem{remark}[theorem]{Remark}
\newtheorem{definition}[theorem]{Definition}

\newcommand{\G}{\mathbold{G}}
\newcommand{\bH}{\mathbold{H}}

\newcommand{\alphabet}{\mathcal{A}}

\newcommand\fixsampler{{\tt Sampler}}

\newcommand{\switch}{{\tt Switch}}

\newcommand{\DisG}{\mathcal{Q}}

\newcommand{\bsigma}{{\mathbold{\sigma}}}
\newcommand{\btau}{{\mathbold{\tau}}}
\newcommand{\bpsi}{{\mathbold{\psi}}}

\newcommand{\initDis}{\UpF}

\newcommand{\DisSpin}{\mathcal{D}}

\newcommand{\psimin}{{\uppsi_{\rm min}}}

\newcommand{\dpsi}{\mathcal{P}}

\newcommand{\bethe}{\upbeta}

\newcommand{\cB}{\mathcal{B}}
\newcommand{\cC}{{\mathcal{C}}}
\newcommand{\cD}{{\mathcal{D}}}
\newcommand{\cE}{{\mathcal{E}}}
\newcommand{\cF}{{\mathcal{F}}}
\newcommand{\cH}{{\mathcal{H}}}
\newcommand{\cJ}{{\mathcal{J}}}
\newcommand{\cK}{{\mathcal{K}}}
\newcommand{\cL}{{\mathcal{L}}}

\newcommand{\cN}{{\mathcal{N}}}

\newcommand{\cR}{{\mathcal{R}}}
\newcommand{\cS}{{\mathcal{S}}}
\newcommand{\cX}{\mathcal{X}}
\newcommand{\cW}{{\mathcal{W}}}
\newcommand{\cY}{\mathcal{Y}}
\newcommand{\cI}{\mathcal{I}}

\newcommand{\bD}{{\bf D}}
\newcommand{\bF}{{\bf F}}

\newcommand{\bJ}{{\bf J}}

\newcommand{\bX}{{\mathbold{X}}}

\newcommand{\bkappa}{\mathbold{\kappa}}
\newcommand{\btheta}{{\mathbold{\theta}}}
\newcommand{\bxi}{{\mathbold{\xi}}}

\newcommand{\qp}{{p}}
\newcommand{\dpr}{{q}}
\newcommand{\visit}{\mathcal{ N}}

\newcommand{\drate}{\mathcal{R}}
\newcommand{\setB}{{\tt SET}}
\newcommand{\weightrange}{[0,2)}

\DeclareSymbolFont{EuclidLetter}{U}{eur}{m}{n}
\DeclareMathSymbol{\UpA}{\mathord}{EuclidLetter}{65}
\DeclareMathSymbol{\UpB}{\mathord}{EuclidLetter}{66}
\DeclareMathSymbol{\UpC}{\mathord}{EuclidLetter}{67}
\DeclareMathSymbol{\UpD}{\mathord}{EuclidLetter}{68}
\DeclareMathSymbol{\UpF}{\mathord}{EuclidLetter}{70}
\DeclareMathSymbol{\UpG}{\mathord}{EuclidLetter}{71}
\DeclareMathSymbol{\UpH}{\mathord}{EuclidLetter}{72}
\DeclareMathSymbol{\UpI}{\mathord}{EuclidLetter}{73}
\DeclareMathSymbol{\UpJ}{\mathord}{EuclidLetter}{74}
\DeclareMathSymbol{\UpK}{\mathord}{EuclidLetter}{75}
\DeclareMathSymbol{\UpL}{\mathord}{EuclidLetter}{76}
\DeclareMathSymbol{\UpM}{\mathord}{EuclidLetter}{77}
\DeclareMathSymbol{\UpP}{\mathord}{EuclidLetter}{80}
\DeclareMathSymbol{\UpQ}{\mathord}{EuclidLetter}{81}
\DeclareMathSymbol{\UpR}{\mathord}{EuclidLetter}{82}
\DeclareMathSymbol{\UpS}{\mathord}{EuclidLetter}{83}
\DeclareMathSymbol{\UpT}{\mathord}{EuclidLetter}{84}
\DeclareMathSymbol{\UpU}{\mathord}{EuclidLetter}{85}
\DeclareMathSymbol{\UpW}{\mathord}{EuclidLetter}{87}
\DeclareMathSymbol{\UpX}{\mathord}{EuclidLetter}{88}
\DeclareMathSymbol{\UpY}{\mathord}{EuclidLetter}{89}

\numberwithin{equation}{section}

\newcommand{\Ind}{{\mathds{1}}}

\newcommand{\spreadpoint}{\newpage}
 \renewcommand{\spreadpoint}{} 
\newcommand{\LastReview}[1]{{\color{magenta}\hspace{.2cm}  [Reviewed: \textrm{#1}] \hspace{.2cm}} }
  \renewcommand{\LastReview}[1]{}
\allowdisplaybreaks

\date{\today}

\begin{document}

\title{On sampling symmetric Gibbs distributions on sparse \\random graphs and hypergraphs.}
\author{Charilaos Efthymiou$^*$}
\thanks{ 
University of Warwick,  Coventry,  CV4 7AL, UK.  Email: {\tt charilaos.efthymiou@warwick.ac.uk}\\ 
$^*$Research supported by EPSRC New Investigator Award, grant EP/V050842/1,  and 
Centre of Discrete Mathematics and Applications (DIMAP), University of Warwick. \\ 
}
\address{}%{Charilaos Efthymiou, {\tt charilaos.efthymiou@warwick.ac.uk}, University of Warwick, Coventry, CV4 7AL, UK.}

\maketitle

\begin{abstract}
In this paper, we present a novel, polynomial time,  algorithm for approximate  sampling from  
symmetric Gibbs distributions on the sparse random graph and hypergraph. The examples of  
symmetric distributions   include but  are not restricted to, 
some  important distributions on  spin-systems and spin-glasses. 
Here, we consider  the $q$-state antiferromagnetic  Potts model for $q\geq 2$, including the (hyper)graph 
colourings.    We  also consider  the uniform distribution over the  Not-All-Equal solutions of a
random $k$-SAT formula.  Finally,  we consider sampling from   the {\em spin-glass} 
distribution called the $k$-spin model, i.e., this is the  
``diluted" version of the well-known
Sherrington-Kirkpatrick model.  Spin-glasses give rise to very intricate distributions which are also 
studied in mathematics, in neural computation, computational biology and many other areas.  To our 
knowledge, this is the first rigorously analysed efficient  algorithm for diluted spin-glasses which 
operates in  a non-trivial range of the parameters of the distribution.

We present, what we believe to be, an elegant  sampling algorithm for  symmetric Gibbs  distributions. 
Our algorithm is  unique in its approach and does not belong to any  of the well-known families of
sampling algorithms.  We derive it  by investigating the power and the limits of the  approach  
that was introduced in [Efthymiou: SODA 2012] and combine it, in a novel way, with powerful notions from
the Cavity method. 

Specifically, for a  symmetric Gibbs distribution $\mu$  on the random (hyper)graph whose parameters are within an 
appropriate  range, our sampling  algorithm has the following  properties: with probability $1-o(1)$ over 
the instances of the input  (hyper)graph,  it   generates  a configuration 
which is distributed within total variation distance $n^{-\Omega(1)}$  from $\mu$. The time complexity 
is $O((n\log n)^2)$, where   $n$ is the size of the input (hyper)graph.

We make    progress regarding impressive  predictions of physicists relating phase transitions of 
 Gibbs distributions with the efficiency of the corresponding  sampling algorithms.
For  most   cases we consider here, our algorithm  outperforms  any  other sampling algorithms 
in terms of  the permitted range of the parameters of the  Gibbs distributions. 

The use of notions and ideas from the Cavity method provides  a new   insight into the  sampling problem. 
Our results imply that there is a lot of potential for further exploiting  the  Cavity method for algorithmic design.

\medskip
\medskip
\noindent
{\bf Key words:} spin-system, spin-glass, sparse random (hyper)graph, approximate sampling, efficient algorithm.\\  \vspace{-.3cm}

\noindent
{\bf AMS subject classifications:} Primary 68R99, 68W25,68W20 Secondary: 82B44
\end{abstract}

\newpage
\tableofcontents

\spreadpoint
\addtocontents{toc}{\setcounter{tocdepth}{1}}
\tableofcontents

% \noindent
% {\color{magenta} 
% Note:  Change the set of condition name from $\setB$ to something different, the same for ${\bf B.1}$ and ${\bf B.2}$\\
% Note:  Should sampling from the BP fixed point take time that depends on $q$?\\
% Note:  Describe the proofs structures with a diagram?!\\
% Note: Emphasize that the algorithm is simple?!\\
% Note: In my CMP paper and subsequent works the number of edges is random, say something about it!\\
% Note: $\G$ is used both for random graph and random factor graph\\
% Note: First algorithm for sampling from spin-glasses? Argue in Abstract, Intro and in application section\\
% % Note: Complement of the event is with bar.\\
% Note: When refer to subsets of the set of weight-functions $\Psi$ need to choose from measurable sets\\
% Note: We use $\partial \alpha_i$ even in $G_{i}$ which does not contain $\alpha_i$. Explain the meaning\\
% Note: Make sure that $i$ is always formally quantified.\\
% Note: Check once more the text before the statement of \cref{thrm:DbcVsDcont} 
% and the text in its proof .\\
% Note: I think that the running times need some revision\\
% Note: contiguity results from other papers?

%  }

\setcounter{page}{1}

\section{Introduction \LastReview{2024-02-19}}\label{sec:intro}

%(add frieze facility location)

Random constraint satisfaction problems (r-CSPs) have been the subject of intense study in combinatorics, 
computer science and statistical  physics. In computer science the study of random CSPs is motivated by a 
wealth of applications, e.g., they are used as algorithmic  benchmarks for  hard problems such as the graph 
colouring, or the $k$-SAT, they are studied as models for statistical  inference, they are also used as 
gadgets for cryptographic constructions, or reductions in complexity theory, e.g., see 
\cite{AchNature,AndJACM,WillSantos,Feige,OneWayFunctions}.

Physicists, independently, have been studying random CSPs as models of {\em disordered systems} using 
the so-called   {\em Cavity Method} (e.g. see \cite{MPZ02,PNAS}). The Cavity method originates from the
groundbreaking ideas in physics which  got Giorgio Parisi the Nobel Prize in Physics in 2021. 
With  its deep intuition and  very impressive predictions, alas lacking  mathematical rigour,  
the Cavity Method  attracted the interest of  computer scientists and mathematicians.
In the last two decades, or so, ideas from the Cavity method have blended the study 
of random CSPs in computer science and have yielded some  beautiful results and breakthroughs in the area 
e.g., \cite{AchCoOg,CKPZ,DSS15,ChromNum,GMU}.

A fundamental notion in physicists’ considerations is that of the Gibbs distribution.   The Cavity method 
makes predictions relating phase-transitions of Gibbs distributions with the efficiency of the  sampling 
algorithms. Establishing rigorously  these connections is a very challenging task and, despite the recent advances, 
many of the central questions remain open.   In this paper, we introduce a novel approach to the sampling 
problem that exploits  {\em intuition} from the Cavity method, as well as  {\em mathematical tools} and ideas 
that  were developed for the study of  random CSPs in conjunction with the Cavity method. Our approach 
yields a sampling algorithm with  notable   performance with respect to the allowed regions of the parameters 
of the problem.

More specifically, we  present an efficient algorithm for sampling from  what we call {\em symmetric Gibbs distributions}.  
This family of  distributions includes important examples such as the (hyper)graph $q$-colourings and 
its generalisation the $q$-state Potts model for $q\geq 2$, the symmetric variants of $k$-SAT such as the 
not-all-equal $k$-SAT (NAE-$k$-SAT). A notable case  is the {\em spin-glass} $k$-spin model, i.e., the same
spin-glasses that Parisi studied in the 80's. Spin-glasses give rise to very intricate distributions which  
have been studied in mathematics, e.g., \cite{guerra2004high,PanchenkoTalagrand,TalagrandAnnals}, but also 
in other areas such as  neural computation, computational biology e.g. see  \cite{GlassesComplexity}. For us, 
the underlying geometry is an instance of the {\em random graph} or {\em hypergraph} of constant expected degree $d>0$.

For most (if not all)  of the above distributions it is  extremely challenging to sample from.  This is not only because the underlying geometry is random. Each one of these distributions exhibits special features that make the analysis of  known sampling techniques  intractable. 
E.g.,   in the interesting region of parameters for  $k$-NAE SAT, or hypergraph colourings,  we  have untypical  configurations  with {\em non-local} freezing of the 
 variables, the spin-glasses are extremely involved due to the random couplings, etc.

An additional motivation for this work comes from our desire  
to investigate the  power  and the limits  of the  (well-known) sampling method that is introduced 
in  \cite{MySODA12}. The method in  \cite{MySODA12} does not exploit the Cavity method ideas. 
On a high level, the approach   summarises  as follows:
having the graph $G$ at the input, the algorithm initially removes all the edges and generates a 
configuration for the empty graph. Then, iteratively, it puts the edges back one by one. 
For $G_i$,  the subgraph we have at iteration $i$,  our objective is   to  have a configuration
$\sigma_i$  which is distributed very close to the Gibbs  distribution on $G_i$.  
The idea is to generate $\sigma_{i}$ by  {\em updating} appropriately   the configuration of $G_{i-1}$, 
i.e.,  update  $\sigma_{i-1}$ to  generate efficiently $\sigma_i$. Once all  the edges are  put back, 
the algorithm outputs the configuration  of $G$.

The algorithm in \cite{MySODA12}  relies heavily on properties that  are special to  
graph colourings, for this reason, it is restricted to this distribution.
 If we'd like to follow a similar  approach to sample from a different distribution, then
we have to design a new algorithm from scratch.   
Our  aim, here,  is to have  a  sampling  algorithm such that  the Gibbs distribution 
we are sampling from is  a {\em parameter} of the input.

The analysis in \cite{MySODA12} follows a more classical approach to sampling 
than what we adopt here. 
%  t
It relies on the {\em correlation decay} condition  called 
{\em tree-uniqueness}  to establish  the accuracy of the algorithm. 
For our purposes, requiring such a condition can  be too  restrictive. 
On one hand, for many of the  distributions   we consider here, 
we are far from establishing    their tree uniqueness region. 
On the other hand, it seems that
Gibbs uniqueness is too restrictive  a condition for distributions on the hypergraph. 
With our approach  here we give a new insight to the problem by  showing that we can exploit  
notions about  the Gibbs  distributions that we typically encounter in  the study of the Cavity method 
and random CSPs.
For example, we use notions like the {\em broadcasting probabilities} encountered in  the study of the extremality of Gibbs 
distributions for random CSPs \cite{CombCA,PNAS},  or the  {\em contiguity} between the Gibbs distribution 
and its corresponding  {\em teacher-student model} for the study of 
the so-called free energy and its fluctuations \cite{quiet,AchCoOg,CKPZ,CoEfCMP}.

What is also notable about the performance of the algorithm is  the region of the parameters  it allows.
As we discuss  shortly,   this region either   coincides with
the tree-uniqueness  region, parametrised w.r.t. the {\em expected  degree} $d$ or 
it even gets  beyond that.  This depends on whether the underlying structure is a graph or a hypergraph. 
To prove that the corresponding MCMC sampler  works anywhere 
near the region  of the parameters that our algorithm allows,  would require  breakthroughs in the 
area of Markov   chains.

The state of the art for  MCMC samplers is   more restrictive with respect to
 the  parameters it  allows, however, it provides  stronger  approximation guarantees than what 
 we obtain here. 
Roughly speaking, our  results   summarise  as follows:  for a symmetric Gibbs distribution 
$\mu$ on the random (hyper)graph which  satisfies our set of conditions, we have an approximation sampling 
algorithm  such that  with probability $1-o(1)$ over the instances of the input  (hyper)graph, 
it  generates  a configuration which is distributed within total variation distance $n^{-\Omega(1)}$ 
from $\mu$. The { time complexity} is   $O((n\log n)^2)$.  

The reader should not confuse the bounds we get here with those in  ``worst-case" instances. 
For  worst case instances, usually, the parametrisation is  w.r.t.  the {\em maximum degree} of the underlying (hyper)graph, 
whereas for the random (hyper)graph, the natural parametrisation is w.r.t. the {\em expected degree}. 
Typically for the random (hyper)graphs here the maximum degree is {\em unbounded}, i.e., $\Theta(\log n/\log\log n)$, 
  while the  expected degree $d$ is a {\em fixed} number.

Concluding,  the idea of   ``adding edges and updating" turns out to be a quite powerful sampling technique,  
particularly when we combine  it with  notions and ideas from the  Cavity method. It allows us to sample
efficiently from distributions  that, prior to this work, we did not know how  to  sample.  Our approach 
leads to,  what we believe to be, a simple and  elegant sampling algorithm which deviates    from   
\cite{MySODA12}  not only on the  phenomena of the Gibbs distributions that  it utilises  but also on
its basic description.  Our work, also,  shows  how powerful the notions from the Cavity method can be,
i.e., even in the context of sampling algorithms.
We believe that there is a lot of potential towards the direction of using  ideas from the Cavity method 
for the sampling problem  to get even stronger  algorithms.

\subsection{General Results}\label{sec:Results}
In order to present our general results, we need to introduce few basic notions. 

\subsubsection*{Gibbs distributions \& Broadcasting Probabilities}
Let the (fixed)  $k$-uniform hypergraph $H_k=(V,E)$.  Clearly, the graph case 
corresponds  to having  $k=2$. A Gibbs distribution on $H_k$ is specified by the set of  spins $\alphabet$ 
and the  weight  functions   $(\psi_{e})_{e \in E}$ such that   $\psi_{e}:\alphabet^k\to\mathbb{R}_{\geq 0}$.
The  Gibbs distribution $\mu=\mu_{H}$  is on the set of configurations $\alphabet^V$
%, i.e., the  assignments of spins to the vertices of $H_k$,  
such that each  $\sigma\in \alphabet^V$ 
gets probability measure 
\begin{align}\nonumber 
\mu(\sigma) &\propto   \prod\nolimits_{e \in E} \psi_e(\sigma(x_{e,1}),\sigma(x_{e,2}), \ldots, \sigma(x_{e,k})) \enspace,
\end{align}
where $x_{e,i}$ is the $i$-th vertex in the hyperedge $e$. We assume a predefined order for
the vertices in each hyperedge. The symbol $\propto$ stands for  ``proportional to".

In many situations, we allow  $\psi_{e}$  to  vary with $e$, e.g.,  in  $k$-NAE-SAT, or the 
$k$-spin model these functions are chosen from a probability distribution.   For this early 
exposition,  the reader may very well assume   that all $\psi_{e}$'s 
are   the  same and fixed.

Roughly speaking, the Gibbs distribution $\mu$ is called  {\em symmetric}, if  for any   $\sigma,\tau\in \alphabet^V$ such that  
 $\sigma$ can be obtained  from $\tau$ by repermuting  the spin classes, we have that 
$\mu(\sigma)=\mu(\tau)$.    E.g., suppose that $\alphabet=\{\pm 1\}$, for a symmetric Gibbs distribution
$\mu$   we have  $\mu(\sigma)=\mu(\tau)$  for any  two  $\sigma, \tau\in \alphabet^V$  such that 
$\sigma(x)=-\tau(x)$  for all $x\in V$.

Given the  weight functions $(\psi_{e})_{e\in E}$ of the Gibbs distribution $\mu$ on $H_k$,  
for each $e\in E$,  we let  $\bethe_{e}$ be the distribution on  
$\alphabet^{e}$ such that 
\begin{align}\label{eq:DefOfBethe}
\bethe_{e}(\sigma) &\propto \psi_{e}(\sigma(x_{e,1}),\sigma(x_{e,2}), \ldots, \sigma(x_{e,k})) 
& \forall \sigma\in \alphabet^{ e} \enspace.
\end{align}
Furthermore,  we let  $\bethe^i_e$ be the distribution $\bethe_e$  conditional on the configuration  
at $x_{e,1}$ being  $i\in \alphabet$. In many settings, the quantities  $\bethe^i_e$, for $i\in \alphabet$, are    
known as the {\em broadcasting  probabilities} of  $\mu$.
Our algorithm  makes  extensive use of the   broadcasting probabilities.

% Here,  the underlying (hyper)graph structure is random. 
We let $\bH=\bH(n,m, k)$ be the random $k$-uniform hypergraph on $n$ vertices and $m$ hyperedges.
For the graph case, i.e.,   $k=2$, we write  $\G(n,m)$. 
% Here,  the underlying (hyper)graph structure is random. 
% 
The expected degree is denoted by 
$d$. We take $d$ to be a constant, i.e., $m=\Theta(n)$. Our results hold for any $d>0$, i.e., we  {\em do not} 
require  that ``$d$ is  sufficiently large" etc. 

\subsubsection*{Features of the algorithm:}
Consider a typical instance of $\bH$, of expected degree $d$, and   $\mu=\mu_{\bH}$ a symmetric  Gibbs distribution  
on $\bH$.    In what follows, we describe the basic features
of the algorithm we propose for sampling from $\mu$. 

We  recall the notion of {\em total variation distance}. For any two distributions $\hat{\nu}$ and $\nu$ on $\alphabet^V$ we have 
\begin{align}\nonumber
||\hat{\nu}-\nu||_{\rm tv}&=(1/2)\sum\nolimits_{\sigma\in \alphabet^V}|\hat{\nu}(\sigma)-\nu(\sigma)| \enspace.
\end{align}
Also, we let $||\hat{\nu}-\nu ||_{\Lambda}$ be the total variation distance of the marginals of $\hat{\nu}$ and 
$\nu$ at the  set $\Lambda\subseteq V$.

\newcommand{\CA}{{\bf B.1}}
\newcommand{\CB}{{\bf B.2}}
\newcommand{\CC}{{\bf B.3}}

For the algorithm to meet our approximation guarantees but also to carry out the analysis, we require
that   $\mu$ satisfies the set of conditions that  we call  $\setB$.  The main conditions in $\setB$ are   $\CA$ and $\CB$. 

$\CA$ is about the broadcasting probabilities  $\bethe^i_e$, $\bethe^j_e$, for any $i,j\in \alphabet$. %,  i.e,  
We say that the condition  $\CA$ is satisfied with slack $\delta>0$ if we have that 
\begin{align}
\max_{i,j\in \alphabet}||\bethe^i_e-\bethe^j_e ||_{\Lambda}&\leq \frac{1-\delta}{d (k-1)},
&  \textrm{where  $\Lambda=\{x_{e,2}, x_{e,3},\ldots, x_{e,k}\}$} \enspace.  \nonumber
\end{align}
The above  implies that  any two broadcasting probabilities of $\mu$ are not too far from each other. 
Specifically,  their total variation distance is smaller than $(1-\delta)$ over  the {\em expected number of 
neighbours} of a given vertex in $\bH$.

The condition $\CB$  requires  to have mutual  {\em contiguity} between  the Gibbs distribution $\mu$  and  the  
so-called teacher-student model.
We generate the pair $(\bH^*, \bsigma^*)$  according to the  
teacher-student  model by working as follows:  choose   $\bsigma^*$    randomly  from   $\alphabet^V$.
Given  $\bsigma^*$, generate the  {\em weighted} random hypergraph $\bH^*$  on $n$ vertices and
$m$ edges, where the weight of each hypergraph instance  depends on   $\bsigma^*$ and   $\mu$. 
Contiguity implies that the {\em typical properties} of the pair $(\bH^*,\bsigma^*)$ are the same as
those  of the pair $(\bH,\bsigma)$,  where $\bH=\bH(n,m, k)$ and $\bsigma$ is distributed as in $\mu$.
More formally, contiguity implies that   for   any sequence of events $(\mathcal{S}_n)_n$ we have that 
\begin{align} \nonumber 
\Pr[(\bH,\bsigma)\in \mathcal{S}_n]&=o(1) &\textrm{iff} &&  \Pr[(\bH^*,\bsigma^*)\in \mathcal{S}_n]&=o(1) \enspace. 
\end{align}

Note that $\CB$, i.e., the contiguity condition, is not directly related to the performance of the algorithm. 
It is a condition we need in order to carry out the analysis of the algorithm. For further discussion on
$\setB$, see   \cref{sec:Conditions}.

In what follows, we say that the Gibbs distribution $\mu$ satisfies $\setB$ with {\em slack} $\delta>0$, to imply
that, apart from $\CB$ being satisfied,  condition $\CA$  is satisfied with  slack $\delta$.

\subsubsection*{Results:} Having seen all the above notions, now, we can state formally our results.

\begin{theorem}\label{thrm:MainA}
For  $\delta\in (0,1]$, for  integer $k\geq 2$,  for any $d\geq 1/(k-1)$ and integer $m={dn}/{k}$ the following is true  for our algorithm: 
Consider the random $k$-uniform  hypergraph  $\bH=\bH(n,m,k)$.
Let $\mu=\mu_{\bH}$ be a symmetric Gibbs distribution  on $\bH$ which satisfies   $\setB$ with slack $\delta$.

With probability $1-o(1)$,   over the input instances  $\bH$ and weight functions  on the edges of $\bH$,
our algorithm generates a configuration whose distribution $\bar{\mu}$ is such that
\begin{align}  \nonumber
||\bar{\mu} -\mu ||_{\rm tv} &\leq   n^{-\frac{\delta}{55\log(dk)}} \enspace.
\end{align}
\end{theorem}

\noindent
As mentioned above, the theorem does not require $d$ to be a  ``sufficiently large constant". We chose $d\geq 1/(k-1)$, 
because otherwise the underlying graph  structure is very simple and the problem is trivial. 

Let us remark that we did not try to optimise the exponent of the error bound in \cref{thrm:MainA}.

\begin{theorem}\label{thrm:MainB}
For  $k\geq 2$ and  $d\geq 1/(k-1)$ and  integer $m={dn}/{k}$, consider the random $k$-uniform  hypergraph  $\bH=\bH(n,m,k)$.
The time complexity of our algorithm on input $\bH$ is $O\left((n\log n)^2\right)$.
\end{theorem}

\noindent
 \cref{thrm:MainB} follows as a corollary of  \cref{thrm:FinalDetTime}, see \cref{sec:FullPerformance}.

See  \cref{sec:Applications} for the results of applying the above theorems  on specific distributions.  For a high level   description of the algorithm, see \cref{sec:HighLevel}.

\subsection*{Related work}

The idea of ``adding edges and updating" for sampling was first introduced in \cite{MySODA12} for sampling
colourings of random graphs.
The techniques and tools we introduce here for the sampling problem,
 rely on results developed in the study  Cavity method and random CSP's in 
 \cite{AchCoOg,CoEfCMP,CoKaMu20,CKPZ}. 

There are  two  other works which follow  the  approach of ``adding edges and updating" and use  the same
correlation decay approach to  \cite{MySODA12}.
One is   \cite{MySICOMP}, an improvement of \cite{MySODA12}, which is  about   colourings  of the random graph of sufficiently 
large expected degree $d$.  
The other one  is   \cite{TRTrivialAlg} for the Potts model on the related random $\Delta$-regular graph, for large $\Delta$.
From the second paper,  it is conceivable  that we can get an efficient algorithm  only for the ferromagnetic  Potts model
on the random  graph, provided that the expected degree $d$ is large.
Apart from colourings and ferro-Potts on the graph,  we cannot rely on any of these two approaches for our endeavours.  
Both of them rely on the special properties of the distribution they are sampling from, thus they don't allow
 for other distributions. Furthermore,  their tree-uniqueness  requirement  restricts their use to 
considering only graphs, rather than  hypergraphs.
Our work here improves  on both results in  \cite{MySICOMP,TRTrivialAlg} as it allows for any expected degree $d>0$, 
i.e., rather than sufficiently large $d$.

There are  other approaches to  sampling from Gibbs distributions    which are different than the one we consider  here. Notably,  the most popular ones  rely on the Markov Chain  Monte Carlo Method (MCMC) \cite{jerrumB,FrVigSurvey}. 
The literature on MCMC sampling algorithms (not only for random graphs)  includes  some  beautiful  results, just to mention a few 
\cite{OptMCMCIS,RCGeneralBound,BlancaPottsGnp,RCPlanar,WeimingDA21,RCLattice,MSSensoring,COLPaperA,MySODA18,BestBound,DFUniformity,AndreasColouringDA21}.

The reader should not confuse the $k$-spin model on the random graphs and hypergraph of constant expected degree
$d$, we consider here, with the {\em mean-field} spin-glasses such as the Sherrington-Kirkpatrick model which is on
the complete graph. 
The two models are quite different from each other.

Our results about the colourings are  related to the work in  \cite{MySODA18} for  MCMC sampling.  
In that respect, our approach outperforms, by far,  \cite{MySODA18}   in terms of  the 
range of the allowed parameters of the Gibbs  distributions. 
However,  we note that the MCMC algorithm achieves better approximation  
guarantees in the  (more restricted) regions it operates. 

\subsection*{Notation} Let the hypegraph $H_k=(V,E)$ and the Gibbs distribution $\mu$ on the set of configurations
$\alphabet^V$.  For a configuration $\sigma$,  we let $\sigma(\Lambda)$ denote the configuration that $\sigma$
specifies on the set of vertices $\Lambda$.
We let   $\mu_{\Lambda}$  denote the marginal of $\mu$ at the  set $\Lambda$.
For a configuration $\sigma\in \alphabet^{V}$ we let $\mu(\cdot \ |\ \Lambda, \sigma)$, denote the distribution 
$\mu$ conditional on the configuration at $\Lambda$ being $\sigma(\Lambda)$. Also, we interpret the conditional
marginal $\mu_{\Lambda}(\cdot \ |\ \Lambda', \sigma)$, for $\Lambda'\subseteq V$, in the natural way.

\spreadpoint

\newcommand{\BMIsing}{ {{\beta}_{\rm Ising} }}

\section{Applications \LastReview{2024-02-05} }\label{sec:Applications} 
\noindent
In  this part of our work, we present a few applications of our algorithm. The list of 
distributions we consider below is not meant to be exhaustive.  The main criterion for choosing  the following distributions, apart from being very important in the field,  is the common frame for analysis we have from \cite{CKPZ,CoKaMu20,CoEfCMP} which we can 
apply directly here.

\subsubsection*{The antiferromagnetic Ising Model}\label{sec:AntiIsing}

The Ising model on the $k$-uniform hypergraph $H_k=(V,E)$ is a   distribution  on  the set of configurations $\{\pm1\}^V$ such that
each  $\sigma\in \{\pm1\}^V$ is assigned probability measure
\begin{equation}\nonumber  
\mu(\sigma) \propto  \exp\left ( \beta \cdot 
 \sum\nolimits_{e \in E}\prod\nolimits_{x,y \in e}  \Ind\{\sigma(x)=\sigma(y)\}   + h\cdot  \sum\nolimits_{x\in V}\sigma(x)
\right) \enspace,
\end{equation}

\noindent
where $\beta\in \mathbb{R}$ is the  {\em inverse temperature} and $h$ is the {\em external field}.
It is straightforward   that the Ising model  is symmetric  when $h=0$.
We assume    $\beta<0$, which  corresponds to the {\em antiferromagnetic} Ising model.

In what follows, for a positive integer $k$ and  for  $\Delta\geq 2^{k-1}/(k-1)$ we  let the  function
\begin{equation} \nonumber
\BMIsing(\Delta, k)=
\log{\textstyle \left(\frac{\Delta(k-1)+1-2^{k-1}}{\Delta(k-1)+1}\right)}
\enspace. 
\end{equation}
% The above extends  the {\em tree uniqueness}  bound for a general size edge $k>0$.
Note that $\BMIsing(\Delta, 2)$  signifies  the  uniqueness  for the antiferromagnetic Ising model  on the 
$\Delta$-ary tree, i.e.,  uniqueness corresponds to having
\begin{align}\nonumber
\BMIsing(\Delta, 2)<\beta\leq 0\enspace. 
\end{align}
It is possible that   $\BMIsing(\Delta, k)$ for $k>2$ is related to  the tree uniqueness on
the $k$-uniform hyper-tree. To the best of our knowledge, there is no rigorous proof for this, yet.

\cref{thrm:MainA,thrm:MainB} imply the following result for the Ising model. 
%%%
\begin{theorem}\label{thrm:FerroIsing}
For  integer  $k\geq 2$,  for any $d\geq 1/(k-1)$, for  $m={dn}/{k}$  the  following 
is true:  

Assume   that   $\beta\in \mathbb{R}$ with $d, k$ satisfy  one   of the following  cases:
\begin{enumerate}
\item $2^{k-1}-1<d(k-1)$ and    $\textstyle \BMIsing\left(d, k\right) < \beta < 0$, 
\item $2^{k-1}-1\geq d(k-1)$ and   $\beta < 0$. 
\end{enumerate}
Consider the random $k$-uniform  hypergraph  $\bH=\bH(n,m,k)$.   Let $\mu=\mu_{\bH}$ 
be the  antiferromagnetic Ising model on  $\bH$,   with inverse temperature  
$\beta$ and  external field $h=0$.

There exists $\updelta_0>0$ which depends only on the choice of  
$k,d, \beta$,  such  that with 
probability  $1-o(1)$  over the input instances  $\bH$,  our algorithm generates    
a configuration with distribution $\bar{\mu}$  such that
\begin{equation}  \nonumber
||\bar{\mu} -\mu ||_{\rm tv} \leq  n^{-\frac{\updelta_0}{55\log(dk)}} \enspace.
\end{equation}
The time complexity of  the algorithm is $O\left((n\log n)^2\right)$ with  probability 1. 
\end{theorem}

The proof  of Theorem \ref{thrm:FerroIsing} appears in Section \ref{sec:AppIsing}.

\newcommand{\BMPotts}{ {{\beta}_{\rm Potts} }}

\subsubsection*{The antiferromagnetic  Potts Model and the  Colourings}\label{sec:App:Potts}
The $q$-state Potts model on the $k$  uniform hypergraph $H_k=(V,E)$ is a generalisation of the Ising model. 
Particularly,   each  $\sigma\in [q]^V$, where  
$[q]=\{1,2,\ldots, q\}$,  is assigned probability measure
\begin{equation} \nonumber %\label{eq:GPotts} 
\mu(\sigma) \propto  \exp\left ( \beta \cdot \sum\nolimits_{e \in E}  \prod\nolimits_{x,y \in e}
\Ind\{\sigma(x)=\sigma(y)\}  \right) \enspace.
\end{equation}
where $\beta\in \mathbb{R}$ is the  {\em inverse temperature}.  
The antiferromagnetic Potts model   corresponds to   having $\beta<0$.

A very interesting case of the Potts model is the {\em colouring model}. This is the uniform distribution 
over the {\em proper}  $q$-colourings of the underlying (hyper)graph $H_k$, i.e., we do not allow 
configurations with monochromatic edges. The colouring model corresponds to the Potts model with
$\beta=-\infty$.

For the reader to appreciate our  results, we provide a chart of the uniqueness for Potts
on the  $\Delta$-ary  tree. This is a blend of rigorous results and conjectures.  
We have uniqueness if and only if  one of the following  holds:
\begin{itemize}
\item  $q>\Delta+1$ and $\beta < 0$ which includes $\beta=-\infty$,
\item  $q<\Delta+1$ and   $\log\left(\frac{\Delta+1-q}{\Delta+1}\right) <\beta<0$. 
\end{itemize}

\noindent
The tree uniqueness for colourings, i.e., $\beta=-\infty$,  is  from  \cite{TreeUniqColor}.
For finite temperature, i.e., $\beta\neq-\infty$,the non-uniqueness  condition 
$\log\left(\frac{\Delta+1-q}{\Delta+1}\right) > \beta$ follows from 
\cite{PDLMA,PDLMB, AndJACM}. Establishing the uniqueness 
seems  to be   a challenging problem.  
Recently, there has been a significant  progress with the best estimate being the ones
in \cite{SSMPottRegts}, building  on the  work in \cite{SSMColouring}.

For  $k\geq 2$ and  $\Delta>\frac{q^{k-1} -1}{k-1}$, we let the function  
$$
\BMPotts(\Delta, q,k)=\textstyle \log\left(\frac{\Delta(k-1)+1-q^{k-1}}{\Delta(k-1)+1}\right) \enspace.
$$
As we discuss above,  $\BMPotts(\Delta, q,2)$ is related to the uniqueness of the Potts model on the tree. 
It is open whether  the quantity $\BMPotts(\Delta, q,k)$, for $k>2$, is related to the uniqueness of
the Potts model on the k-uniform hypertree. Most likely it signifies a point which lies  beyond uniqueness,  
particularly for $k\gg 2$.

\cref{thrm:MainA,thrm:MainB} imply the following result for the $q$-state Potts model. 
\begin{theorem}\label{thrm:Potts}
For  integer  $k\geq 2$,  for any $d\geq 1/(k-1)$, for integer $m={dn}/{k}$  the  following 
is true:  
Assume   that   $\beta\in \mathbb{R}$ and the integer $q\geq 2$ satisfy  one   of the following  cases:
\begin{enumerate}
\item  $q^{k-1}-1< d(k-1)$ and  $\BMPotts(d,q, k)< \beta < 0$,
\item $q^{k-1} -1>  d(k-1)$ and $\beta< 0$, including $\beta=-\infty$,  
\item $q^{k-1} -1 = d(k-1)$ and $\beta< 0$ is bounded from below. 
\end{enumerate} 
Consider the random $k$-uniform  hypergraph  $\bH=\bH(n,m,k)$.   Let $\mu=\mu_{\bH}$ be the $q$-state 
antiferromagnetic Potts model on  $\bH$  with inverse temperature  $\beta$.  There exists
$\updelta_0>0$, which depends only on our choices of $k,d,\beta$ and $q$ such  that  with probability 
$1-o(1)$ over the input instances  $\bH$,  
our algorithm generates a configuration whose distribution $\bar{\mu}$ is such that
\begin{equation}  \nonumber
||\bar{\mu} -\mu ||_{\rm tv} \leq  n^{-\frac{\updelta_0}{55\log(dk)}} \enspace.
\end{equation}
The time complexity of  the algorithm is $O((n\log n)^2)$ with  probability 1. 
\end{theorem}

The proof of Theorem \ref{thrm:Potts} appears in Section \ref{sec:thrm:Potts}.

\subsubsection*{The NAE-$k-$SAT}\label{sec:NAESAT}
For integer $k\geq 2$,    let  $\bF_{k}(n,m)$ be a random propositional formula over  the Boolean variables 
$x_1, \ldots, x_n$.   Particularly, $\bF_{k}(n,m)$  is obtained by inserting $m$ independent random  clauses
of length $k$ such that no variable appears twice in the same clause. 
Here,  we  consider formulas with  $m=dn/k$ clauses for a fixed number $d$, i.e., on average every variable 
occurs in  $d$ clauses.

We focus on  the ``Not-All-Equal" satisfying assignments of $\bF_{k}(n,m)$.  A Boolean assignment 
$\sigma$ of $x_1, \ldots, x_n$ is NAE-satisfying for $\bF_k(n,m)$  if under both  $\sigma$ and its binary 
inverse  $\bar{\sigma}$  all $m$ clauses evaluate to ``true". 

The random NAE-$k$-SAT problem is one of the standard examples of random CSPs and has received a great 
deal of  attention. In particular, in an influential paper, Achlioptas and Moore \cite{TwoMomentsPaper} 
pioneered the use of the {\em second-moment method}
for estimating the partition functions of random CSPs with the example of random NAE-$k$-SAT.

Applying \cref{thrm:MainA,thrm:MainB} 
on the  {uniform distribution} over the NAE satisfying assignments of $\bF_{k}(n,m)$,
we get the following result. 

\begin{theorem}\label{thrm:NAESAT}
For  $\delta\in (0,1]$, for  $k\geq 2$, for any $1/(k-1) \leq d <(1-\delta)\frac{2^{k-1}-1}{k-1}$  and for integer 
$m={dn}/{k}$, the  following is true: % for our algorithm:  

Consider $\bF_{k}(n,m)$ and   let $\mu$ be the uniform distribution over the NAE 
satisfying assignments of $\bF_{k}(n,m)$.
With probability  $1-o(1)$ over the input instances  $\bF_{k}(n,m)$,   our algorithm generates 
a configuration whose distribution $\bar{\mu}$ is such that
\begin{equation}  \nonumber
||\bar{\mu} -\mu ||_{\rm tv} \leq  n^{-\frac{\delta}{55\log(dk)}} \enspace.
\end{equation}
The time complexity of  the algorithm is $O((n\log n)^2)$ with  probability 1. 
\end{theorem}

The proof of Theorem \ref{thrm:NAESAT} appears in Section \ref{sec:thrm:NAESAT}.

As a point  of reference for the performance of our {\em sampling } algorithm, note that it  works in a region of 
parameters which is  comparable (very close) to those of  the {\em search} algorithms for the problem, e.g.  see   \cite{WalkSat}. 

\subsubsection*{The $k$-spin model} \label{sec:App:QSpin}

For integer $k\geq 2$,  consider the  $k$ uniform hypergraph $H_k=(V,E)$ and let $\bJ = (\bJ_e)_{e\in E}$ be a 
family of independent, standard Gaussians, i.e., $\cN(0,1)$. The {\em $k$-spin model} on $H_k$ at inverse temperature $\beta>0$ is the 
distribution that assigns each configuration  $\sigma\in \{\pm1\}^{V}$ the probability measure 
\begin{align}\label{eqkSpin1}
 \mu(\sigma)&\propto   \prod\nolimits_{\alpha \in E} \exp\left ( \beta  \bJ_a
 \prod\nolimits_{y\in \alpha}\sigma(y)\right) \enspace.
\end{align}
The $k$-spin model is symmetric when $k\geq 2$ is an even integer (cf \cite{CoEfCMP}). 
Here we consider the above distribution when the underlying (hyper)graph  is an instance of $\bH=\bH(n,m,k)$ of expected 
degree $d$.

\newcommand{\GlassInf}{{\Upphi}}

For $x\in \mathbb{R}$, we let
\begin{equation}\label{Def:FKSpinGlass}
\GlassInf(x)=\frac{|e^{x}-e^{-x}|}{e^{-x}+e^x} \enspace.
\end{equation}

\begin{theorem}\label{thrm:KSpin}
For  $\delta\in (0,1]$, for  even integer $k\geq 2$, for any $d\geq 1/(k-1)$ and  for any $\beta \geq 0$ such that
\begin{align}
\mathbb{E} \left [\GlassInf\left(\beta \bJ\right) \right] &\leq \textstyle \frac{1-\delta}{d(k-1)} \enspace,
\end{align}
where the expectation is with respect to  the standard Gaussian random variable $\bJ$, 
the  following is true: % for our algorithm:  

Consider $\bH=\bH(n,m,k)$, where $m=dn/k$, and   let $\mu$ be the $k$-spin model on $\bH$ at inverse temperature $\beta$.
With probability  $1-o(1)$ over the input instances  $\bH$ and the weight  functions on the edges of $\bH$,   
our algorithm generates  a configuration 
whose distribution $\bar{\mu}$ is such that
\begin{equation}  \nonumber
||\bar{\mu} -\mu ||_{\rm tv} \leq   n^{-\frac{\delta}{55\log(dk)}} \enspace.
\end{equation}
The time complexity of  the algorithm is $O((n\log n)^2)$ with  probability 1. 
\end{theorem}

\noindent
The proof of Theorem \ref{thrm:KSpin} appears in Section \ref{sec:thrm:KSpin}.

\spreadpoint

\section{Algorithmic Approach - High Level Description \LastReview{2024-02-15}}\label{sec:HighLevel}

To facilitate the high-level exposition of the algorithm, assume in this section that we are dealing with a fixed graph.
Let's recall the algorithm: on input  $G$, the algorithm initially removes all the  edges and generates a configuration 
for the empty graph. Then, iteratively,   it puts the edges  back one by one. If $G_i$ is the subgraph we have at iteration 
$i$,  the aim is   to  have a configuration $\sigma_i$  which is distributed very close to the Gibbs  distribution  on $G_i$,
for every $i$.  The configuration  $\sigma_{i}$ is generated   by updating appropriately $\sigma_{i-1}$,
the configuration of $G_{i-1}$.  Once all  edges are  put back, the algorithm outputs the configuration of $G$.

One of the main  challenges  is to specify the {\em update rule} that  generates  $\sigma_{i}$ from $\sigma_{i-1}$.
We describe the  rule we propose  by considering the following, simpler, setting. Consider  two {\em high-girth}, fixed, 
graphs $G=(V,E)$ and $G'=(V,E')$.  Assume that  $G$ and $G'$  differ on a single edge,   i.e.  compared to $G$, the 
graph  $G'$ has the extra edge $e=\{u,w\}$.  Let $\mu$ and $\mu'$ be the Gibbs distributions of $G$ and $G'$, 
respectively.  We use the update rule to generate efficiently $\btau$ a sample from $\mu'$ by using 
  $\bsigma$,  a sample from $\mu$.
To simplify matters further, assume   that we already know $\btau(u)$ and $\btau(w)$, while they are such  that 
$\btau(u)=\bsigma(u)$ and $\btau(w)\neq \bsigma(w)$.  Henceforth,  we focus on describing how to get
 $\btau$ for the  rest of the vertices in $V$.

The plan is to {\em iteratively} visit each vertex $z$  and specify $\btau(z)$.  
At each iteration $t$, we only know the configuration of $\btau$  for the vertices 
in  the set $\visit_t$, i.e., the vertices that       have already been visited.
Let   $\DisSpin=\{\btau(w), \bsigma(w)\}$, i.e.,   $\DisSpin$  is the set of the spins 
of the initial disagreement.  

At iteration $t$,  we pick a vertex $z$ which is outside $\visit_t$
but   has a neighbour $x\in \visit_t$ which is disagreeing, i.e., $\btau(x)\neq \bsigma(x)$.
For the moment, assume  that such a vertex exists. 
If $\bsigma(z)\notin \DisSpin$, then we just set $\btau(z)=\bsigma(z)$. Otherwise, i.e., 
if $\bsigma(z)\in \DisSpin$,   then we work as follows: there is a probability $\qp_z$, that depends 
on the configuration of $\bsigma$ and $\btau$ at $\visit_t$, and we set 
\begin{equation}\nonumber  
\btau(z)=
\left \{
\begin{array}{lcl}
\DisSpin \setminus \{\bsigma(z)\} && \textrm {with prob. }\ \qp_z\\
\bsigma(z) && \textrm {with prob. }\ 1-\qp_z \enspace.
\end{array}
\right .
\end{equation}
The first line implies that $\btau(z)$ gets  the opposite spin of  $\bsigma(z)$ with respect to $\DisSpin$. 
E.g.,  if $\DisSpin=\{{\tt red}, {\tt blue}\}$ and $\bsigma(z)={\tt red}$, then $\btau(z)={\tt blue}$.
Once  $\btau(z)$ is decided, set $\visit_{t+1}=\visit_t\cup \{z\}$ and continue with the
next iteration. 

It could be that  in iteration $t$, there is no  vertex  $z$ outside $\visit_t$ which has a 
disagreeing neighbour inside $\visit_t$. In this case,  for every $z$ for which we 
have not specified $\btau(z)$, we  set $\btau(z)=\bsigma(z)$.

The probability $\qp_z$ is determined  in terms of a {\em maximal coupling} between the marginals 
of $\mu'$ and $\mu$ at $z$, conditional on $\btau(\visit_t)$ and $\bsigma(\visit_t)$. 
We denote these marginals  as $\mu'_{z}(\cdot \ |\ \visit_{t}, \btau)$ and $\mu_{z}( \cdot \ |\ \visit_{t}, \bsigma)$, 
respectively.  We have 
\begin{equation}\nonumber 
\qp_z=\max \left \{0, 
1-\frac{\mu'_{z}(\bsigma(z) \ |\ \visit_{t}, \btau)}{\mu_{z}( \bsigma(z) \ |\ \visit_{t}, \bsigma)}
\right\} \enspace. 
\end{equation}

\noindent
One can show that the above generates a {\em  perfect} sample from the distribution 
$\mu'$. 

The obstacle with the above  approach is the computation of  the probabilities $\qp_z$, efficiently. 
In our setting,  we don't  know how to estimate them because they involve marginals of Gibbs distributions. 
To circumvent this problem, we use {\em different} probabilities.  That is, we follow the 
previous steps and when at the  iteration $t$ we examine a vertex $z$ for which $\bsigma(z)\in \DisSpin$, 
we set  $\btau(z)$ such that
\begin{equation}\label{eq:RealStep4Q}
\btau(z)=
\left \{
\begin{array}{lcl}
\DisSpin \setminus \{\bsigma(z)\} && \textrm {with prob. }\ \dpr_z\\
\bsigma(z) && \textrm {with prob. }\ 1-\dpr_z \enspace,
\end{array}
\right .
\end{equation}
i.e., instead of $\qp_z$ we use  $\dpr_z$. To specify $\dpr_z$,  recall  that we choose $z$ because 
it has a  disagreeing neighbour $x\in  \visit_t$.  The probability $q_z$  is expressed in terms of the  simpler distribution   $\bethe_{\alpha}$,  where $\alpha$ is  the edge between $z$ and $x$, i.e.,   we have
\begin{equation}\label{eq:RealMarginal4Q}
\dpr_z=\max \left \{0,\ 
1-\frac{\bethe_{\alpha,z}(\bsigma(z) \ |\ x, \btau)}{\bethe_{\alpha,z} (\bsigma(z)\ |\ x, \bsigma)} \right\} \enspace.
\end{equation}
Recall  from  our notation that $\bethe_{\alpha,z}(\cdot  \ |\ x, \btau)$ is the marginal of $\bethe_{\alpha}$ at $z$, 
conditional on $x$ being set $\btau(x)$.  From \eqref{eq:DefOfBethe} we have that the distribution $\bethe_{\alpha}$ 
is very simple  and can be computed very fast.

A natural question at this point is what motivates the use of $\dpr_z$ instead of $\qp_z$. We observe that if our 
graphs $G$  and $G'$ were trees,  then we would have had that  $\dpr_z=\qp_z$. That is, for trees our 
update rule generates perfect samples from  $\mu'$. In some sense, our approach amounts to approximating 
the probabilities $\qp_z$, which are difficult to compute, with those of the tree, which we can compute very fast. 
In light of our assumption that our graphs $G$ and $G'$ are of high-girth,  i.e., locally tree-like, this approximation 
seems quite natural.   

Furthermore, there is a natural way of quantifying how accurate the rule is, i.e., how close is the distribution
of the output configuration to the distribution $\mu'$. This can be done in terms of what we call 
the {\em failure probability}.
Let $\DisG$ be the set of vertices that change configuration during the update, i.e., their
configuration under $\btau$ is different than that under $\bsigma$. Somehow, our update rule runs into trouble
when $\DisG$ induces a subgraph which contains one of the long cycles of $G$, or  $\DisG$  reaches  $u$. 
In this case, we consider that the algorithm {\em fails}. That is, the update rule outputs either a 
configuration $\btau\in \alphabet^V$,  or a fail status.  We establish that  the  accuracy  of the update, i.e., 
the total variation distance between the distribution of the output configuration  and $\mu'$ 
is proportional to   the failure probability.

\subsection{$\setB$ Vs Accuracy}\label{sec:SetVsAccuracy}
We provide a high-level discussion explaining the intuition for using $\setB$. 
Specifically,  we focus on explaining why the failure probability is small with $\setB$.

Consider the random graph $\G=\G(n,m)$ of expected degree $d$. Let $\mu$ be a symmetric  Gibbs distribution 
on $\G$.  Consider the spins $c, \hat{c}\in\alphabet$ different from each other.  
Let  $\bsigma$ be distributed as  in $\mu$ conditional on $\bsigma(u)=c'$, for  some vertex $u$ in $\G$.  
We use  the update process we describe before   to generate  a  configuration $\btau$ which is 
distributed as in $\mu$ conditional on $\btau(u)=c$. 
Our focus  is  on the probability of failure for the update. 

Practically, the aim here is to argue that with condition $\setB$, the set 
of disagreeing vertices in  the update process  grows {\em subcritically} at each iteration. 
In the update,    the disagreements start from vertex $u$ and iteratively  propagate over  the graph. 
Suppose that we are at the early stage of the process, i.e.,   not too many vertices have been visited. 
Assume further that at the following iteration,  the process picks  vertex $z$ which is 
next a disagreement $x$,  i.e., we have $\btau(x)\neq \bsigma(x)$.

If the process hasn't revealed the configuration of $\bsigma$ for too many vertices, 
then it is not far-fetched to assume that the marginal of the configuration at $z$  is 
close to   $\bethe_{e,z}(\cdot \ |\ x, \bsigma(x))$, where $e=\{x, z \}$. 
Furthermore, this would imply that the  disagreement probability at $z$ is at most
$$
\max\nolimits_{c,\hat{c}\in \alphabet}||\bethe_e(\cdot \ |\ x, c)-\bethe_e(\cdot\ |\ x, \hat{c}) ||_{z} \enspace.  
$$
Interestingly,   we have an upper bound for the disagreement probability, i.e., the above total variation  distance,  from condition $\CA$ (in $\setB$).  Particularly,  $\CA$   implies  that the quantity  is  $<1/d$. 
Hence,  the above intuition, if it is correct,  implies that the disagreements grow  subcritically.

Unfortunately, with our assumptions about $\mu$, it is too difficult
to argue that the marginal probability at $z$ is indeed close to   $\bethe_{e,z}(\cdot \ |\ x, \bsigma)$.

We circumvent this problem by utilising the teacher-student model.  That is,  consider the 
pair  $(\G^*, \bsigma^*)$ from the 
teacher-student model. We study the propagation of disagreements  for the  update process on the pair 
$(\G^*,\bsigma^*)$.  
There, it is much simpler to argue that the distribution of $z$ is very close 
to $\bethe_{e,z}(\cdot \ |\ x, \bsigma)$.  The condition $\CA$, which still applies to the teacher-student model,    
implies 
that the growth of  disagreements in $\G^*$ is subcritical.  
In turn, this implies that the failure probability for the  update applied to $(\G^*, \bsigma^*)$ is very small.  
Subsequently, we employ contiguity, i.e., $\CB$,  to argue that  if the probability of failure for  the case of 
$(\G^*,\bsigma^*)$  is very small,  then the probability  of failure for the update rule when it is applied to the
pair $(\G, \bsigma)$ cannot be much larger. 

From the above it should be clear that condition $\CB$ is for the sake of the analysis and so much for
the performance of the algorithm itself.

\subsection{Plan of the Analysis}
The more involved in obtaining result  in this paper is \cref{thrm:MainA}. In \cref{fig:StructureAnalysis} we 
show the structure of the proof of this result. Note that the figure only includes the main ingredients and,
in many cases, we use smaller results.  The running time of the algorithm is analysed in \cref{sec:thrm:FinalDetTime}. 

The results from \cref{sec:Applications} are corollaries  from \cref{thrm:MainA}, their proofs appear in \cref{sec:ProofApps}.

\begin{figure}
\centering
	\centering
		\includegraphics[width=.7\textwidth]{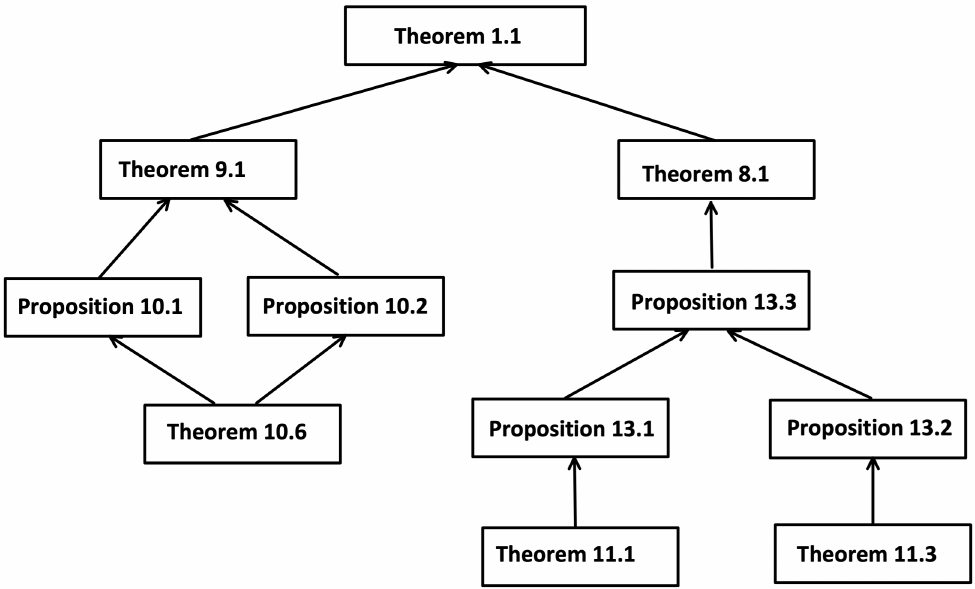}
		\caption{Structure of the Analysis}
	\label{fig:StructureAnalysis}
\end{figure}

\spreadpoint

\newcommand{\FG}{G}

\section{Factor graphs and Gibbs distributions \LastReview{2023-12-19}}\label{sec:FactorGrapsGibbs}

In order to  incorporate in out  analysis  both graphs and hypergraphs, we use the  notion of {\em factor graph}. 
\begin{definition}[Factor graph]
Let $\alphabet$ be the set of {\em spins}, the integer $k\geq2$, while let $\Psi$ be a set of 
weight functions $\psi:\alphabet^k\to\weightrange$.
A {\em $\Psi$-factor graph} $\FG=(V,F,(\partial a)_{a\in F},(\psi_a)_{a\in F})$ consists of
\begin{itemize}
	\item a {finite} set $V$ of {\em variable nodes},
	\item a {finite} set $F$ of {\em factor nodes},
	\item an ordered $k$-tuple $\partial a=(\partial_1a,\ldots,\partial_ka)\in V^k$ for each $a\in F$,
	\item a family $(\psi_a)_{a\in F}\in \Psi^{F}$ of weight functions.
\end{itemize}

\noindent
The {\em Gibbs distribution} of $G$ is the probability distribution on $\alphabet^V$ defined by
\begin{align} \nonumber
\mu_G(\sigma) &=\psi_G(\sigma)/{Z(G)} & \forall \sigma\in\alphabet^V \enspace,
\end{align}
where
\begin{align}\label{eq:GibbsWeights}
\psi_G(\sigma)&= \prod\nolimits_{a\in F}\psi_a(\sigma(\partial_1a),\ldots,\sigma(\partial_ka)) & &\mbox{and } &
Z(G)&= \sum\nolimits_{\tau\in\alphabet^V} \psi_G(\tau) \enspace.
\end{align}
\end{definition}

\noindent
We refer to $Z(G)$ as the {\em partition function}.

The use of the interval $[0,2)$ in the above definition may seem arbitrary.  
However,  this choice allows us to use the weight functions to either reward or penalise   certain value 
combinations and $1$  being the `neutral'  weight.  This is natural in glassy models such as the $k$-spin model. At the same time having an 
explicit upper bound on the  values of $\psi$ it makes some derivations simpler, without harming the 
generality of our results.  We emphasise that the value $0$ corresponds to having {\em hard constraints}.

To see how  the distributions from Section \ref{sec:Applications}   can be cast as  factor graph models 
that satisfy  the above constraints consider the  Potts-Ising model on the graph. 
For  integer $q\geq2$ and a real $\beta>0$, let $\alphabet=[q]$ and
\begin{equation}\label{eqPottsPsi}
\psi_{q,\beta}:(\sigma_1,\sigma_2)\in\alphabet^2\mapsto\exp(-\beta\cdot \Ind\{\sigma_1=\sigma_2\}) \enspace.
\end{equation}
Letting $\Psi$ be the singleton $\{\psi_{q,\beta}\}$,  the Potts model on a given graph $G=(V,E)$ can be cast as a $\Psi$-factor 
graph  as follows: we just set up the factor graph $\widehat{G}=(V,F, (\partial e)_{e\in E},(\psi_e)_{e\in E})$ whose variable nodes 
are the vertices of the original graph $G$ and whose constraint nodes are the edges of $G$.
For an edge $e=\{x,y\}\in E$ we let $\partial e=(x,y)$, where, say, the order of the neighbors is chosen randomly and 
$\psi_e=\psi_{q,\beta}$. The other distributions, apart from the $k$-spin model follow similarly.

For the $k$-spin model we have to argue about the constraint $\psi:\alphabet^k\to\weightrange$,  for every $\psi\in \Psi$.
Recall that for the $k$-spin model we have $\alphabet=\{\pm1\}$. For $J\in\mathbb{R},\beta>0$ we could define the weight function 
$\tilde\psi_{J,\beta}(\sigma_1,\ldots,\sigma_k)=\exp(\beta J\sigma_1\cdots\sigma_k)$ to match 
the definition (\ref{eqkSpin1}) of the $k$-spin model. However, these functions do not necessarily 
take values in $[0,2)$. To remedy this problem, we introduce  $\psi_{J,\beta}(\sigma_1,\ldots,\sigma_k)=
1+\tanh(J\beta)\sigma_1\cdots\sigma_k)$. Then (cf.~\cite{PanchenkoTalagrand}) we have
	\begin{align}\label{eqcosh}
	\tilde\psi_{J,\beta}(\sigma_1,\ldots,\sigma_k)&=\cosh(J\beta)\cdot \psi_{J,\beta}(\sigma_1,\ldots,\sigma_k) \enspace.
	\end{align}
Thus, let $\Psi=\{\psi_{J,\beta}:J\in\mathbb{R}\}$, let $\bpsi=\psi_{\bJ,\beta}$, where $\bJ$ is a standard 
Gaussian.

\begin{figure}
\centering
	\centering
		\includegraphics[width=.3\textwidth]{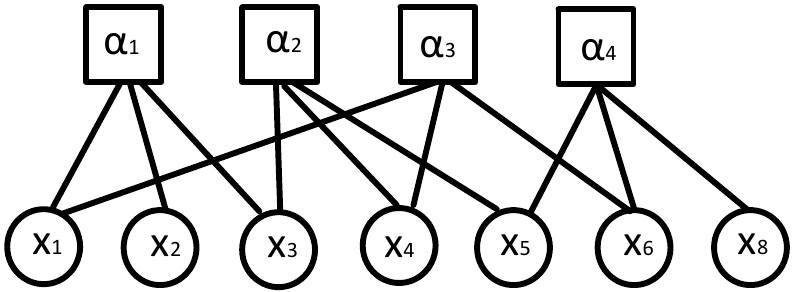}
		\caption{Factor Graph}
	\label{fig:FactorG}
\end{figure}

A $\Psi$-factor graph $G$ induces a bipartite graph with vertex sets $V$ and $F$,  where  $a\in F$ is 
adjacent to $\partial_1a,\ldots,\partial_ka$.  See an example  of a factor graph in Figure \ref{fig:FactorG}. 
We follow the convention to depict the   variable  nodes using cycles and the factor nodes using squares. 
We use common graph-theoretic terminology   and refer to, e.g., the vertices  $\partial_1x ,\ldots,\partial_kx$ 
as the {\em neighbours} of the node $x$. Furthermore,  the length of shortest paths in the bipartite graph induces a 
metric on the nodes of $G$.

\subsection{The random factor graph and its Gibbs distribution}

Here, we consider Gibbs  distributions on the  random factor graph.
%,
To define these concepts formally  we observe that any weight function 
$\psi:\alphabet^k\to[0,2)$ can be viewed as a point in $|\alphabet|^k$-dimensional Euclidean space.
We thus endow the set of all possible weight functions with the $\sigma$-algebra induced by the Borel algebra.
Further, for a weight function $\psi:\alphabet^k\to[0,2)$ and a permutation $\theta:\{1,\ldots,k\}\to\{1,\ldots,k\}$ 
we define $\psi^\theta:\Omega^k\to [0,2)$, $(\sigma_1,\ldots,\sigma_k)\mapsto\psi(\sigma_{\theta(1)},\ldots,\sigma_{\theta(k)})$.
{\em Throughout the paper, we assume that $\Psi$ is a measurable set of weight functions 
such that for all $\psi\in\Psi$ and all permutations $\theta$ we have $\psi^\theta\in\Psi$.}

We fix a probability distribution $\dpsi$ on $\Psi$. We always denote by $\bpsi$ an element of 
$\Psi$ chosen from $\dpsi$, and we set
\begin{equation}\label{eq:DefOfChi}
	q=|\alphabet| \quad\mbox{and}\quad \chi=q^{-k} \cdot  \sum\nolimits_{\sigma\in\alphabet^k}  \mathbb{E}[\bpsi(\sigma)] \enspace.
\end{equation}
For the factor graph $G$, we let $\psimin=\psimin(G)$ be the minimum value of $\psi_{\alpha}(\tau)$, where $\tau$ varies over
the support of $\psi_{\alpha}$ and $\alpha$ varies over the set of factor nodes $F$ in $G$,  that  is, 
\begin{align}\label{eq:DefOfPsiMin}
\psimin=\min_{\alpha}\min_{\tau}\{\psi_{\alpha}(\tau)\ |\ \psi_{\alpha}(\tau) >0 \}\enspace.
\end{align}

With the above conventions in mind, for  $n,m>0$  integers,  we define the random $\Psi$-factor 
graph  $\G=\G(n,m, k, \dpsi)$ as follows: the set of variable nodes is $V_n=\{x_1,\ldots,x_n\}$, the set of 
constraint nodes is  $F_m=\{a_1,\ldots,a_m\}$ and the neighbourhoods  $\partial a_i\in V_n^k$ are chosen 
uniformly and independently for $i\in [m]$.  Furthermore, the weight functions $\psi_{a_i}\in\Psi$ are 
chosen  from the distribution $\dpsi$  mutually independently and independently of the neighbourhoods  
$(\partial a_i)_{i=1,\ldots,m}$.

In this work we focus on the cases where $m=\Theta(n)$. Particularly, we assume that there is a fixed number
$d\geq 1/(k-1)$ such that $m=dn/k$. Note that $d$ is the expected degree of the variable nodes of $\G$.  We assume
that $d\geq 1/(k-1)$ because otherwise the structure of $\G$ is, typically,  very simple and the sampling problem becomes trivial. 
\\ \vspace{-.2cm}

\noindent
{\em Symmetric Gibbs distributions:}
Throughout this work, we assume that we are dealing with a $\Psi$-factor graph $\G$ which gives rise to 
a {\em symmetric} Gibbs distribution $\mu_{\G}$. For $\mu_{\G}$ to be symmetric, each $\psi\in \Psi $ 
satisfies the following conditions: 
\begin{description}
\item[SYM-1] For any two element set $\DisSpin \subseteq \alphabet$ and  for $\sigma,\tau \in \alphabet^k$
such that 
\begin{align} \label{eq:SymmetricWeightA}
\tau(i) &= \left\{
\begin{array}{lcl}
\sigma(i) & & \textrm{if } \sigma(i)\notin \DisSpin\\
\DisSpin\setminus\{\sigma(i)\} && \textrm{otherwise}
\end{array}
\right . & \forall i\in [k] \enspace,
\end{align}
we have that $\psi(\tau)=\psi(\sigma)$.
\item[SYM-2] 
For all $i\in [k]$ and $s \in \alphabet$  we have 
		\begin{equation}\label{eq:SymmetricWeightAB}
		 \sum\nolimits_{\tau\in\alphabet^k }\Ind\{\tau_i=s\}\psi(\tau)=q^{k-1} \cdot \chi
		\end{equation}
		and for every permutation $\theta$ and every measurable $\mathcal{M}\subset \Psi$ we have 
		$\dpsi(\mathcal{M})=\dpsi(\{\psi^\theta:\psi\in\mathcal{ M}\})$.
\end{description}

\subsubsection*{Disagreement Rate:} 
Let the random $\Psi$-factor graph $\G(n,m,k, \dpsi)$. For each
factor node $\alpha$, let the distribution $\bethe_{\alpha}$ be defined by 
\begin{align}\label{eq:BetheVsWeight}
\bethe_{\alpha}(\eta)&\propto \psi_{\alpha}(\eta)  
&\forall \eta\in \alphabet^V \enspace.
\end{align} 
We define the disagreement rate at the factor node $\alpha\in F$ such that 
\begin{equation}\label{eq:RateDisDef} 
\drate_{\alpha}=
\mathbb{E}\left[\max\nolimits_{\sigma, \tau\in \alphabet^{\partial \alpha}} 
|| \bethe_{\alpha} (\cdot\ |\ \partial_1 a, \sigma)-\bethe_{\alpha}(\cdot \ |\ \partial_1\alpha, \tau) ||_{\{\partial_{>1}\alpha\}}\right] \enspace,
\end{equation}
where the set $\partial_{>1}\alpha=\partial \alpha\setminus \{\partial_1 \alpha\}$.  The expectation is with respect to  the 
randomness of the weight function  
$\psi_\alpha$ in $\bethe_{\alpha}$ which is distributed as in $\dpsi$.

\subsubsection*{Teacher-Student model \& Contiguity:}
Typically the structure of the random $\Psi$-factor graph $\G$ is quite complex. This poses formidable 
challenges in the study of  the Gibbs distribution  $\mu_{\G}$. A natural way of accessing   $\mu_{\G}$ 
is by means of the so-called {\em teacher-student} model \cite{LF} and the notion of {\em mutual contiguity}. 

Suppose that $\sigma:V_n\to \alphabet$. We introduce  a random factor graph $\G^*(n,m,\dpsi,\sigma)$ 
with variable nodes $V_n$ and factor nodes $F_m$ such that, 
{\em independently} for each $j=1, \ldots, m$, the neighbourhood 
$\partial \alpha_j$ and the weight function 
$\psi_{\alpha_j}$ are chosen from the following joint distribution: for any 
$y_1, \ldots, y_k\in V_n$ 
and any measurable set $\mathcal{S}\subseteq \Psi$ we have
%%%
\begin{align}\label{def:Probs4GStar}
\Pr[\partial \alpha_j =(y_1, \ldots, y_k), \psi_{\alpha_j}\in \mathcal{S}] &=
\frac{\mathbb{E}\left[ \Ind\{\bpsi\in \mathcal{S}\} \cdot \bpsi(\sigma(y_1), \sigma(y_2), \ldots,\sigma(y_k)) \right]}
{\sum_{z_1, \ldots,z_k\in V_n}\mathbb{E}\left[\bpsi(\sigma(z_1), \sigma(z_2), \ldots,\sigma(z_k))\right]} \enspace.
\end{align}
The independence of the individual factor nodes implies  that
\begin{align}\label{def:OfGStar}
\Pr[\G^*(n,m,\dpsi,\sigma)=G] & =\frac{\psi_G(\sigma)}{\mathbb{E}[\psi_{\G(n,m,k,\dpsi)}(\sigma)] }\Pr[\G(n,m,k,\dpsi)=G] \enspace.
\end{align}

\noindent
The teacher-student model is a distribution over factor-graph/configuration pairs induced by the following 
experiment
\begin{description}
\item[TCH1] choose an assignment $\bsigma^*:V_n\to \alphabet$,   the ``ground truth", uniformly at random,
\item[TCH2] generate $\G^*=\G^*(n,m,\dpsi,\bsigma^*)$.
\end{description}
We say that the pair $(\G^*, \bsigma^*)$ is distributed as in the teacher-student model.

We can use the teacher-student model to investigate the typical properties of the Gibbs samples 
of the factor graph $\G$ by means of a well-known technique called ``quite planting" \cite{AchCoOg,SilentPlanting,quiet}.
 This idea has been used critically in rigorous works on specific examples of random factor graph models, 
 e.g., \cite{AchCoOg,AneFriez,CoEf,MolloyJACM}.

Formally, quiet planting applies if the factor graph/assignment pair $(\G^*,\bsigma^*)$ comprising the 
ground truth $\bsigma^*$ and  the outcome $\G^*=\G^*(n,m,\dpsi,\bsigma^*)$ of ${\bf TCH1}$ -- ${\bf TCH2}$ 
and the pair $(\G,\bsigma)$ consisting of the random $\Psi$-factor graph $\G=\G(n,m,k,\dpsi)$ and a Gibbs 
sample $\bsigma $ of $\G$ are {\em mutually contiguous}.   We say that $(\G^*,\bsigma^*)$ and  
$(\G,\bsigma)$ are mutually contiguous if  for any sequence  of events $(\mathcal{S}_n)_n$, we have  
\begin{align}\label{Def:GeneralMutualContiguity}
\lim_{n\to\infty}\Pr[(\G^*,\bsigma^*)\in \mathcal{S}_n]  &=0 &&\textrm{iff}&\lim_{n\to\infty}\Pr[(\G,\bsigma)\in \mathcal{S}_n]=0 \enspace.
\end{align}
We make an extensive use of the notion of contiguity. However, we use a more quantitative   version that what is stated above, e.g., 
see condition $\CB$ in the following section.

\spreadpoint
\section{The Conditions in  $\setB$ \LastReview{2024-02-02}}
\label{sec:Conditions}

In this section, we define precisely  the set of conditions for the Gibbs distribution $\setB$, which implies that the desirable approximation guarantees for our algorithms.

Consider the random $\Psi$-factor graph $\G=\G(n,m,k, \dpsi)$ of expected degree $d$, 
such that    the Gibbs distribution  $\mu_{\G}$ is symmetric.    
Consider, also,   the standard sequence $\G_0, \ldots, \G_m$ generated from $\G$.
Furthermore,   for any  $\omega_n=\omega(n)$    let $\cC_i(\omega)$  be  the event  that $\log  Z(\G_i)\geq \log \mathbb{E}[Z(\G_i)]-\omega$, where  $i=0, \ldots, m$.

\begin{definition}[$\setB$]
  We say that $\G_0, \ldots, \G_m$ satisfies $\setB$ with slack  $\delta\in [0,1]$, 
if the following hold:
\begin{description}
\item[$\CA$]  For each factor node $\alpha$  in $\G$, 
we have that $\drate_{\alpha} \leq \frac{1-\delta}{d(k-1)}$.
\item[$\CB$]  For any $\omega_n \to\infty$,  we have that 
$\Pr\left[\wedge_{t\in [m]}\cC_t(\omega_n)\right]=1-o(1)$. Furthermore,  
for  any sequence of events $(\mathcal{S}_n)_n$, we have
\begin{align}
\Pr[(\G_i,\bsigma)\in \mathcal{S}_n\ |\ \cC_i (\omega_n)]  &\leq \omega_n \cdot \Pr[(\G^*_i,\bsigma^*)\in \mathcal{S}_n] \enspace,
\end{align}
where $\bsigma$ is distributed as in the Gibbs distribution on $\G_i$ and  $(\G^*_i,\bsigma^*)$ is generated according to the teacher-student model, with $i$  factor nodes in $\G^*_i$, where $i=0, \ldots, m-1$.
\item[$\CC$]  The distribution  $\dpsi$ satisfies the following: 
for any $\delta>0$ and for $\bpsi\sim \dpsi$, we have that   
\begin{align}\label{eq:TailDPsi}
\Pr {\textstyle \left [\exists \tau\in \textrm{support($\bpsi$)} \ \textrm{s.t.}\  
\bpsi(\tau) \leq {\textstyle n^{-\left( \frac{\delta}{\log dk}\right)^{10}}} 
\right ] } &\leq n^{-3/2} \enspace.
\end{align}
\end{description}
\end{definition}

Let us remark that, for sufficiently large $n$, the $k$-spin model satisfies \eqref{eq:TailDPsi}.

\subsection{The region of $\setB$}\label{sec:RegionSet}
Let the random $\Psi$-factor graph   $\G=\G(n,m,k, \dpsi)$ be of expected degree $d$ such that  the Gibbs distribution  $\mu=\mu_{\G}$ is one of the distributions  in  Section \ref{sec:Applications}.  
A natural question is whether we can obtain  a simpler characterisation of the region of $\mu_{\G}$ that satisfies  $\setB$. Admittedly, the   conditions  in $\setB$ do not seem to be so related to each other.

We  focus on $\CA$ and $\CB$ as $\CC$ is only a crude tail  bound for $\bpsi$ 
which  is mildly restrictive. We provide a series of arguments that allow us to 
compare the strength of the two
conditions $\CA$ and $\CB$. The considerations we use  get us deep  into the theory 
phase transitions for random Constraint Satisfaction Problems. We mainly use 
notions and results from \cite{CoEfCMP,CKPZ,PNAS}.

We start by considering  $\CB$. 
We show below, i.e., see \cref{thrm:ContiquitySeq}, that the condition $\CB$  holds within the {\em replica symmetry} region of  the Gibbs distribution $\mu_{\G}$.  
To be more specific,   define $d_{\rm cond}$ to be the infimum over the expected degrees $d>0$ such that 
\begin{align}\label{def:DCond}
\limsup_{n\to\infty}n^{-2}\times 
{ \mathbb{E} \left[ \sum\nolimits_{x_1, x_2\in V_n} || \mu_{\{x_1, x_2\}} -\mu_{x_1}\otimes \mu_{x_2} ||_{\rm tv}\right]}>0 \enspace. 
\end{align}
Having expected degree  $d<d_{\rm cond}$, corresponds to being in the replica symmetric region for $\mu_{\G}$.
Intuitively, replica symmetry implies the following:
for  ${\bf x_1}, {\bf x_2}$,  two randomly chosen variable nodes in $\G$, 
 {\em typically}, the marginal of $\mu$ at these two vertices, i.e.,  $\mu_{\{ {\bf x_1}, {\bf x_2}\}}$, 
 is very close to  the product measure with marginals $\mu_{\bf x_1}$ and $\mu_{\bf x_2}$, i.e., the 
 measure $\mu_{\bf x_1}\otimes \mu_{\bf x_2}$. Here   ``typically" refers  to the choice of both 
  graph instances $\G$ and  pairs of  variable nodes ${\bf x_1}, {\bf x_2}$.

The critical value $d_{\rm cond}$ signifies  the so-called   
{\em condensation phase transition}.
In practice,  $d_{\rm cond}$ is  relatively large, e.g., for the $q$-colourings, or the 
NAE-$k$-SAT,  the corresponding value of $d_{\rm cond}$ is very close to 
the {\em satisfiability threshold} e.g. see \cite{NAEThresh,CondColor}.

\cref{thrm:ContiquitySeq}, below, establishes the connection between
$\CB$ and the replica symmetry. 
 
\begin{theorem}\label{thrm:ContiquitySeq}
Let the $\Psi$-random factor graph $\G(n,m,k,\dpsi)$, of expected  degree $d>0$, give rise to the Gibbs distribution
$\mu_{\G}$ which is  any of  distributions   in  Section \ref{sec:Applications}.
For  $d_{\rm cond}=d_{\rm cond}(k,\dpsi)>0$,  if   $0<d<d_{\rm cond}$, then  
$\mu_{\G}$  satisfies $\CB$
\end{theorem}

\noindent
The proof of Theorem \ref{thrm:ContiquitySeq} appears in Section \ref{sec:thrm:ContiquitySeq}.

We now focus on $\CA$. We argue that when this condition holds, then 
we also have replica symmetry.

To this end,  
we focus on a different property  of the Gibbs distribution $\mu_{\G}$, called non-reconstruction. That is, 
\begin{align}\label{eq:NonReconDef}
\lim_{h\to\infty} n^{-1}\cdot \sum\nolimits_{v\in V}\mathbb{E}\left[  \max\nolimits_{\sigma, \tau} ||\mu(\cdot \ |\ v, \sigma)-\mu(\cdot\ |\ v, \tau) ||_{\{S_{v,2h}\}} \right]=0\enspace, 
\end{align}
where $S_{v,2h}$ is the set of variable nodes at distance $2h$ from $v$. 
It turns out that, if $\CA$ holds, then  \eqref{eq:NonReconDef} holds,
as well.

Eq. \eqref{eq:NonReconDef}  implies that for  a typical variable node $v\in V$, 
under the Gibbs distribution $\mu_{\G}$,  the configuration of the variable 
nodes at the sphere of radius $2h$ around $v$, is asymptotically independent 
of that at $v$. 
It is standard that non-reconstruction implies replica symmetry, as two randomly
chosen variable nodes in $\G$ are typically far apart. Hence, we have that  non-reconstruction is  the stronger condition. 

For $\delta>0$, and the Gibbs distribution induced by $\G(n,m,k,\dpsi)$,
let $d_{\rm  BC}=d_{\rm BC}(\delta, k,\dpsi )$  be  the infimum over the 
expected degree $d$ such that $\mu_{\G}$ does {\em not} satisfy $\CA$ with slack 
$\delta$.   Our  aim is to show that $d_{\rm  BC}\leq d_{\rm cond}$. 

The line of arguments towards showing this inequality is as follows:
Note that  for any $d<d_{\rm  BC}$ we have $\CA$. We then argue that $\CA$ implies 
non-reconstruction for $\mu_{\G}$.  In turn, the non-reconstruction implies replica symmetry
and this concludes that $d_{\rm  BC}\leq d_{\rm cond}$.

The connection between $\CA$ and non-reconstruction is not straightforward, since
$\CA$ is about distributions $\bethe_{\alpha}$ and non-reconstruction is about $\mu_{\G}$. 
We show the aforementioned inequality in the following result. 

\begin{theorem}\label{thrm:DbcVsDcont}
For $\delta>0$,  let the $\Psi$-random factor graph $\G(n,m,k,\dpsi)$, of expected  
degree $d>0$, give rise to the Gibbs distribution $\mu_{\G}$ which is  any of  
distributions   in  Section \ref{sec:Applications}.

For  $d_{\rm BC}=d_{\rm BC}(\delta,k,\dpsi)>0$ and $d_{\rm cond}=d_{\rm cond}(k,\dpsi)>0$
we have that $d_{\rm BC}\leq d_{\rm cond}$.
\end{theorem}
The proof of Theorem \ref{thrm:DbcVsDcont} appears in Section \ref{sec:thrm:DbcVsDcont}.

\spreadpoint
\section{The Sampling Algorithm \LastReview{2023-11-21}}\label{sec:WSAlgorithm}

In this basic description of the algorithm, consider a fixed 
$\Psi$-factor graph $G=(V,F,(\partial a)_{a\in F},(\psi_a)_{a\in F})$, while assume that 
that the weight functions $(\psi_a)_{a\in F}$  are fixed and give rise to the {\em symmetric} 
Gibbs distribution $\mu_G$. Also,  assume that  each factor node is of degree  $k\geq 2$,
while assume that $G$  is of large girth.

Initially, the algorithm creates the sequence  $G_0, \ldots G_m$, where $m=|F|$.
The sequence is such that $G_0$ has no factor nodes, i.e., it only has isolated variable nodes, 
while $G_m$ is identical to $G$. Furthermore, any  two consecutive terms $G_{i}$ and $G_{i+1}$ differ 
in that $G_{i+1}$ has the extra factor node   $\alpha_i\in F$. Assume that $\alpha_i$ is an arbitrary 
factor node in $G_{i+1}$. 
Let $\mu_i$ be the Gibbs distribution that corresponds to $G_i$.

For each $G_i$, the algorithm generates the configuration $\bsigma_i$ which is distributed close 
to $\mu_i$.  The output of the algorithm is   the configuration  $\bsigma_m$.  
The algorithm generates each configuration $\bsigma_{i+1}$ 
by using  $\bsigma_{i}$.
As far as $\bsigma_0$ is concerned,   the algorithm  generates it  
by setting, independently,  for each variable node $x\in V$,
\begin{align}\label{eq:GenerateSIGMA0}
\bsigma_{0}(x) &=\textrm{a uniformly random element of $\alphabet$} \enspace.
\end{align}
Assume that we have  $\bsigma_i$ and we want to generate   $\bsigma_{i+1}$.
As a first step, the algorithm decides  $\bsigma_{i+1}(\partial \alpha_i)$, recall 
that $\partial \alpha_i$ is  the set of variable nodes which are  attached to the
new factor node $\alpha_i$. We set $\bsigma_{i+1}(\partial \alpha_i)$ according to the
following distribution
\begin{align}\label{eq:FirstStep}
\Pr[\bsigma_{i+1}(\partial \alpha_i)=\tau] &= \bethe_{\alpha_i}(\tau) %\propto \psi_{\alpha_i}(\tau) 
\qquad \qquad  \forall \tau\in \alphabet^{\partial \alpha_i} \enspace.
\end{align}
Note that we choose $\bsigma_{i+1}(\partial \alpha_i)$ from a distribution which is different than 
the marginal  of $\mu_{i+1}$  at  $\partial \alpha_i$. 

Setting $\bsigma_{i+1}(\partial \alpha_i)$ as in \eqref{eq:FirstStep} we expect to have vertices
$x\in \partial \alpha_i$ such that $\bsigma_{i+1}(x)\neq \bsigma_{i}(x)$. 
Assume that after the step in \eqref{eq:FirstStep}, we have $\bsigma_i(\partial \alpha_i)=\eta$
and $\bsigma_{i+1}(\partial \alpha_i)=\kappa$.  
Let $\initDis=\{x_1,\ldots, x_{\ell}\}$ contain every variable node in $\partial \alpha_i$ 
at  which $\eta$ and $\kappa$ disagree.   Considers a sequence 
of configurations  $\eta_0,  \ldots, \eta_{\ell}$  at $\partial \alpha_i$ such that  $\eta_0=\eta $, 
$\eta_{\ell}=\kappa$, while each  $\eta_j$ is obtained  from $\eta$ by  changing the assignment of the 
 nodes $z\in \{x_1,\ldots, x_j\}$ from  $\eta(z)$ to $\kappa(z)$.  Then,  apply the following
iteration: for $\btau_0=\bsigma_i$, set 
\begin{align}\label{eq:UpdateIteration}
\btau_j& =\switch(G_i, \btau_{j-1}, \eta_{j-1}, \eta_j) &  \textrm{for }j=1,\ldots, \ell \enspace.
\end{align}
At the end of the above iteration, we set $\bsigma_{i+1}=\btau_{\ell}$.

For two configurations $\tau_x$ and  $\hat{\tau}_x$ be two  configurations at $\partial \alpha_i$ that  differ  in {\em exactly one}  
the variable node  $x\in \partial \alpha_i$,   for $\btau$ being distributed  as in  $\mu_i(\cdot \ |\ \partial \alpha_i, \tau_x)$,
the  process $\switch(G_i, \btau, \tau_x, \hat{\tau}_x)$  generates  a configuration which will be distributed 
very close to $\mu_i(\cdot \ |\ \partial \alpha_i, \hat{\tau}_x)$.
 We describe the details of this process in  \cref{sec:WSUpdate} that follows.

The  pseudocode in Algorithm \ref{MySampler} is a synopsis of  the above.

\begin{algorithm} 
\caption{$\fixsampler$}\label{MySampler}
\begin{algorithmic}[1]
\Require graph $G$
\State ${\tt compute}$ $G_0, \ldots, G_m$
\State ${\tt set}$ $\bsigma_0$ according to \eqref{eq:GenerateSIGMA0}
\For{ $\ i=0, \ldots, m-1$ }
\State ${\tt set}$ $\bsigma_{i+1}(\partial \alpha_i)$ according to \eqref{eq:FirstStep}
\State ${\tt generate}$ $\eta_0, \eta_1, \ldots, \eta_{\ell}$ w.r.t. $\bsigma_{i+1}(\partial \alpha_i)$ and $\bsigma_{i}(\partial \alpha_i)$
\State $\btau_0\gets \bsigma_i$
\For {$j=1, \ldots, \ell $} 
\State $\btau_j \gets \switch(G_i, \btau_{j-1}, \eta_{j-1}, \eta_j)$
\EndFor
\State $\bsigma_{i+1}\gets \btau_{\ell}$
\EndFor
\Ensure $\bsigma_m$
\end{algorithmic}

\end{algorithm}

\subsection{The process $\switch$}\label{sec:WSUpdate} 

\newcommand{\bbeta}{\mathbold{\eta}}
\newcommand{\transprob}{\UpP}
\newcommand{\SWG}{{G}}
\newcommand{\SWIn}{\bsigma}
\newcommand{\SWOut}{\btau}

To avoid complex notation with many   indices, the description of $\switch$ is disentangled from the 
description of the algorithm in the previous section. 

We  consider  the
$\Psi$-factor graph $\SWG$ of large  girth $g$. Let  $\mu=\mu_{G}$ be  the corresponding Gibbs 
distributions which  assume that is symmetric.  Consider a (small) set $\Lambda$ of distant variable 
nodes in $G$ and let $\eta, \kappa \in \alphabet^{\Lambda}$ which  differ only on the assignment of 
a {\em single} variable node $x\in  \Lambda$.
We   consider the process  $\switch(G, \bsigma, \eta, \kappa)$, where $\bsigma$ is distributed as in 
$\mu(\cdot \ |\ \Lambda, \eta)$. We let $\btau$ be the configuration at the output of the 
process, while  $\nu_{\eta,\kappa}$ denotes the distribution of $\SWOut$.

$\switch(G, \bsigma, \eta, \kappa)$ is an {\em iterative} process. It starts from $x$, the disagreeing node between
$\eta$ and $\kappa$,  and iteratively visits   nodes of the graph. It  uses the sets of  nodes $\visit$ and $\DisG$.
At each iteration, $\visit$ contains the nodes (variable and factor) that the process has visited. Hence, 
the process  has specified what   $\btau$ is for  the variable nodes in $\visit$. The set 
$\DisG \subseteq \visit$ contains all the {\em disagreeing} variable nodes in $\visit$, i.e., every  $z\in \visit$ 
such  that $\btau(z)\neq \bsigma(z)$.  

Initially, we set $\btau(\Lambda)=\kappa$, while  $\visit=\{\Lambda \}$ and $\DisG=\{x\}$. 
We let $\DisSpin=\{\eta(x), \kappa(x)\}$, i.e., $\DisSpin$  contains the spins of the disagreement of 
$\kappa$ and $\tau$.  At iteration $t$, we choose a factor node $\beta\notin \visit$ which is adjacent 
to a variable node in $\DisG$. If $\partial \beta$ contains more than one variable nodes whose 
configuration under $\btau$ is known,  then we consider that  $\switch$  {\em fails} and the process 
terminates. Otherwise,  $\switch$ decides on the assignment under $\btau$  for the remaining  nodes in 
$\partial \beta$. 
 W.l.o.g. assume that $\partial_1 \beta$ is the single  node   in $\partial \beta$ whose configuration under 
 $\btau$ is already known. The process decides on the assignment of $\partial_{>1} \beta=\{\partial_2\beta, \ldots, \partial_k\beta\}$ as follows:
With probability $1-\dpr_{\beta}$, it sets 
\begin{align}\label{eq:UpdateRuleA}
\btau(\partial_r \beta) &=\bsigma(\partial_r \beta) & \textrm{for each } r=2,3, \ldots, k \enspace.
\end{align}
With the complementary probability, i.e., with probability $\dpr_{\beta}$, it sets 
\begin{align}\label{eq:UpdateRuleB}
\btau(\partial_r \beta) &=
\left \{
\begin{array}{lcl}
\DisSpin \setminus \{\bsigma(\partial_r \beta)\} && \textrm {if $\bsigma(\partial_r \beta)\in \DisSpin$} \\
\bsigma(\partial_r \beta ) && \textrm {otherwise}, 
\end{array}
\right . & \textrm{for } r=2,3, \ldots, k \enspace. 
\end{align}
 
 \noindent
The probability $\dpr_{\beta}$ is defined by  
\begin{equation}\label{eq:BroadCastDisagreementProb}
\dpr_{\beta} =\max \left \{0, 1-\frac{\bethe_{\beta}(\hat{\bsigma}(\partial \beta)\ |\ \partial_{1}\beta, \btau)}
{\bethe_{\beta}(\bsigma(\partial \beta )\ |\ \partial_{1}\beta, \bsigma)} \right\} \enspace,   
\end{equation}
where  the configuration $\hat{\bsigma}$ is such that $\hat{\bsigma}(\partial_1 \beta)=\btau(\partial_1 \beta)$, while
for any $j\neq 1$ we have $\hat{\bsigma}(\partial_j \beta)=\bsigma(\partial_j \beta)$.

 \cref{fig:UpdateRuleB} shows an example where  factor node $\beta$ is  updated according to 
\eqref{eq:UpdateRuleB}. The top configuration is $\bsigma$ and the bottom is $\btau$. Note that
all the assignments that are not in $\DisSpin$ remain the same, while the assignments  $\bsigma(x)\in \DisSpin$ switch, 
from blue to green and the other way around.

After having decided  $\btau(\partial \beta)$, the  process  updates the sets $\visit$ and $\DisG$, 
appropriately.  That is, it inserts into $\visit$ the factor node $\beta$ and the variable nodes  $ \partial \beta$. 
Also, it inserts into $\DisG$ all  the disagreeing nodes from $\partial_{>1} \beta$. This concludes the iteration $t$.
 \begin{figure}
\centering
	\centering
		\includegraphics[width=.35\textwidth]{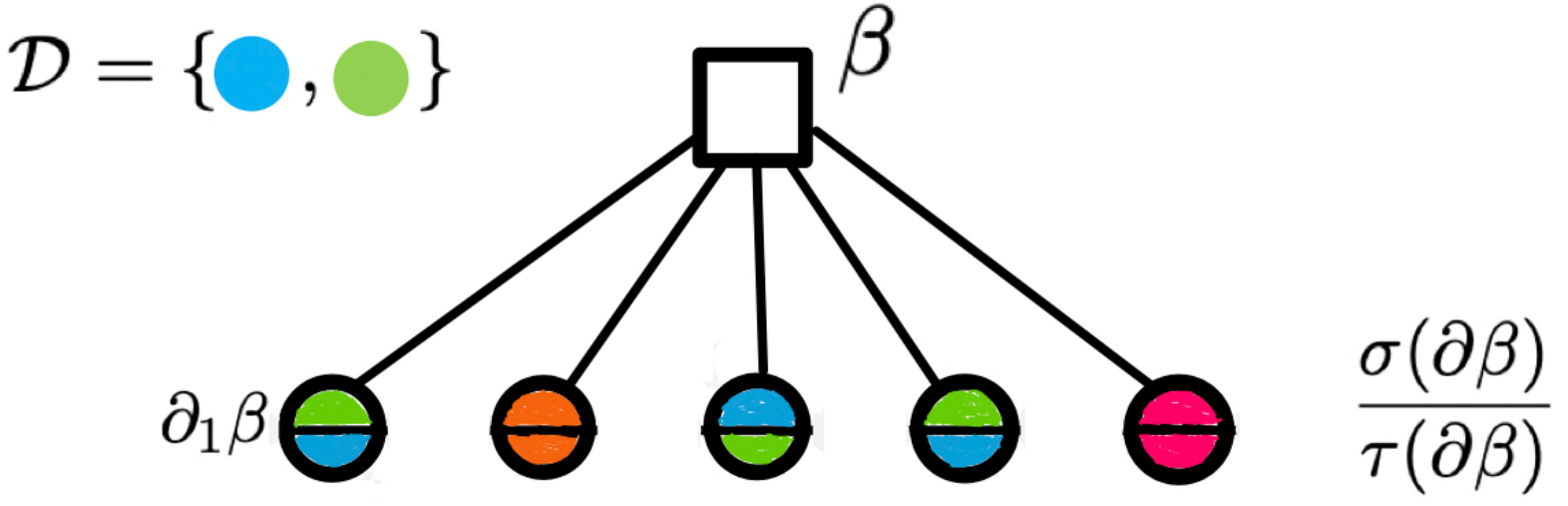}
		\caption{Configuration $\btau(\beta)$ when \eqref{eq:UpdateRuleB} applies.}
	\label{fig:UpdateRuleB}
\end{figure}

At the beginning of an iteration, it could be that for all factor nodes $\beta\notin \visit$, we have
that   $\partial \beta \cap \DisG=\emptyset$.  In this case,  the iterations stop.  However,  there
can be  variable nodes whose assignment under $\btau$ is not specified. If this is the case, for each 
variable node $z$ for which $\btau(z)$ is not known the process sets
\begin{align} \nonumber 
\btau(z)&=\bsigma(z) \enspace.  
\end{align}
After the above step, $\switch$ returns $\btau$.

The pseudo-code in Algorithm \ref{Myswitch}  is a synopsis of the above description of $\switch$.
Note that,  we let $\eta\oplus \kappa$ denote the set of variable nodes on
which the configurations $\eta,\kappa$ disagree. 

\begin{algorithm} 
\caption{\switch}\label{Myswitch}
\begin{algorithmic}[1]
\Require $G$, $\sigma$, $\eta$, $\kappa$ 
\State $\tau(\Lambda)\gets\kappa$
\State $\visit\gets \Lambda$ and $\DisG\gets\eta\oplus \kappa$ 
\While{there is factor node $ \beta\notin \visit$ such that $\partial \beta\cap  \DisG\neq \emptyset$}
\If{$|\partial \beta \cap \visit| >1$}
\State \Return $\tt Fail$
\EndIf
\State $M\gets \partial \beta\setminus\visit$
\State ${\tt set}$ $\theta(M)$ according to \eqref{eq:UpdateRuleB}
\State ${\tt set}$ $\tau(M)$ such that 
\begin{align} \nonumber
 \textstyle\tau(M) &\gets \left \{
\begin{array}{lcl}
\sigma(M) && \textrm{w.p. } 1- \dpr_{\beta}\\
\theta(M) && \textrm{w.p. } \dpr_{\beta}
\end{array}
\right .  & 
\end{align}
\State $\visit \gets \visit\cup M\cup\{\beta\}$ 
\State $\DisG\gets \DisG\cup (\tau(M)\oplus \sigma(M))$
\EndWhile
\For{every  $z$ such that  $\tau(z)$ is not specified}
\State $\tau(z)\gets \sigma(z)$
\EndFor
\Ensure $\tau$
\end{algorithmic}
\end{algorithm}

\newcommand{\BadWSUpdate}{ \mathcal{B}}
\newcommand{\Fvisit}{\mathcal{M}}
\newcommand{\maxq}{\mathcal{Q}}

\subsubsection*{Performance of $\switch$}
We now study the accuracy of $\switch$.   The accuracy of this process is closely related to the 
failure probability.  Let
\begin{equation} \nonumber 
\maxq_x = { \max\nolimits_{\theta,\xi} } \Pr[\textrm{$\switch(G,\bbeta, \theta, \xi)$ fails}] \enspace, 
\end{equation}
where $\theta, \xi$  vary over configurations of $\Lambda$ which differ on %a single variable node  
$x\in \Lambda$, while $\bbeta$ is distributed as in $\mu(\cdot\ |\ \Lambda, \theta)$.

\begin{lemma}\label{lemma:UpdateWSAccuracy}
We have that 
\begin{equation} \nonumber
 ||\mu(\cdot \ |\ \Lambda, \kappa ) -\nu_{\eta,\kappa} ||_{\rm tv} \leq    7 |\Lambda| \cdot |\alphabet|^{|\Lambda|} \cdot \maxq_x
 \enspace, 
\end{equation}
where $\nu_{\eta,\kappa}$is the distribution of   the output of  $\switch(\SWG, \bsigma, \eta, \kappa)$. 
\end{lemma}

\noindent
Lemma \ref{lemma:UpdateWSAccuracy} is a special case of Proposition \ref{prop:Error4RUpdate}. For a proof 
of Lemma \ref{lemma:UpdateWSAccuracy} we refer the  reader to  this result. 

\begin{lemma}\label{lemma:UpdateWSTimeComplexity}
The time complexity of the process $\switch(\SWG, \bsigma, \eta, \kappa)$ is $O(m+n)$,  where 
$m$, $n$ are the numbers of factor and variable nodes in $\SWG$,  respectively.
\end{lemma}

\begin{proof}
 We make some standard   assumptions  about the representation of  the input.  Assume that we can iterate 
over the  nodes in $G$ in time $O(n)$.  Also, for each $z\in V\cup F$, we can iterate over its neighbours in  time 
$O({\rm deg}(z))$.  Furthermore, for each $z\in V$ we can  access  the configuration $\bsigma(z)$ and $\btau(z)$
in time $O(1)$. Finally, we assume that each of the operations  in \eqref{eq:UpdateRuleA} and \eqref{eq:UpdateRuleB} 
can be implemented in $O(1)$ time.

Consider first the iterative part of $\switch$.  Assume that we have a queue $\mathbold{S}$ of 
factor nodes.  Initially, $\mathbold{S}$ only contains  $\partial x$. Recall that $x$ is the  
node on which the  configurations $\eta, \kappa \in \alphabet^{\Lambda}$  disagree.  At each 
iteration, the algorithm pops $\beta$, the element at the top of $\mathbold{S}$, and updates the 
configuration of $\partial \beta$ appropriately. Then, if there is no failure, the algorithm pushes into 
$\mathbold{S}$ the neighbours of every disagreeing node in $\partial \beta$.

Each factor node can only be pushed and popped $O(k)=O(1)$ times during the execution of the algorithm, i.e., it
can be pushed at most as many times as its degree. Each one of these pushes and pops requires $O(1)$ time.
Furthermore, once we pop from $\mathbold{S}$ the factor node $\beta$ we need $O(k)$ 
steps to decide the configuration of $\partial \beta$.  Hence, the algorithm spends $O(k^2)$ time
for each factor node $\alpha$ of $G$.  Then, since we have  $m$ factor nodes,  we conclude  that the iterative 
part  of $\switch$ requires $O(k^2 m)=O(m)$ steps.

Deciding the assignment under $\btau$ for the variable nodes in $V\setminus \visit$, i.e., after the iterative part,  
requires $O(n)$ steps. Recall that we can check each variable node $z$  
if $\btau(z)$ is set in $O(1)$. If $\btau(z)$ is not set, we  have  $\btau(z)=\bsigma(z)$ in $O(1)$ steps.

The above implies that $\switch(G, \bsigma, \eta,\kappa)$ requires $O(m+n)$ steps
when it does not fail. The lemma follows by noting that when $\switch$ fails, the number 
of steps is less than $O(m+n)$. 
\end{proof}

\subsection{Performance of $\fixsampler$}\label{sec:SamplerWS}

\newcommand{\GoodEv}{\mathcal{ G}}
 \newcommand{\ball}{{\rm Ball}}

We use the results from Section \ref{sec:WSUpdate} to describe the performance of the algorithm
both in terms of accuracy and time efficiency. We start with the accuracy.
Let
$$
 \maxq_i= \sum\nolimits_{x\in \partial  \alpha_i} \max\nolimits_{\kappa_x,\eta_x}   \Pr[\textrm{$\switch(G_i,\bsigma_{x}, \eta_x, \kappa_x)$ fails}] \enspace, 
$$
where $\kappa_x,\eta_x$ vary over configurations of $\partial \alpha_i$ which differ on $x$
and $\bsigma_x$ is  distributed as in $\mu_i(\ |\ \partial \alpha_i, \eta_x)$.

\begin{lemma}\label{lemma:AccuracySampler}
Consider $\fixsampler$ on input the $\Psi$-factor graph $G=(V,F,(\partial a)_{a\in F},(\psi_a)_{a\in F})$.
Let $\mu$ be the Gibbs distribution on $G$ and assume that $\mu$ is symmetric.  Let $\bar{\mu}$ be
the distribution induced by the output of $\fixsampler$. We have that
\begin{align}\nonumber 
 || \mu-\bar{\mu}||_{\rm tv} &\leq 10 \ k  |\alphabet|^k  \cdot\sum\nolimits_{i\in [m]}\maxq_i\enspace, & \textrm{where $m=|F|$} \enspace.
\end{align}
\end{lemma}
\noindent
Lemma \ref{lemma:AccuracySampler} is a special case of Theorem \ref{theorem:FinalDetAcc}.  For a proof we refer the reader to
this result.

Furthermore, we have the following result for the time complexity of $\fixsampler$. 
\begin{lemma}\label{lemma:TimeComplexitySampler}
Consider $\fixsampler$ on input $\Psi$- factor graph $G=(V,F,(\partial a)_{a\in F},(\psi_a)_{a\in F})$.
The time complexity of $\fixsampler$ is $O(m(n+m))$, where $m=|F|$ and $n=|V|$.
\end{lemma}

\begin{proof}
The lemma follows immediately from Lemma \ref{lemma:UpdateWSTimeComplexity}. We
only need to observe the following: The algorithm  has to decide $\bsigma_i(\partial \alpha_i)$, for $i\in [m]$,
and each one of these decisions requires $O(k)=O(1)$ steps.  Furthermore,  
$\fixsampler$ makes  at most $m\times k$ calls of the process  $\switch$, 
and each call requires $O(m+n)$ steps.  In total, the running time is 
$O(m+ m(m+n))=O(m(m+n))$.
\end{proof}

\newcommand{\ShortDist}{{(\log_{dk} n)/10} }
\newcommand{\rsampler}{{\tt RSampler}}
\newcommand{\rupdate}{{\tt RUpdate}}
\newcommand{\rswitch}{{\tt RSwitch}}
\newcommand{\rcswitch}{{\tt CycleSwitch}}
\newcommand{\rmaxq}{\UpR}
\newcommand{\rsq}{\UpR\UpS}
 \newcommand{\rcsq}{\UpC\UpS}

\newcommand{\CG}{\mathcal{G}}

\newcommand{\DiaSet}{{\Updelta}}

\spreadpoint
 \begin{figure}
 \begin{minipage}{.29\textwidth}
 \centering
		\includegraphics[width=.64\textwidth]{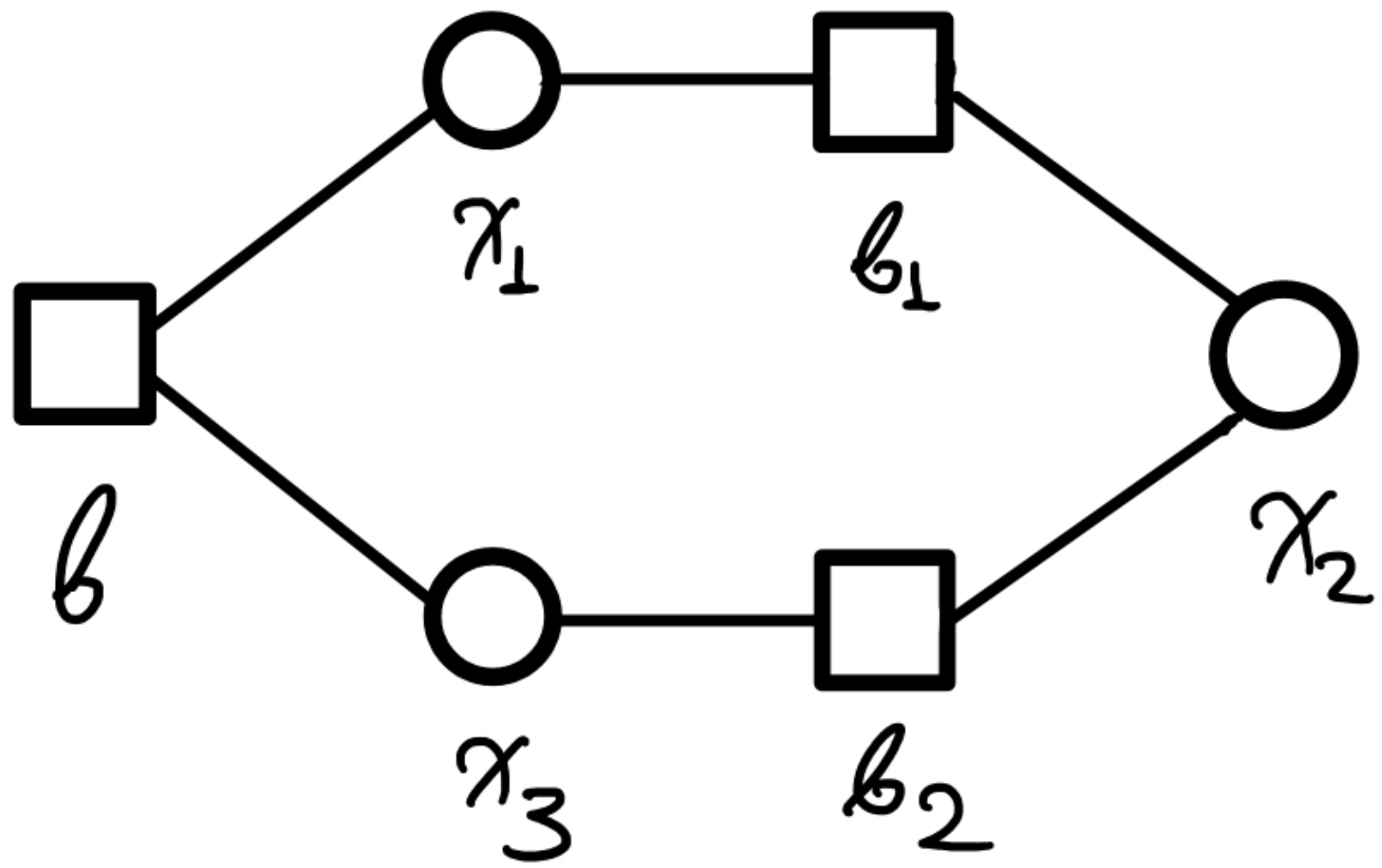}
		% \caption{$\beta$ in short cycle}
            \caption{}
	\label{fig:SCycleCase2A}
\end{minipage}
 \begin{minipage}{.389\textwidth}
 \centering
		\includegraphics[width=.48\textwidth]{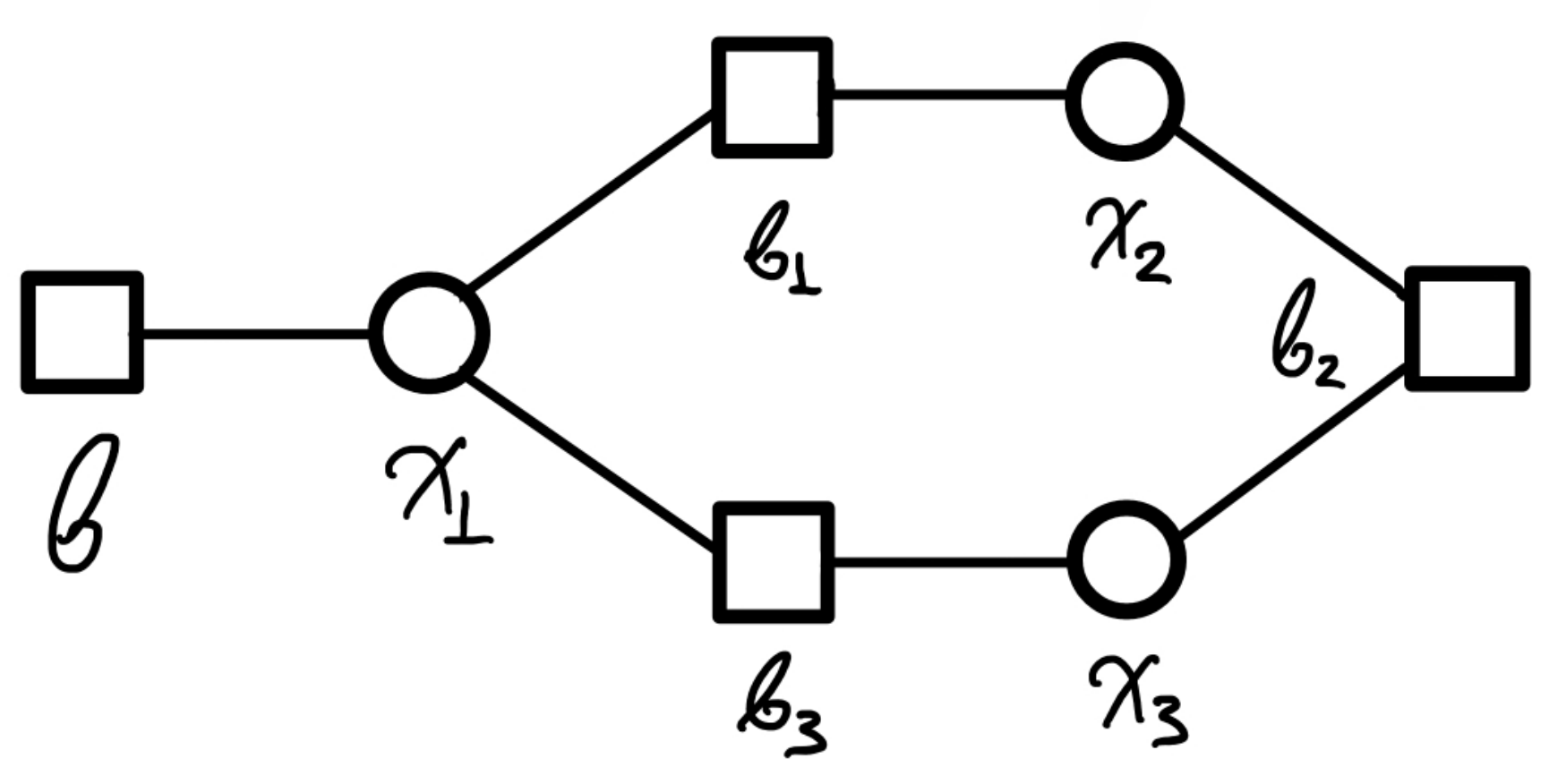}
            \caption{}
	\label{fig:SCycleCase2B}
\end{minipage}
 \begin{minipage}{.30\textwidth}
 \centering
		\includegraphics[width=.51\textwidth]{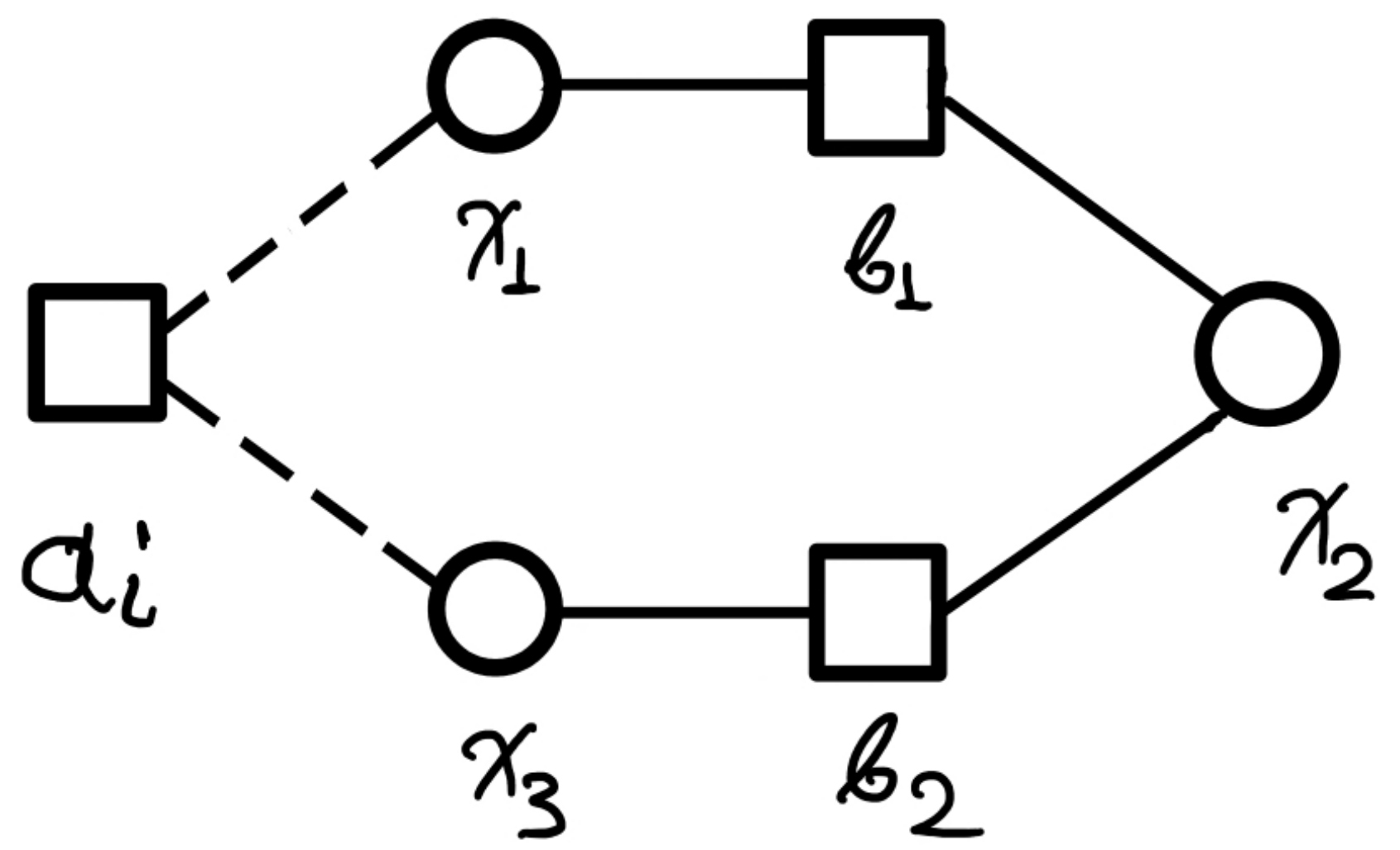}
            \caption{}
	\label{fig:SCycleCase1}
\end{minipage}
\end{figure}

\section{Sampling from  Random Factor Graphs \LastReview{2023-11-24}}\label{sec:AppGnpNew}

Building on the results from Section \ref{sec:WSAlgorithm},  we consider  the case where the input graph 
is a typical instance of the random $\Psi$-factor graph  $\G=\G(n,m,k, \dpsi)$ of expected degree 
$d$ and    $k\geq 2$.

In this setting, the new element   is that, typically, $\G$ contains a small number of short cycles 
which are far apart from each other. Recall that, so far we have been considering graphs  of 
high girth, i.e., with  no  short cycles at all. 

The existence of short cycles in $\G$ needs some caution. 
If we applied  $\fixsampler$ directly on a typical instance of $\G$ the accuracy would have dropped  a lot  
due to the presence of the aforementioned  short cycles. Specifically,  it would have been  common for the algorithm to create 
disagreements which  involve variable nodes of a  short cycle and this   makes it very likely for the algorithm to fail.  
We distinguish two cases in which the short cycles affect the algorithm.  The first one is in $\switch$. 
At some iteration,   the process  may choose a factor node $\beta$ which  either belongs to a short cycle,
or  has a neighbour which belongs to a short cycle.  E.g., see   \cref{fig:SCycleCase2A} and \ref{fig:SCycleCase2B}. 
The second case arises when the addition of  $\alpha_i$ introduces a short cycle in  $\G_{i+1}$ which does 
not exist in $\G_i$, e.g. see   \cref{fig:SCycleCase1}.

From now on, a cycle is considered to be ``short" if its length is less than $\ShortDist$. 
We let $\CG=\CG(n,d,k)$ be the family of instances of $\G(n,m,k, \dpsi)$  such that there 
are no  two short cycles  which share nodes.   To argue about  the short-cycle structure of $\G$,  
we use the following result.

\begin{lemma}\label{lemma:CycleStructureGNMK}
With probability  $1-O(n^{-2/3})$ over the instances of $\G=\G(n,m,k,\dpsi)$ we have that $\G\in \CG$.  
\end{lemma}
The  above lemma is standard. For a proof see \cref{sec:lemma:CycleStructureGNMK}.

We introduce a variation of  $\fixsampler$  that  handles the graphs in $\CG$, we  call it $\rsampler$. 
This algorithm  prevents the short cycles from deteriorating  the accuracy on
the condition that they are apart from each other by handling the above cases.

\subsection{The algorithm $\rsampler$:}\label{sec:RSAMPLERDetails}
Let the  fixed $\Psi$-factor graph $G=(V,F,(\partial a)_{a\in F},(\psi_a)_{a\in F})$ such that $G\in \CG$.
Suppose $G$ is   the input of $\rsampler$.

The basic set-up of   $\rsampler$  is the same as that of $\fixsampler$.
That is, it  creates the  sequence  $G_0,  \ldots, G_m$ 
in the standard way,
% by removing a randomly chosen factor node from  $G_{i+1}$ to obtain $G_i$.
% As before, we call this randomly chosen factor node  $\alpha_i$, 
while   $\mu_i$ is the Gibbs distribution that is induced by  $G_i$.
$\rsampler$  generates configuration  $\bsigma_{i+1}$ for  $G_{i+1}$ by using $\bsigma_{i}$,
 while $\bsigma_0$  is acquired as in \eqref{eq:GenerateSIGMA0}.

 We describe how it uses   $\bsigma_i$ to generate 
efficiently  the configuration $\bsigma_{i+1}$, in the new setting.
If $\alpha_i$, the edge we need to add to $G_i$ to obtain $G_{i+1}$ does  not  introduce a 
new short cycle in $G_{i+1}$, then $\rsampler$ sets  $\bsigma_{i+1}(\partial \alpha_i)$  
according  to \eqref{eq:FirstStep}.  

If $\alpha_i$ does introduce a new cycle in $G_{i+1}$, 
which is $C=x_1, \beta_1, x_2, \beta_2, \ldots, x_{\ell}, \alpha_i$.
That is,  the addition of  $\alpha_i$ into $G_i$  connects the ends of  the path ${\rm P}=x_1, \beta_1, x_2, \beta_2,  \ldots, x_{\ell}$,   
i.e.,  we have $x_1, x_{\ell}\in\partial \alpha_i$, and $2\ell <  \ShortDist$. 
This is similar to what we have in  \cref{fig:SCycleCase1}. 
Let $H$ be  the subgraph  of $G_{i+1}$ that is induced by both the variable and factor nodes in 
the cycle $C$, as well as all the  variable nodes that are adjacent to the factor nodes in $C$.
E.g., see  in  \cref{fig:GraphH}  the graph $H$ when $\ell=4$ and $k=3$. 
Furthermore, let $\mu_H$ be the Gibbs distribution  induced by  $H$.   

Rather than just deciding  $\bsigma_{i+1}(\partial \alpha_i)$,   $\rsampler$ decides  $\bsigma_{i+1}(H)$, 
such that 
%%%%
\begin{equation}\label{eq:ProbInitUpdateCycle}
\Pr[\bsigma_{i+1}(H)=\tau] =\mu_{H}(\tau), \qquad \forall \tau \in \alphabet^{V(H)} \enspace.
\end{equation}

It is clear that the distribution of  $\bsigma_{i+1}(H)$ is not  the same as the marginal of $\mu_{i+1}$ at $H$.
In the proof of \cref{theorem:FinalDetAcc}, we quantify the error that this discrepancy introduces. 
Furthermore, in the proof of \cref{thrm:FinalDetTime}, we show    how we get $\bsigma_{i+1}(H)$  efficiently.

Having acquired $\bsigma_{i+1}(H)$, or $\bsigma_{i+1}(\partial \alpha_i)$,  depending on the situation, 
the algorithm decides   the configuration for the remaining  variable nodes in $G_{i+1}$. 
Let the set  of variable nodes $\DiaSet$ be defined as follows: if 
$\alpha_i$ does  not  introduce a new short cycle in $G_{i+1}$, then  $\DiaSet=\partial \alpha_i$.
Otherwise, i.e., when $\alpha_i$  introduces a new short cycle in $G_{i+1}$, then 
$\DiaSet$ is equal to the variable nodes of $H$, i.e., we have that $\DiaSet=V(H)$.

 \begin{figure}
 \begin{minipage}{.45\textwidth}
\centering
	\centering
		\includegraphics[width=.5\textwidth]{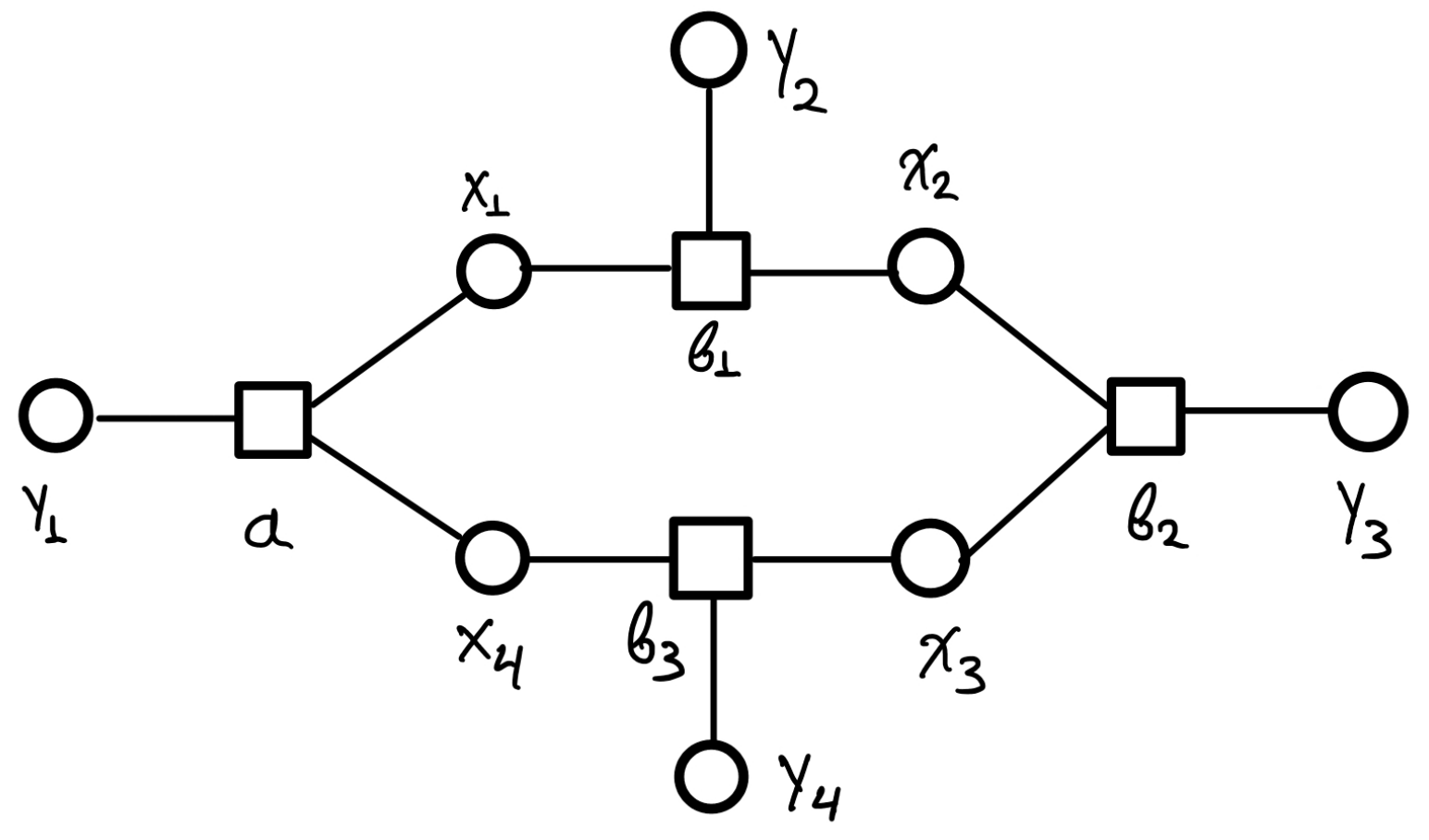}
		\caption{The graph H.}
%		\vspace*{-1cm}
	\label{fig:GraphH}
\end{minipage}
 \begin{minipage}{.45\textwidth}
\centering
	\centering
		\includegraphics[width=.3\textwidth]{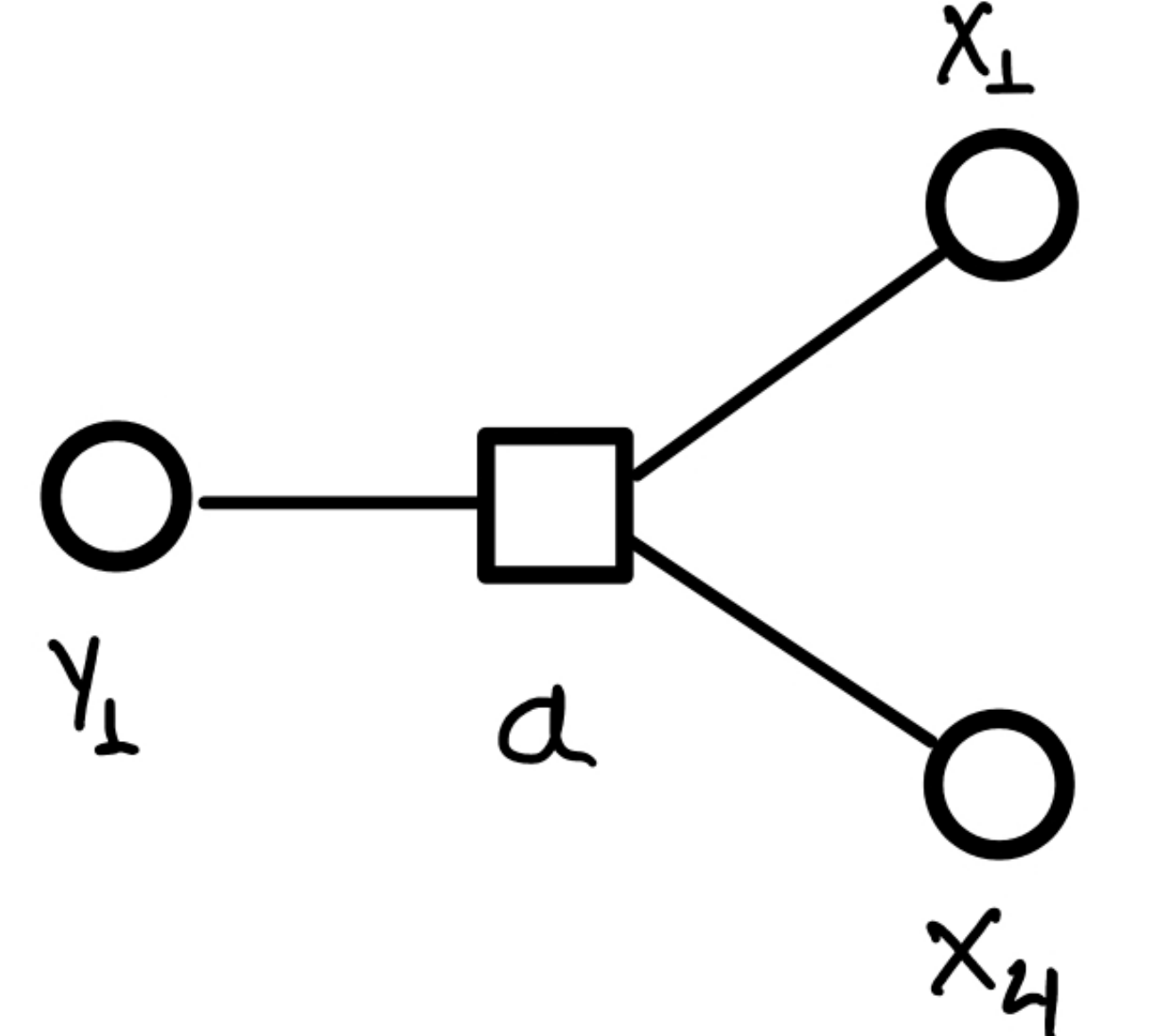}
		\caption{Single factor node}
%		\vspace*{-1cm}
	\label{fig:SingleFactor}
\end{minipage}
\end{figure}

Let   $\initDis=\{y_1, \ldots, y_{r}\}$  contain the  nodes in $\DiaSet$  which  the two configurations 
$\bsigma_{i}(\DiaSet)$ and $\bsigma_{i+1}(\DiaSet)$ disagree. 
Consider the  sequence of configurations   $\kappa_0, \kappa_1, \ldots, \kappa_{r}$ at $\DiaSet$
such that  $\kappa_0=\bsigma_{i}(\DiaSet)$,  while    $\kappa_j$ is obtained from
$\bsigma_{i}(\DiaSet)$ by changing  the assignment of the variable  nodes  
$z\in \{y_1, y_2, \ldots, y_j\}$ from $\bsigma_{i}(z)$ to $\bsigma_{i+1}(z)$.

Letting $\bar{G}_i$ be  obtained from $G_i$ by removing all the factor nodes $\beta$ such that 
$\partial \beta\in \DiaSet$,  we have the  following:  for $\btau_0\leftarrow \bsigma_{i}$,  set
\begin{align}\label{eq:RUpdateIteration}
\btau_j& \leftarrow\rswitch(\bar{G_i}, \btau_{j-1}, \kappa_{j-1}, \kappa_j) & \textrm{for }j=1,\ldots, r \enspace.
\end{align}
Then,  it sets $\bsigma_{i+1}=\btau_{r}$.

The process $\rswitch$ we introduce here is similar to $\switch$, but it has the extra  capability that it 
can deal with the cases of short cycles shown  in   \cref{fig:SCycleCase2A} and \ref{fig:SCycleCase2B}  
without increasing the failure probability. 

So far we have shown how we can deal with the short-cycle case  shown in  \cref{fig:SCycleCase1}.
The cases  that correspond to  \cref{fig:SCycleCase2A} and \ref{fig:SCycleCase2B} are handled 
by $\rswitch$. We describe $\rswitch$ in the following section. 

In Algorithm \ref{rsampler} we provide a synopsis of the above description of $\rsampler$.

\begin{algorithm} 
\caption{\rsampler}\label{rsampler} 
\begin{algorithmic}[1]
\Require $G$   
\If{$G\notin \CG$}
\State \Return $\tt Fail$
\EndIf
\State ${\tt generate}$ $G_0, \ldots, G_m$
\State ${\tt set}$ $\bsigma_0$ according to \eqref{eq:GenerateSIGMA0}
\For{$i=0, \ldots, m-1$}
\If{$\alpha_i$ creates a short cycle }
\State ${\tt set}$ $\DiaSet \gets V(H)$
\State ${\tt set}$ $\bsigma_{i+1}(\DiaSet)$ as in \eqref{eq:ProbInitUpdateCycle}
\Else
\State ${\tt set}$ $\DiaSet=\partial \alpha_i$
\State ${\tt set}$ $\bsigma_{i+1}(\DiaSet)$ according to \eqref{eq:FirstStep}
\EndIf
\State ${\tt generate}$ $\kappa_0, \kappa_1, \ldots, \kappa_{r}$ w.r.t. 
$\bsigma_{i+1}(\DiaSet)$ and $\bsigma_{i}(\DiaSet)$
\State $\btau_0\gets \bsigma_i$
\For {$j=1, \ldots, r $}
\State $\btau_j \gets \rswitch(\bar{G}_i, \btau_{j-1}, \kappa_{j-1}, \kappa_j)$ 
\EndFor
\State $\bsigma_{i+1}\gets \btau(r)$
\EndFor
\Ensure $\bsigma_m$
\end{algorithmic}
\end{algorithm}

\subsection{The process $\rswitch$}\label{sec:DefRswitch}
We define the process $\btau=\rswitch(\bar{G_i}, \bsigma, \eta, \kappa)$ such that $\kappa, \eta$ are two 
configuration  at $\DiaSet$ which differ only on the variable node $x\in \DiaSet$.   The 
input configuration $\bsigma$   is distributed as in $\mu_i(\cdot\ |\ \DiaSet, \eta)$, while the output 
of $\rswitch$ is $\btau$.  Recall that $\bar{G}_i$ is obtained from $G_i$ by removing all the factor nodes $\beta$ such that 
$\partial \beta\in \DiaSet$.  For what   follows, it does not make any difference if $\DiaSet$ is  $\partial \alpha_i$, or $V(H)$.

$\rswitch$   starts from $x$, the disagreeing node, and iteratively visits 
nodes of $\bar{G}_i$. It uses the sets of  nodes $\visit$ and $\DisG$, similarly to $\switch$. 
That is, at each iteration, $\visit$ contains 
the (factor and variable) nodes which the process has already visited.  This means that  for  every variable
node $z\in\visit$ we have $\btau(z)$. The  set $\DisG \subseteq \visit$ contains all the disagreeing 
variable nodes  in $\visit$. Initially, we set  $\btau(\DiaSet)=\kappa$, while  $\visit=\DiaSet$ 
and $\DisG=\{x\}$.  Throughout the process, we consider the set of disagreeing spins 
$\DisSpin =\{\bsigma(x), \btau(x)\}$.

At iteration $t$, $\rswitch$ chooses a factor node $\beta\notin \visit$ which is adjacent to a variable 
node in $\DisG$. If  $\partial \beta$  contains more than one  variable nodes for  which   $\btau$ is specified,  
then we consider that $\rswitch$ {fails} and the process terminates.

If $\beta$ does not belong to a short cycle  or does not have a neighbour which belongs to a short cycle, 
then the configuration $\btau(\partial \beta)$ is decided as in \eqref{eq:UpdateRuleA} and 
\eqref{eq:UpdateRuleB}, i.e.,  in the same way as in  $\switch$.

Assuming that  $\beta $  belongs to the short cycle $C$,  i.e., similarly to    \cref{fig:SCycleCase2A},
let $\mathcal{C}$ be the set of variable nodes which are adjacent to a factor node  in  $C$. Note  that 
$\mathcal{C}$ may include  nodes  outside the cycle $C$.

Recall that $\rswitch$ chooses  $\beta$ because  there exists $z \in \partial \beta$  such that  
$z\in \DisG$.  In this case, we have an additional failure condition, i.e.,   $\rswitch$ fails if 
$\visit\cap \mathcal{C}\neq \{z \}$. 
 $\btau(\mathcal{C})$ is specified   iteratively,  by choosing  a factor node $\alpha$  in $C$ such 
that  $\btau(\partial \alpha)$ is not fully specified, while there is  $y\in \partial \alpha$ at  which 
$\btau(y)\neq \bsigma(y)$.  For every $y \in \partial \alpha$ for which  $\btau(y)$ is not  specified set 
\begin{align}\label{eq:ShortCycleUpdtA}
\btau(y) &\gets
\left \{
\begin{array}{lcl}
\DisSpin \setminus \{\bsigma(y) \} && \textrm {if $\bsigma(y)\in \DisSpin$} \\
\bsigma(y) && \textrm {otherwise} \enspace.
\end{array}
\right . & 
\end{align}
The above iteration starts with the factor node  $\beta$. It can be that the  iteration in \eqref{eq:ShortCycleUpdtA} stops even 
though there are $y \in \mathcal{C}$ such that $\btau(y)$  is not  specified.  When this happens, for each one of those 
$y \in \mathcal{C}$ set
\begin{align}\label{eq:ShortCycleUpdtB}
\btau(y) &\gets\bsigma(y) \enspace.
\end{align}
  \cref{fig:UpdateSCycle}  illustrates an example of  the above rule. For each  node, 
the configuration at the bottom corresponds to $\bsigma$, while the top configuration  corresponds to 
$\btau$. The disagreement initially is at $x_0$ and propagates inside the cycle. The iteration in 
\eqref{eq:ShortCycleUpdtA} can only get up to $\partial \beta_2$ at the top side of the cycle and  
$\partial \beta_5$ at the bottom. The rest of the nodes are considered only at \eqref{eq:ShortCycleUpdtB}.  
Note that  the disagreements only involve the spins in $\DisSpin$.

After all the above,  the sets $\visit$ and $\DisG$ are updated accordingly. That is, we insert into   $\visit$ all 
the factor nodes in the cycle as well as $\mathcal{C}$. Furthermore, each  node in $\mathcal{C}$ which  is
disagreeing is also inserted into $\DisG$.

The case where $\beta$ is not in $C$ but  has a neighbour in $C$ (see example  in   \cref{fig:SCycleCase2B})  
is very similar. Define $\mathcal{C}$ to contain every variable node which is adjacent to a factor node in $C$ 
plus $\partial \beta$.  If there is a variable node in $\mathcal{ C}$, other than the single disagreement in $\partial \beta$, 
which belongs to  $\visit$, then the process fails.  Otherwise, it uses the  iterations shown in
\eqref{eq:ShortCycleUpdtA}   and \eqref{eq:ShortCycleUpdtB}. Note though that these are applied to   
to the factor nodes in  $C$  plus $\beta$.

 \begin{figure}
\centering
	\centering
		\includegraphics[width=.45\textwidth]{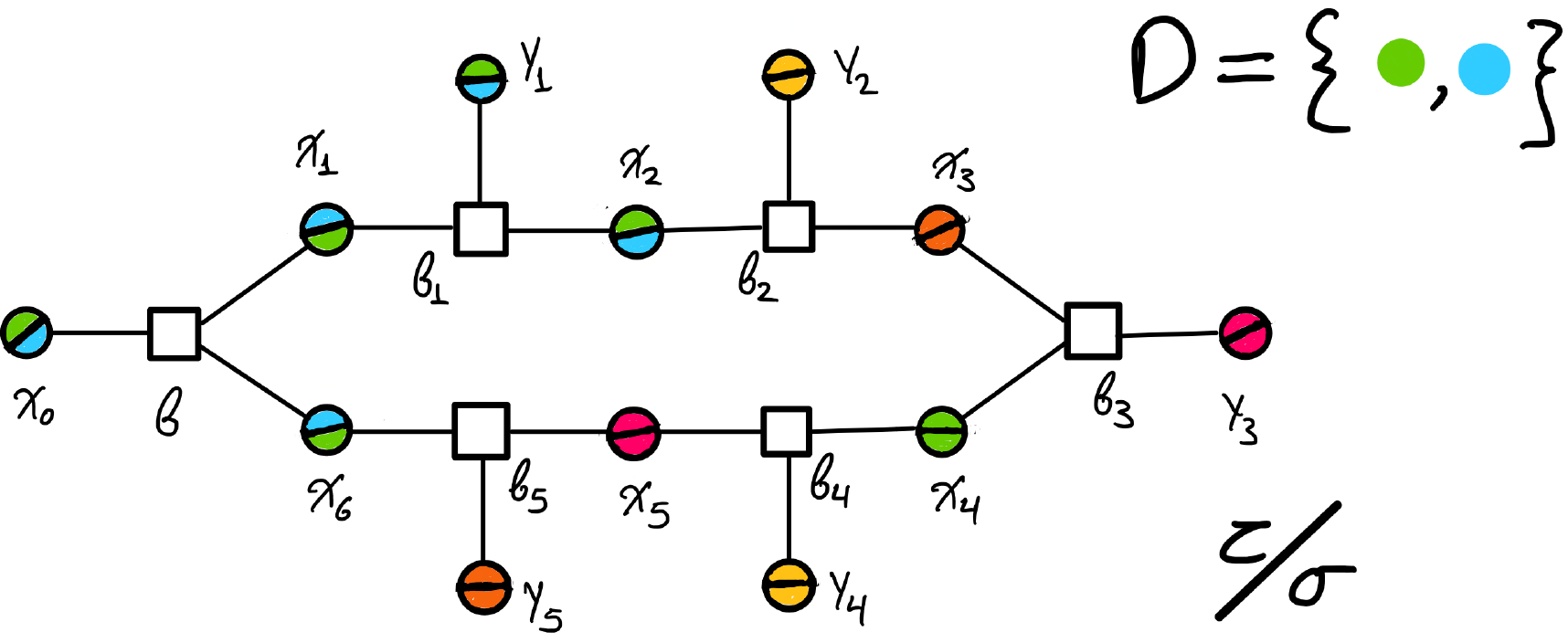}
		\caption{Update of the configuration of a short cycle.}
%		\vspace*{-1cm}
	\label{fig:UpdateSCycle}
\end{figure}

\spreadpoint

\section{Performances of $\rsampler$ \LastReview{2024-01-29} }\label{sec:FullPerformance}
With $\rswitch$, we have the complete picture of $\rsampler$ and  proceed to study its
performance in terms of  accuracy and time complexity.  We assume $\rsampler$ with   a fixed $\Psi$-factor graph  $G\in \CG$ at the input.

\newcommand{\CylAi}{M}
\newcommand{\iter}{{\tt iteration}}

Starting with the accuracy,  consider the sequence $G_0,  \ldots, G_m$ generated by $\rsampler$.
For  brevity, we introduce the process $\iter$  which  corresponds to  the code  from  line 6 to  16 in Algorithm \ref{rsampler}.
That is the part of the pseudo-code of $\rsampler$ that treats each $G_i$ separately. 
 
Specifically, for $\tau\in \alphabet^V$ and $\eta, \kappa\in \alphabet^{\partial \alpha}$ such that $\tau(\partial \alpha_i)=\eta$,  we let    $\iter(G_i,\tau, \eta, \kappa)$, correspond to the $i$-th iteration  of $\rsampler$  where   $\bsigma_i=\tau$ and  $\bsigma_{i+1}(\partial \alpha_i)=\kappa$. The output of the process is a configuration of $G_{i+1}$. 
Note that with this   process, we condition on the configuration of $\partial \alpha_i$ regardless of whether a short cycle is introduced by $\alpha_i$, or  not. 

When the addition of  $\alpha_i$ introduces a new short cycle $C$, then we let  $\CylAi$ be the set of the two variable   
nodes in $\partial \alpha$ which  also belong to $C$. If $\alpha_i$ does not introduce a new short cycle, 
then  $\CylAi$ is the empty set.
Furthermore, we specify the following two quantities:  For  each $x\in \partial \alpha_i\setminus \CylAi$ we let 
\begin{equation}\label{eq:DefMAXQ4Rswitch}
\rsq_x={ \max\nolimits_{\eta, \kappa} }  \Pr[\textrm{$\iter(G_i,\bsigma, \eta, \kappa)$ fails}] \enspace,
\end{equation}
where  $\kappa,\eta \in \alphabet^{\partial \alpha_i}$  differ only on $x$, while 
$\bsigma$  is distributed as in $\mu_{i}(\cdot \ |\ \partial \alpha_i, \eta)$.

For  $\alpha_i$ such that $\CylAi$ is non-empty,  i.e.,  $\alpha_i$  introduces a short cycle $C$ in $G_{i+1}$,
define 
\begin{equation}\label{eq:DefMAXQ4RCswitch}
\rcsq_{\CylAi}= { \max\nolimits_{\eta, \kappa}  }  \Pr[\textrm{$\iter(G_i,\bsigma, \eta, \kappa)$ fails}] \enspace,
\end{equation}
where  $\kappa,\eta\in \alphabet^{\partial \alpha_i}$ differ only on $\CylAi$, %    
while $\bsigma$ is distributed as in  $\mu_{i}(\cdot \ |\ \partial \alpha_i, \eta)$.

Then, we define $\rmaxq_i$ ``the error at iteration $i$"  of $\rsampler$  by 
\begin{equation}\label{def:DefRIWithCycle}
\rmaxq_i= \rcsq_{\CylAi}+\sum\nolimits_{x\in \partial \alpha_i\setminus \CylAi}\rsq_x \enspace.
\end{equation}
If $\CylAi=\emptyset$, i.e., $\alpha_i$ does not introduce a short cycle, then we follow the convention
that $\rcsq_{\CylAi}=0$.

Finally, for the factor graph $G$, we let $\psimin=\psimin(G)$ be the minimum value of $\psi_{\alpha}(\tau)$, where $\alpha$ varies over the set of factor nodes $F$ in $G$ 
and $\tau$ varies over the support of $\psi_{\alpha}$, i.e., see \eqref{eq:DefOfPsiMin}.

\begin{theorem}\label{theorem:FinalDetAcc}
Consider $\rsampler$ on input the $\Psi$-factor graph $G=(V,F,(\partial a)_{a\in F},(\psi_a)_{a\in F})$
such that $G\in \CG$.
Let $\mu=\mu_G$ be the Gibbs distribution on $G$ and assume that $\mu$ is symmetric. 
Let $\bar{\mu}$ be the distribution of the output of $\rsampler$. 
Provided that $\rmaxq_i$ is sufficiently small, for every $i\leq |F|$,  we have that
\begin{equation}\nonumber 
 || \mu-\bar{\mu}||_{\rm tv}\leq  14k|\alphabet|^k \cdot 
 \left( 1+\chi \cdot \psimin^{-1}  \right)\cdot \sum^{|F|}_{i=1}   \rmaxq_{i-1} \enspace.
\end{equation}
\end{theorem}

The proof of  \cref{theorem:FinalDetAcc} appears in  \cref{sec:theorem:FinalDetAcc}.

As far as the time complexity of  $\rsampler$ is concerned, we have the following result.

\begin{theorem}\label{thrm:FinalDetTime}
Consider $\rsampler$ on input $\Psi$-factor graph $G=(V,F,(\partial a)_{a\in F},(\psi_a)_{a\in F})$. 
The time complexity of $\rsampler$ is  $O\left(m(n+m) (\log n)^2+(n+m)^2\right)$, where $m=|F|$ and $n=|V|$.
\end{theorem}

\noindent
Note that  \cref{thrm:FinalDetTime} does not assume that $G\in \CG$.
When we calculate the running time, we account for the number of steps required 
for the algorithm to check whether $G\in \CG$. 
The  proof   appears in \cref{sec:thrm:FinalDetTime}.

\spreadpoint

\section{Proof of  \cref{thrm:MainA} \LastReview{2024-02-06}}\label{sec:thrm:MainA}

We prove  \cref{thrm:MainA} by using the terminology of factor graphs we have
been developing in the last couple of sections.

\begin{theorem}\label{prop:UniqueAccuracy}
For  $\delta\in (0,1)$, $k\geq 2$ and $d\geq 1/(k-1)$ the following is true:
Let $\G=\G(n,m,k,\dpsi)$ be such that $m=dn/k$, while let $\mu=\mu_{\G}$ be symmetric.   
Consider $\rsampler$ with input $\G$. 
If  $\G$  satisfies   $\setB$  with slack $\delta$, then for any  $\omega=\omega(n)$ such that 
$\lim_{n\to\infty}\omega=\infty$  we have that
\begin{align}\nonumber 
\mathbb{E}\left[ \rmaxq_i\ |\ \G \in \CG, \ \wedge_{t\in [m]}\ \mathcal{C}_t(\omega)  \right] &\leq  \omega \cdot (\log n)^9 \cdot n^{-\left(1+\frac{\delta}{41\log(dk)} \right)} 
&& 0\leq i < m \enspace.
\end{align}
\end{theorem}

\noindent
The proof of  \cref{prop:UniqueAccuracy} appears in  \cref{sec:prop:UniqueAccuracySetB}.

\begin{proof}[Proof of  \cref{thrm:MainA}]
As mentioned before, we prove \cref{thrm:MainA} by using the factor graph terminology.
That is, we assume that we have the $\G=\G(n,m,k,\dpsi)$ such that $m=dn/k$, while let $\mu=\mu_{\G}$ be 
symmetric. Further, assume that $\G$  satisfies   $\setB$ with slack $\delta>0$.  

Following the standard notation from \cref{sec:AppGnpNew,sec:Conditions},
we let  $\cH$ be the event that  $\G\in \CG$ and $\wedge_{t\in [m]}\mathcal{C}_t(\omega)$  
where  $\omega=O\left(n^{(\delta^{-1}\log(dk))^{-10} }\right)$. 
Also, let $\cS$ be the event that 
 $\psimin \geq n^{-(\delta^{-1}\log dk)^{-10}}$. 
Finally,  let $\mathcal{B}$ be the  event that  
$|| \mu-\bar{\mu}||_{\rm tv}\geq  n^{-\frac{\delta}{55\log(dk)}}$, where 
$\bar{\mu}$ is the distribution of the configuration at the output.

The theorem follows by showing that 
\begin{align}
\Pr[\cB \ |\ \cH, \ \cS] &\leq n^{-\frac{\delta}{220\log(dk)}} \enspace, 
\label{eq:target4MainTheoremSetA} \\
 \Pr[\cH,\ \cS]&=1-o(1) \enspace. \label{eq:target4MainTheoremSetAAA}
\end{align}

\noindent
Clearly, $\Pr[\cH, \cS]=1-o(1)$ is true. This follows from  \cref{lemma:CycleStructureGNMK}
and recalling that, since we assume that  $\G$ satisfies  $\setB$, we have 
$\Pr[\wedge_{t\in [m]}\mathcal{C}_t(\omega)]=1-o(1)$ and 
$\Pr[\psimin \geq n^{-(\delta^{-1}\log dk)^{-10}}]= 1-o(1)$.
This proves \eqref{eq:target4MainTheoremSetAAA}. 
We now focus on establishing \eqref{eq:target4MainTheoremSetA}.

From  \cref{theorem:FinalDetAcc} and the linearity of expectation, we have that
\begin{align}
\mathbb{E}\left[ || \mu-\bar{\mu}||_{\rm tv}\ |\ \cH, \ \cS\right] 
&\leq  16 k|\alphabet|^k \cdot  \chi \cdot n^{(\delta^{-1}\log dk)^{-10}}\cdot
 \sum\nolimits_{i\in [m]}  \mathbb{E}  \left[  \rmaxq_i \ |\ \cH,\ \cS \right] \enspace, 
 \label{eq:ExpctErrorLinearityIteration}
\end{align}
where   we use that that  $\psimin \geq n^{-(\delta^{-1}\log dk)^{-10}}$. 
Furthermore, since $\rmaxq_i\geq 0$, we have that
\begin{align}\label{eq:ExpctErrorIterationBoundMain}
\mathbb{E}  \left[  \rmaxq_i \ |\ \cH,\  \cS\right] &\leq 
(\Pr[\mathcal{S} \ | \mathcal{H}])^{-1}\cdot 
\mathbb{E}  \left[  \rmaxq_i \ |\   \mathcal{H}  \right] \leq (5/4) \cdot \mathbb{E}  \left[  \rmaxq_i \ |\  \mathcal{H}\right]\enspace.
\end{align}
For  the last inequality we use  that 
$\Pr[\cS \ |\ \cH]\geq 4/5$.  To see why this bound holds note that
\begin{equation} \nonumber 
\Pr[\bar{\cS} \  |\ \cH] = \frac{\Pr[\bar{\cS}, \ \cH]}{\Pr[\cH]} \leq 
\frac{\Pr[\bar{\cS}]}{\Pr[\cH]} \leq  2n^{-3/2} \enspace. 
\end{equation}
Plugging \eqref{eq:ExpctErrorIterationBoundMain} into \eqref{eq:ExpctErrorLinearityIteration} 
and recalling that $m=\frac{dn}{k}$, we get 
\begin{align}
\mathbb{E}\left[ || \mu-\bar{\mu}||_{\rm tv}\ |\  \cH, \ \cS\right] &\leq 
n^{-\frac{\delta}{42\log(dk)}} \enspace.   \label{eq:ExpctErrorCondSCom}
\end{align}
We get \eqref{eq:target4MainTheoremSetA} from the above and Markov's inequality, 
i.e., we have that 
\begin{align}\nonumber
\Pr\left[\cB\ |\ \cH,\ \cS \right] =
\Pr\left[|| \mu-\bar{\mu}||_{\rm tv}\geq  n^{-\frac{\delta}{55\log(dk)}}  \ |\ \cH,\ \cS \right] &\leq n^{-\frac{\delta}{220\log(dk)}} \enspace.
\end{align}
The theorem follows. 
\end{proof}

\spreadpoint
\section{Bounds on the expected error  - Proof of   \cref{prop:UniqueAccuracy}
\LastReview{2024-02-20}
}\label{sec:prop:UniqueAccuracySetB}

\newcommand{\bala}{{\tt Bal}}
\newcommand{\FactVarInPath}{\mathcal{L}}

Recall the process $\iter$ we introduce in \cref{sec:FullPerformance} to study the accuracy of 
$\rsampler$.  
For $\tau\in \alphabet^V$ and $\eta, \kappa\in \alphabet^{\partial \alpha_i}$  such that   $\tau(\partial \alpha_i)=\eta$,  we have that    $\iter(G_i,\tau, \eta, \kappa)$,  correspond to the $i$-th iteration   of $\rsampler$  where   $\tau$ is a configuration of $G_i$,   $\eta,\kappa$ are configurations at $\partial \alpha_i$, while we have that $\tau(\partial \alpha_i)=\eta$.
Typically,  $\eta$ and $\kappa$ disagree on the assignment of perhaps  more than one variable nodes.

For the sake of definiteness in the analysis,  instead of considering the two configurations 
$\eta,\kappa$,  i.e.,  for the third and fourth argument,  to be at the set of variable nodes 
$\partial \alpha_i$,  we consider them to be at $\DiaSet$.  Recall that $\DiaSet$ is defined
to be the set $\partial \alpha_i$, when $\alpha_i$ does not introduce any short cycle in $G_{i+1}$.
When we are dealing with $\alpha_i$ that introduces a short cycle in $G_{i+1}$, then
$\DiaSet$ consists of the variable nodes  that are adjacent to the factor nodes in the 
short cycle that $\alpha_i$ introduces in $G_{i+1}$. Sometimes, to stress the dependence on 
$\DiaSet$, we write  $\iter(G_i, \tau, \eta(\DiaSet), \kappa(\DiaSet))$. 
In the analysis, we allow the configurations $\eta$ and $\tau$ at $\DiaSet$ to disagree in more
than one variable node.

When there are multiple disagreements between $\eta(\DiaSet)$ and $\kappa(\DiaSet)$, 
we further consider that the process $\iter$ fails if during the iterative calls of 
$\rswitch$ (lines 14 \& 15 of $\rsampler$) two, or more instances of  $\rswitch$ 
update the same vertex.  This assumption gives us slightly more pessimistic bounds
on the failure probability. However, it simplifies the analysis substantially. 

Let $\pi=z_1, \ldots z_{\ell}$ be a sequence of $\ell\geq 1$ distinct nodes, 
where  variable and factor nodes in $\pi$ alternate. Assume 
that $z_1\in \eta\oplus\kappa$, 
while there is no other $z_j\in \DiaSet$. 

We are interested in the cases where  the nodes in $\pi$ form a path, while $z_{\ell}$ is either connected  to a node in $\DiaSet$, or it is connected to node $z_s$ for $s\in [\ell]$ such that the nodes $z_s, \ldots z_{\ell}$  form a cycle of length greater than $\ShortDist$, i.e., a long cycle.  Then, we say that  $\pi$ forms an $s$-{\em critical path}, i.e., implying that $z_{\ell}$ is connected to 
$z_s$. Having $s=0$ we imply that  $z_{\ell}$ is connected to $\DiaSet\setminus\{z_1\}$.
For  $0\leq s \leq \ell$,  we let $\UpJ^s_{\pi}$ be the indicator of the event  that $\pi$ 
is an $s$-critical path.

Our focus is on the probability that the $s$-critical path  induced by $\pi$ causes 
$\iter(G_i, \tau, \eta(\DiaSet), \kappa(\DiaSet))$ to fail.  Specifically, assume that the 
process  considers  exclusively the nodes in $\pi$. That is,   starting from a disagreement at $z_1$ it updates,  iteratively, the nodes in $\pi$ by choosing a factor node in 
$\pi$ that is next to a disagreeing node in $\pi$.   If there is no such factor node, it  stops.  
Note that  this process could fail, i.e.,  by just updating all nodes in $\pi$.  Our focus is on 
the probability of failure in this setting. 

While  the nodes in  $\pi$ are updated by the process,   if disagreements reach $z_s$,  
then they  can further  propagate from 
$z_s$ to $z_{s+1}$, but also from $z_s$ to $z_{\ell}$.  Hence, there are a few  alternatives to 
the order that the process chooses the factor nodes. We let $\UpU^s_{\pi}$ be the probability of 
failure for the process $\iter(G_i, \tau, \eta, \kappa)$ that updates only $\pi$, while the order 
of the factor nodes is chosen so that the failure probability is maximised.  

In order for $\UpU^s_{\pi}$ to be meaningful, we need to have that the connection
between $z_{\ell}$ and $z_s$ forms a long cycle, i.e., of length at least
$\ShortDist$. If this is not the case, we follow the convention to consider 
by default that $\UpU^s_{\pi}=0$.
We refer to $\UpU^{s}_{\pi}$ as the probability for the $s$-critical path induced by 
$\pi$ to be  {\em fatal} for the process $\iter$ $(G_i, \tau, \eta, \kappa)$. 
Note that $\UpU^s_{\pi}$ depends only on $G_i$ and $\tau$.

We also let $\UpJ^{\infty}_{\pi}$ to  be the indicator that $\pi$ forms a path in $\G^*$,
i.e., with no specifications on how $z_{\ell}$ is further connected. 
On the event $\UpJ^{\infty}_{\pi}=1$,
we let $\UpU^{\infty}_{\pi}$  be the probability of the path induced by $\pi$ being
disagreeing, i.e., the configuration of the variable nodes in $\pi$ is different than
that specified by the input configuration $\tau$.

\subsection{Fatal paths in the planted model}

For what follows, we let the set $\bala \subseteq \alphabet^{V}$ consist of 
every  {\em balanced} configuration   $\sigma\in \alphabet^V$.
That is, $\sigma\in \bala$ if for every $s\in \alphabet$ we have that 
\begin{align}
  \left| n^{-1} \cdot |\sigma^{-1}(s)|-q^{-1}\right| \leq n^{-2/3} \enspace,  \label{def:OfBallance}
\end{align}
where $\sigma^{-1}(s)\subseteq V$ is the set of variable nodes $x$ such that $\sigma(x)=s$.

 In order to prove  \cref{prop:UniqueAccuracy} we need to study the process $\iter$ with input
$(\G^*_i,\bsigma^*)$, i.e.,  obtained as in the teacher-student model, for $i\in [m]$. Specifically,  set $\partial \alpha_i=(x_1, \ldots,x_k)$ be a {\em fixed} set of $k$ variable nodes in $\G^*_i$.  
Let  $\cK$ be the event that there is no path  of length $\leq  \ShortDist$ connecting 
any  two  nodes in $\partial \alpha_i$.
On the event, $\cK$ and for fixed  $z\in \partial \alpha_i$,  consider  the  
$\iter(\G^*_i, \bsigma^*, \bsigma^*(\partial \alpha_i),  \bkappa^*)$,  where   
$\bkappa^*\in \alphabet^{\partial \alpha_i}$  is such that $\bsigma^*(z)\neq\bkappa^*(z)$, 
while, for  every  $x\in \partial \alpha_i\setminus\{z\}$, we have $\bkappa^*(x)=\bsigma^*(x)$.  
We choose $\bkappa^*(z)$   so that the failure probability of the process is maximized.  

With respect to the aforementioned process   
we consider the  quantity  
\begin{align}\label{eq:DefOfXPiStar}
 \bX^*_z&= \sum\nolimits_{1\leq \ell\leq (\log n)^{5/2}}
 \sum\nolimits_{\pi\in \Pi_{\ell,z}}
 \sum\nolimits_{s:\  s\neq \infty} \UpU^{s}_{\pi} \times \UpJ^s_{\pi}+
  \sum\nolimits_{\pi\in \Pi_{\ell_0,z}} \UpU^{\infty}_{\pi} \times \UpJ^{\infty}_{\pi}\enspace,
\end{align}
where for integer $\ell\geq 1$, $\Pi_{\ell,z}$ is the set of permutations of 
nodes $\pi=x_1,  \ldots, x_{\ell}$,  such that $x_1=z$,  while the factor and variable nodes 
alternate and there is no $x_{j}\in \pi$, where $j>1$, such that $x_{j}\in \DiaSet$. 
 Also, we have that $\ell_0=(\log n)^{3}$.

\begin{proposition}\label{thrm:ExpectedFatalPaths}
For $0\leq i < m$,  for  $\delta\in (0,1]$,  assume that $\mu_i$ satisfies $\setB$  
with slack   $\delta$.  For the process 
$\iter(\G^*_i, \bsigma^*, \bsigma^*(\partial \alpha_i),  \bkappa^*)$ 
we have that
\begin{align}
\mathbb{E}\left[\Ind\{\cK\} \times \bX^*_{z}  \ | \  \cB,\  \G^*_i\in \CG \right] &
\leq 3(\log n)^{11/2}\cdot  n^{-\left(1+\frac{\delta}{40\log (d k)}\right)} \enspace,
\end{align}
where $\cB$  denotes the event  that $\bsigma^*$ is balanced (see \eqref{def:OfBallance}).
\end{proposition}

The expectation of $\bX^*_z$ in \cref{thrm:ExpectedFatalPaths} is with respect
to the randomness of $(\G^*_i,\bsigma^*)$.  
The proof of \cref{thrm:ExpectedFatalPaths} appears in \cref{sec:thrm:ExpectedFatalPaths}.

We also investigate the case where $\alpha_i$ introduces a new short cycle. 
Consider  the process $\iter$ with input $(\G^*_i,\bsigma^*)$, for $i\in [m]$, and 
set $\partial \alpha_i=(x_1, \ldots,x_k)$ be fixed variable nodes in $\G^*_i$, as before. 
Let  $\cL$ be the event that 
\begin{itemize}
\item there is exactly one path $P$ of length $\leq \ShortDist-2$ that connects two nodes in $\partial \alpha_i$
 \item there is no short cycle that  intersects with any node  in ${\rm P}$,  or $\partial \alpha_i$. 
\end{itemize}
With a slight abuse of notation,
for any set of variable nodes $\Xi$ and any set of factor nodes $H$,   denote with $\FactVarInPath(\Xi, H)$  the events  $\cL$ and that  the set of factor nodes in  the path $P$ is $H$, while $\DiaSet=\Xi$.
The overloading of the symbol $\cL$ should not create any confusion.

On the event   $\FactVarInPath(\Xi, H)$,   consider  the 
the process $\iter(\bar{\G}^*_i, \bsigma^*, \bsigma^*(\DiaSet),  \bkappa^*(\DiaSet))$,
where $\bar{\G}^*_i$ is obtained from $\G^*_i$ be removing all the edges that have both 
their ends in $\Xi\cup H$. Also,  $\bkappa^*(\DiaSet)$ is obtained according to 
\eqref{eq:ProbInitUpdateCycle} conditional on $\bkappa^*(\partial \alpha_i)$ being as follows:
we have $\bkappa^*(x)=\bsigma^*(x)$  for all $x\in \partial \alpha_i\setminus \CylAi$,
while for each $z\in \CylAi$ we choose
$\bkappa^*(z)\neq \bsigma^*(z)$ so that it maximizes the failure probability for the process. 
Recall that $\CylAi\subseteq \partial \alpha_i$ includes the two variable nodes in 
$\partial \alpha_i$ that are in the short cycle.

\begin{proposition}\label{thrm:ExpectedFatalPathsWithCycles}
For $0\leq i < m$,  for  $\delta\in (0,1]$,  assume that $\mu_i$ satisfies $\setB$  with slack  $\delta$.
Consider the process $\iter(\G^*_i, \bsigma^*, \bsigma^*(\DiaSet),  \bkappa^*(\DiaSet))$. 
For  $\Xi\subset V$,   $H\subset F$ and for any $z\in \Xi$, we have 
\begin{align*}
\mathbb{E}[\bX^*_{z}\times \Ind\{ \FactVarInPath(H,\Xi) \} \ | \ \cB,\  \G^*_i\in \CG]  &
\leq  4(\log n)^{11/2} \cdot n^{-\left(1+\frac{\delta}{40\log (d k)}\right)} \cdot
\Pr[ \FactVarInPath(H,\Xi) \ |\ \G^*_i \in \CG] 
\enspace.
\end{align*}
\end{proposition}
The proof of \cref{thrm:ExpectedFatalPathsWithCycles} appears in \cref{sec:thrm:ExpectedFatalPathsWithCycles}.

When $\alpha_i$ introduces a short cycle in $\G^*_i$ sometimes instead of writing 
$\iter(\bar{\G}^*_i, \bsigma^*, \bsigma^*(\DiaSet),  \bkappa^*(\DiaSet))$, we abuse the
notation and  write $\iter({\G}^*_i, \bsigma^*, \bsigma^*(\DiaSet),  \bkappa^*(\DiaSet))$. 
This should no cause any confusion.

\subsection{Proof of \cref{prop:UniqueAccuracy}}\label{sec:ActualProofprop:UniqueAccuracy}

For the sake of brevity,  we let $\cE$ denote  the events that $\G \in \CG$ and $\wedge_{t\in [m]}{\cC}_t(\omega_n)$. Recall that the aim is  to  bound 
$\mathbb{E}\left[ \rmaxq_i\ |\ \cE \right]$, for $0\leq i < m$.

We set $\partial \alpha_i$ to be a fixed set of $k$ variable nodes. 
Having   $\G \in \CG$ implies that   for all $0\leq i < m$, 
we always have one of the events $\cK$ and $\cL$. Hence, since 
$\cK$ and $\cL$ are disjoint, we have
\begin{align}\label{eq:Basis4prop:UniqueAccuracyNew}
\mathbb{E}\left[ \rmaxq_i\ |\ \cE \right]
&=
\mathbb{E}\left[ \rmaxq_i\times \Ind\{\cK\}\ |\ \cE   \right]+
\mathbb{E}\left[ \rmaxq_i\times \Ind\{\cL\}\ |\ \cE  \right] \enspace.
\end{align}
We prove \cref{prop:UniqueAccuracy} by showing that
\begin{align}
\mathbb{E}\left[ \rmaxq_i \times \Ind\{\cK\}\ |\ \cE  \right]&\leq 
4 k\omega(\log n)^{6}\cdot n^{-\left(1+\frac{\delta}{40\log (d k)}\right)} \enspace,
\label{eq:Targert4RiOfCalK}\\
\mathbb{E}\left[ \rmaxq_i \times  \Ind\{\cL\}\ |\ \cE  \right]&\leq 
 4 \omega(\log n)^{8} \cdot n^{-\left(1+\frac{\delta}{41\log (d k)}\right)} \enspace.
\label{eq:Targert4RiOfCalL}
\end{align}
Specifically, \cref{prop:UniqueAccuracy} follows by plugging 
\eqref{eq:Targert4RiOfCalK}  and \eqref{eq:Targert4RiOfCalL} into  
\eqref{eq:Basis4prop:UniqueAccuracyNew} and using that $k=\Theta(1)$.

\begin{proof}[Proof of  the bound in \eqref{eq:Targert4RiOfCalK}]

It suffices to show that for any $z\in \partial \alpha_i$ we have 
\begin{align} \label{eq:TargerA4UniqueAccuracyB}
\mathbb{E}\left[\rsq_z\times \Ind\{\cK\} \ |\  \cE  \right] &\leq 
 4 \omega(\log n)^{6} \cdot n^{-\left(1+\frac{\delta}{40\log (d k)}\right)} \enspace.
\end{align}
Note that $\rsq_z$ is considered with respect to the process  $\iter(\G_i, \bsigma, \eta,\kappa)$
where $\eta,\kappa$ are configurations at $\partial \alpha_i$ that differ only at $z$, while they 
are chosen so that the probability of failure is maximised. In this setting,  
$\bsigma$ is distributed as in $\mu_i(\cdot \ |\ \partial \alpha_i, \eta)$.

Recall that $\cB$ is  the event that $\bsigma\in \bala$, i.e., $\bsigma$ is balanced. 
Since $\rsq_z\leq 1$, we have 
\begin{align}  \label{eq:Base4RSzBallance}
\mathbb{E}\left[\rsq_z\times \Ind\{\cK\} \ |\  \cE  \right] &\leq 
\mathbb{E}\left[\rsq_z\times \Ind\{\cK, \ \cB\} \ |\  \cE  \right] +
\Pr[\bar{\cB}\ |\ \cE] \enspace.
\end{align}
Similarly,  let $\mathcal{J}_i$  be the event that  $\G_i$ satisfies 
that  $||\mu_i-\zeta ||_{\partial \alpha_i} \leq (|\alphabet|^{-k}/2)$,  where 
$\zeta$ is the  uniform distribution  over $\alphabet^{\partial \alpha_i}$.
Using once more that  $\rsq_z \leq 1$,  we have 
\begin{align}\label{eq:Base4RSzSetB} 
\mathbb{E}[\rsq_z  \times \Ind\{ \cK,\  \cB\} \  |\  \cE ] 
&\leq  \mathbb{E}[\rsq_z \times \Ind\{  \cK,\ \cB\} \  |\  \cE,\   \cJ_i] +
\Pr[ \bar{\cJ}_i, \ \cK, \ \cB  \ |\  \cE ]\enspace,
\end{align}
where $\bar{\cB}, \bar{\cJ}_i$ are the complements of the events $\cB$ and 
${\cJ}_i$, respectively.
From \eqref{eq:Base4RSzSetB} and \eqref{eq:Base4RSzBallance}, we get that
\begin{align}
\mathbb{E}\left[\rsq_z\times \Ind\{\cK\} \ |\  \cE  \right] &\leq 
 \mathbb{E}[\rsq_z \times \Ind\{  \cK,\ \cB\} \  |\  \cE,\   \cJ_i] 
 + \Pr[ \bar{\cJ}_i, \ \cK, \ \cB  \ |\  \cE ] + \Pr[\bar{\cB}\ |\ \cE] \enspace.
\end{align}
Then,   \eqref{eq:TargerA4UniqueAccuracyB} follows from the above and  showing that
\begin{align} 
%%%
%%%
\Pr[\bar{\cB} \ |\ \cE] &\leq \textstyle \exp\left(-n^{1/4}\right)\enspace, \label{eq:TargerA4UniqueAccuracyBalance}\\
% %%%
\mathbb{E}\left[\rsq_z\times \Ind\{\cK,\  \cB\} \ |\ \cE,\ \cJ_i \right] &\leq 
 2\omega(\log n)^{6}\cdot  n^{-\left(1+\frac{\delta}{40\log (d k)}\right)}\enspace, \label{eq:TargerA4UniqueAccuracy} \\
%%%
\Pr[ \bar{\cJ}_i, \ \cK, \ \cB  \ |\  \cE ] &\leq   \omega(\log n)^{6} \cdot n^{-\left(1+\frac{\delta}{40\log (d k)}\right)} 
\label{eq:TargerA4UniqueAccuracyProducEdge}
\enspace. 
\end{align}

As far as \eqref{eq:TargerA4UniqueAccuracyBalance} is concerned, this follows 
by noting that in the teacher-student model $\bsigma^*\notin\bala$, with probability  $o(\exp(n^{-1/4}))$. This can be obtained by a simple application of 
Chernoff's  bound. Then,  using contiguity, i.e., condition $\CB$,
we get  \eqref{eq:TargerA4UniqueAccuracyBalance}.

We continue with  \eqref{eq:TargerA4UniqueAccuracy}.
Let $\bF_z=\bF_z(\G_i,\bsigma)$ be the number of fatal paths in 
$\iter(\G_i, \bsigma, \eta,\kappa)$ that emanate from $z$. We have that
\begin{align}\label{eq:RSQVsFzExpecationsA}
\mathbb{E} \left [\rsq_z \times \Ind\{ \cK,\ \cB\} \  |\  \cE, \ \cJ_i \right] &\leq 
 \Pr \left[ \bF_{z}\times  \Ind\{ \mathcal{K},\ \cB\} >0
  \ |\   \cE,\  \cJ_i\right ]\enspace.  
\end{align}

We  let $\bF^+_z$ be the number of fatal paths   which are of length at most 
$(\log n)^{5/2}$. Also,  let $\bF^-_z$ be the number of paths of disagreement that
are of length exactly  $(\log n)^{3}$. Recall that a path of disagreement is any path
such that all its variable nodes, at the output, have a different configuration
than that specified by  $\bsigma$. 

It is standard to see that if there are no paths of disagreement of length
$(\log n)^3$ there are no fatal paths of length greater than $(\log n)^2$. 
This implies that $\Ind\{\bF_z>0\} \leq \Ind\{\bF^+_z+\bF^-_z>0\}$. 

Then, \eqref{eq:RSQVsFzExpecationsA} implies that
\begin{align}
\mathbb{E} \left [\rsq_z \times \Ind\{ \cK,\ \cB\} \  |\  \cE, \ \cJ_i \right] &\leq 
 \Pr \left[ (\bF^+_{z}+\bF^-_z)\times  \Ind\{ \mathcal{K},\ \cB\} >0
  \ |\   \cE,\  \cJ_i\right ] \nonumber \\
  &\leq \mathbb{E} \left[ (\bF^+_{z}+\bF^-_z)\times  \Ind\{ \mathcal{K},\ \cB\} >0
  \ |\   \cE,\  \cJ_i\right ]\enspace. \label{eq:RSQVsFzExpecations}
\end{align}
Rather than bounding the expectation above, first, we focus on 
$\iter(\G_i, \bsigma, \bsigma (\partial \alpha_i), \bkappa)$.  Note 
that the third parameter now has changed,  implying that the input is from 
the (unconditional) Gibbs distribution $\mu_i$. Also,  $\bkappa\in \alphabet^{\partial \alpha_i}$ 
is such that  $\bkappa(x)=\bsigma(x)$ for  every  $x\in \partial \alpha_i \setminus \{z\}$, 
while $\bkappa(z)\neq \bsigma(z)$ is chosen so that the failure probability is 
maximised. With respect to this process, consider the variable
\begin{align*}
\bX_z &=\sum\nolimits_{1\leq \ell\leq (\log n)^5/2}\sum\nolimits_{\pi\in \Pi_{\ell,z}}
 \sum\nolimits_{s:\ s\neq \infty} \UpU^{s}_{\pi} \times \UpJ^s_{\pi} 
 + \sum\nolimits_{\pi\in \Pi_{\ell_0,z}} \UpU^{\infty}_{\pi} \times \UpJ^{\infty}_{\pi}\enspace,  
\end{align*}
where $\ell_0=(\log n)^3$.   The above variable $\bX_z$ is defined similarly to $\bX^*_z$ 
in \eqref{eq:DefOfXPiStar}.

\begin{claim}\label{claim:Bound4AverFatalPaths}
We have that 
\begin{align*}
 \mathbb{E} \left[ \bX_{z}\times  \mathbf{1}\{ \mathcal{K},\ \cB\}  \ |\   \cE,\  \cJ_i\right ] 
&\leq 
\left(\Pr[\cJ_i, \   \cK,\ \cB\ |\ \cE ]\right)^{-1} \cdot  
 6\omega(\log n)^{11/2} \cdot n^{-\left(1+\frac{\delta}{40\log (d k)}\right)} 
 \enspace.     
 \end{align*}
\end{claim}
As noted earlier, an important difference between $(\bF^+_z+\bF^-_z)$ and $\bX_z$ 
is the process with respect to which we consider them. That is, the corresponding
processes  differ in their third parameter. 
For the process of  $\bX_z$  the configuration  at $\partial \alpha_i$ 
is chosen according  to the Gibbs distribution $\mu_i$, whereas  for 
that of $(\bF^+_z+\bF^-_z)$, the choice is arbitrary.  In light of this 
discrepancy, we utilise the following  result. 

\begin{claim}\label{claim:FinalBound4ExpXzK}
We have that 
\begin{align}\label{eq:claim:FinalBound4ExpXzKRev}
\Pr[\bar{\cJ}_i, \  \mathcal{K},\  \cB \ |\  \cE ] 
&\leq  \omega(\log n)^{6}\cdot n^{-\left(1+\frac{\delta}{40\log (d k)}\right)} \enspace.
\end{align}
For any $\eta\in \alphabet^{\partial \alpha_i}$ we have that
\begin{equation}\label{eq:claim:FinalBound4ExpXzKRevB}
 \mathbb{E} \left [ \bX_{z} \times \Ind\{\cK,\  \cB\}
  \ |\   \cE,\  \cJ_i, \   \bsigma(\partial \alpha_i)=\eta \right]
\leq 2|\alphabet|^k \cdot \mathbb{E}[ \bX_{z} \times \Ind\{  \cK,\ \cB\}  \ |\ \cJ_i, \ \cE ] 
 \enspace.     
\end{equation}
\end{claim}

\noindent
Then,  \eqref{eq:claim:FinalBound4ExpXzKRevB} implies that 
\begin{align}
\mathbb{E} \left[ (\bF^+_{z}+\bF^-_z)\times  \Ind\{ \mathcal{K},\ \cB\}
  \ |\   \cE,\  \cJ_i\right ] &\leq  2|\alphabet|^k \cdot \mathbb{E}[ \bX_{z} \times \Ind\{  \cK,\ \cB\}  \ |\ \cJ_i, \ \cE ]  \nonumber \\
  &\leq  4 |\alphabet|^k \cdot    6\omega(\log n)^{11/2} \cdot n^{-\left(1+\frac{\delta}{40\log (d k)}\right)}\enspace. 
  \label{eq:TargerA4UniqueAccuracyOneBefore}
\end{align}
For \eqref{eq:TargerA4UniqueAccuracyOneBefore} we use 
\cref{claim:Bound4AverFatalPaths} and \eqref{eq:claim:FinalBound4ExpXzKRev}.
Specifically, since
$\Pr[\bar{\cJ}_i, \  \mathcal{K},\  \cB \ |\  \cE ] +\Pr[{\cJ}_i,\ \mathcal{K},\  \cB \ |\  \cE ]=\Pr[\mathcal{K},\  \cB \ |\  \cE ]= 1-o(1) $, it is immediate that \eqref{eq:claim:FinalBound4ExpXzKRev} implies 
that  $\Pr[{\cJ}_i, \  \mathcal{K},\  \cB \ |\  \cE ]=1-o(1)$.

From \eqref{eq:TargerA4UniqueAccuracyOneBefore}  and \eqref{eq:RSQVsFzExpecations}, we get \eqref{eq:TargerA4UniqueAccuracy}. 
Finally, \eqref{eq:TargerA4UniqueAccuracyProducEdge}  follows from \eqref{eq:claim:FinalBound4ExpXzKRev}.
\end{proof}

\begin{proof}[Proof of the bound in  \eqref{eq:Targert4RiOfCalL}] 
Recall that $\partial \alpha_i=(x_1, x_2, \ldots, x_k)$, is a fixes set of nodes.x

On the event  $\cL$, we let ${\rm P}$ be the unique short path that connects the two variable 
nodes in $\partial \alpha_i$, while recall that $\CylAi={\rm P}\cap \partial \alpha_i$.

The definition of $\rmaxq_i$ in \eqref{def:DefRIWithCycle}, 
implies that \eqref{eq:Targert4RiOfCalL} follows by showing that 
\begin{align}
 \mathbb{E}[ \rcsq_{\CylAi} \times \Ind\{ \cL  \}  \ |\  \cE ] & \leq   
 3\omega(\log n)^{8} \cdot  n^{-\left(1+\frac{\delta}{41\log (d k)}\right)} \enspace,  \label{eq:RsqCycleSetBFinalBoundTargetB} \\
\mathbb{E}[\rsq_z \times \Ind\{ \cL\} \  |\ \cE ]
& \leq 4\omega(\log n)^{6} \cdot  n^{-\left(1+\frac{\delta}{40\log (d k)}\right)} & \forall z\in \partial \alpha_i\setminus \CylAi \enspace.
\label{eq:RsqCycleNotSetBFinalBoundTarget}
\end{align}
We get \eqref{eq:RsqCycleNotSetBFinalBoundTarget} with  the same derivations as 
those for  \eqref{eq:TargerA4UniqueAccuracyB}. 
For this reason, we only focus on \eqref{eq:RsqCycleSetBFinalBoundTargetB}.

Recall that  $\cB$ denotes the event that $\bsigma$ is balanced.
Also,   let $\cN_i$ be the event that      $||\mu_i-\zeta ||_{\partial \alpha_i\setminus \{ x_a\}} 
\leq (|\alphabet|^{-k+1}/2)$ and  $\psimin\geq n^{-(\log dk)^{-10}}$, where
 $x_a\in \CylAi$ is chosen arbitrarily.
Since $\rcsq_{\CylAi}\leq 1$, we have 
\begin{align}   
\mathbb{E}\left[\rcsq_{\CylAi}\times \Ind\{ \cL\} \ |\  \cE  \right] &\leq 
\mathbb{E}\left[\rcsq_{\CylAi}\times \Ind\{ \cL,\  \cB\} \ |\  \cE  \right] +
\Pr[\bar{\cB} \ |\ \cE] \nonumber \\ 
&\leq 
\mathbb{E}\left[\rcsq_{\CylAi}\times \Ind\{ \cL,\  \cB\} \ |\  \cN_i,\ \cE  \right] +
 \Pr[\bar{\cN}_i, \  \cB \ |\  \cE ]+\Pr[\bar{\cB}\ |\ \cE]
\enspace.
\end{align}
In light of the above and \eqref{eq:TargerA4UniqueAccuracyBalance}, 
we get \eqref{eq:RsqCycleSetBFinalBoundTargetB}    by  showing that 
\begin{align} 
 \mathbb{E}[ \rcsq_{\CylAi}  \times \Ind\{ \cL,\ \cB \}  \ |\  \cE ] & \leq   
2 \omega(\log n)^{8}  \cdot n^{-\left(1+\frac{\delta}{41\log (d k)}\right)} \enspace,
\label{eq:RsqCycleSetBFinalBoundTarget} \\
 \Pr[\bar{\cN}_i, \  \cB \ |\  \cE ]&\leq  2\omega(\log n)^{6} \cdot  n^{-\left (1+\frac{\delta}{40\log (d k)} \right)}
 \enspace. 
 \label{eq:PrNiCycleSetBFinalBoundTarget}
\end{align}

\noindent
For a set of variable nodes $\Xi$ and a   set of factor nodes $H$,  recall that $\FactVarInPath(\Xi, H)$ 
denote the events  $\cL$ and that  the set of factor nodes in  the path $P$ is $H$, while 
$\DiaSet=\Xi$.
We have that 
\begin{align}\label{eq:CSmVsXi}
\mathbb{E}[\rcsq_{\CylAi}  \times  \Ind\{ \cL,\  \cB\} \  |\ \cN_i,\  \cE ] &\leq 
  \sum\nolimits_{\Xi, H}
\mathbb{E}[\rcsq_{\CylAi}  \times  \Ind\{\FactVarInPath(\Xi, H), \ \cB \}   \  |\  \cN_i,\   \cE ] \enspace.   
\end{align}
Note that  $H$ varies such that $|H|< \log n$.

On the event $\FactVarInPath(\Xi,H)$, we  let $\bar{\G_i}$  be the factor 
graph  obtained from $\G_i$ by removing every edge that has both its ends in $\Xi\cup H$. 
We consider  $\iter(\bar{\G}_i, \bsigma, \bbeta(\DiaSet), \bkappa(\DiaSet))$ where
$\bbeta(\DiaSet)$ and ${\bkappa}(\DiaSet)$ are obtained as follows: 
We have an arbitrary configuration $\eta$ at $\partial \alpha_i$.
Then, $\bsigma$ is distributed as in $\mu_i(\cdot\ |\ \partial \alpha_i,\eta)$.
Also, $\bbeta(\DiaSet)$ is the configuration of $\bsigma$ at $\DiaSet$. 
The configuration ${\bkappa}(\DiaSet)$  is obtained according to \eqref{eq:ProbInitUpdateCycle} conditional on ${\bkappa}(\partial \alpha_i)$ being as follows:
we have $\bkappa(x)=\bbeta(x)$  for all $x\in \partial \alpha_i\setminus \CylAi$,
while for each $x\in \CylAi$ we choose ${\bkappa}(x)\neq \bbeta(x)$ so that the failure probability for the process is maximised.

Working as in \eqref{eq:RSQVsFzExpecations}, 
for $z\in \Xi$,  we  have that
\begin{align}
 \mathbb{E}[ \rcsq_M \times \Ind\{\FactVarInPath(\Xi, H),\ \cB  \}  \ |\   \cE, \  \cN_i]
&\leq  \sum\nolimits_{z\in \Xi} 
\mathbb{E}[(\bF^+_{z}+\bF^-_z) \times \Ind\{ \FactVarInPath(\Xi, H),\ \cB  \} \ |\  \cE, \ \cN_i  ]   \enspace, 
 \label{eq:EpxCSMVsSumFzCycleCase}
\end{align}
where $\bF^+_z$ is the number of fatal paths which are of length at most 
$(\log n)^{5/2}$ that emanate from $z$, while  $\bF^-_z$ is the number of paths 
of disagreement that are of length  $(\log n)^{3}$ that emanate from the same
vertex.

We prove \eqref{eq:RsqCycleSetBFinalBoundTarget}  by using 
\eqref{eq:EpxCSMVsSumFzCycleCase} and showing that for any $H$, $\Xi$
and any $z\in \Xi$, we have 
\begin{align}
\mathbb{E}[(\bF^+_{z}+\bF^-_z) \times \Ind\{ \FactVarInPath(\Xi, H),\ \cB  \} \ |\  \cE, \ \cN_i  ]   &\leq 
20\omega(\log n)^{13/2}\cdot  n^{-\left(1+\frac{\delta}{41\log (d k)}\right)}\times  \Pr[\FactVarInPath(\Xi,H)\ |\  \G_i \in \CG] 
\enspace. \label{eq:Target4ExpFzWithShortCycle}  
\end{align}
Specifically,  
we use   \eqref{eq:Target4ExpFzWithShortCycle} and \eqref{eq:EpxCSMVsSumFzCycleCase}
and the fact that $|H|\leq  \log n$, $|\Xi|\leq k\ \log n$ to get  that
\begin{align}
 \mathbb{E}[ \rcsq_M \times \Ind\{\FactVarInPath(\Xi, H),\ \cB  \}  \ |\   \cE, \  \cN_i] &\leq  
 20k\omega(\log n)^{15/2}\cdot  n^{-\left(1+\frac{\delta}{41\log (d k)}\right)}\times  \Pr[\FactVarInPath(\Xi,H)\ |\  \G_i \in \CG] 
 \enspace. \nonumber
\end{align}
Plugging the above inequality into \eqref{eq:CSmVsXi}
we get \eqref{eq:RsqCycleSetBFinalBoundTarget}. 

It remains to  show that \eqref{eq:Target4ExpFzWithShortCycle}  is true.
Rather than bounding the expectation of $\bF^+_{z}+\bF^-_{z}$, first, we focus on the process $\iter(\bar{\G}_i, \bsigma, \bsigma(\DiaSet), \btheta(\DiaSet))$. The difference with this process is that the configuration
at $\partial  \alpha_i$ is specified by $\bsigma$ and it is not arbitrary. Given $\bsigma$, the configuration $\btheta(\DiaSet)$ is obtained similarly to what we describe before for $\bkappa$.

With respect to this process, for $z\in \Xi$ and  $\ell_0=(\log n)^3$,  consider the variable  
\begin{align}\label{eq:BXzShortCycle}
\bX_z &=\sum\nolimits_{1\leq \ell\leq (\log n)^5/2}\sum\nolimits_{\pi\in \Pi_{\ell,z}}
 \sum\nolimits_{s:\ s\neq \infty} \UpU^{s}_{\pi} \times \UpJ^s_{\pi} 
 + \sum\nolimits_{\pi\in \Pi_{\ell_0,z}} \UpU^{\infty}_{\pi} \times \UpJ^{\infty}_{\pi}\enspace. 
\end{align}
Using  \cref{thrm:ExpectedFatalPathsWithCycles} and following similar steps to those in the proof of \cref{claim:Bound4AverFatalPaths}  we get that
\begin{align}\label{eq:ExpcXZWithoutBoundaryCycle}
\mathbb{E}[ \bX_{z}  \times \Ind\{\FactVarInPath(\Xi, H), \ \cB\} \ |\  \cE,\cN_i ] 
&\leq 
10\omega(\log n)^{13/2} 
\cdot  n^{-\left(1+\frac{\delta}{40\log (d k)}\right)}
\cdot  \frac{\Pr[\FactVarInPath(\Xi,H)\ |\ \G_i \in \CG]}{\Pr[\cN_i,\ \cB\ |\ \cE]} \enspace.
\end{align}

In order to get the desired bound for $\bF^+_{z}+\bF^-_{z}$, we need to assume that the configuration at $\partial \alpha_i$ can be arbitrary, i.e.,  rather than being specified by $\bsigma$.
To this end,  we use the following claim. 

\begin{claim}\label{claim:FinalBound4ExpCML}%
 We have that
\begin{align}  \label{eq:claim:FinalBound4ExpCMLProb}
\Pr[\bar{\cN_i},\ \cB \ |\ \ \cE ] 
&\leq  
 2\omega(\log n)^{6} \cdot  n^{-\left (1+\frac{\delta}{40\log (d k)} \right)}  \enspace.
\end{align}
For any $\eta\in \alphabet^{\partial \alpha_i}$,  we have that 
\begin{align}\label{eq:FinalBound4ExpXzLNew}
\lefteqn{
 \mathbb{E}[ \bX_{z} \times  \Ind\{\FactVarInPath(\Xi, H),\  \cB \}  \ |\  \cE, \  \cN_i,\ 
 \bsigma(\partial \alpha_i)=\eta]
 }  \hspace{4 cm}\nonumber \\
&\leq 
\textstyle \left (|\alphabet|^{k-1}\cdot  \chi\cdot  n^{(\log dk)^{-10}} \right)\cdot 
\mathbb{E}[ \bX_{z}  \times \Ind\{\FactVarInPath(\Xi, H), \ \cB\} \ |\  \cE,\cN_i ] 
\enspace.
\end{align}
\end{claim}

Then, \eqref{eq:FinalBound4ExpXzLNew},  implies that
\begin{align}
\mathbb{E}[ (\bF^+_{z}+\bF^-_z) \times \Ind\{ \cF(\Xi, H),\ \cB  \} \ |\  \cE, \ \cN_i  ]
  &\leq  \textstyle \left (|\alphabet|^{k-1}\cdot  \chi\cdot  n^{(\log dk)^{-10}} \right)\cdot 
\mathbb{E}[ \bX_{z}  \times \Ind\{\FactVarInPath(\Xi, H), \ \cB\} \ |\  \cE,\cN_i ] 
\nonumber \enspace. 
\end{align}
Plugging \eqref{eq:ExpcXZWithoutBoundaryCycle} into the above inequality
and using that $\Pr[{\cN_i},\ \cB \ |\ \ \cE ]=1-o(1)$, we 
get \eqref{eq:Target4ExpFzWithShortCycle}.
Since $\Pr[\bar{\cN_i},\ \cB \ |\ \ \cE ] +\Pr[{\cN_i},\ \cB \ |\ \ \cE ]=
\Pr[\cB \ |\ \ \cE ]=1-o(1)$, it is immediate that \eqref{eq:TargerA4UniqueAccuracyBalance} 
implies that $\Pr[{\cN_i},\ \cB \ |\ \ \cE ]=1-o(1)$. 
Finally,   \eqref{eq:PrNiCycleSetBFinalBoundTarget} 
follows from \eqref{eq:claim:FinalBound4ExpCMLProb}. 

All the above imply  that \eqref{eq:Targert4RiOfCalL} is true.
\end{proof}

\subsection{Proof of claims from \cref{sec:ActualProofprop:UniqueAccuracy}}
\begin{proof}[Proof of \cref{claim:Bound4AverFatalPaths}]
Consider $\iter$ applied  to the teacher-student pair $(\G^*_i,\bsigma^*)$.  
Particularly, consider    $\iter(\G^*_i, \bsigma^*, \bsigma^*(\partial \alpha_i),  \bkappa^*)$, 
where   $\bkappa^*\in \alphabet^{\partial \alpha_i}$ is such that 
 $\bkappa^*(x)=\bsigma^*(x)$ for all $x\in \partial \alpha_i \setminus \{z\}$,
while $\bsigma^*(z)\neq\bkappa^*(z)$  and $\bkappa^*(z)$ is chosen so that the failure probability of the process is maximised. 

Let $\bX^*$ be the variable in \eqref{eq:DefOfXPiStar} defined with respect to the above process. 
\cref{thrm:ExpectedFatalPaths} implies that
\begin{align}\label{eq:ExpectedFatalPathsPlantedNoCycle}
\mathbb{E}[\bX^*_{z} \times \Ind\{\cK\}  \ | \  \cB,\  \G^*_i\in \CG] &\leq 
3(\log n)^{11/2} \cdot  
n^{-\left(1+\frac{\delta}{40\log (d k)}\right)} \enspace. 
\end{align}

\noindent
Contiguity, i.e.,  \cref{thrm:ContiquitySeq}, implies that for $\omega\to\infty$, arbitrarily slow,  for any $\gamma\in \mathbb{R}_{\geq 0}$  we have 
\begin{align}
\Pr [\bX_z=\gamma,\  \mathcal{K},\  \cB,\ \G_i\in \CG\ |\ \mathcal{C}_i(\omega)] 
&\leq \omega \cdot  \Pr [\bX^*_z=\gamma,\  \mathcal{K},\  \cB,\ \G^*_i\in \CG]
\nonumber \\
&\leq  \omega \cdot   \Pr [\bX^*_z=\gamma, \   \mathcal{K}\ |\  \cB,\ \G^*_i\in \CG] \enspace. 
\nonumber 
\end{align}
Since  $\bX_{z}, \bX^*_{z}\geq 0$, the above  implies that
\begin{eqnarray}
\mathbb{E}[\bX_{z} \times  \Ind\{\mathcal{K},\  \cB,\ \G_i\in \CG\}\ | \ \mathcal{C}_i(\omega) ] 
&\leq &\omega\cdot  \mathbb{E}[\bX^*_{z} \times \Ind\{ \cK\} \ | \ \cB,\ \G^*_i\in \CG] 
\enspace. \label{eq:ExpcXzWithProductClear}
\end{eqnarray}
Furthermore, we have that 
\begin{align}
\mathbb{E}[\bX_{z}  \times \Ind\{ \cK,\ \cB\} \ |\   \cE ] &\leq
\frac{\mathbb{E}[\bX_{z} \times  \Ind\{  \mathcal{K},\  \cB,\ \cE\}\ | \ \mathcal{C}_i(\omega) ] }
{\Pr[\cE\ |\  \mathcal{C}_i(\omega)] } \leq 
2  \cdot 
\mathbb{E}[\bX_{z} \times  \Ind\{\mathcal{K},\  \cB,\ \cE\}\ | \ \mathcal{C}_i(\omega) ]  \enspace.
 \label{eq:ExpcXZWithoutBoundary}
\end{align}
We use that   $\Pr[\cE\ |\  \mathcal{C}_i(\omega)]$ is lower bounded by $1/2$, i.e., 
since each one of the events  $ \cE, \cC_i(\omega)$ occurs with  probability $1-o(1)$.  
Finally, we have that 
\begin{align}
\mathbb{E} \left[ \bX_{z} \times \Ind\{\cK,\  \cB\}  \ |\  \cJ_i, \ \cE \right] 
&\leq 
\frac{\mathbb{E} \left[ \bX_{z}  \times \Ind\{ \cK,\ \cB\} \ |\  \cE \right ]}{\Pr[\cJ_i \ |\ \cE ]}
\leq 
\frac{\mathbb{E}[ \bX_{z} \times \Ind\{ \cK,\ \cB\} \ |\  \cE ]}
{\Pr[\cJ_i, \   \cK,\ \cB\ |\ \cE ]} \enspace. \nonumber
\end{align}
The claim follows by plugging \eqref{eq:ExpectedFatalPathsPlantedNoCycle} and 
\eqref{eq:ExpcXzWithProductClear} in the above inequality. 
\end{proof}

\begin{proof}[Proof of  \cref{claim:FinalBound4ExpXzK}]

Using  \cref{lemma:MarginalLambdaVsUniformDistrWRTRswitch},  
specifically \eqref{eq:Redaux2SingleBoundary} an \eqref{eq:UpdateSingeUpdate} 
in the proof of this lemma, 
we note that the event   $\sum_{z\in \partial \alpha}\bX_{z}  < |\alphabet|^{-k}/2$ implies  $\cJ_i$. 
Hence, we  have that 
\begin{align}
\Pr\left [ \bar{\cJ}_i, \  \cK,\ \cB \ |\  \cE \right] 
&\leq 
\Pr \left[ \Ind\{ \cK,\  \cB \}\times  \sum\nolimits_{z\in \partial \alpha_i}\bX_{z} 
\geq  |\alphabet|^{-k}/2 \ |\  \cE \right] 
\enspace. \label{eq:ProbGoodBoundaryGiClear}
\end{align}
Using   Markov's inequality  we have that 
\begin{align}\label{eq:ProbGoodBoundaryGiClearB}
\Pr \left [ \Ind\{\cK,\ \cB\}\times  \sum\nolimits_{z\in \partial \alpha_i}\bX_{z} \geq |\alphabet|^{-k}/2 \ |\   \cE \right]
&\leq \mathbb{E} \left[ \Ind\{ \cK,\ \cB\}\times \sum\nolimits_z\bX_{z}   \ |\   \cE  \right ] \cdot  2|\alphabet|^k\nonumber \\
&\leq  12 |\alphabet|^{k}k\omega(\log n)^{11/2} \cdot  n^{-\left (1+\frac{\delta}{40\log (d k)} \right)} \enspace,
\end{align}
where the last inequality follows by working as in \cref{claim:Bound4AverFatalPaths}.
Plugging the above into \eqref{eq:ProbGoodBoundaryGiClear} and noting that $|\alphabet|,k\in \Theta(1)$, we get  \eqref{eq:claim:FinalBound4ExpXzKRev}.

We proceed with proving \eqref{eq:claim:FinalBound4ExpXzKRevB}.
 For $\bsigma$ distributed as in $\mu_i$ and any 
$\kappa\in \alphabet^{\partial \alpha_i}$ we have that 
\begin{align}\label{eq:ProbLB4GinQ}
\Pr[\bsigma(\partial \alpha_i)&=\kappa\ |\ \cJ_i, \ \cE]  \geq |\alphabet|^{-k}-
||\mu_G-\zeta ||_{\partial \alpha_i} \geq |\alphabet|^{-k}/2 \enspace.
\end{align}
Also, since $\bX_z\geq 0$,  we have that 
\begin{eqnarray}
 \mathbb{E}[ \bX_{z} \times \Ind\{\cK, \ \cB\}  \ |\  \cE,\   \cJ_i, \  \bsigma(\partial \alpha_i)=\eta]
&\leq & \frac
{ \mathbb{E}[ \bX_{z} \times \Ind\{  \cK,\ \cB\}  \ |\ \cE,\  \cJ_i]}
{\Pr[\bsigma(\partial \alpha_i)=\eta\ |\ \cE,\   \cJ_i ]} \nonumber \\
&\leq & 2|\alphabet|^k \cdot \mathbb{E}[ \bX_{z} \times \Ind\{\cK,\ \cB\}  \ |\ \cJ_i, \ \cE ] \enspace,
\label{eq:BoundaryAmplification}
\end{eqnarray}
where in the last derivation we use \eqref{eq:ProbLB4GinQ}.
The above implies \eqref{eq:claim:FinalBound4ExpXzKRevB}. 
The claim follows. 
\end{proof}

\begin{proof}[Proof  of  \cref{claim:FinalBound4ExpCML}]

On the event $\FactVarInPath(\Xi,H)$, let  $\CylAi=\{x_a, x_b\}$,
while let $\cK^*$ be the event that the vertices in $\partial \alpha_i\setminus\{x_b\}$ are
at minimum distance $\ShortDist$. 

Also, let $\cN^{(a)}_i$ be the event that   $\G_i$ is such  that   $||\mu_i-\zeta ||_{\partial \alpha_i\setminus \{ x_a \}}  \leq (|\alphabet|^{-k+1}/2)$, while let $\cN^{(b)}_i$ be the event that 
$\psimin\geq n^{-(\log dk)^{-10}}$.  Clearly, we have that $\cN_i=\cN^{(a)}_i\cap \cN^{(b)}_i$.
A simple  union bound implies 
\begin{align}\label{base4claim:FinalBound4ExpCML}
\Pr[\bar{\cN}_i,\  \cB \ |\ \cE ] 
&\leq \Pr\left[\bar{\cN}^{(a)}_i,\ \cB \ |\ \cE \right] +\Pr\left[\bar{\cN}^{(b)}_i,  \cB \ | \ \cE \right] 
\leq \Pr\left[\bar{\cN}^{(a)}_i,\ \cB \ |\ \cE \right] 
+\frac{\Pr\left[\bar{\cN}^{(b)}_i \right]}{\Pr[\cE]} \enspace.
\end{align}

For the analysis purposes, we write the event $\FactVarInPath(\Xi,H)$ as the intersection
of the events $\FactVarInPath$ and $\cF(\Xi,H)$, where the latter event indicates 
that the set of factor nodes in  the path $P$ is $H$, while $\DiaSet=\Xi$.
We have
\begin{align}
\Pr\left[\bar{\cN}^{(a)}_i, \cL,\ \cF(\Xi,H),\ \cB \ |\  \cE \right] &=
\Pr\left[\bar{\cN}^{(a)}_i, \ \cL,\ \cB \ |\  \cF(\Xi,H), \ \cE \right] \times  
\Pr[ \cF(\Xi,H) \ |\  \cE ]  \enspace.
\label{eq:ProbCalLVsCalKStar}
\end{align}
Since $P$ is short path, the rightmost probability term, above, is non-zero only
when $|\Xi|, |H|=O(\log n)$. For the rest of the proof  we assume that the cardinality of 
the two sets $\Xi,H$ is $O(\log n)$.

Recall that $\cK^*$ denotes the event that the vertices in $\partial \alpha_i\setminus\{x_b\}$ are
at minimum distance $\ShortDist$. We have that 
\begin{align}\label{eq:claim:FinalBound4ExpCMLProbBoundBBBB}
\Pr\left[\bar{\cN}^{(a)}_i, \ \cL,\ \cB \ |\  \cF(\Xi,H), \ \cE \right]
&\leq \Pr\left[\bar{\cN}^{(a)}_i, \cK^*,\ \cB \ |\  \cF(\Xi,H), \ \cE \right]\enspace,
\end{align}
where the inequality follows since $\cL\subseteq \cK^*$.

Arguing as in  \cref{claim:FinalBound4ExpXzK}, we have that 
$\sum_{z\in \partial \alpha\setminus \{x_b\}}\bX_{z}  < |\alphabet|^{-k+1}/2$ implies  $\cN^{(a)}_i$.  
Working as in \cref{claim:Bound4AverFatalPaths} we have that
\begin{align}
\Pr\left[\bar{\cN}^{(a)}_i, \cK^*,\ \cB \ |\  \cF(\Xi,H), \ \cE \right] &\leq 
 \Pr \left[\Ind\{\cK^*,\ \cB \}\times \sum\nolimits_{z\in \partial \alpha\setminus \{x_b\}} \bX_z\geq |\alphabet|^{-k+1}/2 \ \ |\ \cF(\Xi,H), \ \cE \right] \nonumber \\
 &\leq 2|\alphabet|^{k-1} \cdot \mathbb{E}\left[\Ind\{\cK^*,\ \cB \}\times \sum\nolimits_{z\in \partial \alpha\setminus \{x_b\}} \bX_z\  |\ \cF(\Xi,H), \ \cE \right] \enspace,  \nonumber
\end{align}
where the last derivation is Markov's inequality. 

We bound the expectations above  similarly to how we work  in \cref{claim:Bound4AverFatalPaths}. This gives that 
\begin{align}
  \Pr[ \bar{\cN}_i, \ \partial \alpha^*_i \in \cK^*, \ \cB \ |\ \cF(\Xi,H),\  \cE ] & \leq 
 \omega(\log n)^{6}\cdot  n^{-\left (1+\frac{\delta}{40\log (d k)} \right)} \enspace.
  \label{eq:claim:FinalBound4ExpCMLProbBoundBB}
\end{align}
From \eqref{eq:claim:FinalBound4ExpCMLProbBoundBB} and \eqref{eq:claim:FinalBound4ExpCMLProbBoundBBBB}
we get that
\begin{align}
\Pr\left[\bar{\cN}^{(a)}_i, \ \cL,\ \cB \ |\  \cF(\Xi,H), \ \cE \right] 
&\leq  \omega(\log n)^{6}\cdot 
n^{-\left (1+\frac{\delta}{40\log (d k)} \right)} \enspace.
\end{align}
Plugging the above into \eqref{eq:ProbCalLVsCalKStar} and summing over all choices for
$\Xi$ and $H$ we get that
\begin{align}
\Pr\left[\bar{\cN}^{(a)}_i,\ \cB \ |\ \cE \right] &\leq 
\omega(\log n)^{6}\cdot 
n^{-\left (1+\frac{\delta}{40\log (d k)} \right)}\enspace. 
\end{align}
Furthermore, condition $\CC$ implies that $\Pr\left[\bar{\cN}^{(b)}_i \right]= O\left (n^{-3/2}\right)$,
while we have that $\Pr[\cE]=1-o(1)$. Plugging all the above into \eqref{base4claim:FinalBound4ExpCML} gives \eqref{eq:claim:FinalBound4ExpCMLProb}.

As far as \eqref{eq:FinalBound4ExpXzLNew} is concerned, note that 
 the  dependence between the configurations at  $x_a$ and $x_b$,  we cannot  get  a relation like  \eqref{eq:BoundaryAmplification}.  
Arguing as in the proof of   \cref{prop:AccuracyRCUpdate},
for any $\kappa\in \alphabet^{\partial \alpha_i}$ we have 
\begin{equation}\label{eq:CorrelatedLowerboubdBSIGMA}
\Pr[\bsigma(\partial \alpha_i)=\kappa\ |\ \cE,\ \cN_i ] \geq 
\left( |\alphabet|^{k-1} \cdot \chi \right)^{-1} \cdot   \psimin   \geq 
\textstyle \left (|\alphabet|^{k-1}\cdot  \chi \right)^{-1} \cdot  n^{-(\log dk)^{-10}} \enspace,
\end{equation}
 where in the last inequality we use that $\cN_i$ implies that $\psimin\geq n^{-(\log dk)^{-10}}$. We also have 
\begin{eqnarray}
 \mathbb{E}[ \bX_{z}\times \Ind\{ \FactVarInPath(\Xi,H),\  \cB\}  \ |\  \cE, \ \cN_i,\  \bsigma(\partial \alpha)=\eta] 
 &\leq & \frac{ \mathbb{E}[ \bX_{z}\times  \Ind\{\FactVarInPath(\Xi,H),\ \cB\}    \ |\  \cE,\ \cN_i ] }
 { \Pr[\bsigma(\partial \alpha_i)=\eta\ |\ \cE,\ \cN_i ] }
 \enspace.
\label{eq:BoundaryAmplificationLambda}
\end{eqnarray}
Then, we get \eqref{eq:FinalBound4ExpXzLNew}  
by combining  \eqref{eq:BoundaryAmplificationLambda} and 
\eqref{eq:CorrelatedLowerboubdBSIGMA}. 
The claim follows. 
\end{proof}

\subsection{Proof of  \cref{thrm:ExpectedFatalPaths,thrm:ExpectedFatalPathsWithCycles}} 
Consider $\iter(\G^*_i, \bsigma^*, \bsigma^*(\DiaSet),  \bkappa^*(\DiaSet))$  and let 
$z\in \bsigma^*(\DiaSet)\oplus \bkappa^*(\DiaSet)$. For $\ell\geq 1$, 
let $\pi=x_1, \ldots x_{\ell}$ be such that $\pi\in \Pi_{\ell,z}$.

On the event  $\UpJ^s_{\pi}=1$ and for $r\in [\ell]$ such that $x_r\in \pi$ 
is a factor node,  let $\UpU^{(r,s)}_{\pi}$ be the  failure probability for
$\iter(\G^*_i, \bsigma^*, \bsigma^*(\DiaSet),  \bkappa^*(\DiaSet))$ 
that updates  only $\pi$, while the order of the factor nodes that are chosen 
is such that  $x_r$ is the last to be updated.  
It is useful to note that 
 \begin{align}\label{eq:UsVsUrs}
 \UpU^s_{\pi}&=\textstyle \max\nolimits_{r}\left\{\UpU^{(s,r)}_{\pi} \right\} \leq \sum\nolimits_{r}\UpU^{(s,r)}_{\pi}\enspace. 
 \end{align}

For the special case where $s=0$,  the quantity $\UpU^{(r,s)}_{\pi}$  is meaningful 
for general $r$ only when   we have multiple disagreements in $\DiaSet$.
Also, 
with a slight abuse of the notation, we write $\UpU^{(s,r)}_{\pi}$ even when 
$s=\infty$. In this case, we will just ignore the parameter $r$.

To simplify our notation and the statement of our result we follow the convention to
assume that $\UpU^{(s,r)}_{\pi}=0$ when  $r,s$ take on values such that 
$\UpU^{(s,r)}_{\pi}$ is not meaningful in the way we describe them above.

\begin{theorem}\label{thrm:DisagreementDecayPathSetB}
For $0\leq i< m$, for  $\delta\in (0,1]$,  assume that $\mu_i$ satisfies $\setB$  with slack  $\delta$.  
For any positive integer $\ell \leq (\log n)^5$, for $z \in \DiaSet$ and any $\pi\in \Pi_{\ell,z}$ 
the following is true: 

There is a constant $\widehat{C}>0$  such that for any $0\leq s\leq \ell$, or $s=\infty$, and 
 $0\leq r\leq \ell$ we have that 
\begin{align}
 \mathbb{E}[\UpU^{(s,r)}_{\pi} \times \UpJ^{s}_{\pi} \ | \  \cB,\ \G^*_i \in \CG]  & \leq  \widehat{C} \cdot  \Upsigma(s)\cdot
 n^{-\ell}
 \left( (1-\delta) {k} / {d}\right)^{\lfloor \ell/2 \rfloor} \enspace, \nonumber
\end{align}
where $\Upsigma(s)=|\DiaSet|$ for $s\neq \infty$ and $\Upsigma(\infty)=n$. 
\end{theorem}
The proof of \cref{thrm:DisagreementDecayPathSetB} appears in \cref{sec:thrm:DisagreementDecayPathSetB}.

Furthermore, for $\pi\in \Pi_{\ell,z}$ we let the quantity 
\begin{align}\label{eq:DefOfBDStar}
\bD^*_{\pi}&=\sum\nolimits_{s: s\neq \infty}\sum\nolimits_{r} \UpU^{(s,r)}_{\pi} \times \UpJ^s_{\pi}
\enspace.
\end{align}
\cref{thrm:DisagreementDecayPathSetB} combined with \eqref{eq:UsVsUrs} 
implies  for any  integer $0\leq \ell \leq (\log n)^5$ and  $\pi\in \Pi_{z,\ell}$, we have 
\begin{align}\label{eq:Fromthrm:DisagreementDecayPathSetB}
\mathbb{E}\left[\bD^*_{\pi}   \ | \ \cB,\ \G^*_i \in \CG\right] 
&\leq \widehat{C} \cdot \ell \cdot (\ell+| \DiaSet|) \cdot n^{-\ell} \cdot
 \left( (1-\delta){k}/{d}\right)^{\lfloor \ell/2 \rfloor}\enspace.
\end{align}

\subsection{Proof of \cref{thrm:ExpectedFatalPaths}}\label{sec:thrm:ExpectedFatalPaths}
For   brevity, let  $\CG$ denote the event   $\G^*_i\in \CG$.
For  what follows, we let the quantities $r_0=\ShortDist$, $r_1=(\log n)^{5/2}$
and $r_2=(\log n)^3$. Also, we let the quantities
\begin{align}\nonumber
\UpS_1 &=   \sum\nolimits_{r_0\leq \ell \leq  r_1} \sum\nolimits_{\pi\in \Pi_{\ell,z}} 
\mathbb{E}[  \bD^*_{\pi} \times \Ind\{\cK\} \ | \ \cB,\ \CG] & \textrm{and} &&
\UpS_2&= \sum\nolimits_{\pi\in \Pi_{r_2,z}} 
\mathbb{E}[  \bD^*_{\pi} \times \Ind\{\cK\} \ | \ \cB,\ \CG] \enspace. 
\end{align}

From \eqref{eq:DefOfXPiStar} and \eqref{eq:DefOfBDStar} we have that   
\begin{equation} \label{eq:ExpXzKVsSa1S2}
\mathbb{E}[\bX^*_{z} \times \Ind\{\cK\} \ | \  \cB,\ \CG] 
\leq \UpS_1+\UpS_2 \enspace. \nonumber 
\end{equation}
For $\UpS_1$ we do not need to count $\ell <r_0$ because of the indicator $\Ind\{\cK\}$.

We use  \cref{thrm:DisagreementDecayPathSetB}  and  \eqref{eq:Fromthrm:DisagreementDecayPathSetB}  
to bound the quantity on the r.h.s. of the above inequality.
Specifically, since  the cardinality of $\Pi_{\ell, z}$  is at most 
$n^{\ell-1}  \left( \frac{d}{k}\right)^{\lfloor \ell/2\rfloor}$, 
for any $\pi\in \Pi_{\ell,z}$, we have that
\begin{equation} \label{eq:PropTest}
\UpS_1
 \leq    \sum\nolimits_{r_0\leq \ell \leq r_1} n^{\ell-1} \left( {d}/{k}\right)^{\lfloor \ell/2\rfloor} \cdot  \mathbb{E}\left [   \bD^*_\pi  \ |\ \cB,\ \CG   \right] \enspace.  
\end{equation}
Then, combining \eqref{eq:Fromthrm:DisagreementDecayPathSetB}  and  \eqref{eq:PropTest} we get that
\begin{align*} 
\UpS_1 &\leq   (\log n)^{11/2} \cdot n^{-1}\sum\nolimits_{r_0 \leq \ell \leq r_1} \textstyle \left( 1-\delta/2\right)^{\lfloor (\ell-1)/2\rfloor} %\nonumber \\
 \leq  2(\log n)^{11/2} \cdot n^{-\left(1+\frac{\delta}{40\log(dk)}\right)} \enspace. 
\end{align*}
where in the last inequality we use that  $r_0=\ShortDist$.
Working similarly, we get that $\UpS_2\leq n^{-(\log n)}$.

\cref{thrm:ExpectedFatalPaths} follows by plugging the two bounds for $\UpS_1$ and $\UpS_2$ 
into \eqref{eq:ExpXzKVsSa1S2}.
\hfill $\Box$

\subsection{Proof of \cref{thrm:ExpectedFatalPathsWithCycles}}
\label{sec:thrm:ExpectedFatalPathsWithCycles}

For   brevity, let  $\CG$ denote the event   $\G^*_i\in \CG$.
For  what follows, we let the quantities $r_0=\ShortDist$, $r_1=(\log n)^{5/2}$
and $r_2=(\log n)^3$.
Also, we let the quantities
\begin{align}\nonumber
\UpS_1 &=   \sum\nolimits_{r_0\leq \ell \leq  r_1} \sum\nolimits_{\pi\in \Pi_{\ell,z}} 
\mathbb{E}[  \bD^*_{\pi} \times \Ind\{\FactVarInPath(H,\Xi) \} \ | \ \cB,\ \CG] & \textrm{and} & 
\nonumber \\
\UpS_2&= \sum\nolimits_{\pi\in \Pi_{r_2,z}} 
\mathbb{E}[  \bD^*_{\pi} \times \Ind\{ \FactVarInPath(H,\Xi) \} \ | \ \cB,\ \CG] \enspace. 
\nonumber
\end{align}
From \eqref{eq:DefOfXPiStar} and \eqref{eq:DefOfBDStar} we have that   
\begin{equation}\label{eq:ExpXzKVsSa1S2Cycles}
\mathbb{E}[\bX^*_{z} \times \Ind\{ \FactVarInPath(H,\Xi)  \} \ | \  \cB,\ \CG] 
\leq \UpS_1+\UpS_2 \enspace.  
\end{equation}

Recall that the event $\FactVarInPath(H,\Xi)$ occurs with positive probability only if 
$|H|, |\Xi|=O(\log n)$.  Also, recall that $\pi$ does not intersect with $H,\Xi$ apart from
the variable node $z\in H$, while also $|\pi|=O(\log^{3} n)$. 

Then, from the definition of the teacher-student model, \eqref{def:Probs4GStar},  
and the fact that the sets involved are only $O(\log^3 n)$ it is standard that 
\begin{equation} \nonumber 
\mathbb{E}[\bD^*_{\pi } \times \Ind\{ \FactVarInPath(H,\Xi)  \}  \ | \ \cB,\  \CG] \leq 
(1+o(1))\Pr[\FactVarInPath(H,\Xi)  \ |\  \cB,\ \CG] 
\times \mathbb{E}[\bD^*_{\pi }  \ | \ \cB,\  \CG] \enspace.
\end{equation}
Then, we have that 
\begin{equation}
\UpS_1\leq (1+o(1))\Pr[ \FactVarInPath(H,\Xi)  \ | \ \cB,\ \CG]  \times 
\sum\nolimits_{r_0\leq \ell \leq  r_1} \sum\nolimits_{\pi\in \Pi_{\ell,z}} 
\mathbb{E}[\bD^*_{\pi} \ | \ \cB,\ \CG]\enspace.
\nonumber 
\end{equation}
Working as in  \cref{thrm:ExpectedFatalPaths} to bound the double sum above,  we get that
\begin{align}
 \UpS_1 
&\leq  3(\log n)^{11/2}\cdot 
\Pr[ \FactVarInPath(H,\Xi)  \ | \ \cB,\ \CG] \cdot 
 n^{-\left( 1+\frac{\delta}{10\log(dk)}\right)}\enspace, 
\nonumber 
\end{align}
while with very similar argument we obtain that 
$\UpS_2 \leq   \Pr[ \FactVarInPath(H,\Xi)  \ | \ \cB,\ \CG] \cdot  n^{-(\log n)}$.

\cref{thrm:ExpectedFatalPathsWithCycles} follows by plugging the 
two bounds we obtained for $\UpS_1$ and $\UpS_2$ into
\eqref{eq:ExpXzKVsSa1S2Cycles}.
\hfill $\Box$

\spreadpoint

\section{Detailed Balance}\label{sec:Prel4FinalAcc} \label{sec:RswitchResults} 

In this section, we show that the method $\iter$ satisfies a property which is  
reminiscent of the detail balance equation
for the Markov Chains.  Specifically, we study two different cases of detail balance for $\iter$, i.e., 
those shown in  Theorems \ref{thrm:DBE4Rswitch} and \ref{thrm:DBE4RCupdate}, respectively.

In order to avoid too many indices, we choose to drop them. %,  
Hence, suppose that at the $i$-th iteration of $\rsampler$ we deal with  the  {\em $\Psi$-factor graph} 
$G=(V,F,(\partial a)_{a\in F},(\psi_a)_{a\in F})$ with Gibbs distribution $\mu$. Also, assume that we 
insert into $G$ the factor node $\alpha$. 

As per standard notation, when the addition of  $\alpha$ introduces a new short cycle $C$, then we let  $\CylAi$ be the set of the two variable   
nodes in $\partial \alpha$ which  also belong to $C$. If $\alpha$ does not introduce a new short cycle, then we follow
the convention to assume that $\CylAi$ is the empty set.

For   $\eta, \kappa$ be two configurations at  $\partial \alpha$ and  $\theta, \xi\in \alphabet^V$ consider   $\iter({G}, \theta, \eta, \kappa)$, 
while  let
\begin{equation}\label{eq:TransitionProb4Rswitch}
\transprob_{\eta,\kappa}(\theta, \xi )=\Pr[\xi=\iter(G, \theta, \eta, \kappa)] \enspace. 
\end{equation}
For what follows assume that $\eta,\kappa$ are always {\em admissible}. That is, there are always configurations in $\theta,\xi\in \alphabet^V$
such that $\theta(\partial \alpha)=\eta$ and $\xi(\partial \alpha)=\kappa$, while both $\mu_{G}(\theta), \mu_G(\xi)>0$. 

To simplify our derivations, we allow   $\theta(\partial \alpha)\neq \eta$, or  $\xi(\partial \alpha)\neq \kappa$ in the definition above  by assuming that  we have   $\transprob_{\eta,\kappa}(\theta, \xi )=0$.

\begin{theorem}[Detailed balance]\label{thrm:DBE4Rswitch}
For any $x\in \partial \alpha\setminus \CylAi$, for any $\eta,\kappa \in \alphabet^{\partial \alpha}$   which differ  at $x$ 
we have that
\begin{align}\nonumber 
\mu_G(\theta)\transprob_{\eta,\kappa}(\theta, \xi) &=
\mu_G(\xi )\transprob_{\kappa,\eta}(\xi, \theta) & \forall \theta, \xi \in \alphabet^{V} \enspace.
\end{align}
\end{theorem}

The proof of Theorem \ref{thrm:DBE4Rswitch} appears in Section \ref{sec:thrm:DBE4Rswitch}.

\begin{remark}
For  $\eta,\kappa$ as in Theorem \ref{thrm:DBE4Rswitch}, note that  $\iter(G, \theta, \eta, \kappa)$ 
corresponds to  $\rswitch(G, \theta, \eta, \kappa)$. In that respect, it makes sense to claim that
$\rswitch$ satisfies the detailed balance property. 
\end{remark}

Assume now that the factor node $\alpha$ that is inserted into $G$ introduces a new short cycle, which we call $C$. 
As per standard notation, we let ${H}$ be the  which  is induced by the variable and factor nodes of the cycle $C$, as well as 
the variable nodes which are adjacent to  the  factor nodes in $C$. Also, let 
$\mu_H$  be the Gibbs distributions induced by $H$.

\begin{theorem}[Extended detailed balance]\label{thrm:DBE4RCupdate}
For admissible $\eta,\kappa \in \alphabet^{\partial \alpha}$  that differ at $\CylAi$
we have that
\begin{align}\nonumber  
\frac{\mu_G(\theta)} {\nu_{\eta}}\transprob_{\eta,\kappa}(\theta, \xi) &=
\frac{\mu_G(\xi) }{\nu_{\kappa}}\transprob_{\kappa,\eta}(\xi, \theta) & \forall \theta, \xi \in \alphabet^V  \enspace,
\end{align}
where $\nu_{\eta}$ is equal to  $\mu_{H, \partial \alpha }( \eta)$, similarly 
$\nu_{\kappa}$ is equal to  $\mu_{H, \partial \alpha }( \kappa )$. 
\end{theorem}

Note that the assumption that $\eta,\kappa$ are admissible, above, implies that both $\nu_{\eta},\nu_{\kappa}>0$. 

The proof of Theorem \ref{thrm:DBE4RCupdate} appears in Section \ref{sec:thrm:DBE4RCupdate}.

\subsection{Proof of Theorem \ref{thrm:DBE4Rswitch}}\label{sec:thrm:DBE4Rswitch}

Assume  that $\mu_G(\theta), \mu_G(\xi)>0$,  while  $\theta(\partial \alpha)=\eta$ and $\xi(\partial \alpha)=\kappa$.

For the setting we consider here, we note that $\iter(G, \theta, \eta, \kappa)$ corresponds to the execution
of $\rswitch(G, \theta, \eta, \kappa)$.  Hence, we consider $\transprob_{\eta,\kappa}(\theta, \xi)$ 
it terms of $\rswitch$, i.e.,  
\begin{align}\label{eq:RSWCHTVsTrnsProb}
\transprob_{\eta,\kappa}(\theta, \xi )=\Pr[\xi=\rswitch(G, \theta, \eta, \kappa)] \enspace. 
\end{align}
The analogous of course holds for $\transprob_{\kappa,\eta}(\xi,\theta)$ and 
$\rswitch(G, \xi, \eta, \kappa)$ that generates $\theta$.

In light of \eqref{eq:RSWCHTVsTrnsProb} it is immediate that if 
$\transprob_{\eta,\kappa}(\theta, \xi)=0$, then 
we also have  $\transprob_{\kappa,\eta}(\xi,\theta)=0$. Hence,  in this case,  
the detailed balance equation is
true. For what follows, we assume that  $\transprob_{\eta,\kappa}(\theta, \xi)>0$.

Recall that the process $\rswitch(G, \theta, \eta, \kappa)$ has two parts. The first one is the iterative part, i.e., 
from the initial disagreement  at  $x$, the process, iteratively, reassigns spins to  variable nodes  until 
the disagreement cannot propagate anymore.  In the second part, the process decides for 
the variable nodes that have not been considered in the iterative part, i.e., they keep the same configuration as in $\theta$.

Hence,  letting $\Fvisit\subseteq \visit$ contain only the factor nodes  in $\visit$,
we have that  
\begin{align}\label{eq:WeightCompOutsideRswitch}
\psi_{\beta}(\theta(\partial \beta ))&=\psi_{\beta}(\xi (\partial \beta)) &\forall \beta\notin  \Fvisit \enspace.
\end{align}

%%%
We let $\Fvisit_{\rm ext}\subseteq \Fvisit$ contain every factor node $\beta$ which has only one disagreeing 
neighbour. Also, we let $\Fvisit_{\rm int}=\Fvisit\setminus \Fvisit_{\rm ext}$. From the definition of the symmetric
weight functions in \eqref{eq:SymmetricWeightA} and the update rule in \eqref{eq:UpdateRuleA} and 
\eqref{eq:UpdateRuleB}, as well as the rule in \eqref{eq:ShortCycleUpdtA}  and \eqref{eq:ShortCycleUpdtB} 
for the short cycles, we have that
%%%
\begin{align}\label{eq:WeightCompInternalRswitch}
\psi_{\beta}(\theta(\partial \beta ))&=\psi_{\beta}(\xi (\partial \beta)) & \forall \beta \in \mathcal{ M}_{\rm int} \enspace.
\end{align}
For $\beta\in \mathcal{ M}_{\rm ext}$ we do not necessarily have an equality similar to the one above.

At this point, we  remark that all  sets $\visit, \Fvisit, \Fvisit_{\rm ext}$ and 
$\Fvisit_{\rm int}$  in the process   $\rswitch(G,\theta, \eta, \kappa)$  that generates $\xi$ 
are fully  specified by the configurations $\theta$ and $\xi$. 
In that respect, considering the ``reverse" process 
$\rswitch(G,\xi, \kappa, \eta)$ that generates $\theta$, the corresponding sets $\visit, \Fvisit$, $\Fvisit_{\rm ext}$ and 
$\Fvisit_{\rm int}$  are identical.

From  the observations in \eqref{eq:WeightCompOutsideRswitch} and \eqref{eq:WeightCompInternalRswitch}, 
we get that 
\begin{align}\label{eq:GibbsRatio4DBRswitch}
\frac{\mu (\theta )} {\mu(\xi )} &= 
 \prod\nolimits_{\beta \in \Fvisit_{\rm ext}} \frac{\psi_{\beta}(\theta(\partial \beta))}{\psi_{\beta}(\xi (\partial \beta))} \enspace.
\end{align}

We continue with studying  the ratio $\frac{\transprob_{\eta,\kappa}(\theta, \xi)}{\transprob_{\kappa,\eta}(\xi, \theta)}$.
Consider  $\rswitch(G,\theta, \eta, \kappa)$ that generates $\xi$. Assume that for the iterative part, 
the process considers the factor nodes in $\Fvisit$ in some predefined order, i.e., there is a rule which indicates 
which factor node $\beta$ to choose next among the available ones at each iteration. 
Assume that the same rule applies to  $\rswitch(G,\xi, \kappa, \eta)$ when it generates $\theta$.

Let $\beta_1, \beta_2, \ldots$ be the factors nodes in the orders that are considered by the two processes, 
i.e., at the iteration $t$ each one of them considers the factor node $\beta_t$.

For $\rswitch(G,\theta, \eta, \kappa)$ and $t>0$, we define $\mathcal{K}_t$ to be the event:
that the process decides that the output configuration of $\partial \beta_t$ is $\xi(\partial \beta_t)$. 
Letting 
$\transprob_{\eta,\kappa}(\beta_t)=\Pr[ \mathcal{K}_t\ |\ \cap_{j<t}\mathcal{K}_j]$,
we have that 
\begin{align}  
\transprob_{\eta,\kappa}(\theta, \xi) &= \prod\nolimits_{\beta_t \in \Fvisit} \transprob_{\eta,\kappa}(\beta_t) \enspace. \nonumber
\end{align}
The above implies that 
\begin{align}\label{eq:Base4TransProbRatioBRswitch}
 \frac{\transprob_{\eta,\kappa}(\theta, \xi)}{\transprob_{\kappa,\eta}(\xi, \theta)} =
 \prod\nolimits_{\beta_t \in \Fvisit} \frac{\transprob_{\eta,\kappa}(\beta_t)}{\transprob_{\kappa,\eta}(\beta_t)}
&= \prod\nolimits_{\beta_t \in \Fvisit_{\rm int}} \frac{\transprob_{\eta,\kappa}(\beta_t)}{\transprob_{\kappa,\eta}(\beta_t)}
\times \prod\nolimits_{\beta_t \in \Fvisit_{\rm ext}} \frac{\transprob_{\eta,\kappa}(\beta_t)}{\transprob_{\kappa,\eta}(\beta_t)} \enspace.
\end{align}
For estimating the ratios in \eqref{eq:Base4TransProbRatioBRswitch} we use the following claim.

\begin{claim}\label{claim:TransitionMargsRatioRswitch}
For any $\beta_t \in \mathcal{M}_{\rm int}$ we have that $\transprob_{\eta,\kappa}(\beta_t)=\transprob_{\kappa,\eta}(\beta_t).$
Also, for any $\beta_t \in \mathcal{M}_{\rm ext}$ we have that
\begin{align}\label{eq:claim:TransitionMargsRatioRswitch}
\frac{\transprob_{\eta,\kappa}(\beta_t)}{\transprob_{\kappa,\eta}(\beta_t)}=\frac{\psi_{\beta_t}(\xi(\partial \beta_t))}{\psi_{\beta_t}(\theta(\partial \beta_t))} \enspace.
\end{align}
\end{claim}

\noindent
Combining Claim \ref{claim:TransitionMargsRatioRswitch} and
 \eqref{eq:Base4TransProbRatioBRswitch} we get that
\begin{equation}\label{eq:PrTrnsRatio4DBRswitch}
\frac{\transprob_{\eta,\kappa}(\theta, \xi)}{\transprob_{\eta,\kappa}(\xi, \theta)} =
 \prod\nolimits_{\beta_t \in \Fvisit_{\rm ext}} \frac{\psi_{\beta_t}(\xi (\partial \beta_t))}{\psi_{\beta_t}(\theta (\partial \beta_t))} \enspace.
\end{equation}
Then, from \eqref{eq:GibbsRatio4DBRswitch} and \eqref{eq:PrTrnsRatio4DBRswitch} it is immediate that
$
\frac{\mu (\theta )} {\mu(\xi )} \times \frac{\transprob_{\eta,\kappa}(\theta, \xi)}{\transprob_{\eta,\kappa}(\xi, \theta)}=1, 
$
which proves Theorem \ref{thrm:DBE4Rswitch}. \hfill $\Box$

\begin{proof}[Proof of Claim \ref{claim:TransitionMargsRatioRswitch}]

Throughout this proof we abbreviate $\beta_t$ to $\beta$.    First, we consider the case of 
$\beta \in \Fvisit_{\rm int}$, but not in a short cycle.
We write both $\transprob_{\eta,\kappa}(\beta)$ and $\transprob_{\kappa,\eta}(\beta)$, in terms of the weight
function $\psi_{\beta}$. Particularly, using \eqref{eq:BroadCastDisagreementProb} we have that
\begin{align}\label{eq:TransProb4Beta}
\transprob_{\eta,\kappa}(\beta) &= \max \left \{0, 1-\frac{\psi_{\beta}(\theta^*(\partial \beta))}
{\psi_{\beta}(\theta(\partial \beta))} \right\} & \textrm{and} &&
\transprob_{\kappa,\eta}(\beta)&=\max \left \{0, 1-\frac{\psi_{\beta}(\xi^*(\partial \beta))}
{\psi_{\beta}(\xi(\partial \beta))} \right\} \enspace,
\end{align}
where  $\theta^*(\partial \beta)$ and $\xi^*(\partial \beta)$ are defined as follows: 
Note that the set $\DisSpin$ of the spins of disagreement is the same for  both $\rswitch(G,\theta, \eta, \kappa)$
 and  $\rswitch(G,\xi, \kappa, \eta)$. There is $j\in [k]$ such that   $\DisSpin=\{\theta (\partial_j\beta), \xi(\partial_j \beta) \}=$
$\{\theta^* (\partial_j\beta), \xi^*(\partial_j \beta) \}$,  while
%%%
\begin{align}\nonumber
\theta^*(\partial_j \beta)&= \xi(\partial_j \beta),  &\xi^*(\partial_j\beta)& = \theta (\partial_j\beta)
\end{align}
and
\begin{align}\nonumber
\xi^*(\partial_r \beta) &=\xi(\partial_r \beta), & \theta^*(\partial_r \beta) &=\theta(\partial_r \beta)& \forall r\in [k]\setminus \{j\} \enspace.
\end{align}
That is, $\theta^*(\partial \beta)$ and $\theta(\partial \beta)$ differ only on $\partial_j\beta$. The same holds for 
$\xi^*(\partial \beta)$ and $\xi(\partial \beta)$.

The above implies that  we get $\xi^*(\partial \beta)$ from $\theta^*(\partial \beta)$ by exchanging the spin-classes of
the elements in $ \DisSpin$. Note that a similar relation holds between $\xi(\partial \beta)$ and $\theta(\partial \beta)$.
Hence,  for such $\xi^*(\partial \beta)$, $\theta^*(\partial \beta)$ and $\xi(\partial \beta)$, $\theta(\partial \beta)$
 \eqref{eq:SymmetricWeightA} implies that 
\begin{align}\label{eq:Weights4TransProb4Beta}
\psi_{\beta}(\theta^*(\partial \beta))&=\psi_{\beta}(\xi^*(\partial \beta)) &  \textrm{and} &&
\psi_{\beta}(\theta (\partial \beta))&=\psi_{\beta}(\xi(\partial \beta)) \enspace.
\end{align}
Combining \eqref{eq:TransProb4Beta} and \eqref{eq:Weights4TransProb4Beta} we get that
$\transprob_{\eta,\kappa}(\beta)=\transprob_{\kappa,\eta}(\beta)$.

The case where $\beta \in \Fvisit_{\rm int}$ and also belongs to a short cycle follows immediately 
since the choices in \eqref{eq:ShortCycleUpdtA} are deterministic. 

We proceed with the case where $\beta\in \Fvisit_{\rm ext}$ but not in a short cycle. 
As before, we write both $\transprob_{\eta,\kappa}(\beta)$ and $\transprob_{\kappa,\eta}(\beta)$, in terms of the weight
function $\psi_{\beta}$. Particularly, using \eqref{eq:BroadCastDisagreementProb} we have that
\begin{align}\label{eq:TransProb4BetaExternal}
\transprob_{\eta,\kappa}(\beta) &= \min \left \{1, \frac{\psi_{\beta}(\xi(\partial \beta))}
{\psi_{\beta}(\theta(\partial \beta))} \right\}, 
&\transprob_{\kappa, \eta}(\beta)&=\min \left \{1, \frac{\psi_{\beta}(\theta(\partial \beta))}
{\psi_{\beta}(\xi(\partial \beta))} \right\} \enspace.
\end{align}

If  $\psi_{\beta}(\xi(\partial \beta)) \geq \psi_{\beta}(\theta(\partial \beta))$,  then \eqref{eq:TransProb4BetaExternal} implies that
$\transprob_{\eta,\kappa}(\beta) = 1$ and 
$\transprob_{\kappa,\eta}(\beta)= \frac{\psi_{\beta}(\theta(\partial \beta))}{\psi_{\beta}(\xi(\partial \beta))}$,
which in turn implies \eqref{eq:claim:TransitionMargsRatioRswitch}. 
Similarly, if $\psi_{\beta}(\xi(\partial \beta)) \leq \psi_{\beta}(\theta(\partial \beta))$,
then \eqref{eq:TransProb4BetaExternal} implies that
$\transprob_{\eta,\kappa}(\beta) = \frac{\psi_{\beta}(\xi(\partial \beta))}{\psi_{\beta}(\theta(\partial \beta))}$ 
and $\transprob_{\kappa,\eta}(\beta)=1$
which in turn implies \eqref{eq:claim:TransitionMargsRatioRswitch}.

It only remains to consider the case where $\beta\in \mathcal{M}_{\rm ext}$ and at the same time $\beta$ 
belong to a short cycle. Then, from the rule \eqref{eq:ShortCycleUpdtA} and \eqref{eq:ShortCycleUpdtB} 
the following is immediate: If $z$ is the disagreeing vertex in $\partial \beta$, then, since 
$\beta\in \mathcal{M}_{\rm ext}$, there is no $x\in \partial \beta\setminus\{z \}$ such that $\theta(x)\in \DisSpin$. 
Then, from \eqref{eq:SymmetricWeightA}, we conclude that 
$\psi_{\beta}(\theta(\partial \beta))=\psi_{\beta}(\xi(\partial \beta))$. Furthermore, we have that 
$\frac{\transprob_{\eta,\kappa}(\beta)}{\transprob_{\kappa,\eta}(\beta)}=1$, because the rule 
\eqref{eq:ShortCycleUpdtA} is deterministic. The above observations imply 
\eqref{eq:claim:TransitionMargsRatioRswitch}.

From all the above conclude the proof of Claim \ref{claim:TransitionMargsRatioRswitch}.
\end{proof}

\subsection{Proof of Theorem \ref{thrm:DBE4RCupdate}}\label{sec:thrm:DBE4RCupdate}

Consider the process $\iter(G, \theta, \eta, \kappa)$ for   $\eta,\kappa\in \alphabet^{\partial \alpha}$ such that
$\eta\oplus \kappa\subset \CylAi$.  Our focus is on this process when it outputs the configuration $\xi$.
Similarly to what we had in the proof of Theorem \ref{thrm:DBE4Rswitch},
we consider what we call the reverse process $\iter(G, \xi, \kappa, \eta)$ that outputs
the configuration $\theta$.

It is straightforward to verify that in our setting the following is true: if 
$\transprob_{\eta,\kappa}(\theta, \xi)=0$, then we also have that 
$\transprob_{\kappa,\eta}(\xi,\theta)=0$. Hence, we focus on the case where
both $\transprob_{\eta,\kappa}(\theta, \xi), \transprob_{\kappa,\eta}(\xi,\theta)>0$.

Let
$\hat{\transprob}_{\eta,\kappa}(\theta, \xi)$ be the probability that the process 
changes  the configuration at $V(H)$  from $\theta_{H}$ to $\xi_H$, where
$\theta_{H}=\theta(V(H))$ and $\xi_H=\xi(V(H))$.
Similarly,  we define $\hat{\transprob}_{\kappa,\eta}(\xi, \theta)$, with respect to 
the (reverse)  process $\iter(G, \xi,\kappa, \eta)$. 
From the pseudo-code of $\iter$ in Algorithm \ref{rsampler} -line 8- we have  that
\begin{align}\label{eq:TransProbKVsWeights}
\hat{\transprob}_{\eta,\kappa}(\theta, \xi)&=\frac{\psi_{H}(\xi_H)}{Z^{\kappa}_{H}}
&\textrm{and}&
&\hat{\transprob}_{\kappa,\eta}(\xi, \theta)&=\frac{\psi_{H}(\theta_H)}{Z^{\eta}_{H}} \enspace,
\end{align}
where $\psi_{H}$ is the product of the weight functions $\psi_{\beta}$ with $\beta$ varying 
over the factor nodes in ${H}$. Also, $Z^{\eta}_{H}$ is the sum of $\psi_{H}(\sigma)$,  over  $\sigma$, 
such that  $\sigma(\partial \alpha)=\eta$. We define $Z^{\kappa}_{{H}}$ similarly.

Assume that $\iter(G, \theta, \eta, \kappa)$  changes  the assignment at $V(H)$ from  $\theta_H$ to $\xi_H$. 
Recall that, subsequently,  the process works as follows: let 
the set $\initDis=\{z_1, z_2, \ldots, z_t\}$  be such that   $\initDis=\theta_H \oplus\xi_H$, i.e., the set
of  nodes at which  $\theta_H$,  $\xi_H$ disagree. 

The process considers 
the sequence of configurations $\theta_0,  \ldots, \theta_{t}$ of ${H}$ such that 
$\theta_0=\theta_H$ and $\theta_{t}=\xi_H$, while each $\theta_j$ is obtained 
from $\theta_H$ by changing the assignment of the variable nodes $y\in \{z_1, \ldots, z_j\}$
from $\theta(y)$ to $\xi(y)$. 

Then, it applies the iteration at lines 14 and 15, in the pseudo-code of $\rsampler$ (see Algorithm \ref{rsampler}).
That is, letting $\bbeta_0=\theta$, it sets
\begin{align}\label{eq:RCSwitchIterationAgain}
\bbeta_j& =\rswitch(\bar{G}, \bbeta_{j-1}, \theta_{j-1}, \theta_j) & \textrm{for }j=1,\ldots,t \enspace,
\end{align}
where $\bar{G}$ is obtained  from $G$ be deleting all the factor nodes that also belong to ${H}$.

Consider  $\iter(G, \xi, \kappa, \eta)$, i.e., the  reverse process.  
Then, the corresponding iteration to \eqref{eq:RCSwitchIterationAgain}  is as follows: let 
$\hat{\bbeta}_0=\xi$, set 
\begin{align} 
\hat{\bbeta}_j& =\rswitch(\bar{G}, \hat{\bbeta}_{j-1}, \theta_{t-(j-1)}, \theta_{t-j}) & \textrm{for }j=1,\ldots, t \enspace. \nonumber
\end{align}

\begin{claim}\label{claim:DB4Rupdate}
We have that
$\mu_{\bar{G}}\cdot(\theta)\Pr[\bbeta_t=\xi]=\mu_{\bar{G}}(\xi)\cdot \Pr[\hat{\bbeta}_t=\theta].$
\end{claim}

It is easy to check that 
\begin{align}\nonumber
\transprob_{\eta,\kappa}(\theta, \xi) &=\hat{\transprob}_{\eta,\kappa}(\theta, \xi)\times \Pr[\bbeta_t=\xi] &\textrm{and} &&
\transprob_{\kappa,\eta}(\xi,\theta) &=\hat{\transprob}_{\kappa, \eta}(\xi,\theta)\times \Pr[\hat{\bbeta}_t=\theta] \enspace.
\end{align}
Combining the above with \eqref{eq:TransProbKVsWeights} and Claim \ref{claim:DB4Rupdate} we get that
\begin{align}
\frac{\transprob_{\eta,\kappa}(\theta, \xi )}
{\transprob_{\kappa,\eta}(\xi,\theta)}&=\frac{\psi_{H}(\xi_{H})}{Z^{\kappa}_{H }}\times \frac{Z^{\eta}_{H}}{\psi_{H}(\theta_H)} \times \frac{\psi_{\bar{G}}(\xi)}{\psi_{\bar{G}}(\theta)} =
 \frac{Z^{\eta}_{H}}{Z^{\kappa}_{H }} \times \frac{\psi_{{G}}(\xi)}{\psi_{{G}}(\theta)} \enspace, \label{eq:TransProsRatioRCUpdate}
\end{align}
where for the second equality we use that $\psi_{H }(\xi_H) \times \psi_{\bar{G}}(\xi)=\psi_{G}(\xi)$ and
$\psi_{H }(\theta_H) \times \psi_{\bar{G}}(\theta)=\psi_{G}(\theta)$. 
Furthermore, from the definition of the corresponding quantities and straightforward derivation, we get 
\begin{align}\label{eq:MeasureDef4DBRCUpdate}
\frac{\mu_G(\xi)}{\mu_G(\theta)}&=\frac{\psi_{{G}}(\xi)}{\psi_{{G}}(\theta)} & \textrm{and}&&
\frac{Z^{\eta}_H}{Z^{\kappa}_H}&=\frac{\nu_{\eta}}{\nu_{\kappa}}\enspace. 
\end{align}

The theorem follows by   plugging  \eqref{eq:MeasureDef4DBRCUpdate} into \eqref{eq:TransProsRatioRCUpdate}.
\hfill $\Box$

\begin{proof}[Proof of Claim \ref{claim:DB4Rupdate}]
We prove the claim by using the detailed balance property of $\rswitch$, i.e., Theorem \ref{thrm:DBE4Rswitch}.
Consider an $t$-tuple of configurations $\bxi=(\xi_0,\xi_1, \xi_2, \ldots, \xi_t)$ such that $\xi_r\in \alphabet^V$
and $\mu_{\bar{G}}(\xi_r)>0$,  for $r=0,\ldots t$. Let
\begin{align}\nonumber
{\rm P}_j({\bxi})&=\Pr[\xi_j=\rswitch(G, \xi_{j-1}, \theta_{j-1}, \theta_j)] & \textrm{for }j=1,\ldots, t \enspace.
\end{align}
Similarly, for the reverse process, we let
\begin{align}\nonumber
{\rm Q}_j(\bxi)&=\Pr[\xi_{t-j}=\rswitch(G, \xi_{t-(j-1)}, \theta_{t-(j-1)}, \theta_{t-j})] & \textrm{for }j=1,\ldots, t \enspace.
\end{align}

Let $\mathcal{L}$ be the set of $t$-tuples of configurations as above such that $\xi_0=\theta$ and $\xi_t=\xi$.
We have that
\begin{align}\label{eq:DBRUpdateStepA}
\Pr[\bbeta_t=\xi]& =  \sum\nolimits_{\bxi\in \mathcal{ L}} \Pr[\wedge^{t}_{j=1}\bbeta_j=\xi_j ], &
\Pr[\hat{\bbeta}_t=\theta]& =  \sum\nolimits_{\bxi\in \mathcal{ L}}   \Pr[\wedge^{t}_{j=1}\hat{\bbeta}_j=\xi_{t-j} ] \enspace.
\end{align}
Furthermore, from the definition of the corresponding quantities, for every $\bxi\in \mathcal{ L}$,  we have
\begin{align}\label{eq:DBRUpdateStepB}
 \Pr[\wedge_{j\in [t]}\bbeta_j=\xi_j ] &=  \prod\nolimits_{j\in [t]}{\rm P}_j(\bxi), &
 \Pr[\wedge_{j\in [t]}\hat{\bbeta}_j=\xi_{t-j} ]&=  \prod\nolimits_{j\in [t]}{\rm Q}_j(\bxi) \enspace.
\end{align}
From Theorem \ref{thrm:DBE4Rswitch} we get the following: For any $\bxi\in \mathcal{ L}$ we have that
\begin{align}\label{eq:DetailBalanceJ}
\mu_{\bar{G}}(\xi_{j-1}){\rm P}_j(\bxi)&=\mu_{\bar{G}}(\xi_{j}){\rm Q}_{t-j}(\bxi) &\textrm{for }j=1,\ldots, t \enspace.
\end{align}
Multiplying all the  equalities in \eqref{eq:DetailBalanceJ}, we get that
\begin{align}\nonumber 
 \mu_{\bar{G}}(\xi_0)\prod\nolimits_{j\in [t]}{\rm P}_j(\bxi)=\mu_{\bar{G}}(\xi_t)\prod\nolimits_{j\in [t]}{\rm Q}_j(\bxi) \enspace.
\end{align}
Note that for each $\bxi\in \mathcal{ L}$ we have $\xi_0=\theta$ and $\xi_t=\xi$. Summing over 
 $\bxi\in \mathcal{ L}$ the above equations, we have
\begin{align}\nonumber 
 \mu_{\bar{G}}(\theta)\sum\nolimits_{\bxi\in \mathcal{ L}}\prod\nolimits_{j\in [t]}{\rm P}_j(\bxi)
 =\mu_{\bar{G}}(\xi)\sum\nolimits_{\bxi\in \mathcal{ L}}\prod\nolimits_{j\in [t]}{\rm Q}_j(\bxi) \enspace. 
\end{align}
The claim follows by substituting the sums in the equality above using \eqref{eq:DBRUpdateStepA} 
and \eqref{eq:DBRUpdateStepB}.
\end{proof}

\spreadpoint
\section{Correlation Decay Vs Failure Probabilities}\label{sec:CorrelVsFail}

In order to avoid too many indices, when there is no danger of confusion we choose to drop them. 
Suppose that at the $i$-th iteration of $\rsampler$, we deal with  the  {\em $\Psi$-factor graph} 
$G=(V,F,(\partial a)_{a\in F},(\psi_a)_{a\in F})$.  Also, assume that we insert into $G$ the factor node $\alpha$. 

When the addition of  $\alpha$ introduces a new short cycle $C$, then we let  $\CylAi$ be the set of the two variable   
nodes in $\partial \alpha$ which  also belong to $C$. If $\alpha$ does not introduce a new short cycle, then we follow
the convention to assume that $\CylAi$ is the empty set.

Here we study properties of the Gibbs measure $\mu_G$ which are useful for the analysis. ]
To this end,  utilise  Theorems \ref{thrm:DBE4Rswitch} and \ref{thrm:DBE4RCupdate}.

\begin{lemma}\label{lemma:Point2SetWithRswitch}
Let $\Lambda \subseteq \partial \alpha\setminus \CylAi$.  For any  $x\in \Lambda$ and any $\eta, \kappa\in \alphabet^x$ we have that
\begin{equation} \nonumber  
||\mu_G(\cdot \ |\ x,\eta)- \mu_G(\cdot \ |\ x,\kappa) ||_{\Lambda\setminus\{x\}} \leq 2 \rsq_x \enspace.
\end{equation}
\end{lemma}
The proof of Lemma \ref{lemma:Point2SetWithRswitch} appears in Section \ref{sec:lemma:Point2SetWithRswitch}.

Lemma \ref{lemma:Point2SetWithRswitch}  bounds the effect of the configuration at $x$ on the distribution of the 
configuration at $\Lambda\setminus \{x\}$ w.r.t. the measure $\mu_G$.
Furthermore, we  have the following result.
\begin{lemma} \label{lemma:MarginalLambdaVsUniformDistrWRTRswitch}
For $\Lambda \subseteq \partial \alpha\setminus \CylAi$, let    $\zeta$ be  the uniform distribution on  $\alphabet^{\Lambda}$.
We have that 
$
||\mu_G-\zeta ||_{\Lambda} \leq  2\sum\nolimits_{z \in \Lambda} \rsq_z.
$
\end{lemma}
The proof of Lemma \ref{lemma:MarginalLambdaVsUniformDistrWRTRswitch} appears in
Section \ref{sec:lemma:MarginalLambdaVsUniformDistrWRTRswitch}.

For what follows, let $\bar{H}$ be the subgraph of $G$ that is induced by the variable nodes and
the factor nodes in what becomes a short cycle $C$ after the insertion of $\alpha$ into $G$. 
Note that $\bar{H}$ is acyclic  as it does not include $\alpha$, however it includes $\partial \alpha$ as
these variable nodes are already in $G$. The same of course  holds for the set $\CylAi$. 

\begin{lemma}\label{lemma:FactoringMargOfPathEnd}
We have that
$||\mu_{G}-\mu_{\bar{H}} ||_{\CylAi} \leq 2\rcsq_{\CylAi}$.
\end{lemma}
The proof of Lemma \ref{lemma:FactoringMargOfPathEnd} appears in Section \ref{sec:lemma:FactoringMargOfPathEnd}.

\begin{lemma}\label{lemma:FactoringOfBarH}
We have that
$ || \mu_{G}-\mu_{\bar{H}}||_{\partial \alpha} \leq  
2\cdot \left( \rcsq_{\CylAi}+\sum\nolimits_{x\in \partial \alpha\setminus \CylAi}\rsq_{x}\right).$
\end{lemma}

The proof of Lemma \ref{lemma:FactoringOfBarH} appears in Section \ref{sec:lemma:FactoringOfBarH}

\subsection{Proofs from Section \ref{sec:CorrelVsFail}}\label{sec:ProofCorrelVsFail}

\subsubsection{Proof of Lemma \ref{lemma:Point2SetWithRswitch}}\label{sec:lemma:Point2SetWithRswitch}
W.l.o.g. assume that $\Lambda=\partial \alpha\setminus M$. 
It suffices to show that for  $\bsigma_{\eta}$ and $\bsigma_{\kappa}$  distributed as in $\mu(\cdot\ |\ x, \eta)$ and 
$\mu(\cdot\ |\ x, \kappa)$, respectively,   there is a coupling  such that 
\begin{equation}\label{eq:MarginalLambdaVsUniformDistr:CouplingCond}
\Pr[\bsigma_{\eta}(M) \neq \bsigma_{\kappa}(M)] \leq 2\rsq_x \enspace,
\end{equation}
where   $M=\Lambda\setminus \{x\}$.

Let $\bbeta_0=\iter(G, \bsigma_{\eta}, \bsigma_{\eta}(\partial \alpha ), \mathbold{\theta})$, where
$\mathbold{\theta}$ is a configuration at $M$ such that $\mathbold{\theta}(x)=\kappa$,
while $\mathbold{\theta}(M)=\bsigma_{\eta}(M)$. Also, let 
$\bbeta_1=\iter(G, \bsigma_{\eta}, \eta, \kappa)$.

The choice of the parameters implies that  $\iter(G, \bsigma_{\eta}, \bsigma_{\eta}(\partial \alpha ), \mathbold{\theta})$
and $\rswitch(G, \bsigma_{\eta}, \bsigma_{\eta}(\partial \alpha ), \mathbold{\theta})$ are identical 
processes.
The same holds for  $\iter(G, \bsigma_{\eta}, \eta, \kappa)$ and $\rswitch(G, \bsigma_{\eta}, \eta, \kappa)$.

Note that both  $\iter(G, \bsigma_{\eta}, \bsigma_{\eta}(\partial \alpha ), \mathbold{\theta})$ and 
$\iter(G, \bsigma_{\eta}, \eta, \kappa)$ have the same input configuration $\bsigma_{\eta}$.  For the first
process, having $\mathbold{\theta}$ at the input   implies that it fails  if it attempts to change the configuration 
in $M$. This is not true for $\iter(G, \bsigma, \eta, \kappa)$  as the  process is  allowed to change 
$M\setminus \{x\}$ since $\kappa$ involves only   $x\in M$.  We couple $\bbeta_0, \bbeta_1$ by means of a 
coupling  between  $\iter(G, \bsigma_{\eta}, \bsigma_{\eta}(\partial \alpha ), \mathbold{\theta})$  and  
$\iter(G, \bsigma_{\eta}, \eta, \kappa)$, i.e., the processes that generate them. 

We couple the processes that generate $\bbeta_0, \bbeta_1$ as close as possible. 
Since both processes  take the same input $\bsigma_{\eta}$ and $\mathbold{\theta}(x)=\kappa(x)$,
they evolve identically for as long as they do not consider $M$ in their iterative part.
Specifically,  we only have $\bbeta_1\neq \bbeta_0$  
 when both processes attempt to change  the assignment at $M$. Then, the first process fails, whereas the second one continuous.

Then, we  prove \eqref{eq:MarginalLambdaVsUniformDistr:CouplingCond} by  working as follows: 
We couple $\bbeta_1$ with ${\bsigma}_{\kappa}$  optimally.  Also, we couple $\bbeta_0$ with  $\bbeta_1$ as described above.  
Finally, we couple  $\bsigma_{\eta}$ with $\bbeta_0$ so that  $\bsigma_{\eta}(M)$ and $\bbeta_0(M)$
are as close as possible. 

Note that we have $\bsigma_{\eta}(M)=\bbeta_0(M)=\bbeta_1(M)$ 
when the two processes do not fail, i.e., in this case,  the configuration at $M$ in the two processes does 
not change from its initial configuration  $\bsigma_{\eta}(M)$. Since $\bbeta_0=\bbeta_1$ only when the two processes
do not fail, we   conclude that  if $\bbeta_0=\bbeta_1= {\bsigma}_{\kappa}$, then  we also have that ${\bsigma}_{\kappa}(M)=\bsigma_{\eta}(M)$. 
Hence, we get that 
\begin{equation}\label{eq:MarginalLambdaVsUniformDistr:CouplingBound}
\Pr[\bsigma_{\eta}(M)\neq {\bsigma}_{\eta}(M)] \leq \Pr[\bbeta_0\neq \bbeta_1, \ \textrm{or}\ \bbeta_1\neq \bsigma_{\kappa}] 
\leq \Pr[\bbeta_0\neq \bbeta_1]+ \Pr[ \bbeta_1\neq \bsigma_{\kappa}] \enspace,
\end{equation}
where the second inequality is from the union bound. In light of \eqref{eq:MarginalLambdaVsUniformDistr:CouplingBound}, 
we get \eqref{eq:MarginalLambdaVsUniformDistr:CouplingCond} by showing that 
\begin{align}\label{eq:TargetABC4MarginalLambdaVsUniformDistr}
\Pr[\bbeta_0\neq \bbeta_1] & \leq \rsq_x, & \Pr[\bbeta_1\neq {\bsigma}_{\kappa}] &\leq \rsq_x \enspace.
\end{align}
In what follows, 
 let $\BadWSUpdate(G, \tau, \xi, \hat{\xi})$ denote  the event that  $\iter(G, \tau, \xi, \hat{\xi})$ fails,
for any choice of configurations $\tau, \xi, \hat{\xi}$,

For the leftmost inequality in \eqref{eq:TargetABC4MarginalLambdaVsUniformDistr},  note that  $\bbeta_0\neq \bbeta_1$ if at least one 
of $\iter(G, \bsigma_{\eta}, \bsigma_{\eta}(\partial \alpha), \mathbold{\theta})$ and 
$\iter(G, \bsigma_{\eta}, \eta, \kappa)$ fails. Since the second process can only fail if the first one fails, 
we  get that 
\begin{align}\label{eq:ProgDisEta0VsEta1}
\Pr[\bbeta_0\neq \bbeta_1] &=\Pr[\BadWSUpdate(G, \bsigma_{\eta}, \bsigma_{\eta}(\partial \alpha ), \mathbold{\theta})] \leq \rsq_x \enspace,
\end{align}
where the last inequality follows by a simple convexity argument.

We proceed with bounding $\Pr[\bbeta\neq \bsigma_{\kappa}]$. Recall that we couple $\bsigma_{\kappa}$ and $\bbeta_1$ 
optimally, which implies the following:
Letting $\lambda$ be the distribution of $\bbeta_1$, we have that 
\begin{equation}\label{eq:ProgDisEtaVsSigmaReducToTVD}
\Pr[\bbeta_1\neq {\bsigma}_{\kappa}] = || \lambda -\mu(\cdot \ |\ x, \kappa )||_{\rm tv} \enspace.
\end{equation}
For any $\theta \in \alphabet^V$ we have that
\begin{align}\label{eq:SMDefOfLambda}
 \lambda (\theta) &=  \sum\nolimits_{\xi \in \alphabet^V} \mu(\xi \ |\ x, \eta )\transprob_{\eta,\kappa}(\xi, \theta)
= \frac{1}{\mu_{x}(\eta)} \sum\nolimits_{\tau \in \alphabet^V} \mu(\xi)\transprob_{\eta,\kappa }(\xi, \theta) \enspace.
\end{align}
Recall from  Theorem \ref{thrm:DBE4Rswitch} that
for any $\theta, \xi \in \alphabet^V$, we have that
$\mu(\xi )\transprob_{\eta,\kappa}( \xi, \theta)=
\mu(\theta)\transprob_{\kappa, \eta}(\theta, \xi).$
Pugging this equality into \eqref{eq:SMDefOfLambda} we get that
\begin{align}
\lambda(\theta) &= \frac{1}{\mu_x(\eta)} \sum_{\xi \in \alphabet^V} \mu(\theta )\transprob_{\kappa,\eta}(\theta, \xi)
\ = \ \frac{\mu_x(\kappa)}{\mu_x(\eta)} \sum_{\xi \in \alphabet^V} \mu(\theta \ |\ x, \kappa)\transprob_{\kappa, \eta}(\theta, \xi) 
\ = \ \mu(\theta \ |\ x, \kappa) \sum_{\xi \in \alphabet^V} \transprob_{\kappa, \eta}(\theta, \xi ) \enspace, \nonumber
\end{align} 
where in the last equality we use that $\mu_x(\kappa)=\mu_x(\eta)=1/|\alphabet|$. Noting that the rightmost summation 
is equal to $1-\Pr[\BadWSUpdate(G, \theta, \kappa, \eta)]$,  we have that
\begin{align}
\lambda(\theta)&= \mu(\theta \ |\ x, \kappa) \left( 1-\Pr\left[\BadWSUpdate(G, \theta, \kappa, \eta) \right] \right) \enspace. 
 \label{eq:RelatationANuMuIndependence}
\end{align}
Also, note that $\lambda$ gives positive measure to the event  $\BadWSUpdate(G, \bsigma_{\eta}, \eta, \kappa)$. %fails. 
Particularly, we have
$$
\Pr[\BadWSUpdate(G, \bsigma_{\eta}, \eta, \kappa)] +  \sum\nolimits_{\theta\in \alphabet^V}\lambda(\theta)=1 \enspace.
$$
Combining the above with \eqref{eq:RelatationANuMuIndependence}  we get that 
\begin{align}
\Pr[\bbeta_1\neq {\bsigma}_{\kappa}] &=||\mu(\cdot\ |\ x, \kappa)- \lambda ||_{\rm tv} 
%%%
 =(1/2)\sum\nolimits_{\theta \in \alphabet^V}
|\mu(\theta \ |\ x, \kappa)- \lambda(\theta ) | + (1/2)\Pr[\BadWSUpdate(G, \bsigma_{\eta}, \eta, \kappa)] \nonumber \\
%%%
&=  (1/2) \sum\nolimits_{\theta \in \alphabet^V} 
 \mu(\theta \ |\ x, \kappa) \Pr[\BadWSUpdate(G, \theta, \kappa, \eta)] + (1/2)\Pr[\BadWSUpdate(G, \bsigma_{\eta}, \eta, \kappa)]
& \mbox{[from \eqref{eq:RelatationANuMuIndependence}]} \nonumber\\
%%%
&= (1/2) (\Pr[ \textrm{$\BadWSUpdate(G, {\bsigma}_{\kappa}, \kappa, \eta)$}] +\Pr[\BadWSUpdate(G, \bsigma_{\eta}, \eta, \kappa)])
\label{eq:Error4RestUpdatesProvedIndependence} \\
&\leq  \rsq_x \enspace. \label{eq:ProgDisEtaVsSigma}
\end{align}
The last inequality follows from the observation that  both quantities on the r.h.s. of  \eqref{eq:Error4RestUpdatesProvedIndependence} 
are upper bounded  by $\rsq_x$. 
From \eqref{eq:ProgDisEtaVsSigma} and \eqref{eq:ProgDisEta0VsEta1} we have that \eqref{eq:TargetABC4MarginalLambdaVsUniformDistr} 
is true. This concludes the proof of Lemma \ref{lemma:Point2SetWithRswitch}. \hfill $\Box$

\subsubsection{Proof of Lemma \ref{lemma:MarginalLambdaVsUniformDistrWRTRswitch}}\label{sec:lemma:MarginalLambdaVsUniformDistrWRTRswitch}

W.l.o.g. we may  assume that $M=\emptyset$ and  $\Lambda=\partial \alpha$, i.e.,  $\Lambda=\{z_1, \ldots, z_k \}$.
Also, let $\Lambda_{>r }=\{z_{r+1}, z_{r+2}, \ldots, z_{k}\}$ and $\Lambda_{<r}=\{z_1,\ldots, z_{r-1}\}$ for integer $r$.

Abbreviating  $\mu_G$ to $\mu$,  it suffices to  show that
\begin{align}\label{eq:Redaux2SingleBoundary}
||\mu - \zeta||_{\Lambda} &\leq 
 \sum\nolimits_{r\in [k-1]}\max\nolimits_{\kappa,\eta \in \alphabet^{z_r}}
|| \mu(\cdot \ |\ z_r, \eta)-\mu(\cdot\ |\ z_r, \kappa) ||_{\Lambda\setminus \{z_r\}} \enspace, 
\end{align}
while
for $j\in [k-1]$ and  for any two configurations $\eta,\kappa\in \alphabet^{z_j}$ we have
that
\begin{equation}\label{eq:UpdateSingeUpdate}
||\mu(\cdot\ |\ z_j, \eta)-\mu_{i}(\cdot \ | \ z_j, \tau) ||_{\Lambda\setminus \{z_j\} } \leq 2\rsq_{z_j} \enspace.
\end{equation}
Clearly, \eqref{eq:UpdateSingeUpdate} is true due to Lemma \ref{lemma:Point2SetWithRswitch},
hence, it remains to prove that \eqref{eq:Redaux2SingleBoundary} is true.

For $r=0, \ldots k-1$, we let $\xi_r$ be the distribution over the configurations $\alphabet^{\Lambda}$ such that
$$
\xi_r= \left( \otimes^{r}_{j=1}\mu_{z_j } \right) \otimes \mu_{\Lambda_{>r}} \enspace.
$$
That is, $\xi_r$ factorises as a product over the components $z_1, z_2, \ldots, \Lambda_{>r}$ with the 
corresponding marginals being $\mu_{z_1}, \ldots, \mu_{z_r}$, and $\mu_{\Lambda_{>r}}$.
Since, for all  $j\in [k ]$,  $\mu_{z_j}$  is the uniform distribution over $\alphabet$, we have that 
$\xi_{k-1}$ is the same as $\zeta$.  Also, note that  $\xi_0=\mu_{\Lambda}$, i.e., this is  the marginal of 
$\mu$ at $\Lambda$.
We have that 
\begin{align}\nonumber 
 ||\mu - \zeta||_{\Lambda} \leq \sum\nolimits_{r\in [k-1]} ||\xi_{r-1} -\xi_{r}||_{\rm tv} \enspace,
\end{align}
by the triangle inequality. 
In light of the above,  \eqref{eq:Redaux2SingleBoundary} follows by showing that 
\begin{align}\label{eq:TargetForXIDiff}
 ||\xi_{r-1} -\xi_{r}||_{\rm tv} &\leq \max\nolimits_{\eta, \kappa}||\mu_{i}(\cdot\ | \ z_r, \eta)-\mu_{i}(\cdot\ | \ z_r, \kappa) ||_{\Lambda\setminus\{z_r\}}
& \forall r\in [k-1]\enspace.
\end{align}
We use coupling to prove \eqref{eq:TargetForXIDiff}.  
Consider  $\btau_1$ and $\btau_2$  distributed as in $\xi_{r-1}$ and $\xi_{r}$, respectively. 
Noting that $\xi_{r-1}$ and $\xi_{r}$ specify the same marginals for the set $\Lambda_{<r}=\{z_1,\ldots, z_{r-1}\}$,
we couple $\btau_1$ and $\btau_2$ on $\Lambda_{<r}$ identically, i.e., 
with probability $1$, for all $j\in [r-1]$ we have that $\btau_1(z_j)=\btau_{2}(z_j)$.
Furthermore,  regardless of the configuration $\btau_1, \btau_2$ at $\Lambda_{<r}$,  the marginals of 
both $\xi_{r-1}, \xi_r$ at $z_r$ are both the uniform distribution. This implies that  we can couple 
$\btau_1(z_r)$ and $\btau_2(z_r)$ identically, too.

We now focus on coupling 
$\btau_1(\Lambda_{>r})$ and $\btau_2(\Lambda_{>r})$.
At this point, we note that the difference in the two distributions $\xi_{r-1}, \xi_{r}$ amounts to 
the fact that the marginal of $\xi_{r-1}$ at  $\Lambda_{>r}$   depends on the configuration at $z_r$, while in $\xi_{r}$ it does not. 
Hence, given the value of $\btau_1(z_r)$ and $\btau_2(z_r)$,  the conditional marginals of $\xi_{r-1}, \xi_{r}$ on the 
set $\Lambda_{>r}$ are different from each other. 
We couple $\btau_1(\Lambda_{>r})$ and $\btau_2(\Lambda_{>r})$ optimally.

In the above coupling, we note that $\btau_1$ and $\btau_2$ can only disagree on
the set $\Lambda_{>r}$. Then, we have 
\begin{eqnarray}
||\xi_{r-1} -\xi_{r}||_{\rm tv} &\leq & \Pr[\btau_1 \neq \btau_2] =\Pr[\btau_1(\Lambda_{>r}) \neq \btau_2(\Lambda_{>r}) ] \nonumber \\
&\leq &\max\nolimits_{\eta, \kappa} \Pr[\btau_1(\Lambda_{>r}) \neq \btau_2(\Lambda_{>r}) \ |\ 
\btau_1(z_{r})=\eta, \ \btau_2(z_{r})=\kappa ]\nonumber \\
&=&\max\nolimits_{\kappa, \eta}|| \xi_{r-1}(\cdot \ |\ z_r, \eta)-\xi_{r}(\cdot \ |\ z_r, \kappa) ||_{\Lambda_{>r}} \enspace,
\label{eq:Step1XIDistance}
\end{eqnarray}
the last equality follows from the fact that we couple $\btau_1(\Lambda_{>r})$ and 
$\btau_2(\Lambda_{>r})$ optimally. We also have that 
\begin{align}\nonumber
|| \xi_{r-1}(\cdot \ |\ z_r, \eta)-\xi_{r}(\cdot \ |\ z_r, \kappa) ||_{\Lambda_{>r}} & \leq  || \mu(\cdot \ |\ z_r, \eta)-\mu(\cdot\ |\ z_r, \kappa) ||_{\Lambda_{>r}} 
\leq  || \mu(\cdot \ |\ z_r, \eta)-\mu(\cdot\ |\ z_r, \kappa) ||_{\Lambda\setminus \{z_r\}} \enspace,  
\end{align}
where the first inequality follows from the definition of $\xi_{r-1}$ and $\xi_r$. 
Plugging the above inequality into \eqref{eq:Step1XIDistance} gives  \eqref{eq:TargetForXIDiff}. 
This  concludes the proof of Lemma \ref{lemma:MarginalLambdaVsUniformDistrWRTRswitch}. 
\hfill $\Box$

\subsubsection{Proof of Lemma \ref{lemma:FactoringMargOfPathEnd}}
\label{sec:lemma:FactoringMargOfPathEnd}

To avoid trivialities assume that $\CylAi\neq \emptyset$.  For  definiteness let $\CylAi=\{x,y\}$.   

Since both $\mu_{G}$ and $\mu_{\bar{H}}$ are symmetric, their corresponding marginals  at $x$ is the uniform distribution over 
$\alphabet$. Then, it is standard that  
\begin{align}\label{eq:Base4lemma:FactoringMargOfPathEnd}
 ||\mu_{G}-\mu_{\bar{H}} ||_{M} & \leq  ||\mu_G(\cdot \ |\ x, \sigma)- \mu_{\bar{H}}(\cdot \ |\ x, \sigma) ||_{\{ y\}}\enspace,
\end{align}
where $\sigma \in \alphabet^{x}$ is chosen so that it maximises the r.h.s. of the above inequality.

For  $\tau \in \alphabet^{x}$, different than $\sigma$,  consider the random variables $\btheta, \bbeta\in \alphabet^V$
and $\bxi_{\sigma} \in \alphabet^{\partial \alpha}$ such that    $\btheta$ is distributed as in $\mu_{G}(\cdot \ |\ x,\tau)$, while 
$\bxi_{\sigma}$ agrees with  $\btheta$ on the configuration of the vertices in $\partial \alpha\setminus x$ and
$\bxi_{\sigma}(x)=\sigma$.  Furthermore, let   $\bbeta=\iter(G, \btheta, \btheta(\partial \alpha), \bxi_{\sigma})$.   
Also,  let  $\lambda$ be the distribution of $\bbeta$. 

Using the triangle inequality we have 
\begin{equation}\label{eq:lemma:FactoringMargOfPathEndTriangle}
||\mu_G(\cdot \ |\ x, \sigma)- \mu_{\bar{H}}(\cdot \ |\ x,  \sigma) ||_{y} \leq 
||\mu_G(\cdot \ |\ x, \sigma)- \lambda ||_{\{ y\} }+||\lambda - \mu_{\bar{H}}(\cdot \ |\ x,  \sigma) ||_{\{y\}} \enspace.
\end{equation}
The lemma will follow by bounding appropriately the two quantities on the r.h.s. of the inequality above.

In what follows, for any configurations $\tau, \xi, \hat{\xi}$ 
we let $\BadWSUpdate(G, \tau, \xi, \hat{\xi})$ denote  the event that  the process $\iter(G, \tau, \xi, \hat{\xi})$ fails.

We start with the leftmost quantity in \eqref{eq:lemma:FactoringMargOfPathEndTriangle}. 
Let $\hat{\sigma}$ be distributed as in $\mu_{\bar{H}}(\cdot \ |\ x, \sigma)$.  We upper bound $||\lambda- \mu_{\bar{H}}(\cdot \ |\ x, \sigma) ||_{\{ y\} }$ 
by using a  coupling  between  $\hat{\bsigma}(y)$ and $\bbeta(y)$, recall that $\bbeta$ is distributed as in $\lambda$.

From the definition of $\iter(G, \btheta, \btheta(\partial \alpha), \bxi_{\sigma})$,  it follows that initially $\bbeta(y)$ is 
chosen according to the same  distribution as $\hat{\bsigma}(y)$. In that respect, we  couple the two configurations identically. 
The configuration at $y$ will not change unless the process $\iter$  fails in a later stage of its execution.  
Fail means that $\bbeta(y)$ changes again to an assignment  different than $\hat{\bsigma}(y)$. 
Hence, we have that
\begin{align}\label{eq:lemma:FactoringMargOfPathEndFirstHalfProof}
||\lambda- \mu_{\bar{H}}(\cdot \ |\ x, \sigma) ||_{\{ y\}} & \leq \Pr[\bbeta(y)\neq \hat{\bsigma}(y)] = 
\Pr[\BadWSUpdate(G, \btheta, \btheta(\partial \alpha), \bxi_{\sigma})] \ \leq \ \rcsq_M \enspace.
\end{align}
The rightmost inequality follows from the observation that the failure probability is at most $\rcsq_M.$

As far as $||\mu_G(\cdot \ |\ x, \sigma)- \lambda ||_{\{ y\}}$ is concerned, we use the detail balance and
we work in a way which is very similar to what we have in the proof of Lemma \ref{lemma:Point2SetWithRswitch}.
Specifically,  for any $\eta \in \alphabet^{V}$, we have that
\begin{align*}
\lambda(\eta) &= \sum_{\xi\in \alphabet^V} \mu_G(\xi\ |\ x, \tau)\transprob_{\tau,\sigma}(\xi, \eta) 
\ =\ |\alphabet| \sum_{\xi\in \alphabet^V} \mu_G(\xi )\transprob_{\tau,\sigma}(\xi, \eta) 
\ = \ \sum_{\xi\in \alphabet^V} \frac{\mu_G(\xi )}{\mu_{\bar{H}, x}(\tau)}\transprob_{\tau,\sigma}(\xi, \eta) \enspace, 
\end{align*}
where, with a slight abuse of notation, $\transprob_{\tau,\sigma}(\xi, \eta)$ stands for the probability for 
$\iter(G, \xi, \xi(\partial \alpha), \xi^*)$ outputting $\eta$, where $\xi^*\in\alphabet^{\partial \alpha}$ is obtained
from $\xi(\partial \alpha)$ by switching the configuration at $x$ from $\tau$ to $\sigma$. 

The second and third equalities  follows from the observation that $\mu_{G,x}(\tau)=\mu_{\bar{H},x}
(\tau)=1/|\alphabet|$. Using Theorem \ref{thrm:DBE4RCupdate} we get  

\begin{align}
\lambda(\eta)&=  \sum_{\xi\in \alphabet^V}  \frac{\mu_G(\eta )}{\mu_{\bar{H},x}(\sigma)}\transprob_{\sigma,\tau}(\eta, \xi) 
\ =\  \sum_{\xi\in \alphabet^V} \frac{\mu_G(\eta )}{\mu_{G,x}(\sigma)}\transprob_{\sigma,\tau}(\eta, \xi) 
\ =\ \sum_{\xi\in \alphabet^V} \mu_G(\eta \ |\ x,  \sigma ) \transprob_{\sigma,\tau}(\eta, \xi)  \enspace, \nonumber
\end{align}
where $\transprob_{\sigma,\tau}(\eta, \xi)$ can be interpreted  in the natural way. Arguing as in
\eqref{eq:RelatationANuMuIndependence} we have that
\begin{align}
\lambda(\eta) &=   \mu_G(\eta \ |\ x, \sigma)(1-\Pr[\BadWSUpdate(G,\eta,\eta(\partial \alpha), \eta^*)])\enspace,\label{eq:lemma:FactoringMargOfPathEndDBHelp}
\end{align}
 where $\eta^*\in\alphabet^{\partial \alpha}$ is obtained
from $\eta(\partial \alpha)$ by switching the configuration at $x$ from $\sigma$ to $\tau$. 

Furthermore, it is straightforward that  
\begin{align}\nonumber 
||\mu_G(\cdot \ |\ x, \sigma)- \lambda ||_{\{ y\}}  & \leq   ||\mu_G(\cdot \ |\ x, \sigma)- \lambda ||_{\rm tv}\enspace. %  \leq  \rcsq_M \enspace. 
\end{align}
Using \eqref{eq:lemma:FactoringMargOfPathEndDBHelp} and working as in  \eqref{eq:ProgDisEtaVsSigma} we get that
$ ||\mu_G(\cdot \ |\ x, \sigma)- \lambda ||_{\rm tv} \leq  \rcsq_M $. Hence, we obtain that
\begin{align}
||\mu_G(\cdot \ |\ x, \sigma)- \lambda ||_{\{ y\}}  & \leq   \rcsq_M \enspace.
\label{eq:lemma:FactoringMargOfPathEndSecondHalfProof}
\end{align}
We omit the derivations for the above as  they almost identical to those in \eqref{eq:ProgDisEtaVsSigma}.

The lemma follows by plugging \eqref{eq:lemma:FactoringMargOfPathEndSecondHalfProof},
\eqref{eq:lemma:FactoringMargOfPathEndFirstHalfProof} and 
\eqref{eq:lemma:FactoringMargOfPathEndTriangle} into \eqref{eq:Base4lemma:FactoringMargOfPathEnd}. 
\hfill $\Box$

\subsubsection{Proof of Lemma \ref{lemma:FactoringOfBarH}}\label{sec:lemma:FactoringOfBarH}

For  brevity, we let $\Lambda=\partial \alpha=\{z_1, \ldots,z_k\}$. 
Also, let   $\Lambda_{>r}=\{z_{r+1}, \ldots z_k\}$  
for integer $r>0$.  Furthermore, w.l.o.g. assume that $\Lambda\cap M=\{z_{k-1}, z_{k}\}$.

It suffices to  show that
\begin{align}\label{eq:IndependenceWithShortCycle}
 ||\mu_{G}-\mu_{\bar{H}} ||_{\Lambda}\leq ||\mu_{G}+\mu_{\bar{H}}||_{M}+
\sum_{r\in [k-2]}\max_{\sigma, \tau\in \alphabet^{z_r}}||\mu_G( \cdot\ |\ z_r, \sigma)-\mu_G( \cdot\ |\ z_r, \tau) ||_{\Lambda\setminus\{z_r\}} \enspace,
\end{align}
while, for $r=1,\ldots, k-2$, we have 
\begin{align}\label{eq:MuGVsMuHLambdaZr}
\max\nolimits_{\sigma, \tau\in \alphabet^{z_r}}||\mu_G( \cdot\ |\ z_r, \sigma)-\mu_G( \cdot\ |\ z_r, \tau) ||_{\Lambda\setminus\{z_r\}} &\leq 2 \rsq_{z_r}
&\textrm{and} &&
||\mu_G+\mu_{\bar{H}}||_{M}&\leq 2\rcsq_M \enspace.
\end{align}
Clearly, \eqref{eq:MuGVsMuHLambdaZr} follows by using Lemmas  \ref{lemma:Point2SetWithRswitch} and
\eqref{lemma:FactoringMargOfPathEnd}.

It remains to show that \eqref{eq:IndependenceWithShortCycle} is true.
For $0\leq j\leq k-2$,  define the distribution $\nu^{(j)}_{{\rm prod}}:\alphabet^{\Lambda}\to[0,1]$ by
\begin{align}
 \nu^{(j)}_{\rm prod} &= {\textstyle \left( \otimes^{j}_{r=1} \mu_{G, z_r}\right)\otimes \mu_{G, \Lambda_{>j}}\enspace .}
\end{align}
That is, $\nu^{(j)}_{\rm prod}$ factorises as a product over the components $z_1, z_2, \ldots z_j$ and $\Lambda_{>j}$
with the corresponding marginals being $\mu_{G, z_1}, \mu_{iG, z_2}, \ldots \mu_{G, z_{j}}$ and $\mu_{G, \Lambda_{>j}}$. 
Note that  
$\nu^{(0)}_{\rm prod}=\mu_{G,\Lambda}$, i.e., this is the marginal of $\mu_G$ on the set $\Lambda$.
From the triangle inequality, we have 
\begin{align}\label{eq:TriangleRedaux2NuiS}
 ||\mu_{G}-\mu_{\bar{H}} ||_{\Lambda} \leq 
\left|\left|\mu_{\bar{H},\Lambda}- \nu^{(k-2)}_{\rm prod} \right|\right|_{\rm tv}+\sum\nolimits_{j\in [k-2]}\left|\left|\nu^{(j-1)}_{\rm prod} -\nu^{(j)}_{\rm prod} \right|\right|_{\rm tv} \enspace.
\end{align}
Working as in  the proof of Lemma \ref{lemma:MarginalLambdaVsUniformDistrWRTRswitch} (i.e., to show \eqref{eq:TargetForXIDiff})
we get that  
\begin{align}\label{eq:TargetForNuIRedux2MuI}
\left|\left|\nu^{(r)}_{\rm prod} -\nu^{(r+1)}_{\rm prod}\right|\right|_{\rm tv} &\leq \max\nolimits_{\eta, \kappa}||\mu_{G}(\cdot\ | \ z_r, \eta)-\mu_{G}(\cdot\ | \ z_r, \kappa) ||_{\Lambda\setminus\{z_r\}},
&\textrm{for}\  r=0, \ldots k-3\enspace.
\end{align}
Furthermore, we note that
the Gibbs distribution $\mu_{\bar{H},\Lambda}$ can be expressed by 
%\begin{equation}
$ \mu_{\bar{H},\Lambda} = \left( \otimes^{k-2}_{r=1} \mu_{G, z_r}\right)\otimes \mu_{\bar{H}, M}.$
%\end{equation}
%
Hence, working as for  \eqref{eq:TargetForNuIRedux2MuI}, we get that 
\begin{equation}\label{eq:target4MuURedux2MuPath}
\left | \left|\mu_{\bar{H},\Lambda}- \nu^{(k-2)}_{\rm prod} \right|\right|_{\rm tv} = \left|\left|\mu_{\bar{H}}- \nu^{(k-2)}_{\rm prod} \right|\right|_{M}= ||\mu_{G}- \mu_{\bar{H}} ||_{M} \enspace.
\end{equation}
Plugging \eqref{eq:target4MuURedux2MuPath} and \eqref{eq:TargetForNuIRedux2MuI} into \eqref{eq:TriangleRedaux2NuiS}
we get \eqref{eq:IndependenceWithShortCycle}. This concludes the proof of Lemma \ref{lemma:FactoringOfBarH}. \hfill $\Box$

\spreadpoint
\section{Accuracy of $\rsampler$ - Proof of \cref{theorem:FinalDetAcc} 
\LastReview{2024-01-28}}\label{sec:theorem:FinalDetAcc}

\subsection{Accuracy for $\iter$}
Let $0\leq i < m$. For $\eta,\kappa\in \alphabet^{\partial \alpha_i}$ and $\bsigma$ distributed as in  $\mu_i(\cdot\ |\ \partial \alpha_i, \eta)$, consider the process 
$\iter(G_i,\bsigma, \eta,\kappa)$, while let   $\nu_{\eta, \kappa}$ be  the distribution of its output.

\begin{proposition}\label{prop:Error4RUpdate}
Suppose that adding  $\alpha_i$ into $G_i$ does not introduce a  short cycle in $G_{i+1}$. For any 
$\eta, \kappa\in \alphabet^{\partial \alpha_i}$  the following is true: 
Provided that $\rsq_{x}$ is sufficiently small,  for all $x\in \partial \alpha$,  we have that 
\begin{equation}\nonumber 
||\mu_i(\cdot \ |\ \partial \alpha_i, \kappa ) -\nu_{\eta,\kappa} ||_{\rm tv} \leq  
7k |\alphabet|^{k} \cdot \sum\nolimits_{x\in \partial \alpha} \rsq_{x} \enspace.
\end{equation}
\end{proposition}

The proof of Proposition \ref{prop:Error4RUpdate} appears in Section \ref{sec:prop:Error4RUpdate}.

\begin{proposition}\label{prop:AccuracyRCUpdate}
Suppose that adding $\alpha_i$ into $G_i$ introduces a  short cycle in $G_{i+1}$.  
 For  any $\eta, \kappa\in \alphabet^{\partial \alpha_i}$  the following is true:
Provided that $\rcsq_M$ and $\rsq_x$, for $x\in \partial \alpha\setminus M$, are sufficiently small, we have that
\begin{equation}\nonumber 
||\mu_{i}(\cdot\ |\ \partial \alpha, \kappa)-\nu_{\eta, \kappa}||_{\rm tv} \leq  
7k|\alphabet|^k\cdot \left( 1+\chi \cdot \psimin^{-1}  \right)\cdot 
\left(\rcsq_M+ { \sum\nolimits_{x\in \partial \alpha\setminus M}\rsq_{x} } \right) \enspace.
\end{equation}
\end{proposition}
The proof of Proposition \ref{prop:AccuracyRCUpdate} appears in Section \ref{sec:prop:AccuracyRCUpdate}.

\subsection{Accuracy of $\rsampler$}\label{sec:RUpdateAcc}
For $0\leq i < m$, let $\bbeta_{i+1}$ be generated according to  the following steps: if $\alpha_i$ does not introduce a short cycle in $G_{i+1}$, then set
$\bbeta_{i+1}(\partial \alpha_i)$ according to  \eqref{eq:FirstStep}. Otherwise, i.e.,
 if $\alpha_i$ introduces the short cycle $C$ in $G_{i+1}$, then set $\bbeta_{i+1}(\partial \alpha_i)$  according to the marginal of $\mu_{H}$ at $\partial \alpha_i$. 
Recall that  $H$ is  the subgraph  of $G_{i+1}$ that is induced by both the variable and factor nodes in 
the cycle $C$, as well as all the  variable nodes that are adjacent to the factor nodes in $C$.
Then set 
\begin{align}\label{eq:RSamplerStep10} 
\bbeta_{i+1}&=\iter(G_i, \btheta_i, \btheta_i(\partial \alpha_i), \bbeta_{i+1}(\partial \alpha_i)) \enspace,
\end{align}
where  $\btheta_i$ be distributed as in $\mu_i$. 
We let $\nu^{(i+1)}$ be the distribution of the configuration $\bbeta_{i+1}$.

\begin{proposition}\label{prop:ErrorPerIterUnicyclic}
Provided that $\rmaxq_i$ is sufficiently small, we have that
\begin{align} \label{eq:prop:ErrorPerIterUnicyclic}
\| {\nu}^{(i+1)}-\mu_{i+1} \|_{\rm tv} &\leq   14 k|\alphabet|^k
\left(   1+\chi \psimin^{-1} \right)  \cdot \rmaxq_i, &
\textrm{for  $i=0, \ldots, m-1$} \enspace.
\end{align}
\end{proposition}
The proof of Proposition \ref{prop:ErrorPerIterUnicyclic} appears in Section \ref{sec:prop:ErrorPerIterUnicyclic}.

\begin{proof}[Proof of  \cref{theorem:FinalDetAcc}]
Consider $\rsampler$ on input $G\in \CG$. Recall, that  $\mu_i$ is the Gibbs distribution on $G_i$.
Also, recall that $\bsigma_i$  is the configuration that $\rsampler$ generates
for $G_i$. Let $\bar{\mu}_i$ be the distribution of $\bsigma_i$. 

We  show that 
\begin{align}\label{eq:Base4lemma:AccuracySampler}
|| \mu_m-\bar{\mu}_m||_{\rm tv} &\leq { \sum\nolimits_{i\in [m]} } \|\mu_i-{\nu}^{(i)} \|_{\rm tv} \enspace.
\end{align}
Recall that ${\nu}^{(i)}$ is the distribution that is induced by $\bbeta_i$ which is defined in \eqref{eq:RSamplerStep10}.
Given  \eqref{eq:Base4lemma:AccuracySampler} the theorem follows 
by using Proposition \ref{prop:ErrorPerIterUnicyclic} to bound each term $\|\mu_i-{\nu}^{(i)} \|_{\rm tv}$. 
It remains to  prove \eqref{eq:Base4lemma:AccuracySampler}.  

Using the triangle inequality we have
\begin{equation}\label{eq:AccuracySampler:Triangle}
|| \mu_m-\bar{\mu}_m||_{\rm tv} \leq || \mu_m-{\nu}^{(m)}||_{\rm tv}+|| {\nu}^{(m)}-\bar{\mu}_m||_{\rm tv} \enspace.
\end{equation}
We bound the quantity $||{\nu}^{(m)} -\bar{\mu}_m||_{\rm tv}$ by coupling $\bbeta_{m}$ and $\bsigma_{m}$,
where $\bbeta_{m}$ is defined in \eqref{eq:RSamplerStep10} and is distributed as in ${\nu}^{(m)}$,
while $\bsigma_{m}$  is the configuration that is generated at the $m$-th (the last) iteration of $\rsampler$ 
and is distributed as in $\bar{\mu}_m$.

For this coupling, we use another two configurations, $\btheta_{m-1}$ and $\bsigma_{m-1}$ which 
are  distributed as in $\mu_{m-1}$ and $\bar{\mu}_{m-1}$, respectively. More specifically, the coupling
amounts to the following: 
First, we couple $\btheta_{m-1}$ and $\bsigma_{m-1}$ optimally. Then, we generate 
$\bbeta_{m}$ by using $\btheta_{m-1}$ and  $\iter$ as shown in  
\eqref{eq:RSamplerStep10}. Similarly, we generate  
$\bsigma_{m}$ by using $\bsigma_{m-1}$  and  $\iter$.
We couple the corresponding  processes that  generate $\bbeta_{m}$ and $\bsigma_m$  
as close as possible. 
This implies that  if $\btheta_{m-1}=\bsigma_{m-1}$, then we also have that 
 $\bbeta_{m}=\bsigma_{m}$.
  
Hence, we immediately get that
 \begin{align}
 ||{\nu}^{(m)}-\bar{\mu}_m||_{\rm tv} &\leq  \Pr[\bbeta_{m}\neq \bsigma_{m}]  
 \leq  \Pr[\btheta_{m-1}\neq \bsigma_{m-1}] 
 =  \|\mu_{m-1}-\bar{\mu}_{m-1}\|_{\rm tv} \enspace, \label{eq:AccuracySampler:ReducByA}
 \end{align}
where the last equality is because  we couple 
$\btheta_{m-1}$ and $\bsigma_{m-1}$ optimally.
Plugging \eqref{eq:AccuracySampler:ReducByA} into \eqref{eq:AccuracySampler:Triangle}, we get 
\begin{align}\label{eq:Final4lemma:AccuracySampler}
\| \mu_m-\bar{\mu}_m\|_{\rm tv} &\leq \| \mu_m-{\nu}^{(m)} \|_{\rm tv}+\|\mu_{m-1}-\bar{\mu}_{m-1}\|_{\rm tv}\enspace.
\end{align}
The theorem follows by applying inductively the same steps for the quantity 
$||\mu_{m-1}-\bar{\mu}_{m-1}||_{\rm tv}$ in \eqref{eq:Final4lemma:AccuracySampler}. 

The above concludes the proof of the theorem.
\end{proof}

\subsection{Proof of  \cref{prop:Error4RUpdate}}\label{sec:prop:Error4RUpdate}
When there is no danger of confusion we drop the index $i$ from $G_i$, $\mu_i$ and $\alpha_i$. 
Furthermore,  assume that for every $x\in \partial \alpha_i$ we have $k|\alphabet|^{k}\rsq_{x}<1/8$.
Finally, for any configurations $\tau, \xi, \theta$ 
we let $\BadWSUpdate(G, \tau, \xi, \theta)$ denote  the event that   $\iter(G, \tau, \xi, \theta)$ fails. 

Consider the sequence  $\kappa_0, \kappa_1, \ldots, \kappa_{r}$ of configurations at $\partial \alpha_i$
as they are described in line 12 of the pseudo-code of $\rsampler$ in Algorithm \ref{rsampler}.
Recall that we assume that $\alpha_i$ does not introduce any short-cycle in $G_i$, hence 
$\DiaSet=\partial \alpha_i$.

Let $\uplambda_j$ be the distribution of  the output of  $\iter(G, \bxi_{j-1}, \kappa_{j-1}, \kappa_j)$
for  $\bxi_{j-1}$ being distributed as in $\mu(\cdot \ | \ \partial \alpha, \kappa_{j-1})$ and $j\in [r]$. 
Note that $\uplambda_j$ is {\em not} the distribution of the configuration $\btau_{j}$ generated at line 15 of the pseudo-code. 

\begin{claim}\label{claim:Reduction2SingleUpdateErrorUnicyclic}
We have that 
$
||\mu(\cdot \ |\ \partial \alpha, \kappa)- \nu_{\eta,\kappa} ||_{\rm tv} \leq  \sum\nolimits_{j\in [r]} ||\mu(\cdot \ |\ \partial \alpha, \kappa_{j})- \uplambda_{j} ||_{\rm tv}.
$
\end{claim}

In light of  \cref{claim:Reduction2SingleUpdateErrorUnicyclic}, it suffices to show that
for any $j\in [r]$ we have  
\begin{align}\label{eq:Base4prop:Error4RUpdate}
 ||\mu(\cdot \ |\ \partial \alpha, \kappa_{j})- \uplambda_{j} ||_{\rm tv} &\leq  7 |\alphabet|^k \cdot \sum\nolimits_{z\in \partial \alpha} \rsq_{z} 
\enspace,
\end{align}
where $x_j\in \partial \alpha$ is the node at which $\kappa_{j-1}$ and $\kappa_j$ disagree. 
For any $\sigma \in \alphabet^V$ we have that
\begin{align} \nonumber 
 \uplambda_{j}(\sigma) 
&= \sum\nolimits_{\tau \in \alphabet^V} \mu(\tau \ |\ \partial \alpha, \kappa_{j-1} )\transprob_{\kappa_{j-1},\kappa_{j}}(\tau, \sigma)
\ = \ \frac{1}{\mu_{\partial \alpha}(\kappa_{j-1})} \sum\nolimits_{\tau \in \alphabet^V} \mu(\tau )\transprob_{\kappa_{j-1},\kappa_{j}}(\tau, \sigma) \enspace.
\end{align} 
Theorem \ref{thrm:DBE4Rswitch} implies that 
$\mu(\tau)\cdot \transprob_{\kappa_{j-1},\kappa_{j}}(\tau, \sigma)=
\mu(\sigma)\cdot \transprob_{\kappa_j, \kappa_{j-1}}(\sigma, \tau).$
Hence, we have that 
\begin{align}
\uplambda_{j}(\sigma) 
&=  \frac{1}{\mu_{\partial \alpha}(\kappa_{j-1})} \sum\nolimits_{\tau \in \alphabet^V} \mu(\sigma )\transprob_{\kappa_{j},\kappa_{j-1}}(\sigma, \tau)
\nonumber\\
&=  \frac{\mu_{\partial \alpha}(\kappa_j)}{\mu_{\partial \alpha}(\kappa_{j-1})} \mu(\sigma \ |\ \partial \alpha, \kappa_{j})
\sum\nolimits_{\tau \in \alphabet^V} \transprob_{\kappa_{j},\kappa_{j-1}}(\sigma, \tau) \nonumber \\
&=  \frac{\mu_{\partial \alpha}(\kappa_j)}{\mu_{\partial \alpha}(\kappa_{j-1})} \mu(\sigma \ |\ \partial \alpha, \kappa_{j})
\left(1-\Pr[\BadWSUpdate(G, \sigma, \kappa_{j}, \kappa_{j-1})] \right) \enspace.
 \label{eq:RUpdateRelatationBetweenNuMuA}
\end{align}

 Additionally, we have
\begin{align}\label{eq:RUpdateConditionRatio}
\frac{\mu_{\partial \alpha}(\kappa_j)}{\mu_{\partial \alpha}(\kappa_{j-1})} &=1+{\tt err}_j, & \textrm{where } 
 |{\tt err}_j| \leq 6|\alphabet|^{k}\cdot  \sum\nolimits_{x\in \partial \alpha} \rsq_{x} \enspace.
\end{align}
The bound on ${\tt err}_j$ as follows:
letting $\zeta$ be the uniform distribution over $|\alphabet|^k$, we have
$$
 \frac{|\alphabet|^{-k}- ||\mu-\zeta ||_{\partial \alpha}}
 {|\alphabet|^{-k}+ ||\mu-\zeta ||_{\partial \alpha}} \leq 
\frac{\mu_{\partial \alpha}(\kappa_j)}{\mu_{\partial \alpha}(\kappa_{j-1})} \leq 
\frac{|\alphabet|^{-k}+ ||\mu-\zeta ||_{\partial \alpha}}{|\alphabet|^{-k}- ||\mu-\zeta ||_{\partial \alpha}} \enspace.
$$
Then, we get the desired bound for ${\tt err}_j$ by using
\cref{lemma:MarginalLambdaVsUniformDistrWRTRswitch} and the assumption
that $k|\alphabet|^{k}\rsq_{x}<1/8$.

From the definition of total variation distance, we get 
\begin{align}
||\mu(\cdot\ |\ \partial \alpha, \kappa_{j})- \uplambda_{j} ||_{\rm tv} 
& = \frac12 \left ( \Pr[\BadWSUpdate(G, \bxi_{j-1}, \kappa_{j-1}, \kappa_{j})]+ { \sum\nolimits_{\sigma \in \alphabet^V}}
|\mu(\sigma\ |\ \partial \alpha, \kappa_{j})- \uplambda_{j}(\sigma ) | \right) \enspace.
  \label{eq:RUpdateAccu:StepOne}
\end{align}
In the above we use that $\uplambda_j$ include the event of failure in its support. 

Furthermore, using \eqref{eq:RUpdateRelatationBetweenNuMuA}, we get 
\begin{align} 
\lefteqn{
  \sum\nolimits_{\sigma \in \alphabet^V}
|\mu(\sigma\ |\ \partial \alpha, \kappa_{j})- \uplambda_{j}(\sigma ) |
} \hspace{1cm} \nonumber \\
&\leq    |{\tt err}_j|+\sum\nolimits_{\sigma \in \alphabet^V} 
 \mu(\sigma \ |\ \partial \alpha, \kappa_{j}) \Pr[\BadWSUpdate(G, \sigma, \kappa_{j}, \kappa_{j-1})]  %\nonumber \\  
 =  |{\tt err}_j|+ \Pr[\BadWSUpdate(G, \bxi_j,\kappa_j, \kappa_{j-1}) ]    \enspace.  \nonumber 
 \end{align}
Plugging the above into \eqref{eq:RUpdateAccu:StepOne} and noting that 
$\Pr[\BadWSUpdate(G, \bxi_j,\kappa_j, \kappa_{j-1}) ] \leq \rsq_{x_j}$, we get that
\begin{equation} %\label{eq:Final4prop:Error4RUpdate}
||\mu(\cdot\ |\ \partial \alpha, \kappa_{j})- \uplambda_{j} ||_{\rm tv}  
\leq
\rsq_{x_j} +|{\tt err}_j| \enspace.  \nonumber
\end{equation}
Combining the above with \eqref{eq:RUpdateConditionRatio}, we get \eqref{eq:Base4prop:Error4RUpdate}.
This concludes the proof of  \cref{prop:Error4RUpdate}. \hfill $\Box$

\begin{proof}[Proof of  \cref{claim:Reduction2SingleUpdateErrorUnicyclic}]
We remind the reader  that  $\uplambda_j$ corresponds to the distribution of the output of 
the process 
$\iter(G, \bxi_{j-1}, \kappa_{j-1}, \kappa_j)$, for  $\bxi_{j-1}$  distributed as in 
$\mu(\cdot \ | \ \partial \alpha, \kappa_{j-1})$.

Let $\nu_{j}$ be the distribution of the output of $\iter(G, \bsigma, \kappa_0, \kappa_j)$, 
for  $\bsigma$  distributed as in $\mu(\cdot \ | \ \partial \alpha, \kappa_{0})$ (recall that $\kappa_0=\eta$).  In this notation, we have that $\nu_{\eta,\kappa}=\nu_{r}$.
Applying the triangular inequality, we have 
\begin{align}
||\mu(\cdot \ |\ \partial \alpha, \kappa)- \nu_{\eta,\kappa} ||_{\rm tv} &= ||\mu(\cdot \ |\ \partial \alpha, \kappa_r )- \nu_{r} ||_{\rm tv} 
%\nonumber \\
\leq ||\mu(\cdot \ |\ \partial \alpha, \kappa_r)- \uplambda_{r} ||_{\rm tv}+
|| \uplambda_{r} -\nu_{r}||_{\rm tv} \enspace.
\label{eq:claim:Reduction2SingleUpdateErrorA}
\end{align}
Let $\bbeta$ and $\hat{\bbeta}$ be distributed as in $\mu(\cdot \ | \ \partial \alpha, \kappa_{r-1})$ 
and $\nu_{r-1}$, respectively. Let $\mathbold{\theta}=\iter(G, \bbeta, \kappa_{r-1}, \kappa_{r})$
and $\hat{\mathbold{\theta}}=\iter(G, \hat{\bbeta}, \kappa_{r-1}, \kappa_r)$. 
From the definition of the corresponding quantities, we have that $\mathbold{\theta}$ is distributed 
as in $\uplambda_{r}$, while $\hat{\mathbold{\theta}}$ is distributed as in $\nu_{r}$.
We use a coupling between $\mathbold{\theta}$ and $\hat{\mathbold{\theta}}$ to  bound the rightmost quantity in \eqref{eq:claim:Reduction2SingleUpdateErrorA}.

We couple  $\mathbold{\theta}$ and $\hat{\mathbold{\theta}}$  by means of
 $\bbeta$ and $\hat{\bbeta}$. That is, we couple optimally $\bbeta$ and $\hat{\bbeta}$
and then, we couple as close as possible the processes $\iter(G, \bbeta, \kappa_{r-1}, \kappa_{r})$
and $\iter (G, \hat{\bbeta}, \kappa_{r-1}, \kappa_r)$.

If  in the coupling   we have $\bbeta=\hat{\bbeta}$, then the two processes that generate 
$\mathbold{\theta}$ and $ \hat{\mathbold{\theta}}$, respectively,  are identical. 
Hence,  we can only have  $\mathbold{\theta}\neq \hat{\mathbold{\theta}}$  if $\bbeta\neq \hat{\bbeta}$.
We conclude that
\begin{equation}\nonumber 
|| \uplambda_{r} -\nu_{r}||_{\rm tv} \leq \Pr[\mathbold{\theta}\neq \hat{\mathbold{\theta}}] \leq 
\Pr[\bbeta\neq \hat{\bbeta}] = || \mu(\cdot \ | \ \partial \alpha, \kappa_{r-1}) - \nu_{r-1}||_{\rm tv} \enspace.
\end{equation}
The last equality follows since  we couple $\bbeta$ and $\hat{\bbeta}$
optimally. Plugging the above into \eqref{eq:claim:Reduction2SingleUpdateErrorA} we get 
\begin{align}\nonumber 
||\mu(\cdot \ |\ \partial \alpha, \kappa_{r})- \nu_{r} ||_{\rm tv} &\leq 
||\mu(\cdot \ |\ \partial \alpha, \kappa_r)- \uplambda_{r} ||_{\rm tv}+|| \mu(\cdot \ | \ \partial \alpha, \kappa_{r-1}) - \nu_{r-1}||_{\rm tv} \enspace.
\end{align}
The claim follows by working inductively on the quantity $|| \mu(\cdot \ | \ \partial \alpha, \kappa_{r-1}) - \nu_{r-1}||_{\rm tv}$,
above, and noting that $\uplambda_1$ and $\nu_1$ correspond to  the same distribution. 
\end{proof}

\subsection{Proof of \cref{prop:AccuracyRCUpdate}}\label{sec:prop:AccuracyRCUpdate}
If  there is no danger of confusion we drop the index $i$ from $G_i$, $\mu_i$ and $\alpha_i$.  

Recall that  we assume that the addition of $\alpha_i$ into $G_i$ introduces the short cycle $C$ in  $G_{i+1}$.
We let $\bar{H}$ be the subgraph of $G_i$ that is induced by the variable nodes and
the factor nodes in what becomes a short cycle $C$ after the insertion of $\alpha$ into $G_i$. 

Let   $\initDis=\{z_1,  \ldots, z_{r}\}$  contain the variable nodes in 
$\partial \alpha_i\setminus \CylAi$ at which  the two configurations  $\eta$ and $\kappa$ disagree. 
Consider the  sequence of configurations   $\kappa_0, \kappa_1, \ldots, \kappa_{r}$ at $\partial \alpha_i\setminus \CylAi$
such that  $\kappa_0=\eta$,  while    $\kappa_j$ is obtained from
$\eta$ by changing  the assignment of the variable  nodes   $x\in \{z_1, \ldots, z_j\}$ from 
$\eta(x)$ to $\kappa(x)$.

For  $\bxi_j$ being  distributed as in $\mu(\cdot \ | \ \partial \alpha, \kappa_{j})$,  let
$\uplambda_j$ be the distribution of the output of the process $\iter(G, \bxi_{j-1}, \kappa_{j-1}, \kappa_j)$, for $j\in [r]$. 
Similarly, let $\uplambda_{r+1}$ be the distribution of the output of  the process 
$\iter(G,\bxi_{r}, \kappa_{r}, \kappa)$.

For any configurations $\tau, \xi, \theta$  we let 
$\BadWSUpdate(G, \tau, \xi, \theta)$ be  the event that
$\iter(G, \tau, \xi, \theta)$ fails.

\begin{claim}\label{claim:Reduction2SingleUpdateRCCase}
We have that 
\begin{align}\nonumber 
 ||\mu_i(\cdot \ |\ \partial \alpha, \kappa)- \nu_{\eta,\kappa} ||_{\rm tv} \leq || \mu(\cdot \ |\ \partial \alpha, \kappa)- \uplambda_{r+1} ||_{\rm tv}
+\sum\nolimits_{j\in [r]} ||\mu_i(\cdot \ |\ \partial \alpha, \kappa_{j})- \uplambda_{j} ||_{\rm tv} \enspace.
\end{align}
\end{claim}
The proof of  \cref{claim:Reduction2SingleUpdateRCCase} is almost identical to the proof of \cref{claim:Reduction2SingleUpdateErrorUnicyclic}, 
for this reason we omit it.  

Furthermore, working as in \eqref{eq:Base4prop:Error4RUpdate} we obtain that 
\begin{equation}\label{eq:BaseA4prop:AccuracyRCUpdate}
 \sum\nolimits_{j\in[r]} ||\mu_i(\cdot \ |\ \partial \alpha, \kappa_{j})- \uplambda_{j} ||_{\rm tv} \leq 
7k|\alphabet|^k\cdot \sum\nolimits_{z\in \partial \alpha\setminus M}\rsq_{z} \enspace.
\end{equation}

We now focus on  bounding  $|| \mu(\cdot \ |\ \partial \alpha, \kappa)- \uplambda_{r+1} ||_{\rm tv}$.
We have seen the following derivations  in various places, before. 
For any $\sigma \in \alphabet^V$ we have that
\begin{align}
% \lefteqn{
\uplambda_{r+1}(\sigma) &= \sum\nolimits_{\tau \in \alphabet^V}
\mu_i(\tau\ |\ \partial \alpha, \kappa_{r})\transprob_{\kappa_{r}, \kappa}(\tau, \sigma) 
% 
% } 
\nonumber\\
&=\frac{1}{\mu_{i, \partial \alpha}(\kappa_{r})}\sum\nolimits_{\tau \in \alphabet^V}\mu_i(\tau)\transprob_{\kappa_{r},\kappa}(\tau, \sigma)
=\frac{\mu_{\bar{H},\partial \alpha}(\kappa_{r})}{\mu_{i, \partial \alpha}(\kappa_{r})}
\sum\nolimits_{\tau \in \alphabet^V}\frac{\mu_{i}(\tau)}{\mu_{\bar{H},\partial \alpha}(\kappa_{r})}\transprob_{\kappa_{r},\kappa}(\tau, \sigma) \enspace. \nonumber
\end{align}
Using  \cref{thrm:DBE4RCupdate} and standard derivation we have seen before, 
we get that
\begin{align}
\uplambda_{r+1}(\sigma)&=
\frac{\mu_{\bar{H},\partial \alpha}(\kappa_{r})}{\mu_{i, \partial \alpha}(\kappa_{r})} \times 
\frac{\mu_{i,\partial \alpha}(\kappa)}{\mu_{\bar{H}, \partial \alpha}(\kappa)} \times
\mu_{i}(\sigma \ |\ \partial \alpha,\kappa) \times (1-\Pr[\BadWSUpdate(G_i,\sigma,\kappa, \kappa_{r})]) \nonumber  \\
&= \left( 1 +{\tt err}\right) \mu_{i}(\sigma \ |\ \partial \alpha,\kappa) \times (1-\Pr[\BadWSUpdate(G_i,\sigma,\kappa, \kappa_{r})]) 
\label{eq:prop:AccuracyRCUpdateStepA}  \enspace,
\end{align}
where in the last equality we set ${\tt err}=\frac{\mu_{\bar{H},\partial \alpha}(\kappa_{r})}{\mu_{i, \partial \alpha}(\kappa_{r})} 
\times \frac{\mu_{i,\partial \alpha}(\kappa)}{\mu_{\bar{H}, \partial \alpha}(\kappa)}-1$.

\begin{claim} \label{claim:ErrBoud4prop:AccuracyRCUpdate}
For any $\kappa, \kappa_r\in \alphabet^{\partial \alpha}$ the following is true:
for sufficiently small $\rcsq_M$ and $\rsq_{x}$, where $x\in \partial \alpha$, 
we have that ${\tt err} \leq   6|\alphabet|^{k-1}\chi \psimin^{-1}  \cdot
\left( \rcsq_M+  \sum\nolimits_{x\in \partial \alpha\setminus M}\rsq_{x} \right)$. 
\end{claim}

Working as in  \cref{prop:Error4RUpdate} we get the following: 
Using the definition of total variation distance 
and plugging \eqref{eq:prop:AccuracyRCUpdateStepA} we have that 
\begin{align}
||\mu_{i}(\cdot\ |\ \partial \alpha, \kappa)-\uplambda_{r+1}||_{\rm tv} 
&= (1/2) \left ( \Pr[\BadWSUpdate(G_i,\bxi_r,\kappa_r, \kappa)] +\sum\nolimits_{\xi\in \alphabet^V}| \mu_{i}(\xi\ |\ \partial \alpha,\kappa) - \uplambda_{r+1}(\xi)|\right) \nonumber \\
&\leq (1/2) \left( \Pr[\BadWSUpdate(G_i,\hat{\bxi},\kappa, \kappa_{r})]+|{\tt err}|+ \Pr[\BadWSUpdate(G_i,\bxi_r,\kappa_r, \kappa)]
\right) \enspace, \label{eq:ResultA4prop:AccuracyRCUpdateBefore}
\end{align}
where $\hat{\bxi}_i$ is distributed as in $\mu_{i}(\cdot \ |\ \partial \alpha,\kappa)$.

Since the two failure probabilities in \eqref{eq:ResultA4prop:AccuracyRCUpdateBefore} are upper 
bounded by $\rcsq_M$, 
the above inequality yields
$$
||\mu_{i}(\cdot\ |\ \partial \alpha, \kappa)-\uplambda_{r+1}||_{\rm tv} \leq \rcsq_M+|{\tt err}| \enspace.
$$
Combining the above with  \cref{claim:ErrBoud4prop:AccuracyRCUpdate},  we get that
\begin{equation}\label{eq:ResultA4prop:AccuracyRCUpdate}
||\mu_{i}(\cdot\ |\ \partial \alpha, \kappa)-\uplambda_{t+1}||_{\rm tv} \leq \rcsq_M+
6|\alphabet|^{k-1}\chi \psimin^{-1} \cdot 
 \left( \rcsq_M+ { \sum\nolimits_{x\in \partial \alpha\setminus M}\rsq_{x}} \right) \enspace.
\end{equation}
The proposition follows  by combining \eqref{eq:ResultA4prop:AccuracyRCUpdate}, \eqref{eq:BaseA4prop:AccuracyRCUpdate}
and  \cref{claim:Reduction2SingleUpdateRCCase}. \hfill $\Box$

\begin{proof}[Proof of  \cref{claim:ErrBoud4prop:AccuracyRCUpdate}]
Assume that  $\rcsq_M, \rsq_z$, for $z\in \partial \alpha\setminus M$ sufficiently small such that 
$\mu^{-1}_{\bar{H},\partial \alpha}(\kappa_{r})||\mu_{i}-\mu_{\bar{H}}||_{\partial \alpha} \leq 1/10$.
Recall that $||\mu_{i}-\mu_{\bar{H}}||_{\partial \alpha}$ is related with  $\rcsq_M, \rsq_z$ because of \cref{lemma:FactoringOfBarH}.

Using  that  $|\mu_{i,\partial \alpha}(\eta)-\mu_{\bar{H},\partial \alpha}(\eta)|\leq ||\mu_{i}-\mu_{\bar{H}} ||_{\partial \alpha}$,
for any $\eta\in \alphabet^{\partial \alpha}$, and  elementary derivations we get
\begin{align}
|{\tt err}| &\leq 
\frac{\mu^{-1}_{\bar{H},\partial \alpha}(\kappa)+\mu^{-1}_{\bar{H},\partial \alpha}(\kappa_{r})}
{1-\mu^{-1}_{\bar{H},\partial \alpha}(\kappa_{r})||\mu_{i}-\mu_{\bar{H}} ||_{\partial \alpha}}
||\mu_{i}-\mu_{\bar{H}} ||_{\partial \alpha} \enspace.  \nonumber
\end{align}
 Using the assumption that $\mu^{-1}_{\bar{H},\partial \alpha}(\kappa_{r})||\mu_{i}-\mu_{\bar{H}}||_{\partial \alpha} \leq 1/10$, we get
\begin{align}
|{\tt err}| &\leq 
\frac{10}{9}\left( \mu^{-1}_{\bar{H},\partial \alpha}(\kappa)+\mu^{-1}_{\bar{H},\partial \alpha}(\kappa_{r})\right)
||\mu_{i}-\mu_{\bar{H}} ||_{\partial \alpha} \nonumber \\
&\leq  
3\left( \mu^{-1}_{\bar{H},\partial \alpha}(\kappa)+\mu^{-1}_{\bar{H},\partial \alpha}(\kappa_{r})\right)
{ \left( \rcsq_M+ \sum\nolimits_{x\in \partial \alpha\setminus M}\rsq_{x} \right)}& \mbox{[from   \cref{lemma:FactoringOfBarH}]}\nonumber \\
&\leq   6|\alphabet|^{k-1}\chi \psimin^{-1}
 \left( \rcsq_M+ { \sum\nolimits_{x\in \partial \alpha\setminus M}\rsq_{x}} \right) \enspace,
\label{eq:ErrBoud4prop:AccuracyRCUpdate}
\end{align}
where  the last derivation follows from the  observation that the subgraph $\bar{H}$ consists of 
$k-2$ isolated  nodes and a path whose ends belong to $\partial \alpha_i$. For such a graph
and for any $\tau\in \alphabet^{\partial \alpha_i}$ which is in the support of $\mu_{\bar{H},\partial \alpha}(\cdot)$
we have that $\mu^{-1}_{\bar{H},\partial \alpha}(\tau)\leq |\alphabet|^{k-1}\chi \psimin^{-1}$. 
The claim follows.
\end{proof}

\subsection{Proof of \cref{prop:ErrorPerIterUnicyclic}}\label{sec:prop:ErrorPerIterUnicyclic}
For the sake of brevity,  in what follows, for any two $\eta,\kappa\in \alphabet^{\partial \alpha_i}$, we let 
$\nu^{(i+1)}_{\eta,\kappa}$ be the distribution $\nu^{(i+1)}(\cdot \ |\ \btheta_i(\partial \alpha_{i})=\eta, \ \bbeta_{i+1}(\partial \alpha_{i})=\kappa)$.

We consider two cases. In the first one, we assume that the addition of 
$\alpha_i$ does not introduce any new short cycle in $G_{i+1}$. In the second one, 
we assume that  it does. For each case, we  show that \eqref{eq:prop:ErrorPerIterUnicyclic} is true.

We start with the first case.  
\begin{claim}\label{claim:RemoveFirstStepUnicyclic}
We have that
$
\left|\left|\mu_{i+1}- {\nu}^{(i+1)} \right |\right |_{\rm tv} \leq 
 \left|\left| \mu_{i+1}-\bethe_{\alpha_{i}} \right|\right|_{\partial \alpha_i}+
 \max_{\sigma, \kappa \in \alphabet^{\partial \alpha_{i}}}
\left |\left|\mu_{i+1}(\cdot\ |\ \partial \alpha_i, \kappa)- {\nu}^{(i+1)}_{\sigma,\kappa}\right|\right|_{\rm tv}.
$
\end{claim}
In light of the above claim, 
it suffices to show that
\begin{align}
 \max_{\sigma, \kappa \in \alphabet^{\partial \alpha_{i}}}
\left |\left|\mu_{i+1}(\cdot\ |\ \partial \alpha_i, \kappa)- {\nu}^{(i+1)}_{\sigma,\kappa}\right|\right|_{\rm tv} &\leq 
 7k|\alphabet|^k\cdot \rmaxq_i \enspace, \label{eq:RestOfGraphUnicyclic} \\
  || \mu_{i+1, \partial \alpha_i}- \bethe_{\alpha_i} ||_{\rm tv} &\leq  4 |\alphabet|^k \cdot \rmaxq_i \enspace.
 \label{eq:FirstEdgeErrorUnicyclic}
\end{align}
The inequality  in \eqref{eq:RestOfGraphUnicyclic} follows  from 
\cref{prop:Error4RUpdate}. For  \eqref{eq:FirstEdgeErrorUnicyclic},  we use the following result.

\begin{claim}\label{claim:FirstUpdateUniformDistUnicyclic}
Let $\zeta$ be the uniform distribution over  $\alphabet^{V}$. 
Provided that $\rmaxq_i$ is sufficiently small, 
we have that
$ || \mu_{i+1, \partial \alpha_i}-\bethe_{\alpha_i} ||_{\rm tv} \leq 
2|\alphabet|^k \cdot ||\mu_{i} - \zeta||_{\partial \alpha_i}.$
\end{claim}

The inequality in \eqref{eq:FirstEdgeErrorUnicyclic} follows from \cref{claim:FirstUpdateUniformDistUnicyclic} 
and by noting that  \cref{lemma:MarginalLambdaVsUniformDistrWRTRswitch} implies that 
\begin{equation}\nonumber  
||\mu_{i} - \zeta||_{\partial \alpha_i} \leq 2\rmaxq_i \enspace.
\end{equation}

We proceed with the second case, i.e., assume that the addition of 
$\alpha_i$ into $G_i$ introduces a new short cycle in $G_{i+1}$ which we call $C$. 
Let $H$ be the subgraph of $G_{i+1}$ which is induced by the nodes of $C$,
as well as the variable nodes that are adjacent to the factor nodes of this cycle. 
Working as in  \cref{claim:RemoveFirstStepUnicyclic}, we get that
\begin{equation}\label{eq:Base4prop:ErrorPerIterUnicyclicCaseB}
||\mu_{i+1}-\nu^{(i+1)} ||_{\rm tv} \leq 
 || \mu_{i+1}- \mu_{H}||_{\partial \alpha_i}+
 \max_{\sigma, \kappa \in \alphabet^{\partial \alpha_{i}}}
\left |\left|\mu_{i+1}(\cdot\ |\ \partial \alpha_i, \kappa)- {\nu}^{(i+1)}_{\sigma,\kappa}\right|\right|_{\rm tv}\enspace, 
\end{equation}
where $\mu_H$ is the Gibbs distributed induced by $H$.
In light of the above, it suffices to show that
\begin{align}
 \max\nolimits_{\sigma, \kappa \in \alphabet^{\partial \alpha_{i}}}
\left |\left|\mu_{i+1}(\cdot\ |\ \partial \alpha_i, \kappa)- {\nu}^{(i+1)}_{\sigma,\kappa}\right|\right|_{\rm tv}
 &\leq   7k|\alphabet|^k\cdot 
 \left( 1+\chi \psimin^{-1}\right) 
 \cdot\rmaxq_i \enspace,  \label{eq:RestOfGraphUnicyclicCaseB}\\
  || \mu_{i+1}- \mu_{H}||_{\partial \alpha_i} &\leq  4 |\alphabet|^{k-1}\cdot \chi \psimin^{-1} \cdot \rmaxq_i \enspace. 
 \label{eq:FirstEdgeErrorUnicyclicCaseB}
\end{align}
From  \cref{prop:AccuracyRCUpdate} we  immediately get  \eqref{eq:RestOfGraphUnicyclicCaseB}.
For  \eqref{eq:FirstEdgeErrorUnicyclicCaseB} we use the following result.

\begin{claim}\label{claim:FirstUpdateUniformDistShortCLC}
For sufficiently small $\rmaxq_i$, we have that 
$ || \mu_{i+1}- \mu_{H}||_{\partial \alpha_i} \leq  2|\alphabet|^{k-1}\chi \psimin^{-1} ||\mu_{i} - \mu_{\bar{H}}||_{\partial \alpha_i}, $
 where recall that  $\bar{H}$ is obtained from $H$ by removing the factor node $\alpha_i$.
\end{claim}

\noindent
We get \eqref{eq:FirstEdgeErrorUnicyclicCaseB} from \cref{claim:FirstUpdateUniformDistShortCLC} by noting that
\cref{lemma:FactoringOfBarH}, implies that 
\begin{equation}\nonumber  
||\mu_{i} - \mu_{\bar{H}}||_{\partial \alpha_i} \leq 2\rmaxq_i \enspace.
\end{equation}
This concludes the proof of the proposition.  \hfill $\Box$

\begin{proof}[Proof of \cref{claim:RemoveFirstStepUnicyclic}.]
Recall from the beginning of \cref{sec:RUpdateAcc}, that $\btheta_{i+1}$ and $\bbeta_{i+1}$ are distributed as in $\mu_{i+1}$ and $\nu^{(i+1)}$, respectively. Specifically, 
$\nu^{(i+1)}$ is defined below \eqref{eq:RSamplerStep10}.

We couple $\btheta_{i+1}$ and $\bbeta_{i+1}$ as follows: At  first, we couple optimally 
$\btheta_{i+1}(\partial \alpha_i)$ and $\bbeta_{i+1}(\partial \alpha_i)$. 
Then, given  the outcome of the first step, we couple 
$\btheta_{i+1}(V\setminus \partial \alpha_i)$ and $\bbeta_{i+1}(V\setminus \partial \alpha_i)$ optimally. 

Let $\cY_1$ be the event that 
$\btheta_{i+1}(\partial \alpha)\neq \bbeta_{i+1}(\partial \alpha)$. 
Similarly, let  $\cY_2$ be the event that $\btheta_{i+1}(V\setminus \partial \alpha)\neq \bbeta_{i+1}(V\setminus \partial \alpha)$.
We have that 
\begin{align}
||\mu_{i+1}-\hat{\nu}_{i+1} ||_{\rm tv} &\leq \Pr[\cY_1\cup \cY_2] 
%\ = \ \Pr[A_1\cup A_2\ |\ A_1]\Pr[A_1]+ \Pr[A_1\cup A_2\ |\ \bar{A}_1]\Pr[\bar{A}_1] \nonumber\\
\leq  \Pr[\cY_1]+ \Pr[\cY_2\ |\ \bar{\cY}_1] \enspace.  \label{eq:claim:RemoveFirstStepA}
\end{align}
Since we couple $\btheta_{i+1}(\partial \alpha_i)$ and $\bbeta_{i+1}(\partial \alpha_i)$ optimally, we have that
\begin{equation}
\Pr[\cY_1]= || \mu_{i+1, \partial \alpha_i}-\bethe_{\alpha_i}||_{\rm tv} \enspace.
\label{eq:claim:RemoveFirstStepB}
\end{equation}
 Similarly, we get that
\begin{eqnarray} \label{eq:claim:RemoveFirstStepC}
\Pr[\cY_2\ |\ \bar{\cY}_1] &\leq & \max_{\sigma, \kappa \in \alphabet^{\partial \alpha_i}}||\mu_{i+1}(\cdot\ |\ \partial \alpha_i, \kappa)- 
\hat{\nu}_{i+1}(\cdot \ |\ \btheta_i(\partial \alpha_{i})=\sigma, \ \bbeta_{i+1}(\partial \alpha_{i})=\kappa) ||_{\rm tv} \enspace.
\end{eqnarray}
The claim follows by plugging \eqref{eq:claim:RemoveFirstStepB} and \eqref{eq:claim:RemoveFirstStepC} into \eqref{eq:claim:RemoveFirstStepA}.
\end{proof}

\begin{proof}[Proof of \cref{claim:FirstUpdateUniformDistUnicyclic}]
We let $\Uplambda=|\alphabet|^{k}\times||\mu_{i} - \zeta||_{\partial \alpha_i}$. We assume that  $\rmaxq_i$ is so small that  $\Uplambda \leq 1/2$.

We have that
\begin{align}\label{eq:Base4claim:FirstUpdateUniformDistUnicyclic}
|\alphabet|^{-k}- ||\mu_{i} - \zeta||_{\partial \alpha_i} &\leq \mu_{i, \partial \alpha_i}(\eta) \leq |\alphabet|^{-k} + ||\mu_{i} - \zeta||_{\partial \alpha_i} 
& \forall \eta\in \alphabet^{\partial \alpha_i} \enspace.
\end{align}
We can express $\mu_{i+1, \partial \alpha_i}(\eta)$ it terms of $\mu_{i, \partial \alpha_i}$ 
by using the standard relation
\begin{equation}\label{eq:FirstStepIdealMarg}
\mu_{i+1, \partial \alpha}(\eta)\propto {\psi_{\alpha_i}(\eta) \cdot \mu_{i, \partial \alpha_i}(\eta)} \enspace.
\end{equation}
%%%%
From \eqref{eq:Base4claim:FirstUpdateUniformDistUnicyclic} and \eqref{eq:FirstStepIdealMarg}we get that
\begin{align}
\mu_{i+1, \partial \alpha}(\eta)
&\leq 
 \frac{\psi_{\alpha_i}(\eta) \left( |\alphabet|^{-k}+||\mu_{i} - \zeta||_{\partial \alpha_i}\right )}
{\sum_{\eta'} \psi_{\alpha_i}(\eta') \left( |\alphabet|^{-k}- ||\mu_{i} - \zeta||_{\partial \alpha_i}\right )} 
  =
 \frac{\psi_{\alpha_i}(\eta)}{\sum_{\eta'} \psi_{\alpha_i}(\eta') } \cdot
\left( 1+ 2\frac{\Uplambda}{1-\Uplambda} \right)  \leq   \left( 1+ 4\Uplambda \right)\cdot \bethe_{\alpha_i}(\eta) \enspace,
\nonumber 
\end{align}
where in the last derivation we use \eqref{eq:BetheVsWeight} and the assumption that 
$\Uplambda\leq 1/2$.

Working similarly for the lower bound,  we get that 
\begin{align*}%\label{eq:StepA4FirstUpdateUniformDistA1821}
| \mu_{i+1, \partial \alpha}(\eta)- \bethe_{\alpha_i}(\eta) | & \leq 
4 \Uplambda \cdot \bethe_{\alpha_i}(\eta), &
\forall \eta\in \alphabet^{\partial \alpha_i} \enspace.  
\end{align*}
Furthermore, plugging the above into  the definition of total variation distance 
we get
\begin{align*}
 || \mu_{i+1}-\bethe_{\alpha_i}||_{\partial \alpha_i} = 
  (1/2)\sum\nolimits_{\eta \in \alphabet^{\partial \alpha_i}}| \mu_{i+1, \partial \alpha_i}(\eta) - \bethe_{\alpha_i}(\eta) | 
\Uplambda   \sum\nolimits_{\eta \in \alphabet^{\partial \alpha_i}} \bethe_{\alpha_i}(\eta) 
 &= 2 \Uplambda \enspace.  
\end{align*}
The claim follows.
\end{proof}

\begin{proof}[Proof of \cref{claim:FirstUpdateUniformDistShortCLC}]
Let $\uprho=\max_{\tau}\left\{\mu^{-1}_{\bar{H}, \partial \alpha_i}(\tau)\right\}$ 
and $\UpQ=\uprho \cdot ||\mu_{i} - \mu_{\bar{H}} ||_{\partial \alpha_i}$.
 For  $\rmaxq_i$ sufficiently small we have $\UpQ\leq 1/2$ 
This follows from \cref{lemma:FactoringOfBarH}. That is,  if $\rmaxq_i$ is small, then 
$||\mu_{i} - \mu_{\bar{H}} ||_{\partial \alpha_i}$ is small as well.

We also have that
\begin{align}\label{eq:Base4claim:FirstUpdateUniformDistShortCLC}
\mu_{\bar{H}, \partial \alpha_i}(\eta)- ||\mu_{i} - \mu_{\bar{H}} ||_{\partial \alpha_i} \leq \mu_{i, \partial \alpha_i}(\eta) &\leq 
\mu_{\bar{H}, \partial \alpha_i}(\eta) + ||\mu_{i} - \mu_{\bar{H}} ||_{\partial \alpha_i}, 
&\forall \eta\in \alphabet^{\partial \alpha_i} \enspace.
\end{align}
Using the standard relation  
\begin{align*}% 
\textstyle \mu_H (\eta)&\propto 
{\psi_{\alpha_i}(\eta) \mu_{\bar{H}, \partial \alpha_i}(\eta)},
& \textrm{for any $\eta\in \alphabet^{\partial \alpha_i}$} \enspace.
\end{align*}
%%%%
%%%%
together with  \eqref{eq:Base4claim:FirstUpdateUniformDistShortCLC}, %and \eqref{eq:FirstStepIdealMarg}, 
we get that  
\begin{align*}
\mu_{i+1, \partial \alpha}(\eta)
&\leq   \frac{\psi_{\alpha_i}(\eta) \mu_{\bar{H}, \partial \alpha_i}(\eta) }
{\sum_{\eta'} \psi_{\alpha_i}(\eta')\mu_{\bar{H}, \partial \alpha_i}(\eta') } \times 
\frac{1+\UpQ}{1-\UpQ}
  \leq   \mu_{{H},\partial \alpha_i}(\eta)\times \left( 1+4\UpQ\right)  \enspace.
\end{align*}
The last inequality from the assumption that $\UpQ\leq 1/2$.  Working similarly for the lower bound,  we get that 
\begin{align*}%\label{eq:StepA4FirstUpdateUniformDist}
| \mu_{i+1, \partial \alpha}(\eta)- \mu_{{H},\partial \alpha_i}(\eta) | &\leq 4\UpQ\times \mu_{{H},\partial \alpha_i}(\eta)
& \forall \eta\in \alphabet^{\partial \alpha_i}
\enspace.
\end{align*}
Using  the definition of the total variation distance and plugging the above in inequality, %\eqref{eq:Step1ForFirstUpdateUniformDistSC} 
we have that 
\begin{align}
 || \mu_{i+1}-\mu_H  ||_{\partial \alpha_i} &= 
  \frac12\sum\nolimits_{\eta \in \alphabet^{\partial \alpha_i}}| \mu_{i+1, \partial \alpha_i}(\eta) - \mu_{{H}, \partial \alpha_i}(\eta) |
 %%%%
%
 %%%%
 = 2 \UpQ\enspace. 
 \nonumber
\end{align}
The claim follows by noting that $\varrho\leq |\alphabet|^{k-1}\chi \psimin^{-1}$.
\end{proof}

\spreadpoint
\section{Running Time of $\rsampler$ - Proof of \cref{thrm:FinalDetTime}
\LastReview{2024-02-20}
}\label{sec:thrm:FinalDetTime}

First, we show that we can check whether $G\in \CG$ in  $O((n+m)^2)$ time. 
If $G\notin \CG$, then there is  a  node which belongs to more than one short-cycles.
We can check whether there is such a node  by initiating  a  Depth First Search (DFS) from each  one of the nodes in $G$ . 
The running time of a single DFS excursion requires $O(N+M)$  where $N$ is the number of nodes 
and $M$ is the number of edges.  In our case, we have $N=n+m$ and $M=k m$.
Furthermore, we repeat DFS for each one of the  $n$ variable nodes of the graph. It is direct that  this check requires $O((n+m)^2)$.

Now, we focus on the running time of each iteration of $\rsampler$, i.e., this is the running time to get $\bsigma_{i+1}$ given $\bsigma_i$.

When $\alpha_i$ introduces a short cycle $C$ in $\G_{i+1}$, recall that we 
let $H$ be  the subgraph  of $\G_{i+1}$ that is induced by both the variable and factor nodes in  the cycle $C$, as well as all the  variable nodes that are adjacent to the factor nodes in $C$.  In this case, we denote $\DiaSet$  the set of variable
nodes in $H$. 

Recall, also, that when $\alpha_i$ does not introduce any short cycle in $\G_{i+1}$,
$\DiaSet$ is  the set $\partial \alpha_i$,  If $\DiaSet=\partial \alpha_i$, the configuration  
$\bsigma_{i+1}(\DiaSet)$ can be generated in  time $O(k)$.
That is, we spend $O(1)$ time for each of the $k$ variable nodes in $\partial \alpha_i$.

If $\DiaSet=V(H)$, then recall that $H$ is a cycle with the variable nodes attached to its factor nodes.  Since we assumed that we are dealing with a symmetric Gibbs distribution, we can eliminate the effect of  the cycle  
by working as follows:
Choose  $x\in \DiaSet$ which is also in the unique cycle of $H$, arbitrarily,  and 
set $\bsigma_{i+1}(x)$ according to the distribution 
\begin{align}\nonumber
\Pr[\bsigma_{i+1}(x)=c]&=|\alphabet|^{-1} &\forall c\in \alphabet\enspace.
\end{align} 
Then, once $x$ has been set, the variable nodes in $\DiaSet\setminus\{x \}$ induce a tree subgraph of $H$.
Specifically, we  can sample from the distribution $\mu_{H}(\cdot \ |\ \{x, c\})$ for any $c\in \alphabet$ by using dynamic programming.
It is standard to show that the dynamic program would require 
 $O(|\alphabet|^k \cdot |\DiaSet|)$ steps, 
e.g., see \cref{sec:SupplemantaryA}. 

From all the above,  we conclude that $\rsampler$ requires $O(|\alphabet|^k \cdot| \DiaSet|)=O(\DiaSet)$ steps to decide $\bsigma_{i+1}(\DiaSet)$. Note that  $|\alphabet|^k$ and  $k$ are $O(1)$.

We continue with the time complexity of  $\rswitch$.  Particularly, consider 
$\rswitch(G_i, \sigma, \eta, \kappa)$ where $\eta,\kappa \in \alphabet^{\DiaSet}$ differ only on 
$x\in \DiaSet$.   $\rswitch$ has common  features with $\switch$ whose performance we study  in 
\cref{lemma:UpdateWSTimeComplexity}.

 If at some  iteration $t$, $\rswitch$ chooses the factor node $\beta$ which is away from a short cycle, then the
process decides  the configuration at $\partial \beta$ in the same manner as $\switch$, which takes
$O(k)$ steps. 

If $\beta$, or $\partial \beta$  intersects with the short cycle $C$, then the process 
needs to decide the configuration  of $O(k |C|)$ variable nodes, where $|C|$ is the length of the cycle. This does not happen in $\switch$. From \eqref{eq:ShortCycleUpdtA} and \eqref{eq:ShortCycleUpdtB} it is  immediate that this iteration requires  $O(k)$ steps for each factor node in $C$.

Using the above and arguing as in \cref{lemma:UpdateWSTimeComplexity}, we get that 
the time complexity of  $\rswitch(G_i, \sigma, \eta, \kappa)$ is $O((m+n))$.
 since we assume that $k=\Theta(1)$.

We note that at each iteration $\rsampler$ makes at most $|\DiaSet|$ calls of $\rswitch$. 
Hence, each iteration of $\rsampler$ takes $O(|\DiaSet|\cdot ( m+n)+| \DiaSet|)=
O(( m+n)\log(n))$, since $|\DiaSet|=O(\log n)$.

Since  $\rsampler$ needs $m$ iterations to create $\bsigma_{m}$, i.e., the output configuration, the total number of steps is $O(m(m+n)\log(n))$. 

We obtain the total running time by adding the time needed to check whether
$G\in \CG$ and the time needed to create $\bsigma_{m}$. The theorem follows.

\spreadpoint

\newcommand{\neighpi}{N_{\pi}}
\newcommand{\scylcpi}{L_{\pi}}

\section{Disagreement Propagation - Proof of \cref{thrm:DisagreementDecayPathSetB}
\LastReview{2024-01-24}}
\label{sec:thrm:DisagreementDecayPathSetB}

Recall that we consider $\pi=x_1, \ldots x_{\ell}$ such that $\pi\in \Pi_{\ell,z}$. 
Furthermore, we need to bound  the quantity 
$\mathbb{E}[\UpU^{(s,r)}_{\pi}\times \UpJ^s_{\pi} \ |\ \cB,\  \G^*_i\in \CG]$
with respect to the process $\iter(\G^*_i, \bsigma^*, \bsigma^*(\DiaSet),  \bkappa^*(\DiaSet))$.

To simplify our notation and the statement of our result, we assume that $\UpU^{(r,s)}_{\pi}=0$, 
when  $r,s$ take  on values such that  $\UpU^{(r,s)}_{\pi}$ is not meaningful.

Let the  set $\neighpi$ consists of every variable node which 
is adjacent to a factor node of $\pi$ in $\G^*_i$. $\neighpi$ does not necessarily
include nodes only in $\pi$.
 Also, let $\scylcpi$ be the set that consists every  factor node $\pi$ which either
 belongs to a short cycle,  or it is at distance one from a short cycle.

For $0\leq s \leq \ell$, or $s=\infty$,  and $r\geq 0$ we have that 
\begin{align}
\lefteqn{
\mathbb{E}[\UpU^{(s,r)}_{\pi}\times \UpJ^s_{\pi} \ |\ \cB,\  \CG] 
}\hspace{.1cm} \nonumber \\
&=  \sum\nolimits_{\sigma\in \alphabet^{\neighpi}}  \mathbb{E}[\UpU^{(s,r)}_{\pi}\ |\  \UpJ^s_{\pi}=1,\  \bsigma^*(\neighpi)=\sigma, \ \cB, \  \CG] 
\times\Pr[\UpJ^s_{\pi}=1,\  \bsigma^*(\neighpi)=\sigma \ |\  \cB,\ \CG] \enspace,  \qquad  \label{eq:Base4thrm:DisagreementDecayPathSetB}
\end{align}
where, for  brevity, we let $\CG$ denote the event that $\G^*_i\in \CG$. 

Furthermore, letting $\cS(\sigma)$ be the event that   $\UpJ^s_{\pi}=1$ and $\bsigma^*(\neighpi)=\sigma$, we have that
\begin{align}
\mathbb{E}[\UpU^{(s,r)}_{\pi}\ |\ \cS(\sigma), \ \cB,\ \CG]   
&=   \sum\nolimits_{J}  
\mathbb{E}[\UpU^{(s,r)}_{\pi}  \ |\   \scylcpi=J,\  \cS(\sigma),\ \cB,\ \CG]  
\times \Pr[\scylcpi=J \ |\ \cS(\sigma),\ \cB,\ \CG ] 
\nonumber \\ &\leq   \sum\nolimits_{J}
\mathbb{E}[\UpU^{(s,r)}_{\pi}  \ |\  \scylcpi=J,\  \cS(\sigma),\ \cB,\ \CG]  
\times \Pr[\scylcpi=J \ |\ \cS(\sigma_{\max}),\ \cB,\ \CG ] 
\enspace, \label{eq:Base4thrm:DisagreementDecayPathSetBNNN}
\end{align}
where $J$ varies over the subset of factor nodes in $\pi$. 
We let $\sigma_{\max}=\sigma_{\max}(J)$ be the configuration 
at $\neighpi$ which maximises the probability of the event $\scylcpi=J$.
Plugging \eqref{eq:Base4thrm:DisagreementDecayPathSetBNNN} into \eqref{eq:Base4thrm:DisagreementDecayPathSetB} and
rearranging, we get 
\begin{eqnarray}
% \lefteqn{
\mathbb{E}[\UpU^{(s,r)}_{\pi}\times \UpJ^s_{\pi} \ |\   \cB,\  \CG] 
% } \nonumber \\ 
&\leq & \sum\nolimits_{J}\Pr[\scylcpi=J \ |\ \cS(\sigma_{\max}),\ \cB,\  \CG ]  \nonumber \\
&&\qquad \times
 \sum\nolimits_{\sigma} \mathbb{E}[\UpU^{(s,r)}_{\pi}\ |\  \cS(\sigma),\ \scylcpi=J,\ \cB,\ \CG] 
\cdot \Pr[\cS(\sigma) \ |\ \cB,\ \CG]\enspace.   \label{eq:2Base4thrm:DisDecayPath}
\end{eqnarray}
We   upper bound  the  rightmost summation in \eqref{eq:2Base4thrm:DisDecayPath} by using the following result.

\begin{proposition}\label{lemma:DisagrCycleColor} 
For $0\leq i<m$,  $\delta\in (0,1]$,  assume that $\mu_i$ satisfies $\setB$  with slack  $\delta$.  
For $\ShortDist\leq \ell \leq (\log n)^5$, for $z \in \DiaSet$ and any $\pi\in \Pi_{\ell,z}$ 
the following holds:  There is a constant $\widehat{C}>0$  such that for  $0\leq s \leq \ell$, or $s=\infty$, 
for $0\leq r\leq \ell$ and any $J$, subset of factor nodes in $\pi$, we have that 
% %
\begin{align}
 \sum\nolimits_{\sigma\in \alphabet^{\neighpi}} 
 \mathbb{E}[\UpU^{(s,r)}_{\pi}\ |\  \cS(\sigma),\ \scylcpi=J,\ \cB,\ \CG] 
&\times \Pr[\cS(\sigma) \ |\  \cB,\  \CG]   \nonumber  \\
 &\leq  \widehat{C} \cdot  \Upsigma(s)  \cdot n^{-\ell}\cdot
 \left( (1-\delta){k} / {d}\right)^{\lfloor \ell/2 \rfloor} \cdot 
  \drate^{-|J|} \enspace, \nonumber 
\end{align}
where $\drate=\frac{1-\delta}{d(k-1)}$ is defined in \eqref{eq:RateDisDef}, 
while $\Upsigma(s)=|\DiaSet|$ for $s\neq \infty$
and $\Upsigma(\infty)=n$. 
\end{proposition}

\noindent
The proof of  
\cref{lemma:DisagrCycleColor} appears in  \cref{sec:lemma:DisagrCycleColor}.

For $\drate=\frac{1-\delta}{d(k-1)}$, set 
\begin{align} \nonumber
\Uplambda= \sum\nolimits_{J}   \drate^{-|J|} \times  \Pr[\scylcpi=J \ |\ \cS(\sigma_{\rm max}),\ \cB,\ \CG ] \enspace,
\end{align}
where, as before, $J$ varies over all subsets of the factor nodes in $\pi$. 

\begin{proposition}\label{lemma:CycleContr2DisagreeSetB}
We have that $\Uplambda= \left( 1+o(1)\right).$
\end{proposition}

\noindent
The proof of  \cref{lemma:CycleContr2DisagreeSetB} appears in  \cref{sec:lemma:CycleContr2DisagreeSetB}.

 \cref{thrm:DisagreementDecayPathSetB} follows by plugging the bounds from  \cref{lemma:DisagrCycleColor,lemma:CycleContr2DisagreeSetB} into 
\eqref{eq:2Base4thrm:DisDecayPath}. \hfill $\Box$

\subsection{Proof of  \cref{lemma:DisagrCycleColor}}\label{sec:lemma:DisagrCycleColor}

\noindent
Recall that  $\pi=x_1, \ldots x_{\ell}$ such that $\pi\in \Pi_{\ell,z}$. 
Also, given $s, r$, we let $\Phi \subseteq [\ell]\setminus\{r \}$ 
be the indices of the factor nodes in $\pi$ apart from $r$.
For brevity, we let
\begin{align}\label{eq:DefOfUpXI}
\Upxi= \sum\nolimits_{\sigma\in \alphabet^{\neighpi}} 
 \mathbb{E}[\UpU^{(s,r)}_{\pi}\ |\  \cS(\sigma),\ \scylcpi=J,\ \cB,\ \CG] 
\times \Pr[\cS(\sigma) \ |\  \cB,\  \CG]   \enspace.
\end{align}
\cref{lemma:DisagrCycleColor} follows by bounding appropriately $\Upxi$.

Recall that all the above are considered with respect to the process 
$\iter({\G}^*_i, \bsigma^*, \bsigma^*(\DiaSet),  \bkappa^*(\DiaSet))$,
$z$ is a disagreement at $\bsigma^*(\DiaSet)\oplus \bkappa^*(\DiaSet)$, 
while $\pi\in \Pi_{\ell,z}$. Assume that $\bsigma^*(z)=c$ and 
$\btau^*(z)=\bar{c}$. Hence, the set of disagreeing spins is $\DisSpin=\{c,\bar{c}\}$.

We start  by considering the case $s\neq \infty$. 
Recall, also,  that  $\cS(\sigma)$ stands for the event that $\UpJ^s_{\pi}=1$ and 
$\bsigma^*(\neighpi)=\sigma$.  On the event $\UpJ^s_{\pi}=1$, let $T$ be the
tree or forest that is induced by the nodes in $\pi\setminus \{x_r\}$. 
Assume that the root is $x_1$, while if there are more than one tree components, 
then the root is a node $x_{\ell}$, i.e., the one that is connected to  $\DiaSet$

Assume w.l.o.g. that the process updates from parent to children, i.e., using a preorder traversal
of the tree.  For every node $x_{j}$  in  $T$,  let ${\tt p}(x_j) $ be the parent node.

Recall that $\Phi\subset [\ell]$ is the set of indices of the factor nodes in $\pi$.  
For $j\in \Phi$,  let $\cX_j$ be the event that the descendant(s) of  the factor node 
$x_j$ is disagreeing.  Also, let $\cF_{\pi}$ be the $\sigma$-algebra generated by the 
weight functions $\bpsi_{j}$, for all $j\in \Phi$. We have that
\begin{align}\label{eq:ExpUVsPrCalX}
\mathbb{E}[\UpU^{(s,r)}_{\pi}\ |\  \cF_{\pi},\ \cS(\sigma),\ \scylcpi=J,\ \cB,\  \CG] &= \Pr[\wedge_{j}\cX_j\ |\ \cF_{\pi},\ \cS(\sigma),\ \scylcpi=J,\ \cB, \   \CG] \enspace.
\end{align}
Because of the conditioning above, the probability term on the r.h.s. is only with respect to the random choices of $\iter$.

Note that the disagreements involve  only on the variable nodes whose configuration is in
$\DisSpin=\{c,\bar{c}\}$. A factor node $x_j$ propagates the disagreement from its 
parent ${\tt p}(x_j)$  to its children  with configuration in $\DisSpin$ when it is 
updated.  If $x_j\notin \scylcpi$  the probability that (all) the children of $x_j$  
become disagreements is   equal to $\dpr_{x_j}(\sigma(\partial x_j))$, 
also considered in \eqref{eq:BroadCastDisagreementProb},   defined as follows:
assume w.l.o.g. that $\sigma({\tt p}(x_j))=c$, then  
for all $\tau \in \alphabet^{\partial x_{j}}$ such that $\tau({\tt p}(x_j))=c$,  we have
\begin{align} \nonumber %\label{eq:DisProb4Perc}
\dpr_{x_j}(\tau)&=
\max \left \{0, 1-\frac{\bethe_{x_{j}}\left( \bar{\tau}\ |\ {\tt p}(x_j), \ \bar{c}\right)}
{\bethe_{x_{j}}(\tau\ |\ {\tt p}(x_j), \  c)} \right\} \enspace,
\end{align}
where  $\bar{\tau}\in\alphabet^{\partial x_j}$ is such that $\bar{\tau}({\tt p}(x_j))=\bar{c}$,  while
$\bar{\tau}(y)=\tau(z)$ for all  $z \in \partial x_j\setminus \{ {\tt p}(x_j)\}$.  
For any other $\tau \in \alphabet^{\partial x_{j}}$, i.e., with $\tau({\tt p}(x_j))\neq c$, 
we have  $\dpr_{x_j}(\tau)=0$.

For $x_j\in \scylcpi$,  we use the trivial bound $1$ for the probability of disagreement in our
estimations.

% Recall that $\dpr_{x_j}(\cdot)$. 
Using the product rule, we get that 
\begin{align}
\Pr[\wedge_{j\in \Phi}\cX_j\ |\ \cF_{\pi},\ \cS(\sigma),\ \scylcpi=J,\  \cB,\   \CG] 
% &=   { \prod^{|\phi|}_{j=1}} \Pr[\cX_j\ |\ \wedge_{t<j} \cX_{\phi_t}, \ \cF_{\pi},\ \cS(\sigma),\ \scylcpi=J,\ \cB,\  \CG] \nonumber \\
&\leq  { \prod\nolimits_{j\in \Phi: \ x_{j}\notin J} } \dpr_{x_j}(\sigma\left(\partial x_{j})\right)
\enspace, \nonumber 
\end{align}
where the weight function for each $\dpr_{x_j}$ is specified by $\cF_{\pi}$.
The above together with \eqref{eq:ExpUVsPrCalX} imply that
\begin{align}\label{eq:ExpUCondFAlg} 
\mathbb{E}[\UpU^{(s,r)}_{\pi}\ |\  \cF_{\pi},\ \cS(\sigma),\ \scylcpi=J,\ \cB,\  \CG] 
&\leq   \prod\nolimits_{j\in \Phi:\  x_{j}\notin J}  \dpr_{x_j}(\sigma(\partial x_{j}))\enspace.
\end{align}

At this point,  we shift our focus on the term  $\Pr[\cS(\sigma) \ |\  \cB,\  \CG]$ in \eqref{eq:DefOfUpXI}
and obtain the following result. 

\begin{lemma}\label{lemma:TreeMarginal}
There exists $C_1>0$ such that, for   measurable set $\cW_j\subseteq \Psi$, where 
$j\in \Phi$ and  $\sigma\in \alphabet^{\neighpi}$, the following is true:
For $ \ShortDist\leq \ell \leq (\log n)^5$  and  $0\leq s \leq \ell$ we have that
\begin{eqnarray}
\lefteqn{
\Pr[\UpJ^s_{\pi}=1,\  \bsigma^*(\neighpi)=\sigma,\ \wedge_{j \in \Phi} \bpsi_{j}\in \cW_j  \ |\ \cB,\ \CG] 
}  \hspace{.5cm}\nonumber \\
&\leq &C_1 |\DiaSet| 
{\textstyle \left( \frac{k}{n}\right)^{\lceil \ell/2 \rceil}\left( \frac{k-1}{n-1}\right)^{\lfloor \ell/2 \rfloor} } { \prod\nolimits_{j\in \Phi}  }
\mathbb{E}_{\bpsi_j\sim \dpsi}
\left[ \Ind  \{ \bpsi_j \in \cW_{j} \} 
\times \bethe_{x_{j}} 
(\sigma(\partial x_{j})\ |\ {\tt p}(x_j),\  \sigma) 
\right] \enspace.  \nonumber
\end{eqnarray}
Furthermore,  for  any $1\leq \ell \leq (\log n)^5$, 
we have that
\begin{eqnarray}
\lefteqn{
\Pr[\UpJ^{\infty}_{\pi}=1,\  \bsigma^*(\neighpi)=\sigma,\ \wedge_{j \in \Phi} \bpsi_{j}\in \cW_j  \ |\ \cB,\  \CG] 
}  \hspace{.5cm}\nonumber \\
&\leq &C_1 {\textstyle \left( \frac{k}{n}\right)^{\lfloor \ell/2 \rfloor}\left( \frac{k-1}{n-1}\right)^{\lfloor (\ell-1)/2 \rfloor} } 
 \prod\nolimits_{j\in \Phi} 
\mathbb{E}_{\bpsi_j\sim \dpsi}\left[\Ind \{\bpsi_{j}\in \cW_{j}\} \times \bethe_{x_{j}} (\sigma(\partial x_j)\ |\ {\tt p}(x_j),\  \sigma) \right] \enspace.  \nonumber
\end{eqnarray}
\end{lemma}

\noindent
The proof of  \cref{lemma:TreeMarginal} appears in  \cref{sec:lemma:TreeMarginal}.

Combining \eqref{eq:ExpUCondFAlg}  with \cref{lemma:TreeMarginal} we get  that 
\begin{eqnarray}\label{eq:ExptUVsBBS}
\Upxi \leq  C_1 \cdot |\DiaSet| \cdot \left( \frac{k}{n}\right)^{\lceil \ell/2 \rceil}\cdot \left( \frac{k-1}{n-1}\right)^{\lfloor \ell/2 \rfloor} \cdot  \UpS \enspace,
\end{eqnarray}
where   
\begin{align} \nonumber
\UpS&=\sum\nolimits_{\sigma\in \alphabet^{\neighpi}}     
\prod\nolimits_{j\in \Phi}
\mathbb{E}_{\textrm{$\bpsi_{j}$}\sim \dpsi}
\left[ \left(
\Ind\{x_{j} \notin J\} \cdot \dpr_{x_j}(\sigma(\partial x_j))  + \Ind\{x_{j} \in J\} \right)
\cdot \bethe_{x_j} (\sigma(\partial x_{j})\ |\ {\tt p}(x_j), \sigma) \right]
\enspace.
\end{align}
In order to bound $\UpS$ we use the following result whose proof appears in
\cref{sec:claim:MathbbSBound}.

\begin{lemma}\label{claim:MathbbSBound}
We have that  $\UpS \leq  \drate^{\lfloor \ell/2 \rfloor-2-|J|}$,
where $\drate=\frac{1-\delta}{d(k-1)}$. 
\end{lemma}

\noindent
Plugging  the bound from \cref{claim:MathbbSBound}  into \eqref{eq:ExptUVsBBS}, we get the following: 
There exists $\widehat{C}>0$ such that 
\begin{equation}\nonumber 
\Upxi
\leq  \widehat{C} \cdot |\DiaSet| \cdot n^{-\ell} \cdot
{ \left( (1-\delta){k}/{d}\right)^{\lfloor \ell/2 \rfloor} \cdot 
\drate^{-|J|} } \enspace.
\end{equation}
We work similarly, for the case where $s=\infty$.
\cref{lemma:DisagrCycleColor} follows. \hfill $\Box$

\subsection{Proof of  \cref{lemma:TreeMarginal}}\label{sec:lemma:TreeMarginal}

Firstly,   assume that  $s\notin\{0,\infty\}$ and $x_s$ is a variable node, hence $\ell$ is  an even number.

Let ${\cI}$ be the event that  both 
$x_{1}, \ldots, x_{r-1}$ and  $x_{r+1}\ldots x_{\ell}, x_{s}$ are paths in $\G^*_i$.  
Also,  let  $\bar{\cI}$ be the event that the factor node $x_r$
is adjacent to the variable nodes $x_{r+1}$ and $x_r$. 

Since  $\UpJ^s_{\pi}=1$ corresponds to the event $\cI \wedge  \bar{\cI}$, 
from the product rule, we obtain that 
\begin{align}
\Pr[\UpJ^s_{\pi}=1,\  \bsigma^*(\neighpi)&=\sigma,\ \wedge_{j \in \Phi} \bpsi_{j}\in \cW_j  \ |\ \cB,\  \CG]  
 \nonumber \\
&= \Pr[\cI, \  \bar{\cI},\  \bsigma^*(\neighpi)=\sigma,\ \wedge_{j \in \Phi } \bpsi_{j}\in \cW_j  \ |\  \cB, \ \CG]   \nonumber \\
&=\Pr[\cI,\  \bsigma^*(\neighpi)=\sigma,\ \wedge_{j \in \Phi } \bpsi_{j}\in \cW_j  \ |\  \cB,\ \CG] \label{eq:Base4lemma:TreeMarginalA} \\
&\qquad \times
\Pr[ \bar{\cI}\ |\ \cI,\ \bsigma^*(\neighpi)=\sigma,\ \wedge_{j \in \Phi } \bpsi_{j}\in \cW_j,\ \cB,\ \CG] \enspace.   
\quad \qquad \label{eq:Base4lemma:TreeMarginalB}
\end{align}
We focus on  bounding appropriately the two
probability terms in   and  \eqref{eq:Base4lemma:TreeMarginalB}.

We start with the term in \eqref{eq:Base4lemma:TreeMarginalA}. 
It turns out that is easier to, first,  upper bound the probability term  $\Pr[\cI,\ \bsigma^*(\neighpi)=\sigma, \ 
\wedge_{j \in \Phi} \bpsi_{j}\in \cW_j\ |\ \cB]$, i.e., having  removed conditioning on $\CG$. 

On the event $\cI$ the nodes in $\pi\setminus \{x_r\}$ induce a tree (or a forest) in $\G^*_i$. 
As before, we call this tree $T$.  We bound our probability term by revealing the tree and its configuration
under $\bsigma^*$ in steps by using a preorder traversal. That is, inductively,  having revealed the configuration of 
the parent, we reveal its children in $T$ and their configuration.

Let $\phi$ be a permutation  of the indices in $\Phi$  such that the elements of  $\Phi$ appear in the same order as  
in  the  preorder traversal of $T$.  For the sake of keeping the notation in 
this proof simple, we assume that the indices in $\phi$ appear from smaller to larger.

For $j \in \Phi$,  let  $\cI_j$ be the event that the factor node   $x_j$ is adjacent to the variable nodes $x_{j-1}$ 
and $x_{j}$,  while for $j=\ell$ we have that 
$x_{\ell}$ is connected to $x_{\ell-1}$ and $x_s$. Note that $\cI$ corresponds to having $\wedge_{j\in \Phi}\cI_j$. 

From the product rule, we have that
\begin{eqnarray}
\lefteqn{
\Pr[\cI,\ \bsigma^*(\neighpi)=\sigma, \ \wedge_{j \in \Phi} \bpsi_{j}\in \cW_j \ |\ \cB ] 
} \hspace{2cm}\nonumber \\
&=&\prod\nolimits_{j\in \Phi} 
\Pr[ \bsigma^*(\partial x_{j} )=\sigma(\partial x_{j}),\  \cI_{j }, \ \bpsi_{j} \in \cW_{j}   \ |\ \cB,\ \cH(j)  ] \enspace,
\label{eq:Reduction2MarginalDisagree}
\end{eqnarray}
where,  for $j\in \Phi$, we have that
\begin{align}
\cH(j):= {\textstyle \bigwedge\nolimits_{t<j}  } \bsigma^*(\partial x_{t} ) =\sigma(\partial x_{t}),  \ \cI_{t}, \ 
\bpsi_{t}\in \cW_{t}  \enspace.
\end{align}
With the preoder traversal of $T$, conditioning on $\cH(j)$, 
we  only have information about the ancestors of $x_j$. 

\begin{claim}\label{claim:IndependentMarginal}
For any $j\in \Phi$, we  have that 
\begin{eqnarray}  
\lefteqn{
\Pr[ \bsigma^*(\partial x_{j} )=\sigma(\partial x_j),\  \cI_j, \ \bpsi_{j} \in \cW_{j}   \ |\ \cB,\ \cH(j)]
} \hspace{2cm}\nonumber \\
&=& {\textstyle \left (1+O\left(n^{-1/3} \right) \right )}  \cdot 
\frac{k(k-1)}{n(n-1)} \cdot
\mathbb{E}_{\bpsi_j\sim \dpsi}  \left[\Ind \{\bpsi_{j}\in \cW_{j}\} \times \bethe_{x_j} 
(\sigma(\partial x_j ) \ |\ {\tt p}(x_j),\  \sigma) \right] \enspace.
\nonumber
\end{eqnarray}
\end{claim}

\noindent
Since the number of terms in $\Phi$ is  $\leq \ell$, \cref{claim:IndependentMarginal}
and \eqref{eq:Reduction2MarginalDisagree}  imply that 
\begin{eqnarray}
\lefteqn{
\Pr[\cI,\ \bsigma^*(\neighpi)=\sigma,\ \wedge_{j \in \Phi } \bpsi_{j}\in \cW_j \  |\ \cB ]  
} \hspace{.5cm}\nonumber \\
 &=&  {\textstyle \left(1+O\left(\ell \cdot n^{-1/3}\right) \right)}\cdot
 \left( \frac{k(k-1)}{n(n-1)}\right)^{|\Phi|}  
 \prod\nolimits_{j\in \Phi}  
\mathbb{E}\left[\Ind \{\bpsi_{j}\in \cW_{j}\} \times \bethe_{x_j} (\sigma(\partial x_j)\ |\ 
{\tt p}(x_j), \sigma) \right]
\enspace . \nonumber  
\end{eqnarray}
Returning to the term in \eqref{eq:Base4lemma:TreeMarginalA}, we  have that 
\begin{eqnarray}
\Pr[\cI,\ \bsigma^*(\neighpi)=\sigma,\ \wedge_{j \in \Phi } \bpsi_{j}\in \cW_j \ |\ \cB,\  \CG]  &\leq& 
\frac{\Pr[\cI,\ \bsigma^*(\neighpi)=\sigma,\ \wedge_{j \in \Phi } \bpsi_{j}\in \cW_j \ |\ \cB]  }{\Pr[\CG\ |\ \cB]} \nonumber \\
&=&(1+o(1)) \Pr[\cI,\ \bsigma^*(\neighpi)=\sigma,\ \wedge_{j \in \Phi } \bpsi_{j}\in \cW_j \ |\ \cB] \enspace. \nonumber  
\end{eqnarray}
Furthermore, using  \eqref{def:Probs4GStar} and working as in the proof of 
\cref{claim:IndependentMarginal}
it is immediate to show that there exists fixed constant $\widehat{C}_0>0$ such that 
\begin{align} \nonumber
\Pr[ \bar{\cI}\ |\ \cI,\  \bsigma^*(\neighpi)=\sigma,\ \wedge_{j \in \Phi } \bpsi_{j}\in \cW_j, \ \cB,\  \CG]&\leq \widehat{C}_0\cdot n^{-2} \enspace.
\end{align}
From the three relations above and  \eqref{eq:Base4lemma:TreeMarginalA}, 
\eqref{eq:Base4lemma:TreeMarginalB}, we get that
\begin{eqnarray}
\lefteqn{
\Pr[\UpJ^s_{\pi}=1,\  \bsigma^*(\neighpi)=\sigma,\ \wedge_{j \in \Phi} \bpsi_{j}\in \cW_j  \ |\ \cB,\  \CG] 
}  \hspace{2cm}\nonumber \\
&\leq &\frac{\hat{C}_0}{n^{2}}\cdot  \left( \frac{k(k-1)}{n(n-1)}\right)^{|\Phi|} \cdot
 \prod\nolimits_{j\in \Phi}  
\mathbb{E}\left[\Ind\{\bpsi_{j}\in \cW_{j}\} \times \bethe_{x_{j}} (\sigma(\partial x_{j})\ |\ {\tt p}(x_j), \sigma) \right] \nonumber \\
&\leq& C_1 \cdot \left( \frac{k(k-1)}{n(n-1)}\right)^{\ell/2} \cdot
 \prod\nolimits_{j\in \Phi}  
\mathbb{E}\left[\Ind\{\bpsi_{j}\in \cW_{j}\} \times \bethe_{x_{j}} (\sigma(\partial x_{j})\ |\ {\tt p}(x_j), \sigma) \right] \nonumber  \enspace.
\end{eqnarray}
For the last inequality, we use that $k=\Theta(1)$, while we also use that 
$|\Phi|=\ell/2-1$. This follows from the fact that our initial assumptions
imply that  $\ell$ is an even number.  The quantity $C_1$ is the same as
the one in the statement of the lemma.

Recall that we initially assumed that $0<s<\ell$ and $\ell$ is an even number
The other cases for $\ell$ and $s$, i.e., $s=0$ and $s=\infty$, follow very similarly, 
for this  reason, we omit their derivation.  
\cref{lemma:TreeMarginal} follows. \hfill $\Box$

\begin{proof}[Proof of  \cref{claim:IndependentMarginal}]
To simplify the notation in the proof that follows, assume w.l.o.g. that $j\neq \ell$, 
i.e., the case for $j=\ell$ is identical to the one we consider below alas with more
involved notation. 

From the definition of the event $\cI_j$, it is elementary that
\begin{align}
\lefteqn{
\Pr[ \bsigma^*(\partial x_{j} )=\sigma(\partial x_j),\  \cI_j, \ \bpsi_{j} \in \cW_{j}   \ |\ \cB,\ \cH(j)] } \hspace{3cm}\nonumber \\
&=\sum\nolimits_{\mathbold{z}}
\Pr[\bsigma^*(\partial x_j )=\sigma ,  \partial  x_j=\mathbold{z}, \ 
\bpsi_{j} \in  \cW_{j} \ |\  \cB,\ \cH(j) ] \enspace,
\label{eq:Base4claim:IndependentMarginal}
\end{align}
where note that $\mathbold{z}$ varies over the $k$-tuples $\mathbold{z}=(z_1,\ldots, z_k)$
of distinct nodes. We prove \cref{claim:IndependentMarginal} by estimating the
summads in \eqref{eq:Base4claim:IndependentMarginal}.

We start by considering  $\mathbold{z}=(z_1,\ldots, z_k)$  a $k$-tuple of variable nodes 
which intersects with $\pi$ and $\DiaSet$ exactly at the nodes $x_{j+1}$ and $x_{j-1}$. W.l.o.g. assume 
that $z_{k-1}=x_{j-1}$ and $z_k=x_{j+1}$. 

We show that for $\mathbold{z}$ and any $\tau\in \alphabet^{\mathbold{z}}$  such that  
$\tau({\tt p}(x_j))=\sigma({\tt p}(x_j))$, we have that
\begin{align}
\lefteqn{
\Pr[\bsigma^*(\partial x_j )=\tau ,  \partial  x_j=\mathbold{z}, \ 
\bpsi_{j} \in  \cW_{j} \ |\  \cB,\ \cH(j) ] 
} \hspace{2cm}\nonumber \\
&=
{\textstyle \left(1+O\left(n^{-1/3}\right)\right)  \times \left(n^{\underline{k}}\right)^{-1}} \times
\mathbb{E}  \left[\Ind\{\bpsi_{j}\in \cW_{j}\} \cdot
\bethe_{x_j} (\tau\ |\  {\tt p}(x_j), \ \sigma) \right]
\enspace. \label{eq:FirstTarget4claim:IndependentMarginalOld}  
\end{align}
Due to conditioning on $\cH(j)$, for any $\tau\in \alphabet^{\partial x_j}$  such that 
$\tau({\tt p}(x_j))\neq \sigma({\tt p}(x_j))$ the above probability is trivially zero. 
We have that
\begin{align}
\lefteqn{
\Pr[\bsigma^*(\partial x_j)=\tau, \ \partial  x_j=\mathbold{z}, \  \bpsi_{j}\in \cW_{j} \ |\   \cB,\ \cH(j) ]  
} \hspace{1.5cm} \nonumber \\
&=
\Pr[ \partial  x_j=\mathbold{z},\    \bsigma^*(\mathbold{z})=\tau,\ \bpsi_{j}\in \cW_{j} \ |\  \cB,\ \cH(j) ] 
\nonumber \\
&=
\Pr[ \partial  x_j=\mathbold{z},\ \bpsi_{j}\in \cW_{j} \ |\  \bsigma^*(\mathbold{z})=\tau,\ \cB,\ \cH(j) ]  
\times  \Pr[\bsigma^*(\mathbold{z})=\tau \ |\  \cB,\  \cH(j) ] 
 \enspace.  \label{eq:Target4claim:IndependentMarginal} 
\end{align}
For the first equality, we use  that the event
$\partial  x_j=\mathbold{z},\    \bsigma^*(\partial x_j)=\tau$
is identical to  $\partial  x_j =\mathbold{z},\    \bsigma^*(\mathbold{z})=\tau$. 
We prove \eqref{eq:FirstTarget4claim:IndependentMarginalOld}  by estimating each one of 
the probability terms in\eqref{eq:Target4claim:IndependentMarginal}.

From the definition of $\G^*_i$, i.e., in the paragraph above \eqref{def:Probs4GStar} 
we have that 
\begin{align}\label{eq:PartialYCondColour}
\Pr[\partial  x_{j}=\mathbold{z}, \ \bpsi_{j}\in \cW_{j}  \ |\  \bsigma^*(\mathbold{z})=\tau, \ \cB,\  \cH(j)]
&= \frac
{\mathbb{E}_{\bpsi_j\sim \dpsi} \left [ \Ind \{\bpsi_{j}\in \cW_{j}\}\cdot \bpsi_{j}(\tau) \right]}
{
{\textstyle \left(1+O\left(n^{-1/3}\right)\right)} \cdot n^{\underline{k}} \cdot \chi 
}
\enspace.
\end{align}
Perhaps it is worth explaining how  we obtain the denominator in the above relation.  
We note that conditioning on $\cB$ and $\cH(j)$, for any $j\in \Phi$, we reveal 
information for at most $O(\ell)=O((\log n)^5)$ variable nodes.  Hence, this conditioning
 does not affect the balanceness  of $\bsigma^*$, i.e.,   we have that  
\begin{align}\label{eq:H(j)StillBalanced}
n^{-1}|(\bsigma^{*})^{-1}(c)|&=|\alphabet|^{-1} \cdot {\textstyle \left(1+O\left(n^{1/3}\right)\right)} & \forall c\in \alphabet
\enspace. 
\end{align}
The  denominator in \eqref{eq:PartialYCondColour} is the weighted sum 
of all $k$-tuples of variable nodes $(y_1,\ldots, y_{k})$ weighted by 
$\mathbb{E}_{\bpsi_j\sim \dpsi} [\bpsi_j(\bsigma^*(y_1),\ldots, \bsigma^*(y_k))]$. 
In this summation we use \eqref{eq:H(j)StillBalanced} 
and the definition of $\chi$ in \eqref{eq:DefOfChi}.

Also, it is standard to obtain that  
\begin{align}\label{eq:OutPartialYOnlyColour}
\Pr[ \bsigma^*(\mathbold{z})=\tau\ |\  \cB,\  \cH(j)] & =
{\textstyle \left(1+o\left(n^{-10}\right)\right)}\cdot |\alphabet|^{-(k-1)}\enspace.
\end{align}
Then, \eqref{eq:FirstTarget4claim:IndependentMarginalOld}   follows 
by plugging  \eqref{eq:PartialYCondColour} and \eqref{eq:OutPartialYOnlyColour} into 
\eqref{eq:Target4claim:IndependentMarginal} and noting that 
$\bethe_{x_j} (\tau\ |\  {\tt p}(x_j), \ \sigma)=\frac{\bpsi_j(\tau)}{|\alphabet|^{(k-1)}\cdot \chi}$, 
e.g. see \eqref{eq:SymmetricWeightAB}.

Working similarly to \eqref{eq:FirstTarget4claim:IndependentMarginalOld},  we obtain that for
a $k$-tuple of variable nodes $\mathbold{z}$ whose intersection with $\pi$ and $\DiaSet$, 
apart from $x_{j-1}$ and $x_{j+1}$,  includes another $t>0$ variable nodes satisfies that 
\begin{align}
\Pr[\bsigma^*(\partial x_j )=\tau ,  \partial  x_j=\mathbold{z}, \ 
\bpsi_{j} \in  \cW_{j} \ |\  \cB,\ \cH(j) ] 
&=O\left(1/n^{k}\right)\enspace. 
 \label{eq:FirstTarget4claim:IndependentMarginalLargeInter}  
\end{align}
The claim follows by plugging \eqref{eq:FirstTarget4claim:IndependentMarginalOld} and \eqref{eq:FirstTarget4claim:IndependentMarginalLargeInter}   into \eqref{eq:Base4claim:IndependentMarginal} and noting that the main contribution
to the summation comes from the $k$-tuples $\mathbold{z}$ that only intersect with 
$\pi$ and $\DiaSet$ at $x_{j+1}, x_{j-1}$. 

Specifically,  it is elementary that  there are 
$\left(1+o\left(n^{-1/3}\right)\right)\cdot k\cdot (k-1)\cdot n^{\underline{k-2}}$ 
tuples of $k$ variable nodes $\mathbold{z}$ that  include   $x_{j+1}, x_{j-1}$  
and no other variable node from $\pi$ and $\DiaSet$. On the other hand,  
there are only $O(n^{k-3})$ many  $k$-tuples 
that  include  $x_{j+1}, x_{j-1}$ and other variable nodes  from $\pi$ and $\DiaSet$. 

All the above concludes the proof of \cref{claim:IndependentMarginal}.
\end{proof}

\subsection{Proof of  \cref{claim:MathbbSBound}}\label{sec:claim:MathbbSBound}
We write $\UpS$ as follows:
\begin{equation}\nonumber %\label{Redef:MathbbSDis}
\UpS= \sum\nolimits_{\sigma\in \alphabet^{\neighpi} }
\prod\nolimits_{j\in \Phi}
 \mathbb{E}_{\bpsi_{j}\sim \dpsi}\left[ \UpQ(\sigma,j) \right] \enspace,
\end{equation}
where 
\begin{align}\nonumber
\UpQ(\sigma, j ) &=
\left(
\Ind\{x_{j} \notin J\} \cdot \dpr_{x_j}(\sigma(\partial x_j))  + \Ind\{x_{j} \in J\} \right)
\cdot \bethe_{x_j} (\sigma(\partial x_{j})\ |\ {\tt p}(x_j), \sigma) \enspace.
\end{align}

\noindent
Also , let 
\begin{equation}\nonumber
\UpS_t(\eta)=    \prod\nolimits_{ j\in \Phi:j > t} 
 \mathbb{E}_{\bpsi_{j}\sim \dpsi}\left[ \UpQ(\eta,j) \right] \enspace. 
\end{equation}
For $t=1,\ldots,  |\Phi|$, let  $N_{t}\subseteq \neighpi$ consist of the
variable nodes in $\partial x_j$ for all $j\in \Phi$ such that $j>t$. 
We have  
\begin{align}\label{eq:BBSVsST}
\UpS= \sum\nolimits_{\sigma\in \alphabet^{\partial x_{2}}}
 \mathbb{E}_{\bpsi_{2}\sim \dpsi}\left[\UpQ(\sigma, 2) \right]\times 
  \sum\nolimits_{\tau\in \alphabet^{N_2}} 
\Ind\{\tau(x_3)=\sigma(x_3)\}\cdot \UpS_2(\tau)\enspace. 
\end{align}
For what follows, we assume first that  $x_{2}\notin J$. 

Since  the disagreements propagate over the path 
$\pi$ using spins in  $\DisSpin=\{c,\ \hat{c}\}$, we  have that
\begin{eqnarray}
\sum\nolimits_{\sigma\in \alphabet^{\partial x_{2}}}
\mathbb{E}_{\bpsi_{2}\sim \dpsi}\left[\UpQ(\sigma, 2) \right] 
&\leq & \sum\nolimits_{\sigma\in \alphabet^{\partial x_{2} }:\sigma(x_3)\in \DisSpin }
 \mathbb{E}_{\bpsi_{2}\sim \dpsi}\left[\UpQ(\sigma, 2) \right] \enspace.
 \end{eqnarray}
From the above inequality and the definition of  $\dpr_{x_2}(\cdot)$ in 
\eqref{eq:BroadCastDisagreementProb},  we have that
\begin{align}
\lefteqn{
\sum\nolimits_{\sigma\in \alphabet^{\partial x_{2}}}
 \mathbb{E}_{\bpsi_{2}\sim \dpsi}\left[\UpQ(\sigma, 2) \right] 
 } \hspace{2cm} \nonumber \\
 &= (1/2) \sum\nolimits_{\sigma\in \alphabet^{\partial x_{2} }:\sigma(x_3)=c}
 \mathbb{E}_{\bpsi_{2}\sim \dpsi}\left[ 
|  \bethe_{x_2 }( \sigma \ |\ x_{1}, c)
-\bethe_{x_2}( \hat{\sigma} \ |\ x_{1}, \ \hat{c}) |  
\right] \nonumber \\
&\qquad +(1/2) \sum\nolimits_{\sigma\in \alphabet^{\partial x_{2}}:\sigma(x_3)=\hat{c}}
\mathbb{E}_{\bpsi_{2}\sim \dpsi}\left[ 
|  \bethe_{x_2}( \sigma \ |\ x_{1}, c )-
\bethe_{x_2}( \hat{\sigma} \ |\ x_{1}, \ \hat{c})  |  \right] 
\enspace, 
\end{align}
where  $\hat{\sigma}\in \alphabet^{\partial x_2}$ is such that $\hat{\sigma}(x_{1})=\hat{c}$ and  
$\hat{\sigma}(z)=\sigma(z)$ for all  $z\in \partial x_{2}\setminus \{x_1\}$.
From the above inequality, it is standard to get that
\begin{eqnarray}
\sum\nolimits_{\sigma\in \alphabet^{\partial x_{2}}}
 \mathbb{E}_{\bpsi_{2}\sim \dpsi}\left[\UpQ(\sigma, 2) \right] 
 &\leq &
 \mathbb{E}_{\bpsi_{2}\sim \dpsi}\left[  
 ||\bethe_{x_{2}}( \cdot \ |\ x_{1}, {c})
 -\bethe_{x_{2}}( \cdot  \ |\ x_{1}, \hat{c}) ||_{\partial x_{2}\setminus \{x_{1} \}}
   \right] =    \drate_{x_2}\enspace, \quad \label{eq:SumExpQVsBetas}
\end{eqnarray}
where $\drate_{x_2}$ is defined in \eqref{eq:RateDisDef}.
Plugging the above into \eqref{eq:BBSVsST}, we get that
\begin{equation} \label{eq:UpSBound3X2InJ}
\UpS \leq \drate_{x_2}   \times  \sum\nolimits_{\tau\in \alphabet^{N_2}}  
\Ind \{\tau(y_3)=c_{\max}\} \cdot \UpS_2(\tau) \enspace,
\end{equation}
where   $c_{\max}=  \arg\max_{r\in \DisSpin}
\left \{ \sum_{\tau\in \alphabet^{N_2}}
\Ind\{\tau(x_3)=r \}\cdot \UpS_2(\tau)  \right \}.$
Similarly, for the case where $x_2\in J$, we get that
\begin{equation}\label{eq:UpSBound3X2InJNot}
\UpS \leq   \sum\nolimits_{\tau\in \alphabet^{N_2}}   
\Ind \{\tau(x_3)=c_{\max}\}\cdot \UpS_2(\tau)\enspace.
\end{equation}
From \eqref{eq:UpSBound3X2InJ} and \eqref{eq:UpSBound3X2InJNot} and a simple induction
we have  $\UpS =  \prod\nolimits_{j\in \Phi:\  y_{j}\notin J} \drate_{x_{j}}$.
Then, \cref{claim:MathbbSBound}  follows by noting that $|\Phi|=\lfloor \ell/2 \rfloor-2-|J|$, 
while, since $\mu$ satisfies $\setB$ with slack $\delta$, we have that  
$\drate_{x_{j}}\leq \frac{1-\delta}{d(k-1)}$ for all $j\in \Phi$. 
\hfill $\Box$

\subsection{Proof of \cref{lemma:CycleContr2DisagreeSetB}}\label{sec:lemma:CycleContr2DisagreeSetB}

Recall that  $\scylcpi$ is the set that contains every  factor node $x\in \pi$ which either belongs 
to a short cycle or  is at distance one from a short cycle.

For integers $\gamma,t\geq 0$, let $\cC_{\gamma,t}$  be the set of short cycles of length $\gamma$ 
in $\G^*_i$ such that  each  one of them has an intersection with $\pi$ which is of  length $t$. 
The set $\cup_{\gamma,t}\cC_{\gamma,t}$  needs to consist of disjoint cycles,  since we condition   
on  $\G^*_i\in \CG$.    

It is direct to verify  that each cycle in $\cC_{\gamma,t}$ intersects with $t$ 
(not necessarily factor) nodes in $\pi$ and, hence, we have that 
$|\scylcpi| \leq \sum_{\gamma,t} (t+2)\cdot |\cC_{\gamma,t}|$.  We add two because 
$\scylcpi$ also includes the factor nodes in $\pi$ that are at distance one from
the short cycle. 

Furthermore, since $\drate<1$,  we have that 
\begin{equation} \label{eq:RedaucSRT2Sets}
\Uplambda \leq  \sum\nolimits_{\{ j_{\gamma,t}\}}
\drate^{-\left (\sum_{\gamma,t}(t+2)\cdot j_{\gamma,t}\right ) } \cdot 
\Pr[\wedge_{\gamma,t}|\cC_{\gamma,t}|=j_{\gamma,t}  \ |\ \cS(\sigma_{\max}),\ \cB,\  \CG]
\enspace.
\end{equation}
Recall that $\cS(\sigma_{\max})$ stands for the event $\UpJ^s_{\pi}=1$ and $\bsigma^*(\neighpi)=\sigma_{\max}$. 
Also, we assume that $0\leq s\leq \ell$, or $s=\infty$.

We now focus on bounding the probability term on the r.h.s. of \eqref{eq:RedaucSRT2Sets}. To this end, consider first the probability term 
$\Pr[\wedge_{\gamma,t}|\cC_{\gamma,t}|=j_{\gamma,t}\ |\   \cS(\sigma_{\max}),\cB ]$, i.e.,  without the condition that $\CG$.   We provide an  upper bound on   this term. 
We still assume that $\cup_{\gamma,t}\cC_{\gamma,t}$ specify disjoint cycles.

On the event $\UpJ^s=1$ and for given $\gamma,t$,  let  $\cI=\{ I_1, \ldots,  I_N \}$
be a  collection of $N$ disjoint  subpaths of $\pi$ such that each $I_k$ is of length $t$.
Also, let $\cK(j)$ be the  event that there  is  a cycle  of  length $\gamma$  that intersects with $\pi$  in the interval ${I}_j$  and this cycle is disjoint  from the cycles that intersect ${ I}_{s}$  for $s\in [N ]\setminus \{j\}$.

For   $j\in [N]$ and $S\subseteq[N]\setminus\{j\}$, 
let $\cF_{j,S}$ be the $\sigma$-algebra generated by the cycle whose intersection 
with $\pi$ is ${ I}_s$, for all $s\in  S$. 
Using arguments very similar to those we used in the proof of \cref{lemma:TreeMarginal} we get that following: 
there is a constant $C>0$ such that for any $j\in [N]$  and any $S\subseteq [N]$, we have that
\begin{equation}\nonumber
\Pr[\cK(j)\  |\ {\mathcal F}_{S,j},\  \cS(\sigma_{\max}), \ \cB ] \leq 
C \cdot n^{-1} \cdot (dk)^{\lceil (\gamma-t-1)/2\rceil} \enspace.
\end{equation}
Furthermore, from the product rule, and the above, we get 
\begin{equation}\label{eq:ProbCycleOnPiBound}
 \Pr[\wedge_{j\in [N]} \cK(j) \  |\ \cS(\sigma_{\max}),\ \cB] \leq 
 \left( C\cdot n^{-1} \cdot  (dk)^{(\gamma-t)/2}\right)^N  \enspace, 
\end{equation} 
where recall that  $N$ is the cardinality of $\cI$.

Letting $\cF_{\gamma, t}$ by the $\sigma$-algebra generated by the sets  
$\cC_{x,z}$, where $x\leq \gamma$ and $z<t$,    \eqref{eq:ProbCycleOnPiBound} implies that
\begin{equation}\nonumber
 \Pr[|{\mathcal C}_{\gamma,t} |= j_{\gamma,t}\ | \ \cF_{\gamma, t}, \ \cS(\sigma_{\max}),  \ \cB]  \leq  {\ell \choose j_{\gamma,t}} \left(C \cdot n^{-1}\cdot (dk)^{(\gamma-t)/2 } \right)^{j_{\gamma,t}} \enspace.
\end{equation}
Note that the above is a crude overestimate. 
In turn, the above implies that 
\begin{equation}\label{eq:CalCCondPath}
\Pr[\wedge_{\gamma,t} | \cC_{\gamma,t}|=j_{\gamma,t} \ |\  \cS(\sigma_{\max}),\ \cB] \leq    \prod\nolimits_{\gamma,t} 
 {\ell \choose j_{\gamma,t}}\cdot \left( C\cdot n^{-1}\cdot (dk)^{(\gamma-t)/2} \right)^{j_{\gamma,t}}  \enspace.
\end{equation}
Furthermore, note that
\begin{eqnarray}
\Pr[\wedge_{\gamma,t} | \cC_{\gamma,t}|=j_{\gamma,t}\ |\ \cS(\sigma_{\max}),\ \cB,\ \CG] &\leq& 
\frac{\Pr[\wedge_{\gamma,t} | \cC_{\gamma,t}|=j_{\gamma,t} 
\ |\ \cS(\sigma_{\max}),\ \cB ]}
{(\Pr[\CG\ |\ \cS(\sigma_{\max}),\ \cB])}
\nonumber\\
&=&  (1+o(1))  \cdot 
\Pr[\wedge_{\gamma,t} | \cC_{\gamma,t}|=j_{\gamma,t}  \ |\ \cS(\sigma_{\max}),\ \cB ]
\enspace. \label{eq:CalCCondPathRightCond}
\end{eqnarray}

Plugging \eqref{eq:CalCCondPath} and \eqref{eq:CalCCondPathRightCond}
into  \eqref{eq:RedaucSRT2Sets}, we get that 
\begin{eqnarray}
\Uplambda 
%%%%
&\leq &(1+o(1))  \sum\nolimits_{ \{ j_{\gamma,t}\}}\prod\nolimits_{\gamma,t} 
{\ell  \choose j_{\gamma,t}} \left( \drate^{-(t+2)} \cdot  C \cdot  n^{-1}
\cdot (dk)^{(\gamma-t)/2} \right)^{j_{\gamma,t}} \nonumber.
\end{eqnarray}
Using the observations that $\drate^{-1}=\frac{dk}{1-\delta}$ and ${\ell \choose j_{\gamma,t}} \leq (\ell)^{j_{\gamma,t}}$ we   obtain 
\begin{eqnarray}
\Uplambda
&\leq &  (1+ o(1))   \sum\nolimits_{\{ j_{\gamma,t}\}}\prod\nolimits_{\gamma,t} 
 \left( \ell \cdot \drate^{-2}\cdot C \cdot  n^{-1} \cdot \left( dk \right)^{(\gamma+t)/2}  
 \cdot (1-\delta/2)^{-t} \right)^{j_{\gamma,t}}   \nonumber \\
&\leq &(1+o(1))  \sum\nolimits_{ \{ j_{\gamma,t}\}}\prod\nolimits_{\gamma,t} 
 \left( n^{-0.7} \right)^{j_{\gamma,t}} \enspace,   \nonumber
\end{eqnarray}
where in the last derivation  we use that  
$(\gamma+t)/2\leq \gamma \leq \frac{\log n}{10\log(dk)}$ and $\ell =O((\log n)^5)$.
Also, we noted that $(1-\delta/2)^{-t}\leq n^{1/6}$, i.e.,
since $dk\geq 2$ and $t\leq \gamma$. Furthermore, since $\sum_{\gamma,t}j_{\gamma,t}=K$, where $K$ is  the number of all   cycles that we consider,   the above simplifies as follows:
\begin{align}\nonumber  
\Uplambda  & \leq   (1+o(1)) \sum\nolimits_{K \geq 0}  
\left( (\log n)^3 \cdot n^{-0.7} \right)^{K}  = (1+o(1)) \enspace.
\end{align}
\cref{lemma:CycleContr2DisagreeSetB} follows. 
 \hfill $\Box$

\spreadpoint
\section{Proofs of results in \cref{sec:Applications}
\LastReview{2024-02-05}}\label{sec:ProofApps}

\subsection{Proofs of \cref{thrm:FerroIsing,thrm:Potts} }\label{sec:AppIsing} \label{sec:thrm:Potts}

Since  the Ising model is a special case of the $q$-state Potts model, i.e., $q=2$, we focus on proving   \cref{thrm:Potts}, then  \cref{thrm:FerroIsing} follows as a corollary.

In light of  \cref{thrm:MainA,thrm:MainB},  we get  \cref{thrm:Potts}  by 
arguing  that the $q$-state antiferromagnetic Potts  model with the parameters indicated in the statement of   \cref{thrm:Potts}   satisfies  the conditions in 
$\setB$ with slack $\updelta_0>0$, where $\updelta_0$ depends
on the choice of the parameters  of the problem.   

Among the conditions in $\setB$, it is immediate that  $\CC$ in $\setB$ is
trivially  satisfied.
Furthermore, in light of   \cref{thrm:ContiquitySeq,thrm:DbcVsDcont},  
if $\CA$ holds with slack $\updelta_0>0$, then  $\CB$ also holds.   
In that respect, we  only need to focus on $\CA$.

For any region of the parameters of the $q$-state Potts, the following is true:
For any edge  $e$ and any $i,j\in [q]$ we can couple
$\bethe^i_e$ and $\bethe^j_e$ maximally on the coordinates in $\Lambda=\{x_2, \ldots, x_k\}$
and get  that
\begin{equation}\label{eq:RePottsThrm2New}  
\drate_{e}=||\bethe^{i}_e-\bethe^{j}_e ||_{\Lambda} \leq  
\frac{1-e^{\beta}}{q^{k-1}-1+e^{\beta}} \enspace.
\end{equation}
The above also holds for the colouring model, i.e., $\beta=-\infty$.
For each one of the cases  we consider in  \cref{thrm:Potts}, we  show that  the
rightmost quantity in the inequality above is upper bounded by 
$\frac{1-\updelta_0}{d(k-1)}$.

We start with Case (1).
Our assumption about $\beta$ implies that  there exists $\upzeta>0$ such that  
\begin{equation}\nonumber 
\textstyle \beta=\BMPotts(d,q, k)+
\log\left( 1+\frac{q^{k-1}}{d(k-1)+1} \left(1-\frac{q^{k-1}}{d(k-1)+1}\right)^{-1} \cdot \upzeta\right) \enspace.
\end{equation}
Plugging the above into \eqref{eq:RePottsThrm2New}, elementary calculations yield 
$ \drate_{e}\leq \frac{1-\upzeta}{d(k-1)+\zeta} \leq \frac{1-\upzeta}{d(k-1)}$. 
Clearly, this implies  that  $\CA$ is satisfied 
with slack $\updelta_0$,  for any $0< \updelta_0 \leq   \upzeta$.

Case (2) corresponds to assuming that $q^{k-1}-1>  d(k-1)$ and $\beta< 0$, including $\beta=-\infty$. 
In this setting   the quantity on the right-hand side of  \eqref{eq:RePottsThrm2New}
is monotonically decreasing in $\beta$.  Hence, it suffices to prove that $\CA$ is satisfied 
with slack $\updelta_0>0$  for $\beta=\infty$.  That is,  we only need to consider the 
colouring model.

Since we assume $q^{k-1}-1> d(k-1)$, there exists $\zeta>0$ such that 
$q^{k-1}-1=(1-\zeta)^{-1}d(k-1)$. Plugging this inequality into 
\eqref{eq:RePottsThrm2New} and setting $\beta=-\infty$, we get that 
$\drate_{e}\leq \frac{1-\zeta}{(k-1)d}$. Hence, the condition $\CA$ is  
satisfied with  slack $0< \updelta_0 \leq  \upzeta$.

Case (3) is identical to Case (2), by setting $\upzeta=e^{\beta}$, and 
hence, $\CA$ is  satisfied with  slack $0< \updelta_0 \leq  e^{\beta}$.

All the above, conclude the proof of  \cref{thrm:Potts}. As far as \cref{thrm:FerroIsing} is concerned, we only need to remark that  the case (1)  corresponds to the case (1) of 
\cref{thrm:Potts} where $q=2$. Similarly, case (2) of  \cref{thrm:FerroIsing} corresponds to the cases (2) and (3) of \cref{thrm:Potts} where $q=2$ and $\beta\neq \infty$.
\hfill $\Box$

\subsection{Proof of \cref{thrm:NAESAT}}\label{sec:thrm:NAESAT}
Similarly to the proof of \cref{thrm:FerroIsing,thrm:Potts}, it suffices 
to show that the uniform distribution over the 
NAE solutions of $\bF_{k}(n,m)$, with the parameters indicated in the statement
of \cref{thrm:NAESAT}, satisfies  $\CA$ with slack $\delta_0>0$.

Consider the clause   $\alpha$ in $\bF_k(n,m)$ and the corresponding distribution $\bethe_{\alpha}$ on this clause.
It is standard that   $\bethe_{\alpha}$ corresponds to  the uniform distribution over the NAE satisfying assignments of the clause $\alpha$. 
Let $\bethe^T_{\alpha}$ denote the distribution $\bethe_{\alpha}$ where we condition on the first literal being true. Similarly, let $\bethe^F_{\alpha}$ denote the distribution  where 
the first literal is false.

The support of $\bethe^T_{\alpha}$ consists of $2^{k-1}-1$ assignments. That is, 
all but the assignment that evaluates all literals in $\alpha$ to true NAE satisfy $\alpha$.
Similarly   for $\bethe^F_{\alpha}$,  its support consists of $2^{k-1}-1$ assignments, i.e.,  excluding the assignment that evaluates all literals in $\alpha$ the value false.  

Recalling that both $\bethe^T_{\alpha}$ and $\bethe^F_{\alpha}$ are uniform distributions 
over the NAE satisfying  assignments of $\alpha$ and  the above observation
   implies that 
\begin{align} \nonumber 
\textstyle \drate_{\alpha}&\leq   (2^{k-1}-1)^{-1} \enspace.
\end{align}
Since we have assumed $d\leq (1-\delta)\frac{2^{k-1}-1}{k-1}$, 
 $\CA$ is satisfied with slack $\delta>0$.

The theorem follows. \hfill $\Box$

\subsection{Proof of \cref{thrm:KSpin}}\label{sec:thrm:KSpin}
It is standard   to verify that for any even integer $k\geq 2$, the $k$-spin model is 
symmetric, for further details  see \cite{CoEfCMP}.  For the range of parameters we consider 
in \cref{thrm:KSpin}, we  show that the $k$-spin model  satisfies  $\setB$ with slack 
$\delta>0$.

It is standard to  verify that the condition $\CC$ in $\setB$ is satisfied. 
Furthermore, arguing as in the proof of \cref{thrm:FerroIsing,thrm:Potts}, 
it only remains to show that   $\CA$ holds for slack $\delta>0$.
i.e., then  $\CB$ also holds.

Consider $\bH=\bH(n,m,k)$ and let the hyperedge $\alpha=(x_1, \ldots, x_k)$  in $\bH$.  
Let $\bethe^+_{\alpha}$ denote the  distribution $\bethe_{\alpha}$ where we condition 
 that the configuration at $x_1$ is $+1$.  Similarly, let $\bethe^-_{\alpha}$ 
denote the distribution where we condition  that the configuration at $x_1$ is $-1$.

Let $\Lambda=\{x_2,\ldots,x_k\}$. Using standard  maximal coupling we get that
\begin{equation}\nonumber 
|| \bethe^-_{\alpha}-\bethe^+_{\alpha} ||_{\Lambda}\leq \frac{|e^{\beta\bJ_{\alpha}}-e^{-\beta\bJ_{\alpha}}|}{e^{-\beta\bJ_{\alpha}}+
e^{\beta\bJ_{\alpha}}}= \GlassInf(\beta\bJ_{\alpha}) \enspace, 
\end{equation}
where $\GlassInf_{k}(x)$ is defined in \eqref{Def:FKSpinGlass}. 
From the definition of  $\drate_{\alpha}$, the above implies that 
$ \drate_{\alpha}\leq \mathbb{E}[\GlassInf(\beta\bJ_{\alpha})]$,
where the expectation is w.r.t. the Gaussian random variable $\bJ_{\alpha}$.

The theorem follows since we assume that 
$\mathbb{E}[\GlassInf(\beta \bJ_a) ]\leq \textstyle \frac{1-\delta}{d(k-1)}$.
\hfill $\Box$

\spreadpoint

\section{Proof of Results from \cref{sec:RegionSet} \LastReview{2024-02-02}}

\subsection{Proof of \cref{thrm:ContiquitySeq}}\label{sec:thrm:ContiquitySeq}

\begin{proof}[Proof of \cref{thrm:ContiquitySeq}]
Consider the random $\Psi$-factor graph   $\G=\G(n,m,k, \dpsi)$  of expected degree $d$, 
while the Gibbs distribution $\mu=\mu_{\G}$ corresponds to one of the  
distributions we consider in  \cref{sec:Applications}.

Let $\G_0, \ldots, \G_m$ be  the sequence of subgraphs of $\G$ obtained in the standard way. 
Note that  for each  $i=1, \ldots,m$   we have that $\G_i$ is an instance of $\G(n,i,k, \dpsi)$.
Let $\mu_i$ be the distribution induced by $\G_i$. 

We  show that for any  $\omega\to\infty$ we have
\begin{equation}\label{eq:Target4thrm:ContiquitySeq}
\Pr[\wedge^m_{i=1}\mathcal{C}_i (\omega)]=1-o(1) \enspace.
\end{equation}
Given \eqref{eq:Target4thrm:ContiquitySeq}, the theorem follows by using, technical, but standard arguments.
Particularly it follows by virtually the same arguments presented in  Section 4.3 in \cite{CoEfCMP} and Section 7.4 in \cite{CoKaMu20}.
Hence, \cref{thrm:ContiquitySeq}  relies on showing that the above is true. 

Using the small-subgraph  conditioning  technique  \cite{Janson,RobinsonWormald}, 
Theorem 2.7 in \cite{CoKaMu20} and Theorem 2.7  in \cite{CoEfCMP}   imply that
for any $d<d_{\rm cond}$ and  any $\omega\to \infty$ with $n$,   we have
\begin{align}\label{eq:ConcentrationGm}
\Pr[\mathcal{C}_i(\omega)] &=1-o(1) &  \textrm{for\ } i=1,\ldots, m 
\enspace. 
\end{align}
Eq.  \eqref{eq:ConcentrationGm}  implies the desired concentration of 
 $Z(\G_i)$
for each $i$ {\em separately}. For \eqref{eq:Target4thrm:ContiquitySeq}
we need to prove it  of all the graphs $\G_0, \ldots, \G_m$, 
{\em simultaneously}.  We use a proof by contradiction to prove \eqref{eq:Target4thrm:ContiquitySeq}. Suppose that
there is $\bar{\omega}$ such that $\bar{\omega}\to\infty$, as $n\to\infty$,
and a constant $c>0$,  bounded away from zero such that
\begin{align}\label{eq:Target4thrm:ContiquitySeqContr}
\Pr\left[{\textstyle \bigcup^m_{i=1}} \bar{\mathcal{C}_i}(\bar{\omega}) \right]
& \geq c \enspace. 
\end{align}

First, we prove the following, useful result. 
\begin{lemma}\label{lemma:ExpZmCondCE}
Under the hypothesis in \eqref{eq:Target4thrm:ContiquitySeqContr}, 
there is a constant $\widehat{K}>0$ such that
$\mathbb{E}\left[Z(\G_m) \ |\  {\textstyle \bigcup^m_{i=1}} \bar{\mathcal{C}_i}(\bar{\omega})  \right] \leq  \widehat{K}\cdot(\bar{\omega})^{-1}\cdot\mathbb{E}[Z(\G_m) ]$.
\end{lemma}

For brevity, let $\cE$ denote the event  ${\textstyle \bigcup^m_{i=1} }\bar{\mathcal{C}_i}(\bar{\omega})$. 
From Markov's inequality we have that 
\begin{equation}\nonumber %\label{eq:MarkovZm}
\Pr\left[Z(\G_m) \leq 2 \cdot \mathbb{E}[Z(\G_m) \ |\  \cE]\ |\ \cE \right] \geq 1/2 \enspace.
\end{equation}
Then, we have that
\begin{align}\nonumber
\frac{1}{2}\leq \Pr[Z(\G_m) \leq 2 \cdot \mathbb{E}[Z(\G_m)\ |\  \cE]\ |\ \cE] &=
\frac{\Pr[Z(\G_m) \leq 2\cdot \mathbb{E}[Z(\G_m) \ |\  \cE],  \cE] }{\Pr[\cE]} \nonumber \\
&\leq 
\frac{\Pr[Z(\G_m) \leq 2 \cdot \mathbb{E}[Z(\G_m) \ |\  \cE]]}{\Pr[\cE]} \enspace. \nonumber
\end{align}
Clearly, the above implies that
\begin{equation}\nonumber
\Pr[\cE] \leq 2\cdot  \Pr\left[Z(\G_m) \leq 2 \cdot \mathbb{E}[Z(\G_m) \ |\ \cE] \right] 
\leq 2\cdot  \Pr\left[Z(\G_m) \leq 2 (\bar{\omega})^{-1}\cdot \mathbb{E}[Z(\G_m) ] \right ]=o(1) \enspace.
\end{equation}
The second inequality uses  \cref{lemma:ExpZmCondCE} and the last one follows from \eqref{eq:ConcentrationGm}. 
The theorem follows by noting that the above contradicts 
the hypothesis in \eqref{eq:Target4thrm:ContiquitySeqContr}, 
hence  \eqref{eq:Target4thrm:ContiquitySeq} is true. 
% 
% The theorem follows. 
% 
\end{proof}
% \hfill $\Box$

\begin{proof}[Proof of \cref{lemma:ExpZmCondCE}]
For brevity, let $\cE$ denote the event  ${\textstyle \bigcup^m_{i=1} }\bar{\mathcal{C}_i}(\bar{\omega})$.

 Let $\cD$ be the set of all distributions on the set of spins $\alphabet$,
while, let $\bar{\rho}\in \cD$ be the uniform one. We let $\cR_n\subseteq \cD$ denote the 
set of all the distributions $\rho\in \cD$ such that $n\rho\in \mathbb{R}^{\alphabet}$ is a 
vector of integers.  Also, for $\epsilon=n^{-1/3}$,  let $\cR_{n}(\epsilon)\subseteq\cR_n$ 
contain every $\rho\in \cR_n$ such that   $||\rho-\bar{\rho}||_2\leq \epsilon$. 

For $\sigma\in \alphabet^V$, let $\rho_{\sigma}\in \cR_n$  be such that 
$\rho_{\sigma}(c)$ is equal to the fraction of variable nodes $x$ such that $\sigma(x)=c$, 
for every   $c\in \alphabet$.

For $\rho\in \cD$, let $Z_{\rho}(i)=Z(\G_i)\cdot \mathbb{E}[{\bf 1}\{\rho_{\bsigma}=\rho\}]$, 
where   $\bsigma$  is distributed as in $\mu_i$.  We have that 
\begin{align}\label{eq:GenFirstMmtA}
\mathbb{E}[Z(\G_i) \ |\ \cE]&=\textstyle \sum_{\rho\in \cR_n} \mathbb{E}[Z_\rho(i)\ |\ \cE] 
&\textrm{for \ } i=1,\ldots, m
\enspace.
\end{align}
On the event $\cE$, let $\ell \in [m]$ be the smallest index such that 
$Z(\G_{\ell}) < (\bar{\omega})^{-1} \cdot \mathbb{E}[Z(\G_{\ell})]$.

For each $\rho\in \cR_n$, let $\upomega_{\rho}>0$ be such that  
$Z_{\rho}(\ell)=(\upomega_{\rho})^{-1} \cdot \mathbb{E}[Z_{\rho}(\ell)]$. 
We have that 
\begin{align}\label{eq:ExpZmCondBasicA}
 \sum\nolimits_{\rho\in \cR_n}  (\upomega_{\rho})^{-1} \cdot \upgamma_{\rho}(\ell) &< (\bar{\omega})^{-1}, 
&\textrm{where} & & 
\upgamma_{\rho}(\ell)=\frac{\mathbb{E}[Z_{\rho}(\ell)]}{\mathbb{E}[Z(\G_{\ell})] }\enspace.
\end{align}

We also obtain the following results.

\begin{claim}\label{claim:lemma:ExpZmCondCEA}
For any $\rho\in \cR_n(\epsilon)$,   we have that 
$\mathbb{E}[Z_{\rho} (m)\ |\ \cE ] {\sim }
(\upomega_{\rho})^{-1} \cdot \mathbb{E}[Z_{\rho} (m)]$. 
\end{claim}

\begin{claim}\label{claim:lemma:ExpZmCondCEB}
 There exists a constant $\uptheta>0$ such that for any 
 $\rho\in  \cR_n(\epsilon)$, we have that
$\frac{\upgamma_{\rho}(m) }{\upgamma_{\rho}(\ell)}\leq \uptheta$. 
\end{claim}

\begin{claim}\label{claim:lemma:ExpZmCondCEC}
%For  $\epsilon=n^{-1/3}$, 
We have that  $\sum_{\rho\in \cR_n\setminus \cR_n(\epsilon)}\mathbb{E}[Z_\rho(m)\ |\ \cE]=\exp(-\Omega(n^{1/3})) \cdot \mathbb{E}[Z(\G_m)]$.
\end{claim}

In light of all the above, we get that
\begin{align}
\mathbb{E}\left[Z(\G_m) \ |\ \cE \right] &\leq 
{ \sum\nolimits_{\rho\in \cR_n(\epsilon)} } \mathbb{E}[Z_{\rho} (m) \ |\ \cE ]+ 
{ \sum\nolimits_{\rho\in \cR_n\setminus \cR_n(\epsilon)} } \mathbb{E}[Z_{\rho} (m) \ |\ \cE ] \nonumber \\
&\leq  \left(  \sum\nolimits_{\rho\in \cR_n(\epsilon)} \upomega_{\rho}^{-1} 
\cdot \mathbb{E}[Z_{\rho} (m) ] 
+ {\textstyle \exp\left(-\Omega\left(n^{1/3}\right)\right)} \cdot \mathbb{E}[Z(\G_m)]   \right)  \enspace,
% & \mbox{[from  \cref{claim:lemma:ExpZmCondCEA,claim:lemma:ExpZmCondCEC}]} \nonumber \\
\end{align}
where in the second derivation we use \cref{claim:lemma:ExpZmCondCEA,claim:lemma:ExpZmCondCEC}.
Using \eqref{eq:ExpZmCondBasicA}, we get 
\begin{align}
\mathbb{E}\left[Z(\G_m) \ |\ \cE \right]
&\leq  \mathbb{E}\left[Z(\G_m) \ |\ \cE \right] \cdot  \left(   \sum\nolimits_{\rho\in \cR_n(\epsilon)} \upomega_{\rho}^{-1}\cdot \upgamma_{\rho}(m) 
+ {\textstyle \exp\left(-n^{1/4}\right)}\right) \nonumber \\
%%%
&\leq \mathbb{E}\left[Z(\G_m) \ |\ \cE \right]  \cdot
\left( \uptheta \cdot  \sum\nolimits_{\rho\in \cR_n(\epsilon)} 
\upomega_{\rho}^{-1}\cdot \gamma_{\rho}(\ell ) 
+ {\textstyle \exp\left(-n^{1/4}\right)}\right) & \mbox{[from  \cref{claim:lemma:ExpZmCondCEB}]} \nonumber \\
%%%%%
&\leq \mathbb{E}\left[Z(\G_m) \ |\ \cE \right] \cdot
\left( \uptheta  \cdot \bar{\omega}^{-1}+ {\textstyle \exp\left(-n^{1/4}\right)}\right) & \mbox{[use \eqref{eq:ExpZmCondBasicA}]}\nonumber \\
&\leq 2\uptheta \cdot \bar{\upomega}^{-1} 
\cdot \mathbb{E}\left[Z(\G_m) \ |\ \cE \right]  \enspace. \nonumber 
\end{align}
The lemma follows by setting $\widehat{K}=2\theta$.
\end{proof}

\subsubsection*{Proof of \cref{claim:lemma:ExpZmCondCEA,claim:lemma:ExpZmCondCEB,claim:lemma:ExpZmCondCEC}:}
We prove all three claims together at this part of the paper. 

Starting with \cref{claim:lemma:ExpZmCondCEA},   let the function $\phi:\mathbb{R}^{\alphabet} \to [0,2)$ be such that $\rho \mapsto {\textstyle  \sum_{\tau\in\alphabet^k} } 
 \mathbb{E}[\bpsi(\tau)] \cdot {\textstyle \prod_{i=1}^k}\rho(\tau_i)$. 
In \cite{CoEfCMP}, Section 7, it is shown that
for any $i\in [m]$ and uniformly for all  $\rho\in \cR_{n}( \epsilon)$, we have 
\begin{align}\label{eq:GenFirstMmtVsRho}
\mathbb{E}[Z_\rho(i)]&\sim\frac{\exp(n f_i(\rho))}{\sqrt{(2\pi n)^{q-1}\prod_{c\in\alphabet}\rho(c)}} &
\mbox{and}& & f_i(\rho)&=\mathcal{H}(\rho)+\frac{i}{n}\ln\phi(\rho) \enspace, 
\end{align}
where $\mathcal{H}$ is the entropy function, i.e., for  $\rho\in \cD$ we have $\mathcal{H}(\rho)=-\sum_{c\in \alphabet}\rho(c)\log \rho(c)$.
Similarly, we get 
\begin{align}\label{eq:CondZFirstMmtVsRho}
\mathbb{E}\left[Z_\rho(m)\ |\ Z_{\rho}(\ell) \right]&\sim Z_{\rho}(\ell) \cdot \exp\left(n \cdot \hat{f}_{m,\ell}(\rho)\right)&
\mbox{and}& &  \hat{f}_{m,\ell}(\rho)&=\frac{m-\ell}{n}\ln\phi(\rho) \enspace.
\end{align}
Then, \eqref{eq:GenFirstMmtVsRho} and \eqref{eq:CondZFirstMmtVsRho},   imply that for any $\rho\in \cR_{n}(\epsilon)$ we have
\begin{align}
\mathbb{E}[Z_{\rho} (m)\ |\ \cE ] &\sim \upomega_{\rho}^{-1} \cdot \mathbb{E}[Z_{\rho} (m)] \enspace.
\end{align}
The above proves  \cref{claim:lemma:ExpZmCondCEA}.
We continue to prove \cref{claim:lemma:ExpZmCondCEB}.

For every $\rho\in \cR_n(\epsilon)$, consider the expansion of   $f_{i}(\rho)$ around $\bar{\rho}$. 
In  \cite{CoEfCMP} Section 7, it is proved that 
\begin{equation}\label{eq:FiRho}
f_i(\rho)=f_{i}(\bar{\rho})-
\frac{q}{2}(\rho-\bar{\rho})^T\cdot 
(\UpI-k(k-1)\frac{i}{n}\Upphi)
\cdot (\rho-\bar{\rho})+O(\epsilon^3) \enspace,
\end{equation}
where $\UpI, \Upphi\in \mathbb{R}^{\alphabet\times \alphabet}$. Particularly, $\UpI$ is the
identity matrix, while $\Upphi$ is a stochastic matrix which only depends on the set of weight functions $\Psi$. Furthermore,  for any $x\in \mathbb{R}^{\alphabet}$ such that $x\perp {\bf 1}=0$, we have  $x^T\Upphi x\leq 0$. 
For further details about the derivation of \eqref{eq:FiRho}, see \cite{CoEfCMP}.

Furthermore,  Proposition 7.1 and  Lemma 7.3 in   \cite{CoEfCMP} imply that for any $i\in [m]$, there exist fixed numbers $\lambda_1, \lambda_2, \ldots, \lambda_{q-1}$, 
where $q=|\alphabet|$ such that  
\begin{align}
\sum_{\rho\in \cR_n(\epsilon)}\mathbb{E}[Z_\rho(i)]&=
\frac{q^{n+\frac{1}{2}}\chi^i}{\prod_{j}\sqrt{1-k(k-1)
(\frac{i}{n})\lambda_j}},  &%   
\sum_{\rho\in \cR_n\setminus \cR_n(\epsilon)}\mathbb{E}[Z_\rho(i)]&=\exp(-\Omega(n^{1/3}))\sum_{\rho\in \cR_n(\epsilon)}\mathbb{E}[Z_\rho(i)] \enspace,
\label{eq:TinyContribution2Z}
\end{align}
where the quantity $\chi$ is defined in \eqref{eq:DefOfChi}. 
Combining the  above with  \eqref{eq:GenFirstMmtA} we get that 
\begin{align}\label{eq:FirstMomentClosedExpression}
\mathbb{E}[Z(\G_i)] &\sim  \frac{q^{n+\frac{1}{2}}\chi^i}{\prod_{j}\sqrt{1-k(k-1)(i/n)\lambda_j}} & \mbox{for all $i\in [m]$}\enspace.
\end{align}

From the definition of $\upgamma_{\rho}(i)$,   \eqref{eq:GenFirstMmtVsRho} and \eqref{eq:FirstMomentClosedExpression}, 
for any $\rho\in \cR(\epsilon)$,  we  have that 
\begin{equation}\label{eq:RatioVsTheta}
\frac{\upgamma_{\rho}(m)}{\upgamma_{\rho}(\ell)}\sim
\prod_{j}\sqrt{\frac{1-k(k-1)(\ell/n)\lambda_j}{1-k(k-1)(m/n)\lambda_j}}
\cdot 
\exp\left( \frac{q}{2}k(k-1)\frac{m-\ell}{n}(\rho-\bar{\rho})^T\Phi(\rho-\bar{\rho}) \right) \enspace.
\end{equation}
The above follows from elementary calculations. Note that we need to use that $\chi=\phi(\bar{\rho})$.

Recall  that $x^T\Phi x\leq 0$ for all vectors $x$ such that, $x\perp {\bf 1}=0$.
For any $\rho\in \cR_n$, we have that $\rho-\bar{\rho}\perp {\bf 1}$, since 
both  $\rho, \bar{\rho}$ are  distributions. 
Hence, we conclude that  the exponential in \eqref{eq:RatioVsTheta}  is 
at most one for all $\rho$.  \cref{claim:lemma:ExpZmCondCEB} follows by recalling that 
all $\lambda_j$'s are independent of  $n$.

As far as \cref{claim:lemma:ExpZmCondCEC} is concerned, note that
\begin{align}
\sum\nolimits_{\rho\in \cR_n\setminus \cR_n(\epsilon)}\mathbb{E}[Z_\rho(i)\ |\ \cE]& \leq \left(\Pr[\cE] \right)^{-1} \cdot 
\sum\nolimits_{\rho\in \cR_n\setminus \cR_n(\epsilon)}\mathbb{E}[Z_\rho(i) ] & \mbox{[since $Z_\rho(i)\geq 0$]}\nonumber\\
&\leq c^{-1} \cdot  \sum\nolimits_{\rho\in \cR_n\setminus \cR_n(\epsilon)}\mathbb{E}[Z_\rho(i) ] &
\mbox{[from assumption \eqref{eq:Target4thrm:ContiquitySeqContr}]} \nonumber \\
&\leq   \exp(- \Omega(n^{-1/3}))\cdot  \sum\nolimits_{\rho\in   \cR_n(\epsilon)}\mathbb{E}[Z_\rho(i) ] &
\mbox{[use \eqref{eq:TinyContribution2Z}]} \enspace.  \nonumber
\end{align}
The above proves \cref{claim:lemma:ExpZmCondCEC} since $\sum_{\rho\in   \cR_n(\epsilon)}\mathbb{E}[Z_\rho(i) ] \leq \mathbb{E}[Z_\rho(i) ]$.

\newcommand{\bT}{\mathbold{T}}
\subsection{Proof of \cref{thrm:DbcVsDcont}}\label{sec:thrm:DbcVsDcont}
 \cref{thrm:DbcVsDcont} is a direct corollary from Theorem 2.8 in \cite{CoEfCMP}.

To be more specific, consider the random $\Psi$ factor tree $\bT=\bT(d,k,\dpsi)$ 
which is rooted  at the variable node $r$. $\bT$ can be defined inductively. 
Each variable node $v$ at level $2h\geq 0$, independently,   has  Poisson with 
parameter $d$  descendants which are factor nodes.  These factor nodes are at 
level $2h+1$ of the  tree. 
Each one of these factor nodes has $k-1$ descendants at level $2h+2$,
which are variable nodes.  
The weight functions at $\bT$ are chosen in the standard way we describe in 
\cref{sec:FactorGrapsGibbs} using $\dpsi$. 

Furthermore, assume that the Gibbs distribution $\mu_{\bT}$ that is induced by  $\bT$ 
is symmetric and let 
\begin{align}\label{eq:CorrDef}
\mathrm{corr}^{\star}(d) &=  \lim\nolimits_{h\to\infty} \mathbb{E}\left[  \max\nolimits_{\sigma, \tau} ||\mu_{\bT}(\cdot \ |\ r, \sigma)-
\mu_{\bT}(\cdot\ |\ r, \tau) ||_{\{S_{2h}\}} \right] \enspace.
\end{align}
We have non-reconstruction for  $\mu_{\bT}$ when $\mathrm{corr}^{\star}(d)=0$. Otherwise,
we have reconstruction.  Furthermore, we let the {\em tree reconstruction threshold} 
be defined as	$d_{\rm recon}^\star=\inf\{d>0:\mathrm{corr}^\star(d)>0\}$.

Whether a Gibbs distribution satisfies $\CA$ does not depend on the underlying graph,
but only on the specifications of this distribution. Hence, for any $d<d_{\rm BC}=d_{\rm BC}(\delta, k, \dpsi)$
the Gibbs distribution on the {\em tree} satisfies  $\CA$.

Furthermore, using a simple coupling argument, 
one can show that   $\CA$ implies  non-reconstruction for $\mu_{\bT}$, i.e.,  
$\CA$ implies that   $\mathrm{corr}^{\star}(d)=0$.  
Hence, we have that
\begin{align}
d_{\rm BC} \leq d_{\rm recon}^\star \enspace. 
\end{align}
On the other hand, Theorem 2.8 in \cite{CoEfCMP} implies the non-trivial relation that 
 \begin{align}\label{eq:thrm:DbcVsDcontA}
 d^{\star}_{\rm recon}\leq d_{\rm cond}\enspace. 
\end{align}
Let us remark here that the above is established by connecting the reconstruction/non-reconstruction transition of the Gibbs distribution on $\bT(d,k,\dpsi)$  with that of 
$\G(n,m,k,\dpsi)$ with  expected degree $d$. 

\cref{thrm:DbcVsDcont} follows from the two inequalities above. 
\hfill $\Box$

\vspace{.3cm}
\noindent
{\bf Acknowledgment:} The author would like to thank Amin Coja-Oghlan for the fruitful discussions.

\spreadpoint

\appendix
\section{Proof of some standard results}

\subsection{Cycle structure of the random hypergraph}\label{sec:lemma:CycleStructureGNMK}

\begin{proof}[Proof of \cref{lemma:CycleStructureGNMK}]
%{ 
For  brevity, let $\ell_0=\ShortDist$. If there are two cycles of length  at most $\ell_0$ each,   in $\G$ that 
intersect,  then there are sets $B$ and $\Phi$ of variable and factor nodes, respectively, such that 
the following holds:  letting $|B|=r_1$ and $|\Phi|=r_2$, we have $|r_1-r_2| \leq 1$, while the number 
of edges that these sets  span is  $r_1+r_2+1$.   Furthermore,  we have that $r_1+r_2 \leq 2\ell_0$.

Let $\UpD$ be the event that $\G$ contains sets like $B$ and $\Phi$ we describe, above. 
Since $|B|, |\Phi|=O(\log n)$, it is elementary to verify  that  each edge between 
a variable node in $B$ and a factor node in $\Phi$ appears with probability at most
$\left(1+n^{-1/2}\right)\frac{k}{n}$, regardless of the other edges between the two sets. 

Setting  $r=r_1+r_2$, we have that
\begin{align}
\Pr[D]& \leq 
\sum^{2\ell_0}_{r=4} \sum_{r_1: |2r_1-r|\leq 1}
{n \choose r_1}{m\choose r-r_1}{r_1(r-r_1) \choose r+1}\left(\left(1+n^{-1/2}\right)\frac{k}{n} \right)^{r+1}
\nonumber\\
& \leq 2\sum^{2\ell_0}_{r=4} \sum_{r_1: |2r_1-r|\leq 1} 
\left( \frac{ne}{r_1}\right)^{r_1} \left( \frac{dn}{k}\frac{e}{r-r_1} \right)^{r-r_1} 
\left( \frac{r_1(r-r_1)e}{r+1} \right)^{r+1} \left(\frac{k}{n} \right)^{r+1} \nonumber \\
& \leq 2\frac{ek}{n}\sum^{2\ell_0}_{r=4} \sum_{r_1: |2r_1-r|\leq 1}
e^{2r}  d^{r-r_1} k^{r_1}  r_1^{r-r_1+1} \left( r-r_1  \right)^{r_1+1} 
\left( r+1 \right)^{-(r+1)} \enspace, \nonumber \label{eq:FirstStepDenseSubgraph}
\end{align}
where for the second derivation we use the standard inequality ${N \choose t}\leq (Ne/t)^t$
and that $m=dn/k$. 
Furthermore, noting that our assumption about $r$ and $r_1$ implies that 
$\frac{r-1}{2} \leq r_1\leq \frac{r+1}{2}$, we  have that 
\begin{align}
\Pr[D] 
& \leq  \frac{e\sqrt{dk^3}}{n}\sum\nolimits_{r} \sum\nolimits_{r_1}
 \left(\frac{dke^4}{4}\right)^{\frac{r}{2}} \left(\frac{r+1}{2}\right)^{2}  
\  \leq \   \frac{2e\sqrt{dk^3}}{n}\sum\nolimits_{r}  
 \left(\frac{dke^4}{4}\right)^{\frac{r}{2}} \left(\frac{r+1}{2}\right)^{2} \nonumber\\
&\leq \frac{8 \ell_0^2 e\sqrt{dk^3}}{n} \sum^{2\ell_0}_{r=4} \left(dke^4\right)^{\frac{r}{2}} \nonumber \\
&\leq \frac{C(\log n)^2}{n} \left(dke^4\right)^{\ell_0} \leq n^{-2/3} \enspace, \nonumber
\end{align}
in the one prior to last inequality, we choose $C=C(d,k)>0$ to be a sufficiently large constant, while
we use that $\ell_0=\Theta(\log n)$.  The lemma follows.
\end{proof}

\subsection{Dynamic Programming for Sampling}\label{sec:SupplemantaryA}

 For \cref{claim:Decimation}, consider a factor tree $T$ and let $\mu=\mu_T$ 
be the Gibbs distribution that is induced by $T$.

\begin{claim}\label{claim:Decimation} 
For any $\Lambda \in V(T)$ and any $\eta \in \alphabet^{\Lambda}$,  Dynamic Programming samples  from   
$\mu_{T}(\cdot \ |\ \{\Lambda, \eta\})$  in  $O(|\alphabet|^{k} \cdot |V(T)|)$ steps. 
\end{claim}

\begin{proof}
% %

For a variable node $z\notin \Lambda$,   we obtain the Gibbs marginal   
$\mu_{z}$ by using the following recursive relation: 
for any $c\in  \alphabet$ we have
\begin{equation}\label{eq:DPRecursion} 
\mu_{z}(c\ |\ \Lambda, \eta)  \propto \prod\nolimits_{\beta \in \partial z} 
\sum\nolimits_{\sigma\in \alphabet^{\partial \beta}} 
\Ind\{\sigma(z)=c\}\cdot\psi_{\beta}(\sigma)\cdot
\prod\nolimits_{x\in \partial \beta\setminus\{z\}} \mu_{T_x, x}( \sigma(x)\ |\  
{ \Lambda, \eta}) \enspace,
\end{equation}  
where $T_x$ is the subtree of $T$ that contains the variable node $x$ and its descendants.  Note that $\mu_{T_x}$ is the Gibbs distribution that is induced by
the subtree $T_x$, while $\mu_{T_x, x}$ is the marginal of this distribution at $x$.

A simple induction suffices to verify that  the running time for computing the Gibbs
marginal  $\mu_z(\cdot \ |\ \Lambda, \eta)$ is  $O(|\alphabet|^{k}\cdot  |V(T)|)$, 
where $V(T)$ is the number of variable nodes in $T$.

We use \eqref{eq:DPRecursion}
to generate a sample $\bsigma$ that is distributed as in $\mu_T(\cdot\ |\ \Lambda,\eta)$.

Suppose that the variable node $r$ is the root of the tree $T$.  We run the above recursion to calculate $\mu_r$, i.e., the marginal at the root. 
Note that, for each variable node $x$,  this recursion calculates the marginal of $\mu_{T_x}$ at $x$. 
We  store all these marginals when we calculate $\mu_r$. 
Then, we obtain the configuration $\bsigma$ by working as we describe below. 

We compute $\bsigma$ inductively, starting from the root $r$.
The basis corresponds to computing $\bsigma(r)$. For this, we only need to  sample from 
$\mu_r(\cdot \ |\ \Lambda,\eta)$ to get $\bsigma(r)$.  This  requires $O(1)$ steps. 

Now consider the variable node $z$ whose configuration $\bsigma(z)$ needs to be 
computed. Let  the factor node $p_z$ be the parent of $z$, while  
let the variable node $g_z$ be the grandparent . 
Assume that we have already computed $\bsigma(g_z)$. Furthermore, there might be also 
some variable nodes in $\partial p_z\setminus \{z\}$ for which, we already know their
configuration under $\bsigma$. W.l.o.g. assume that it is only $g_z$ among the nodes in
$\partial p_z$ that has its configuration specified. 

The  Gibbs marginal we need to sample the  configuration  $\bsigma(z)$ can be calculated 
by using  \eqref{eq:DPRecursion},  with the minor difference that we also need to add the 
configuration at  $g_z$ at the boundary condition.  Note that we have already stored the
marginals for the variable nodes at distance $2$ from $z$. Hence, we need 
$O(|\alphabet|^k)$ steps to calculate the desired marginal, while
once we obtain it, we need   $O(1)$ steps  to  sample  $\bsigma(z)$.

From all the above, it is immediate that, given the marginals, we obtain $\bsigma$ in $O(|\alphabet|^{k}\cdot  |V(T)|)$ steps.  

The claim follows. 
\end{proof}

\end{document}